\newcommand{\pointsize}{11pt}

\documentclass[oneside, \pointsize]{amsbook}

\usepackage[
   includehead,
   includefoot,
     left = 1.5in,
      top = 0.5in,
    right = 1in,
   bottom = 1in
]{geometry}
\usepackage{fancyhdr}
\usepackage{setspace}
\usepackage{calc}
\usepackage{amsthm}
\usepackage{pstricks}
\usepackage{graphicx}
\usepackage{pst-node}

\fancyheadoffset[R]{0.5in}

\fancypagestyle{prelim}{%
   \renewcommand{\headrulewidth}{0pt}
   \fancyhf{}
   \pagenumbering{roman}
   \cfoot{-\thepage-}
}

\fancypagestyle{maintext}{%
   \renewcommand{\headrulewidth}{0.4pt}
   \pagenumbering{arabic}
   \fancyhf{}
   \fancyhead[L]{\rightmark}
   \rhead{\thepage}
}

\numberwithin{figure}{chapter}
\numberwithin{table}{chapter}
\numberwithin{equation}{chapter}
\numberwithin{section}{chapter}

\newcommand{\AutoAbs}[1]{\left\vert #1 \right\vert}
\newcommand{\AutoNorm}[1]{\left\Vert #1 \right\Vert}
\newtheorem{thm}{Theorem}[section]
\newtheorem{lem}{Lemma}[section]
\newtheorem{cor}{Corollary}[section]
\newtheorem{prop}{Proposition}[section]
\begin{document}
   \frontmatter

   \pagestyle{prelim}

   % Redefine plain page style so that the first pages of chapters
   % have desired page style.
   %
   \fancypagestyle{plain}{%
      \fancyhf{}
      \cfoot{-\thepage-}
   }%
\begin{center}
   \null\vfill
   \textbf{%
      The Numerical Simulation of General Relativistic Shock Waves\\
      by a\\
      Locally Inertial Godunov Method\\
      Featuring\\
      Dynamical Time Dilation
   }%
   \\
   \bigskip
   By \\
   \bigskip
   Zeke K. Vogler \\
   \bigskip
   B.S. (University of California, Davis) 2001 \\
   \bigskip
   DISSERTATION \\
   \bigskip
   Submitted in partial satisfaction of the requirements for the
   degree of \\
   \bigskip
   DOCTOR OF PHILOSOPHY \\
   \bigskip
   in \\
   \bigskip
   Applied Mathematics \\
   \bigskip
   in the \\
   \bigskip
   OFFICE OF GRADUATE STUDIES \\
   \bigskip
   of the \\
   \bigskip
   UNIVERSITY OF CALIFORNIA \\
   \bigskip
   DAVIS \\
   \bigskip
   Approved: \\
   \bigskip
   \bigskip
   \makebox[3in]{\hrulefill} \\
   Blake Temple \\
   \bigskip
   \bigskip
   \makebox[3in]{\hrulefill} \\
   John Hunter \\
   \bigskip
   \bigskip
   \makebox[3in]{\hrulefill} \\
   Angela Cheer \\
   \bigskip
   Committee in Charge \\
   \bigskip
   2010 \\
   \vfill
\end{center}

   \newpage

   % Begin Double Spacing
   %
   \doublespacing

   \tableofcontents
   \newpage

{\singlespacing
   \begin{flushright}
      Zeke K. Vogler \\
      March 2010 \\
      Applied Mathematics \\
   \end{flushright}
}

\bigskip

\begin{center}
    The Numerical Simulation of General Relativistic Shock Waves by a\\
    Locally Inertial Godunov Method Featuring Dynamical Time Dilation
\end{center}

\section*{Abstract}
We introduce what we call a locally inertial Godunov method with dynamical time dilation, and use it to simulate a new one parameter family of general relativistic shock wave solutions of the Einstein equations for a perfect fluid.  The forward time solutions resolve the secondary reflected wave (an incoming shock wave) in the Smoller-Temple shock wave model for an explosion into a static singular isothermal sphere.  The backward time solutions indicate black hole formation from a smooth underlying solution via collapse associated with an incoming rarefaction wave.  As far as we know, this is the first numerical simulation of a fluid dynamical shock wave in general relativity.
   \newpage

   \section*{Acknowledgments}
I would like to thank all the great teachers I have learned from throughout my academic career.  They made many sacrifices for their students and rarely get the credit and respect they deserve.  Without them, I would have fallen off this path of enlightenment and never written these words.  A special thanks goes to my advisor Blake Temple.  His insight, guidance, and enthusiasm was a real inspiration, and I feel very privileged to have worked with such a great man.  I am also grateful to all my committee members for taking the time to look over my work.  I want to send a very special thanks to the late Evelyn Silvia, whose passion for teaching was a big influence on my teaching endeavors.

I want to send a warm thanks to the math department staff, especially Celia Davis and Perry Gee, whose efforts on making the department seem like a family did not go unnoticed.

I wish to give huge thank you to my family, whose impact on my life over the years I am just beginning to understand and appreciate.  I very special thanks goes to my mom, Kady Vogler.  Her death was the most significant moment of my life, and it is a shame she missed one of my greatest accomplishments.  May her soul rest in peace.

Finally,  I would like to thank to all my friends that have shown me love and support throughout the years.  Their presence keep me sane throughout the challenging world of mathematics.  Though there are many graduate students to thank, a special one goes out to Spiros Michalakis and Ernest Woei; many great ideas and conversations came out of those weekly lunches.  My greatest appreciation goes to the Riess household and the circle of friends that evolved from there.  They are a second family to me, and they keep me grounded with simple pleasures.  I want to pay homage to the game of basketball for giving me an avenue to exercise my muscles besides the brain.  I want to thank all the ballers out there that influenced my game; this list includes, but not limited to, Duane Kouba, the math basketball team, and the ARC legends.
   \mainmatter

   \pagestyle{maintext}

   % Redefine plain page style so that the first pages of
   % chapters have desired page style.
   %
   \fancypagestyle{plain}{%
      \renewcommand{\headrulewidth}{0pt}
      \fancyhf{}
      \rhead{\thepage}
   }%
   \chapter{Introduction}
   \label{ch:intro}
This thesis presents a numerical study of a canonical family of spherically symmetric solutions to the Einstein equations of General Relativity, such that shock wave formation is simulated in forward time and black hole formation is indicated in backward time.  To start the simulations, we introduce a new one parameter family of initial data satisfying the constraints of the Einstein equations, and to perform the simulations, we develop a locally inertial Godunov method incorporating a new dynamical time dilation feature to account for the relativistic effects of curvature.

Our point of departure is the work of Groah and Temple \cite{groate}, who gave the first existence theory for shock wave solutions of the Einstein equations, starting from general initial data of bounded variation in the density and velocity.  The setting is restricted to spherically symmetric spacetimes in standard Schwarzschild coordinates.  For the analysis, Groah and Temple introduced the {\it locally inertial Glimm scheme}, a numerical method based on approximating spacetime by flat Minkowski space in each grid cell, such that the method was amenable to detailed mathematical estimates sufficient to prove convergence to weak (shock wave) solutions of the Einstein equations.   The curvature of spacetime in the locally inertial Glimm scheme is accounted for by coordinate transformations between the locally inertial approximations in each grid cell.  This technique is pedagogically interesting because it parallels the so called {\it correspondence principle}, the physical principle that general relativity should reduce to special relativity in sufficiently small neighborhoods of spacetime.  Groah and Temple concludes in \cite{groate} that the curvature of spacetime necessarily becomes discontinuous at shock waves, resulting in solutions solving the Einstein equations in the weak sense of the theory of distributions.  In this thesis, we develop these locally inertial formulations of the Einstein equations into a viable numerical method, what we refer to as a {\it locally inertial Godunov method with dynamical time dilation}, and we introduce a new one parameter family of interactive solutions on which we demonstrate the convergence of the method.

Our numerical demonstrations are backed up by a general theorem which we prove.  Namely, we prove if a sequence of approximate solutions, generated by our locally inertial Godunov method converges pointwise almost everywhere to a function along with the total variation of the fluid variables remain uniformly bounded under the limit, (exactly what we can demonstrate numerically), then the limit solution is an exact weak solution to the Einstein equations.  By this general theorem, we need only demonstrate numerical convergence to a limit, with bounded oscillations, in order to conclude our simulated solutions accurately represent exact (weak) solutions of the Einstein equations.

The one parameter family of spacetimes on which we test the convergence of the method are interesting in their own right.  This one parameter family is constructed by matching initial data from a critical ($k=0$) Friedmann-Robertson-Walker (FRW) spacetime to initial data for a Tolmann-Oppengeimer-Volkoff (TOV) spacetime at a point where the spacetime is only Lipschitz continuous, the smoothness required for our convergence theorem.  Since our numerical method is based on standard Schwarzschild coordinates, this requires a coordinate transformation from FRW coordinates to standard Schwarzschild coordinates. Such a transformation requires an integrating factor, and we use the integrating factor derived by Smoller and Temple in \cite{smolte}, (we call the resulting standard Schwarzschild coordinates version of FRW the FRW-1 spacetime).  We also derive a new integrating factor that leads to a different standard Schwarzschild coordinates version of FRW, (which we call FRW-2).  The initial discontinuity that separates FRW from TOV in the standard Schwarzschild coordinates initial data then generates a region of interaction between the two spacetimes, a region for which there is no closed form solution, and can only be constructed by numerical simulation.  The ultimate goal of this thesis is then to simulate this region of interaction, and demonstrate convergence of the locally inertial Godunov method.

The numerical simulation of the region of interaction between FRW and TOV is a problem motivated by the original work of Smoller and Temple in \cite{smolte}, further developed in \cite{smolte3,smolte4}.  In \cite{smolte}, the first exact shock wave solution of the Einstein equations is constructed by matching an FRW spacetime to a TOV spacetime across a shock interface, such that the resulting matched solution, given by exact formulas, is a true weak solution of the Einstein equations.  The resulting matched spacetime (we denote as the FRW/TOV metric) with a shock wave was a simple model for an explosion into a {\it static singular isothermal sphere}, relevant to models of star formation \cite{smolte}.  Based on Smoller and Temple's work, we built a computer visualization \cite{smoltevo} to demonstrate the qualitative behavior of their model.  The Smoller and Temple (Sm/Te) construction requires taking a different sound speed on the FRW and TOV sides of the shock wave finely tuned to reduce the region of interaction between FRW and TOV to a single pure shock wave.  As a result, they only required the existence of, not an explicit formula for, the integrating factor.  In our simulation, we take the {\em same} equation of state on both sides, a more realistic assumption, and as a consequence, a region of interaction, for which there are no exact formulas, between the lightlike curves emanating from the initial point of connection between FRW and TOV metrics is generated.  Based on this background, we interpret our numerical FRW/TOV model simulation in forward time as resolving the secondary reflected wave in the Sm/Te model of an explosion into a static singular isothermal sphere.  In particular, our simulation confirms for the first time the reflected wave is a shock wave, not a rarefaction wave.  We also simulated the time reversal of the FRW/TOV solution.  In this case our simulation indicates the formation of a black hole in the resulting solution, which consists entirely of rarefaction waves.  Based on our simulation, one can only expect the observed strong rarefaction wave evolving toward the coordinate center will form a black hole.  Note that a black hole cannot be computed exactly in standard Schwarzschild coordinates because black holes generate coordinate singularities in this coordinate system.  A proposal for future work would be to simulate the black hole all the way by transforming the solution to Eddington-Finkelstien or Kruskal coordinates \cite{wein}.

We start our discussion by briefly introducing the Einstein equations in Chapter \ref{ch:prelim}.  In this introduction, we formulate these equations in the simplest setting for shock waves in General Relativity; this setting is the spherically symmetric Einstein equations with the stress energy tensor for a perfect fluid and an equation of state $p=\sigma\rho$, $\rho=const$.  In this setting, the Einstein equations reduce to a conservation law with source terms coupled with a system of ODEs.

Chapter \ref{ch:family_of_shock_waves} sets out the background information for the FRW, the TOV, and the FRW/TOV matched models.  In particular, we construct a one parameter family of initial data that agrees with the FRW and TOV metrics in standard Schwarzschild coordinates on either side of a single point where they match continuously.  The continuity of the metric at the matching point implies the initial data meets the constraints of the Einstein equations introduced by Groah and Temple in \cite{groasmte}.  More specifically, the initial data is built by connecting a FRW metric with a TOV metric across an initial discontinuity in the fluid variables, such that the metric components are matched Lipschitz continuously.  This one parameter family of initial data enables us to construct the one parameter family of interacting spacetimes that solve the Einstein equations.  To be simulated by our locally inertial Godunov scheme, the metric must be in standard Schwarzschild coordinates, and while the TOV is already in this coordinate system, the FRW metric needs to be mapped over.  We demonstrate two mappings of the FRW metric over to standard Schwarzschild coordinates, denoted as FRW-1 and FRW-2.  With these two coordinate transformations, the one parameter family of shock wave solutions, the FRW/TOV metric, has two different forms based on the coordinate system used.  These two forms are used as a further test of the correctness of our numerical simulations, by verifying the simulations produce the same spacetime in two different coordinate systems.

In Chapter \ref{ch:frac_god_method}, we introduce the locally inertial Godunov method.  This chapter begins with a discussion of the solution to the Riemann problem \cite{smol} for relativistic compressible Euler equations in Minkowski spacetime.  The Riemann problem is the fundamental building block of the Godunov method, \cite{leve}.  We outline the solution to this Riemann problem as demonstrated by Smoller and Temple in \cite{smolte6}.  In order to incorporate the Riemann problem in the Godunov method for simulating non-flat spacetimes, we augment the solution to the Riemann problem with the time dilation feature.  To this end, we determine the effects of time dilation on the Godunov step and incorporate it into our locally inertial Godunov method.  The chapter ends with a layout of the locally inertial Godunov method and all the steps necessary to implement it.  This method is a fractional Godunov scheme that defines the solution inductively and contains four major steps: a Riemann problem step, a Godunov step (with time dilation), an ODE step, and an update step.

In Chapter \ref{ch:converge_of_method}, we prove the main theorem that establishes the consistency of our locally inertial Godunov method with time dilation.  It states, assuming the convergence of our approximate solution constructed using our method with a total variation bound at each time step, the convergent solution is indeed a weak solution to the Einstein equations.  The two main difficulties of this proof are handling the jump in the fluid variables across each time step and showing the jump in the metric across each space step is accounted for by the term added to the ODE step.  The assumptions to this theorem are obtained numerically by the simulation results.  Thus, our theorem is perfectly tailored to our numerical simulation: we numerically establish convergence and a total variation bound, and this main theorem proves that assuming these conditions, we can conclude convergence to a weak solution of the Einstein equations.

In Chapter \ref{ch:cont_models}, we begin the numerical study of the locally inertial Godunov method with dynamic time dilation by simulating the pure non-interacting FRW and TOV spacetimes individually.  Each of these continuous models provide us with a different scenario to test and fine tune our method.  These models are also important for correctly handling the non-interaction regions and the ghost cells for the subsequent shock wave model.  Within this chapter are graphs of the solutions and tables of the convergence results, along with the explanations behind them.

In Chapter \ref{ch:shock_wave_models}, the first glimpse of shock wave solutions to the Einstein equations is shown by simulating the FRW-1/TOV model.  We demonstrate this one parameter family converges and produces quantitatively different solutions as we vary the free parameter.  Regardless of the parameter chosen,  the solution always contains two shock waves, an incoming and outgoing wave, enclosing a region of higher density, denoted as the interaction region.  Over and above the convergence and graphical results, we show the interaction region converges to the cone of sound (the extent of the sound like information emanating from the initial discontinuity), and we also show outside this interaction region the FRW-1 and TOV metrics are preserved under the simulation.  Changing coordinate systems, the FRW-2/TOV model simulation is similar to the FRW-1/TOV model.  We discover and numerically confirm the FRW-2/TOV model is the same solution as the FRW-1/TOV model that differs by a non-linear coordinate transformation, and this second model provides a pedagogically interesting numerical confirmation of the covariance of the Einstein equations in standard Schwarzschild coordinates.

In Chapter \ref{ch:time_rev_model}, we examine the time reversed FRW/TOV model.  The numerical simulation shows this solution consists of two rarefaction waves, one outgoing and one incoming, encompassing a region of lower density.  In particular, the simulation indicates there are no shock waves present.  Therefore, this simulation shows the underlying solution is smooth and thus a strong solution to the Einstein equations, unlike the forward model.  We record the convergence of this solution and the convergence of the interaction region to the cone of sound along with the preservation of the non-interaction regions.  The simulation convinces us the incoming rarefaction wave will form a black hole in finite proper time.  That is, we extend the time of the simulation up until extreme relativistic effects of time dilation, characteristic of a black hole, appear.  We cannot simulate all the way up to black formation in standard Schwarzschild coordinates because black holes are coordinate singularities in standard Schwarzschild coordinates \cite{smolte2}, but the metric component $A$ is monotonically decreasing to a small value before time dilation makes our numerical calculations impractical.  Indeed, we numerically demonstrate the mass function near the black hole formation gets within $1.084\leq1.125 = 9/8$ of the Schwarzschild radius of that mass; thus, we have numerically confirmed the solution gets within the Buchdal statiblity limit of 9/8ths the Schwarzschild radius.  Buchdal's theorem states that whenever the mass of the star gets within 9/8ths of the Schwarzschild radius for that mass, there is no static configuration with a finite pressure at the center capable of holding the mass up \cite{smolte2}.  The implication then is a star that gets within 9/8ths the Schwarzschild radius cannot be prevented from collapsing into a black hole; thus, we view our results here as a numerical physical proof that black hole formation must occur in the time reversed FRW/TOV model.  Thus, the simulation indicates a black hole forming out of a smooth solution to the Einstein equations.  With the success of running the forward and reverse FRW/TOV models, we have simulated general relativistic analogs of both the 1 and 2 family of shock and rarefaction waves, that is, the analogs of all the elementary waves of conservation laws.  This fact gives us confidence in the ability and accuracy of our method to simulate any solution to the Einstein equations for a perfect fluid in standard Schwarzschild coordinates.

In Chapter \ref{ch:units}, we put dimensions back into the problem, giving our simulations of the earlier chapters a physical context.  For convenance and concreteness, we set the Einstein's speed of light constant and Newton's gravitational constant both to one.  To place units on our numerical values requires us to understand the effect this has on our units.  The effect of this is the unification of the units of time, length, and mass into one set of units, which we choose as the units of mass.  Understanding this unification allows us to recover the other units, time and length.  To end this chapter, we consider our simulation on the solar and galactic scale by transforming our numerical values to familiar units, giving us a firm grasp of the scale and magnitude involved within our simulations.

In Appendix \ref{ch:sim_code}, we outline the code used to perform the simulations in this paper.  A CD is attached that contains all this code, which is approximately 8,000 lines.  The appendix is designed to provide the reader an overview of the organization of the code and to help the reader navigate to specific areas of interest.

   \chapter{Preliminaries}
   \label{ch:prelim}
This chapter introduces the background to the topic of this thesis, the {\it locally inertial Godunov method with dynamical time dilation}.  This is a fractional step numerical method that approximates the convective part of the fluid dynamics in each grid cell by the solution of a Riemann problem in a (flat) Minkowski background spacetime, such that each cell is endowed with a time dilation factor.  Averages of the fluid variables are taken, and a second fractional step then evolves these averages according to evolution by a dynamical system that accounts for the undifferentiated source terms.  The main issue is at shock waves the gravitational metric has a jump discontinuity in the derivative, so the curvature of the gravitational metric becomes discontinuous at shocks, and the purpose of this thesis is to confirm our method is effective at numerically simulating general relativistic shock waves.

The  {\it locally inertial Godunov method with dynamical time dilation} is based on the locally inertial formulation of the Einstein equations for spherically symmetric spacetime metrics in standard Schwarzschild coordinates and was first derived by Groah and Temple in \cite{groate2}.  Groah and Temple provided an exposition of their work along with the work of Smoller and Temple in \cite{groasmte}.  All the equations in this chapter are taken from the joint work of Groah, Smoller, and Temple in \cite{groasmte}.   In this study, we consider the stress energy tensor a perfect fluid with an equation of state $p=\sigma\rho$\footnote{C.f. \cite{groasmte} where $\sigma^2$ is used in place of $\sigma$}, where $\sigma\equiv Const$, the simplest possible setting for shock wave propagation in general relativity.  In this case, the particle number density and entropy decouple from the energy and momentum equations, and the fluid equations close at the level of the Einstein equations alone.   This equation of state models flow at constant temperature, and in the case $\sigma=c^2/3$, it models the important cases of free particles in the extreme relativistic limit, as well as the case of pure radiation, c.f. \cite{wein}.  In \cite{groasmte}, Groah and Temple show that the four standard Schwarzschild coordinates Einstein equations are weakly equivalent to the system of PDE's obtained by taking the two equations that express the vanishing divergence of the stresses in place of the two equations that involve time derivatives of the metric.  This results in a system of conservation laws with source terms in which the special relativistic Minkowski variables are taken as the conserved quantities.  Fortunately, the equations close because all time derivatives of metric cancel out for this particular case.  The beauty of these equations is that they reflect the locally inertial character of spacetime.   The {\it locally inertial Godunov method with dynamical time dilation} is then a numerical method tailored to the simulation of solutions of these equations. We prove the consistency of this method, and demonstrate its value as an efficient method for computing shock wave solutions of the Einstein equations.

In Section \ref{sec:intro_einstein}, we introduce the Einstein equations, and we proceed to show how these equations with spherical symmetry and a simple equation of state reduces to a system of PDEs which express the locally inertial character in Einstein's theory.  In Section \ref{sec:gen_problem}, we formulate the conservation law with source terms.  Finally, we pose the general problem that is conducive to shock wave solutions to the Einstein equations, the starting point of our studies.

\section{Introduction to Einstein Equations}
\label{sec:intro_einstein}
The fundamental equations of general relativity, the Einstein field equations, is given by
\begin{equation}\label{einstein_eqns}
    G=\kappa T,
\end{equation}
where $G$ is the Einstein curvature tensor, $T$ is stress energy tensor (the source of gravity), and $\kappa$ is Einstein's coupling constant
\begin{equation}\label{ch1_kappa}
\kappa=\frac{8\pi\mathcal{G}}{c^4},
\end{equation}
where $\mathcal{G}$ is Newton's gravitational constant, and $c$ the speed of light.  For concreteness and convenience, we take the speed of light as $c=1$ and Newton's gravitational constant as $\mathcal{G}=1$ throughout this paper.  The consequences of setting these constants to one is discussed in detail in Chapter \ref{ch:units}.  These equations describe how the mass/energy of an object curves space and how the curvature, in turn, stretches or squeezes matter in the three spatial directions.

The component form of the Einstein equations is
\begin{equation}\label{einstein_eqns_component}
    G_{ij}(x)=\kappa T_{ij}(x),
\end{equation}
where
\begin{equation}\label{einstein_curvature}
    G_{ij}=R^\sigma_{i\sigma j} - \frac{1}{2}R^{\sigma\tau}_{\sigma\tau}g_{ij}
\end{equation}
denotes the individual components of the Einstein curvature tensor.  We let $i,j \in\{0,1,2,3\}$ refer to the individual components in a given coordinate system.  Throughout, the Einstein summation convention is used where repeated up-down indices are summed from 0 to 3, unless otherwise noted.

Modeling the sources of the stress energy tensor by a perfect fluid, $T$ is given by
\begin{equation}\label{stress_energy_tensor}
T^{ij}=(\rho+p)w^{i}w^{j}+pg^{ij},
\end{equation}
where $\mathbf{w}$ denotes the unit 4-velocity vector of the fluid (tangent vector to the world line of the fluid particle), $\rho$ denotes the energy density (measured in the inertial frame moving along with the fluid), and $p$ denotes the fluid pressure.

For simplification, the only metrics under consideration are spacetime metrics $g$ which are spherically symmetric.  A general spherically symmetric metric takes the form
\begin{equation}\label{spherically_symmetric_metric}
ds^2=g_{ij}dx^idx^j=-A(t,r)dt^2+B(t,r)dr^2+2D(t,r)dtdr+C(t,r)d\Omega^2,
\end{equation}
where the metric components $A,B,C,D$ are functions of the time and space variables $(t,r)$ only, $d\Omega^2 = d\theta^{2}+\sin^{2}\theta d\phi^{2}$ represents the standard line element of the unit 2-sphere, and $x\equiv(x^0,\ldots,x^3)\equiv(t,r,\theta,\phi)$ is the spacetime coordinate system.  A spherically symmetric metric is in standard Schwarzschild coordinates\footnote{Note that the choice of notation in \cite{groasmte} for standard Schwazschild coordinates was $ds^{2} =-A(t,r)dt^2+B(t,r)dr^2+r^{2}d\Omega^2$} if it is written in the more simple form of
\begin{equation}\label{ssc}
ds^{2} =-B(t,r)dt^2+\frac{1}{A(t,r)}dr^2+r^{2}d\Omega^2.
\end{equation}
A classic argument transforming a spherically symmetric metric (\ref{spherically_symmetric_metric}) into standard Schwarzschild coordinates (\ref{ssc}), which is thoroughly explained in Chapter \ref{ch:family_of_shock_waves}.

In this paper, the simplest setting for shock wave propagation in General Relativity is assumed: a spherically symmetric metric in standard Schwarzschild coordinates (\ref{ssc}) where the sources are modeled by a perfect fluid (\ref{stress_energy_tensor}) with an equation of state
\begin{equation}\label{eqn_of_state}
p = \sigma\rho, \phantom{4444} 0 < \sigma < c,
\end{equation}
where $\sigma$ is assumed to be constant and $\sqrt{\sigma}$ represents the sound speed.  Using MAPLE to put the standard metric (\ref{ssc}) into the Einstein equations (\ref{einstein_eqns}) gives us the following system of four coupled partial differential equations,
\begin{equation}\label{pde_system1}
\frac{AB}{r^2}\left\{-r\frac{A'}{A} + \frac{1}{A} - 1 \right\} = \kappa B^2T^{00},
\end{equation}
\begin{equation}\label{pde_system2}
\frac{\dot{A}}{rA} = \kappa \frac{B}{A}T^{01},
\end{equation}
\begin{equation}\label{pde_system3}
\frac{1}{r^2}\left\{r\frac{B'}{B} - (\frac{1}{A} - 1) \right\} = \frac{\kappa}{A^2}T^{11},
\end{equation}
\begin{equation}\label{pde_system4}
\frac{A^2}{rB}\left\{r\frac{\ddot{A}}{A^2} + B'' - \Phi \right\} = 2\kappa rAT^{22},
\end{equation}
where the quantity $\Phi$ is
\begin{equation}
\Phi = \frac{2\dot{A}}{A^3} + \frac{\dot{A}\dot{B}}{2A^2B} - \frac{1}{2A}\left(\frac{\dot{A}}{A}\right)^2 - \frac{B'}{r} - \frac{BA'}{rA} + \frac{B}{2}\left(\frac{B'}{B}\right)^2 - \frac{B}{2}\frac{B'}{B}\frac{A'}{A},
\end{equation}
where "prime" represents $\frac{\partial}{\partial r}$ and "dot" represents $\frac{\partial}{\partial t}$.

The mass function is $M\equiv M(t,r)$ is defined through the identity
\begin{equation}\label{mass_identity}
A=\left(1-\frac{2M}{r}\right),
\end{equation}
which is interpreted as the total mass inside radius $r$ at time $t$.  Using this total mass, equation (\ref{mass_identity}) leads to an equivalent form for equation (\ref{pde_system1}) given by
\begin{equation}\label{mass_space_ode}
M'=\frac{1}{2}\kappa r^2BT^{00},
\end{equation}
and equation (\ref{pde_system2}) becomes
\begin{equation}
\dot{M}=-\frac{1}{2}\kappa r^2BT^{01}.
\end{equation}

Assuming a perfect fluid (\ref{stress_energy_tensor}), the components of the stress energy tensor $T^{ij}$ satisfy the following relations
\begin{equation}\label{minkowski_relation00}
T^{00} = \frac{1}{B}T^{00}_M
\end{equation}
\begin{equation}\label{minkowski_relation01}
T^{01} = \sqrt{\frac{A}{B}}T^{01}_M
\end{equation}
\begin{equation}\label{minkowski_relation11}
T^{11} = AT^{11}_M
\end{equation}
\begin{equation}\label{minkowski_relation22}
T^{22} = \frac{\sigma\rho}{x^2},
\end{equation}
where $T^{ij}_M$ represent the components of the stress energy tensor in flat Minkowski spacetime.  When an equation of state takes the form (\ref{eqn_of_state}), the components of $T_M$ are given as
\begin{equation}\label{minkowski_defn00}
T^{00}_M = \frac{c^4+\sigma^2v^2}{c^2-v^2}\rho,
\end{equation}
\begin{equation}\label{minkowski_defn01}
T^{01}_M = \frac{c^2+\sigma^2}{c^2-v^2}cv\rho,
\end{equation}
\begin{equation}\label{minkowski_defn11}
T^{11}_M = \frac{v^2+\sigma^2}{c^2-v^2}c^2\rho.
\end{equation}

We denote the fluid velocity $v$, in place of the 4-velocity vector $\mathbf{w}$, as measured by an observer fixed with respect to the radial coordinate $r$.  Combining (\ref{mass_space_ode}) together with (\ref{minkowski_relation00}),  a formula for the mass is
\begin{equation}\label{mass_defn}
M(t,r) = M(t,r_0)+\frac{\kappa}{2}\int^r_{r_0}{T^{00}_M(t,r)r^2dr}.
\end{equation}
(Equations (\ref{pde_system1})-(\ref{mass_defn}) are taken from equations (1.3.2)-(1.3.15) in the work of Groah, Smoller, and Temple \cite{groasmte}.)  We note this equation (\ref{mass_defn}) combined with (\ref{mass_identity}) is an easier way to compute the metric component $A$ than using equation (\ref{pde_system1}), and these equations are used to compute $A$ in our locally inertial Godunov scheme presented in Chapter \ref{ch:frac_god_method}.

\section{Statement of the General Problem}
\label{sec:gen_problem}
The presence of shock waves in solutions of (\ref{pde_system1})-(\ref{pde_system4}) makes the stress-energy tensor $T$ discontinuous, and consequently, the metric components $A$ and $B$ are Lipschitz continuous at best.  Since (\ref{pde_system4}) contains second derivatives of $A$ and $B$, this equation can only be satisfied in the weak sense.  In \cite{groasmte}, Groah and Temple show when the metric is Lipschitz and the stress-energy tensor is bounded in sup-norm, system (\ref{pde_system1})-(\ref{pde_system4}) is weakly equivalent to a system in which equations (\ref{pde_system2}) and (\ref{pde_system4}) are replaced by
\begin{equation}\label{conservation_law00}
\{T^{00}_M\}_{,0}+\left\{\sqrt{AB}T^{01}_M\right\}_{,1}=-\frac{2}{x}\sqrt{AB}T^{01}_M,
\end{equation}
\begin{equation}\label{conservation_law01}
\begin{split}
\{T^{01}_M\}_{,0}+&\left\{\sqrt{AB}T^{11}_M\right\}_{,1}= \\
-\frac{1}{2}&\sqrt{AB}\left\{ \frac{4}{x}T^{11}_M+\frac{(\frac{1}{A}-1)}{x}(T^{00}_M-T^{11}_M)
+\frac{2\kappa x}{A}(T^{00}_MT^{11}_M-(T^{01}_M)^2)-4xT^{22}\right\}.
\end{split}
\end{equation}
Here, $\{\}_{,0}$ and $\{\}_{,1}$ represent the derivative with respect to the time and space variable, respectively.  The other two equations (\ref{pde_system1}) and (\ref{pde_system3}) are rearranged as
\begin{equation}\label{metric_ode0}
\frac{A'}{A}=\frac{(\frac{1}{A}-1)}{x}-\frac{\kappa x}{A}T^{00}_M,
\end{equation}
\begin{equation}\label{metric_ode1}
\frac{B'}{B}=\frac{(\frac{1}{A}-1)}{x}+\frac{\kappa x}{A}T^{11}_M.
\end{equation}
The variable $x$ is written instead of $r$ when the equations are expressed as a system of conservation laws consistent with the literature on the subject.  This conservation form is an advantageous formulation because equations (\ref{conservation_law00}) and (\ref{conservation_law01}) form a system of conservation laws with source terms, where the conserved quantities $u=(T^{00}_M,T^{01}_M)$ are independent of the metric (c.f. (\ref{minkowski_defn00}) and (\ref{minkowski_defn01})).  Moreover, all the time derivatives of the metric (\ref{pde_system2}) and (\ref{pde_system4}) cancels out, and the source terms to the conservation laws (\ref{conservation_law00}) and (\ref{conservation_law01}) are independent of any derivatives, resulting in a solvable system.  The resulting system (\ref{conservation_law00})-(\ref{metric_ode1}) can be written compactly as
\begin{equation}\label{conservation_law_source_cmpt}
\begin{split}
    u_{t}+&f(\mathbf{A},u)_{x}=g(\mathbf{A},u,x),\\
    &\mathbf{A}'=h(\mathbf{h},u,x).
\end{split}
\end{equation}
Here, the conserved quantities are the Minkowski energy and momentum densities
\begin{equation}
u=(T^{00}_M,T^{01}_M)\equiv(u^0,u^1),
\end{equation}
and the metric components are written as
\begin{equation}
\mathbf{A}=(A,B).
\end{equation}
The flux function becomes
\begin{equation}\label{flux_cmpt}
    f(\mathbf{A},u)=\sqrt{AB}(T^{01}_M,T^{11}_M).
\end{equation}
The source terms of the conservation laws, (\ref{conservation_law00}) and (\ref{conservation_law01}), are
\begin{equation}\label{conservation_law_source_cmpt}
    g(\mathbf{A},u,x)=(g^0(\mathbf{A},u,x),g^1(\mathbf{A},u,x)),
\end{equation}
where
\begin{equation}\label{conservation_law00_source}
g^0(\mathbf{A},u,x)=-\frac{2}{x}\sqrt{AB}T^{01}_M,
\end{equation}
and
\begin{equation}\label{conservation_law01_source}
\begin{split}
g^1(\mathbf{A},&u,x)= \\-&\frac{1}{2}\sqrt{AB}\left\{ \frac{4}{x}T^{11}_M+\frac{(\frac{1}{A}-1)}{x}(T^{00}_M-T^{11}_M)
+\frac{2\kappa x}{A}(T^{00}_MT^{11}_M-(T^{01}_M)^2)-4xT^{22}\right\}.
\end{split}
\end{equation}
The LHS of the metric ODEs, (\ref{metric_ode0}) and (\ref{metric_ode1}), are
\begin{equation}\label{metric_ode_source_cmpt}
    h(\mathbf{A},u,x)=(h^0(\mathbf{A},u,x),h^1(\mathbf{A},u,x)),
\end{equation}
with
\begin{equation}\label{metric_ode0_source}
h^0(\mathbf{A},u,x)=\frac{(1-A)}{x}-\kappa xT^{00}_M,
\end{equation}
and
\begin{equation}\label{metric_ode1_source}
h^1(\mathbf{A},u,x)=\frac{B}{A}\left\{\frac{(1-A)}{x}+\kappa xT^{11}_M\right\}.
\end{equation}
Thus, for weak solutions to the Einstein equations, the coupled PDEs (\ref{pde_system1})-(\ref{pde_system4}) are equivalent to a conservation law with source terms paired with a system of ODEs (\ref{conservation_law_source_cmpt}).  In this thesis the convergence of the locally inertial Godunov method (defined in Chapter \ref{ch:frac_god_method}) is shown to be a consistent and effective first order method for computing solutions of system (\ref{conservation_law_source_cmpt}), equivalent to the Einstein equations when shock waves are present and the gravitational metric is $C^{1,1}$.  Such solutions only solve the Einstein equations in the weak sense of the theory of distributions.  Our general system (\ref{conservation_law_source_cmpt}) applies to the case of spherically symmetric metrics in standard Schwarzschild coordinates, and in the next chapter, a list of candidates meeting this criteria is assembled.

   \chapter[%
      Short Title of 3rd Ch.
   ]{%
      One Parameter Family of Shock Wave Solutions
   }%
   \label{ch:family_of_shock_waves}
This chapter presents the theory behind the one parameter family of shock wave solutions of the Einstein equations whose numerical simulation is the subject of this thesis.  Our point of departure of this chapter is the paper \cite{smolte} in which Smoller and Temple (Sm/Te) construct exact spherically symmetric shock wave solutions to the Einstein equations for a perfect fluid.  Our purpose here is to include a more realistic equation of state, and thereby simulate the secondary reflected wave, which requires numerical simulation because it cannot be expressed in closed form.  The Sm/Te solutions were realized by matching two distinct metrics to form a Lipschitz continuous hybrid metric \cite{smolte5,groasmte2}.  In particular, the critical (k=0) Friedmann-Robertson-Walker (FRW) metric for the expanding universe was matched to the Tolman-Oppenheimer-Volkoff (TOV) metric, a model of the interior of a star, across a shock wave interface.  To perform this matching, one is required to provide a coordinate map of a spherically symmetric metric, like the FRW metric, over to standard Schwarzschild coordinates.   Recall, a general spherically symmetric metric takes the form
\begin{equation}\label{ch2_spherically_symmetric_metric}
ds^2=-A(t,r)dt^2+B(t,r)dr^2+2D(t,r)dtdr+C(t,r)d\Omega^2,
\end{equation}
where the metric components $A,B,C,D$ are functions of the time and space variables $(t,r)$ only.  The expression $d\Omega^2 = d\theta^{2}+\sin^{2}\theta d\phi^{2}$ represents the standard line element of the unit 2-sphere.  A metric in standard Schwarzschild coordinates takes the form
\begin{equation}\label{ch2_ssc}
ds^{2} =-B(\bar{t},\bar{r})d\bar{t}^2+\frac{1}{A(\bar{t},\bar{r})}d\bar{r}^2+\bar{r}^{2}d\Omega^2,
\end{equation}
where again the metric components are functions of the time and space alone.  Throughout this paper, the barred coordinates $(\bar{t},\bar{r})$ are reserved for metrics in a standard Schwarzschild coordinates system (\ref{ch2_ssc}) while the unbarred coordinates $(t,r)$ are for metrics that are only spherically symmetric (\ref{ch2_spherically_symmetric_metric}).  This FRW metric transformation in not only necessary for the matching, but it is also a prerequisite for the locally inertial Godunov scheme since it is only capable of simulating metrics in standard Schwarzschild coordinates.  When constructing these shock wave solutions, Sm/Te proved the existence of the necessary coordinate transformation, but they did not require any detailed information about it.  Here, we do not have this luxury since this information is essential to run the simulation.  Our main goal of this chapter is to uncover these details on transforming the spherically symmetric FRW metric into standard Schwarzschild coordinates.  We discover there exists at least two of these transformations of the FRW metric, denoted as {\it FRW-1} and {\it FRW-2}.  In the subsequent chapters, these FRW metrics expressed in different standard Schwarzschild coordinates are used to test our locally inertial Godnuov method, by simulating these continuous metrics.

In Section \ref{sec:frw1_metric}, we briefly discuss the FRW metric.  The FRW-1 metric is obtained by providing a candidate for the transformation of the FRW metric into standard Schwarzschild coordinates and proving the validity of this transformation.  The fluid variables also get mapped over as functions of the new coordinates.  In Section \ref{sec:frw2_metric}, we review the standard procedure for mapping a spherically symmetric metric to standard Schwarzschild coordinates, and the FRW-2 metric is constructed using this procedure.  The other metric in the matching, the TOV metric, is discussed in Section \ref{sec:tov_metric}.  This metric is already in standard Schwarzschild coordinates, but the main aspects of it needed in this paper are explained in this section.  We conclude with the details of matching the FRW and TOV metrics continuously across a shock wave interface in Section \ref{sec:matching_metrics}.  It is in this section, we make precise the one parameter family of shock wave solutions by defining the shock wave surface.

\section{The FRW-1 Metric}
\label{sec:frw1_metric}
The conformally flat ($k=0$) FRW metric takes the form
\begin{equation}\label{ch2_frw_metric}
ds^2=-dt^2+R^2(t)\{dr^2+r^2d\Omega^2\}.
\end{equation}
The fluid density, as a function of time $t$, associated with this metric takes the form
\begin{equation}\label{frw_density}
\rho(t)=\frac{3}{4\kappa t^2}
\end{equation}
where again
\begin{equation}\label{ch2_kappa}
\kappa=\frac{8\pi\mathcal{G}}{c^4}
\end{equation}
is Einstein's coupling constant (\ref{ch1_kappa}).  To determine the fluid velocity, the fluid in this metric is comoving relative to the background.  More precisely, its velocity vector, $\mathbf{w}=(w^0,w^1,w^2,w^3)$, has the form $w^i=0$, for $i=1,2,3$, or, in words, the movement of the fluid is limited to the time coordinate and none of the space coordinates.  The FRW metric $g$ being diagonal and the velocity vector having unit length implies
\begin{equation}
w^0=\sqrt{-g_{00}} = 1.
\end{equation}
This velocity vector is expressed in terms of the classical coordinate radial velocity $v\equiv dr/dt$ of the particle paths of the fluid as $v=0$.

Recently, Smoller and Temple derived a new set of equations with solutions describing a two parameter family of expanding wave solutions of the Einstein equations which includes the critical Friedmann universe in the standard model of cosmology \cite{smoltePNAS09}.  In deriving this set of equations, they discover a coordinate map for the FRW metric (\ref{ch2_frw_metric}) over to standard Schwarzschild coordinates (\ref{ch2_ssc}) for $\sigma=\frac{1}{3}$, and where $R(t)=\sqrt{t}$ is the scale factor.  We state their result and derivation in the following theorem
\begin{thm}\label{thm:frw1_xform_ssc}
Assume we have an equation of state $p=\frac{1}{3}\rho$ and $k=0$.  Then the FRW metric (\ref{ch2_frw_metric}) under the coordinate transformation
\begin{equation}\label{frw1_xform_ssc}
\begin{split}
\bar{r}=&\sqrt{t}r,\\
\bar{t}=&\left\{1+\frac{\bar{r}^2}{4t^2}\right\}t=t+\frac{r^2}{4},
\end{split}
\end{equation}
goes over to the following metric in standard Schwarzschild coordinates
\begin{equation}\label{ch2_frw1_in_ssc}
ds^{2} =-\frac{1}{1-v^2}d\bar{t}^2+\frac{1}{1-v^2}d\bar{r}^2+\bar{r}^{2}d\Omega^2,
\end{equation}
where the fluid velocity $v$ is related to $\bar{r}/\bar{t}$ by
\begin{equation}\label{xi_v_relation}
\xi\equiv\frac{\bar{r}}{\bar{t}}=\frac{2v}{1+v^2}.
\end{equation}
\end{thm}
\begin{proof}
To map the FRW metric tensor (\ref{ch2_frw_metric}) from the unbarred coordinates $(t,r)$ over to the barred coordinates $(\bar{t},\bar{r})$, we use the tensor transformation law
\begin{equation}
\bar{g}_{\alpha\beta}=g_{ij}\frac{\partial x^i}{\partial \bar{x}^\alpha}\frac{\partial x^j}{\partial \bar{x}^\beta}.
\end{equation}
More specifically, at each point, g transforms by the matrix transformation law
\begin{equation}
\bar{g}=J^TgJ
\end{equation}
for a bilinear form since the Jacobian matrix $J=\frac{\partial x^i}{\partial \bar{x}^\alpha}$ transforms the vector components of the $\bar{x}$-basis $\left\{\frac{\partial}{\partial \bar{x}^\alpha}\right\}$ into their corresponding coordinates of the $x$-basis $\left\{\frac{\partial}{\partial x^i}\right\}$.
The inverse of the Jacobian matrix is easily computed by taking partial derivatives of the coordinate transformation equations (\ref{frw1_xform_ssc}), resulting in
\begin{equation}
J^{-1}=\left(
\begin{array}{cc}
1 & \frac{r}{2}\\
\frac{r}{2\sqrt{t}} & \sqrt{t}
\end{array}
\right).
\end{equation}
So the Jacobian matrix is
\begin{equation}
J = \frac{1}{|J^{-1}|}\left(
\begin{array}{cc}
\sqrt{t} & -\frac{r}{2}\\
-\frac{r}{2\sqrt{t}} & 1
\end{array}
\right),
\end{equation}
where
\begin{equation}\label{jacobian_inverse_determinat}
|J^{-1}| = \frac{4t-r^2}{4\sqrt{t}}.
\end{equation}
The metric in barred coordinates $\bar{g}$ is computed as
\begin{equation}
\begin{split}
\bar{g} = J^TgJ =&
\frac{1}{|J^{-1}|^2}\left(
\begin{array}{cc}
\sqrt{t} & -\frac{r}{2\sqrt{t}}\\
-\frac{r}{2} & 1
\end{array}
\right)
\left(
\begin{array}{cc}
-1 & 0\\
0 & t
\end{array}
\right)
\left(
\begin{array}{cc}
\sqrt{t} & -\frac{r}{2}\\
-\frac{r}{2\sqrt{t}} & 1
\end{array}
\right)\\
=&
\frac{1}{|J^{-1}|^2}\left(
\begin{array}{cc}
-(t-\frac{r^2}{4}) & 0\\
0 & t-\frac{r^2}{4}
\end{array}
\right).
\end{split}
\end{equation}
Using (\ref{jacobian_inverse_determinat}) to solve for
\begin{equation}
|J^{-1}|^2 = \frac{(t-\frac{r^2}{4})^2}{t},
\end{equation}
gives us
\begin{equation}
\bar{g}=
\left(
\begin{array}{cc}
-\frac{1}{1-\frac{r^2}{4t}} & 0\\
0 & \frac{1}{1-\frac{r^2}{4t}}
\end{array}
\right).
\end{equation}
Note, the new metric is a function of $\bar{r}/t$, and it is useful to define a new variable $\eta$ as
\begin{equation}\label{eta_defn}
\eta\equiv\frac{\bar{r}}{t}=\frac{\sqrt{t}r}{t}=\frac{r}{\sqrt{t}}.
\end{equation}
As a function of $\eta$, the metric becomes
\begin{equation}
\bar{g}=
\left(
\begin{array}{cc}
-\frac{1}{1-\frac{\eta^2}{4}} & 0\\
0 & \frac{1}{1-\frac{\eta^2}{4}}
\end{array}
\right),
\end{equation}
or equivalently it is written in standard Schwarzschild coordinates form (\ref{ch2_ssc}),
\begin{equation}\label{frw1_in_ssc_with_eta}
ds^{2} =-\frac{1}{1-\frac{\eta^2}{4}}d\bar{t}^2 +\frac{1}{1-\frac{\eta^2}{4}}d\bar{r}^2+\bar{r}^{2}d\Omega^2,
\end{equation}
with
\begin{equation}
A(\bar{t},\bar{r}) = 1-\frac{\eta^2}{4},\phantom{4444} B(\bar{t},\bar{r}) = \frac{1}{1-\frac{\eta^2}{4}}.
\end{equation}

The metric is also a function of the fluid velocity $v$, shown by finding the relationship between $v$ and our variable $\eta$.  To this end, the old velocity vector $\mathbf{w}$ is related to the new velocity vector $\mathbf{\bar{w}}$ using the tensor transformation law
\begin{equation}\label{velocity_vector_xform_law}
\bar{w}^\alpha=\frac{\partial\bar{x}^\alpha}{\partial x^i}w^i,
\end{equation}
so for our comoving velocity vector $\mathbf{w}=(1,0,0,0)$ this law simplifies to
\begin{equation}\label{velocity_vector_xform_law_comoving}
\bar{w}^\alpha=\frac{\partial\bar{x}^\alpha}{\partial x^0}.
\end{equation}
Using (\ref{velocity_vector_xform_law_comoving}) with our coordinate transformation (\ref{frw1_xform_ssc}), the new velocity vector coordinates are
\begin{equation}
\bar{w}^0 =\frac{\partial\bar{x}^0}{\partial x^0}= \frac{\partial\bar{t}}{\partial t} = 1,\phantom{4444} \bar{w}^1 =\frac{\partial\bar{x}^1}{\partial x^0}= \frac{\partial\bar{r}}{\partial t} = \frac{1}{2}\frac{r}{\sqrt{t}},\phantom{4444} \bar{w}^2=\bar{w}^3=0.
\end{equation}
For a spherically symmetric metric in standard Schwarzschild coordinates (\ref{ch2_ssc}) the radial velocity $v\equiv dr/dt$ is computed as
\begin{equation}
v=\frac{\bar{w}^1}{\bar{w}^2}\frac{1}{\sqrt{AB}}=\frac{1}{2}\frac{r}{\sqrt{t}}=\frac{1}{2}\frac{\bar{r}}{t}.
\end{equation}
For our metric (\ref{frw1_in_ssc_with_eta}), $A=B^{-1}$, implying $\sqrt{AB}=1$.
By (\ref{eta_defn}), the relationship between $v$ and $\eta$ is
\begin{equation}\label{v_eta_relationship}
v=\frac{\eta}{2}.
\end{equation}

We substitute this relation (\ref{v_eta_relationship}) into (\ref{frw1_in_ssc_with_eta}) to prove (\ref{ch2_frw1_in_ssc}).
By defining the variable $\xi=\bar{r}/\bar{t}$, the coordinate transformation (\ref{frw1_xform_ssc}) implies
\begin{equation}
\xi\equiv\frac{\bar{r}}{\bar{t}}=\frac{\sqrt{t}r}{(1+\frac{r^2}{4t})t}=\frac{2v}{1+v^2},
\end{equation}
proving (\ref{xi_v_relation}).
\end{proof}

The FRW metric in standard Schwarzschild coordinates (\ref{ch2_ssc})
\begin{equation}
ds^{2} =-\frac{1}{1-v^2}d\bar{t}^2+\frac{1}{1-v^2}d\bar{r}^2+\bar{r}^{2}d\Omega^2,
\end{equation}
with
\begin{equation}
A(\xi) = 1-v(\xi)^2,\phantom{4444} B(\xi) = \frac{1}{1-v(\xi)^2},
\end{equation}
is a function of $\xi$ alone.  We consider now the fluid variables $(\rho,v)$.  The fluid velocity $v$ is implicitly a function of $\xi$ by equation (\ref{xi_v_relation}), but through manipulation the fluid velocity $v$ is a function of $\xi$ explicitly
\begin{equation}
v(\xi)=\frac{1-\sqrt{1-\xi^2}}{\xi},
\end{equation}
where the minus sign is taken to ensure $v<1$.  The density $\rho$ is a function of the FRW time $t$ (\ref{frw_density}) and is transformed into a function of $\xi$ by multiplying the density function (\ref{frw_density}) by $\bar{r}^2$ to obtain
\begin{equation}\label{rho_fn_v}
\bar{r}^2\rho = \frac{3}{4}\frac{\bar{r}^2}{\kappa t^2} =\frac{3}{4}\frac{\eta^2}{\kappa}=\frac{3v(\xi)^2}{\kappa},
\end{equation}
using (\ref{v_eta_relationship}).  As a side note, this equation along with (\ref{mass_space_ode}) implies that the function $\frac{dM}{d\bar{r}}$ at a point, where $M$ is the total mass, can also be written as a function of $\xi$.  Rearranging equation (\ref{rho_fn_v}), the density function becomes
\begin{equation}
\rho(\xi,\bar{r}) =\frac{3v(\xi)^2}{\kappa\bar{r}^2}.
\end{equation}
These results are recorded as a Corollary to Theorem \ref{thm:frw1_xform_ssc}
\begin{cor}
The fluid variables $(\rho,v)$ corresponding to the FRW metric in standard Schwarzschild coordinates (\ref{ch2_frw1_in_ssc}) are given as
\begin{equation}\label{ch2_frw1_fluid_vars_in_ssc}
\rho(\xi,\bar{r}) =\frac{3v(\xi)^2}{\kappa\bar{r}^2},\phantom{4444} v(\xi)=\frac{1-\sqrt{1-\xi^2}}{\xi},
\end{equation}
where $\xi=\bar{r}/\bar{t}$.
\end{cor}

\section{The FRW-2 Metric}
\label{sec:frw2_metric}
In this section, we show the existence of another mapping of the FRW metric (\ref{ch2_frw_metric}) into standard Schwarzschild coordinates.  As opposed to the a priori knowledge of the FRW-1 transformation in Theorem \ref{thm:frw1_xform_ssc}, the FRW-2 transformation is developed through construction.  There is a classic argument in \cite{wein} to build a coordinate transformation $(t,r)\rightarrow(\bar{t},\bar{r})$ taking a general spherically symmetric metric (\ref{ch2_spherically_symmetric_metric}) over to standard Schwarzschild coordinates (\ref{ch2_ssc}), and this argument is repeated to build the FRW-2 transformation.  In particular, we take the spherically symmetric FRW metric
\begin{equation}\label{ch2_frw_metric2}
ds^2=-dt^2+R^2(t)\{dr^2+r^2d\Omega^2\},
\end{equation}
over to the standard Schwarzschild coordinate form
\begin{equation}
ds^{2} =-B(\bar{t},\bar{r})d\bar{t}^2+\frac{1}{A(\bar{t},\bar{r})}d\bar{r}^2+\bar{r}^{2}d\Omega^2.
\end{equation}
We first match the spheres of symmetry
\begin{equation}
\bar{r}^2d\Omega^2=R^2r^2d\Omega^2,
\end{equation}
so the first component of the coordinate mapping $(t,r)\rightarrow(\bar{t},\bar{r})$ is defined by the following
\begin{equation}\label{r_map}
\bar{r}\equiv\bar{r}(t,r)=R(t)r.
\end{equation}
Rewriting (\ref{r_map}) as
\begin{equation}
r=\frac{\bar{r}}{R},
\end{equation}
we differentiate and square to obtain the following differential form relationships,
\begin{equation}
dr=\frac{1}{R}d\bar{r}-\frac{\dot{R}\bar{r}}{R^2}dt
\end{equation}
and
\begin{equation}\label{dr_squared}
dr^2=\frac{1}{R^2}d\bar{r}^2-\frac{2\dot{R}\bar{r}}{R^3}d\bar{r}dt +\frac{\dot{R}^2\bar{r}^2}{R^4}dt^2
\end{equation}
Substituting (\ref{dr_squared}) into (\ref{ch2_frw_metric2}), results in the FRW metric in $(t,\bar{r})$ coordinates
\begin{equation}\label{t_r_bar_coords}
ds^2=-\left(1-\frac{\dot{R}^2\bar{r}^2}{R^2}\right)dt^2+d\bar{r}^2-\frac{2\dot{R}\bar{r}}{R}dtd\bar{r} +\bar{r}^2d\Omega^2.
\end{equation}
In order to complete the coordinate transformation $(t,r)\rightarrow(\bar{t},\bar{r})$, $\bar{t}=\bar{t}(t,r)$ must be defined to eliminate the cross term $dtd\bar{r}$ in (\ref{t_r_bar_coords}).  Since it is already matched, the $\bar{r}^2d\Omega$ term is ignored in our subsequent calculations to simplify notation.  We show the procedure for eliminating the cross term for a generalized metric of the form
\begin{equation}\label{cross_term_general_metric}
ds^2=-C(t,\bar{r})dt^2+D(t,\bar{r})d\bar{r}^2+2E(t,\bar{r})dtd\bar{r},
\end{equation}
and handle our specific case afterwards.  To eliminate the cross term in (\ref{cross_term_general_metric}), an integrating factor $\Psi=\Psi(t,\bar{r})$ is chosen to satisfy the equation
\begin{equation}\label{general_integrating_factor_eqn}
\frac{\partial}{\partial\bar{r}}(\Psi C)=-\frac{\partial}{\partial t}(\Psi E),
\end{equation}
implying
\begin{equation}\label{exact_differential_eqn}
d\bar{t}=\Psi\{Cdt-Ed\bar{r}\}
\end{equation}
is an exact differential, proven at the end of this section in Lemma \ref{lem:exact_differential}.  Exactness gives us the existence of $\bar{t}(t,\bar{r})$, completing the definition of our coordinate transformation $(t,r)\rightarrow(\bar{t},\bar{r})$.  We consider options for the integrating factor $\Psi(t,\bar{r})$ and the corresponding coordinate transform $\bar{t}=\bar{t}(t,r)$ later; for now, we continue mapping the metric over, by assuming $\Psi$ satisfies (\ref{general_integrating_factor_eqn}).  Squaring both sides of (\ref{exact_differential_eqn}) results in
\begin{equation}
d\bar{t}^2=\Psi^2\{C^2dt^2-2CEdtd\bar{r}+E^2d\bar{r}^2\},
\end{equation}
and rearranging leads us to
\begin{equation}
-Cdt^2+2Edtd\bar{r}=-\frac{1}{C\Psi^2}d\bar{t}^2+\frac{E^2}{C}d\bar{r}^2.
\end{equation}
Substituting this equation into (\ref{cross_term_general_metric}) gives us the diagonal metric in the barred coordinate frame
\begin{equation}\label{diagonal_general_metric}
d\tilde{s}^2=-\frac{1}{C\Psi^2}d\bar{t}^2+\left(D+\frac{E^2}{C}\right)d\bar{r}^2.
\end{equation}
For our specific case, we have
\begin{equation}\label{cde_metric_terms}
C=1-\frac{\dot{R}^2\bar{r}^2}{R^2},\phantom{4444} D=1,\phantom{4444} E=-\frac{\dot{R}\bar{r}}{R},
\end{equation}
and substituting these into the general form (\ref{diagonal_general_metric}) gives us
\begin{equation}
ds^2=-\frac{1}{\Psi^2\left(1-\frac{\dot{R}^2\bar{r}^2}{R^2}\right)}d\bar{t}^2 +\frac{1}{\left(1-\frac{\dot{R}^2\bar{r}^2}{R^2}\right)}d\bar{r}^2 +\bar{r}^2d\Omega^2.
\end{equation}
For the pure radiation phase, where $\sigma=1/3$, the cosmological scale function becomes $R(t)=\sqrt{t}$, and the resulting metric is
\begin{equation}
ds^2=-\frac{1}{\Psi^2\left(1-\frac{\bar{r}^2}{4t^2}\right)}d\bar{t}^2 +\frac{1}{\left(1-\frac{\bar{r}^2}{4t^2}\right)}d\bar{r}^2 +\bar{r}^2d\Omega^2.
\end{equation}
The definition of $\eta$ (\ref{eta_defn}) allows us to rewrite the metric in the barred coordinates as
\begin{equation}\label{frw2_in_ssc_with_eta}
ds^2=-\frac{1}{\Psi^2\left(1-\frac{\eta^2}{4}\right)}d\bar{t}^2 +\frac{1}{\left(1-\frac{\eta^2}{4}\right)}d\bar{r}^2 +\bar{r}^2d\Omega^2.
\end{equation}

Notice this metric (\ref{frw2_in_ssc_with_eta}) closely resembles the FRW-1 case (\ref{frw1_in_ssc_with_eta}) except for inclusion of $\Psi$, indicating our ability to recover the FRW-1 metric by choosing the correct integrating factor.  With this in mind, we look for solutions of the integrating factor equation (\ref{general_integrating_factor_eqn}).  In light of the functions $C(t,\bar{r})$ and $E(t,\bar{r})$ defined in (\ref{cde_metric_terms}), equation (\ref{general_integrating_factor_eqn}) becomes
\begin{equation}
\frac{\partial}{\partial\bar{r}}\left(\Psi\left(1-\frac{\bar{r}^2}{4t^2}\right)\right)=\frac{\partial}{\partial\bar{r}}(\Psi C)=-\frac{\partial}{\partial t}(\Psi E)=\frac{\partial}{\partial t}\left(\Psi\left(\frac{\bar{r}}{2t}\right)\right).
\end{equation}
The following lemma expresses the existence of two solutions for this integrating factor equation
\begin{lem}\label{lem:integrating_factor_eqn}
The following PDE for the function $\Psi(t,\bar{r})$
\begin{equation}\label{integrating_factor_eqn}
\frac{\partial}{\partial\bar{r}}\left(\Psi\left(1-\frac{\bar{r}^2}{4t^2}\right)\right)=\frac{\partial}{\partial t}\left(\Psi\left(\frac{\bar{r}}{2t}\right)\right)
\end{equation}
has at least two solutions, a constant solution $\Psi_1(t,\bar{r})=\Psi_0$ and a dynamical solution
\begin{equation}\label{frw2_integrating factor}
\Psi_2(t,\bar{r})=\Psi_0\sqrt{\frac{t}{4t^2+\bar{r}^2}},
\end{equation}
where $\Psi_0$ is some constant.
\end{lem}
\begin{proof}
The constant function, $\Psi_1(t,\bar{r})=\Psi_0$, solves the PDE (\ref{integrating_factor_eqn}) since
\begin{equation}
\Psi_0\frac{\partial}{\partial\bar{r}}\left(1-\frac{\bar{r}^2}{4t^2}\right)= -\Psi_0\frac{\bar{r}}{2t^2}= \Psi_0\frac{\partial}{\partial t}\left(\frac{\bar{r}}{2t}\right),
\end{equation}

To find the other solution for the integrating factor $\Psi_2(t,\bar{r})$, a simplification of (\ref{integrating_factor_eqn}) is obtained by noticing
\begin{equation}
\frac{\partial}{\partial\bar{r}}\left(1-\frac{\bar{r}^2}{4t^2}\right)= \frac{\partial}{\partial t}\left(\frac{\bar{r}}{2t}\right),
\end{equation}
so $\Psi_2$ satisfying (\ref{integrating_factor_eqn}) is equivalent being a solution to the following PDE
\begin{equation}\label{psi_pde}
\frac{\partial\Psi_2}{\partial\bar{r}}\left(1-\frac{\bar{r}^2}{4t^2}\right)= \frac{\partial\Psi_2}{\partial t}\left(\frac{\bar{r}}{2t}\right).
\end{equation}

We proceed by constructing a function $\Psi_2$ satisfying (\ref{psi_pde}).  Suppose $\Psi_2$ has the form
\begin{equation}\label{psi2_defn}
\Psi_2(t,\bar{r})=\frac{f(\eta)}{\sqrt{\bar{r}}},
\end{equation}
where $f(\eta)$ is an unknown function of the predefined variable $\eta=\bar{r}/t$.  Taking the partial derivatives of \ref{psi2_defn} leads to
\begin{equation}\label{psi_partials}
\frac{\partial\Psi_2}{\partial\bar{r}}=\frac{f'(\eta)}{t\sqrt{\bar{r}}}-\frac{f(\eta)}{2\bar{r}^{\frac{3}{2}}}, \phantom{4444}\frac{\partial\Psi_2}{\partial t}=\frac{f'(\eta)}{\sqrt{\bar{r}}}\left(-\frac{\bar{r}}{t^2}\right),
\end{equation}
where the prime represents differentiation with respect to the variable $\eta$.  Substituting $\eta$ into equation (\ref{psi_pde}) gives us
\begin{equation}\label{psi_pde_eta}
\frac{\partial\Psi}{\partial\bar{r}}\left(1-\frac{\eta^2}{4}\right)= \frac{\partial\Psi}{\partial t}\left(\frac{\eta}{2}\right).
\end{equation}
Plugging in the partials (\ref{psi_partials}) into the PDE (\ref{psi_pde_eta}) and simplifying results in the following ODE
\begin{equation}
\frac{f'(\eta)}{f(\eta)}=\frac{4-\eta^2}{2\eta(4+\eta^2)}.
\end{equation}
Using partial fractions, the solution of this ODE is
\begin{equation}
f(\eta)=f_0\sqrt{\frac{4+\eta^2_0}{\eta_0}\frac{\eta}{4+\eta^2}}.
\end{equation}
Combining the constants together by defining
\begin{equation}
\Psi_0\equiv f_0\sqrt{\frac{4+\eta^2_0}{\eta_0}}
\end{equation}
produces the unknown function $f(\eta)$ as
\begin{equation}
f(\eta)=\Psi_0\sqrt{\frac{\eta}{4+\eta^2}}.
\end{equation}
Therefore, the second integrating factor becomes
\begin{equation}
\Psi_2(t,\bar{r})=\frac{f(\eta)}{\sqrt{\bar{r}}}=\Psi_0\sqrt{\frac{\eta}{\bar{r}(4+\eta^2)}}= \Psi_0\sqrt{\frac{t}{4t^2+\bar{r}^2}},
\end{equation}
proving (\ref{frw2_integrating factor}).
\end{proof}

The first integrating factor $\Psi_1$, the constant solution, gives us a differential form (\ref{exact_differential_eqn}) of
\begin{equation}
d\bar{t}=\Psi\{Cdt-Ed\bar{r}\}= \Psi_0\left\{\left(1-\frac{\bar{r}^2}{4t^2}\right)dt-\frac{\bar{r}}{2t}d\bar{r}\right\}
\end{equation}
Since this is an exact 1-form, the coordinate transformation $\bar{t}(t,\bar{r})$ is obtained through integration as
\begin{equation}
\bar{t}(t,\bar{r})=\Psi_0\left(1+\frac{\bar{r}^2}{4t^2}\right)t,
\end{equation}
or
\begin{equation}
\bar{t}(t,r)=\Psi_0\left(1+\frac{r^2}{4t}\right)t,
\end{equation}
where $\bar{r}=R(t)r=\sqrt{t}r$ is used to connect the two.  By setting $\Psi_0=1$, we recover the FRW-1 coordinate transformation (\ref{frw1_xform_ssc}) into standard Schwarzschild coordinates, providing another proof for its existence.

The second integrating factor $\Psi_2$, the dynamical solution, has the corresponding 1-form (\ref{exact_differential_eqn}) of
\begin{equation}
d\bar{t}=\Psi\{Cdt-Ed\bar{r}\}= \Psi_0\left\{\sqrt{\frac{t}{4t^2+\bar{r}^2}}\left(1-\frac{\bar{r}^2}{4t^2}\right)dt -\sqrt{\frac{t}{4t^2+\bar{r}^2}}\left(\frac{\bar{r}}{2t}\right)d\bar{r}\right\}.
\end{equation}
Again, since this 1-form is exact, there exists a function $\bar{t}(t,\bar{r})$ such that
\begin{equation}
\frac{\partial\bar{t}}{\partial t} = \Psi_0\sqrt{\frac{t}{4t^2+\bar{r}^2}}\left(1-\frac{\bar{r}^2}{4t^2}\right)\phantom{4444}
\frac{\partial\bar{t}}{\partial\bar{r}} =-\Psi_0\sqrt{\frac{t}{4t^2+\bar{r}^2}}\left(\frac{\bar{r}}{2t}\right)
\end{equation}
To solve for $\bar{t}$ we integrate the $\frac{\partial\bar{t}}{\partial\bar{r}}$ equation to obtain
\begin{equation}\label{prelim_t_bar}
\bar{t} = \frac{\Psi_0}{2}\sqrt{\frac{4t^2+\bar{r}^2}{t}} + g(t),
\end{equation}
for some function $g(t)$.  Taking $\frac{\partial}{\partial t}$ of (\ref{prelim_t_bar}) leads to $g(t)=C$, the constant function.  By setting $g(t)=0$, the coordinate transform corresponding to the second integrating factor (\ref{frw2_integrating factor}) is
\begin{equation}\label{ch2_t_bar_xform}
\bar{t}(t,\bar{r}) = \frac{\Psi_0}{2}\sqrt{\frac{4t^2+\bar{r}^2}{t}}.
\end{equation}
Hence, our new coordinate transformation of the FRW metric (\ref{ch2_frw_metric}) into standard Schwarzschild coordinates is
\begin{equation}\label{frw2_xform_ssc}
\begin{split}
\bar{r}=&\sqrt{t}r,\\
\bar{t}=&\frac{\Psi_0}{2}\sqrt{\frac{4t^2+\bar{r}^2}{t}}.
\end{split}
\end{equation}

Similar to the FRW-1 case, there exists a relationship between $v$ and $\eta$.  The velocity vector transformation laws, (\ref{velocity_vector_xform_law}) and (\ref{velocity_vector_xform_law_comoving}), along with our new coordinate transformation (\ref{frw2_xform_ssc}) provides the new velocity vector $\mathbf{\bar{w}}$ coordinates
\begin{equation}
\bar{w}^0 =\frac{\partial\bar{x}^0}{\partial x^0}= \frac{\partial\bar{t}}{\partial t} = \frac{\Psi_0}{(4t+r^2)^{\frac{1}{2}}},\phantom{4444} \bar{w}^1 =\frac{\partial\bar{x}^1}{\partial x^0}= \frac{\partial\bar{r}}{\partial t} = \frac{1}{2}\frac{r}{\sqrt{t}},\phantom{4444} \bar{w}^2=\bar{w}^3=0.
\end{equation}
The radial velocity $v\equiv dr/dt$ is computed as
\begin{equation}
v=\frac{\bar{w}^1}{\bar{w}^2}\frac{1}{\sqrt{AB}} =\frac{r\sqrt{4t+r^2}}{2\Psi_0\sqrt{t}}\frac{1}{\sqrt{AB}} =\frac{\eta}{2\Psi_0}\sqrt{\frac{4t^2+\bar{r}^2}{t}}\frac{1}{\sqrt{AB}} =\frac{\eta}{2\Psi}\frac{1}{\sqrt{AB}}=\frac{\eta}{2},
\end{equation}
where we used (\ref{frw2_in_ssc_with_eta}) that $\sqrt{AB}=1/\Psi$.  Remarkably, this relationship matches the FRW-1 case (\ref{v_eta_relationship}).  Using this relationship, the FRW-2 metric (\ref{frw2_in_ssc_with_eta}) is rewritten in terms of the $v$ as
\begin{equation}
ds^{2} =-\frac{1}{\Psi^2(1-v^2)}d\bar{t}^2+\frac{1}{1-v^2}d\bar{r}^2+\bar{r}^{2}d\Omega^2.
\end{equation}
Unlike the FRW-1 metric, the fluid variables are not a function of the ratio $\xi=\bar{r}/\bar{t}$, so the FRW-2 metric relies on the unbarred coordinate time $t$.  The coordinate transformation (\ref{ch2_t_bar_xform}) is used to find $t$ as a function of the new coordinates $(\bar{t},\bar{r})$,
\begin{equation}
t(\bar{t},\bar{r})=\frac{\bar{t}^2+\sqrt{\bar{t}^4-\bar{r}^2\Psi^4_0}}{2\Psi^2_0}.
\end{equation}

All these results for the FRW-2 metric are recorded in the following thereom
\begin{thm}\label{thm:frw2_xform_ssc}
Assume we have an equation of state $p=\frac{1}{3}\rho$ and $k=0$.  Then the FRW metric (\ref{ch2_frw_metric}) under the coordinate transformation
\begin{equation}\label{ch2_frw2_xform_ssc}
\begin{split}
\bar{r}=&\sqrt{t}r,\\
\bar{t}=&\frac{\Psi_0}{2}\sqrt{\frac{4t^2+\bar{r}^2}{t}}
\end{split}
\end{equation}
associated with the integrating factor
\begin{equation}\label{ch2_frw2_integrating factor}
\Psi(\bar{t},\bar{r})=\Psi_0\sqrt{\frac{t}{4t^2+\bar{r}^2}},
\end{equation}
goes over to the following metric in standard Schwarzschild coordinates
\begin{equation}\label{ch2_frw2_in_ssc}
ds^{2} =-\frac{1}{\Psi^2(1-v^2)}d\bar{t}^2+\frac{1}{1-v^2}d\bar{r}^2+\bar{r}^{2}d\Omega^2.
\end{equation}
The fluid variables $(\rho,v)$ corresponding to this metric are
\begin{equation}\label{ch2_frw2_fluid_vars_in_ssc}
\rho(t) =\frac{3}{4\kappa t^2},\phantom{4444} v(t,\bar{r})=\frac{\eta}{2}=\frac{\bar{r}}{2t},
\end{equation}
where the unbarred time coordinate is the following function of $(\bar{t},\bar{r})$
\begin{equation}\label{ch2_frw2_t}
t(\bar{t},\bar{r})=\frac{\bar{t}^2+\sqrt{\bar{t}^4-\bar{r}^2\Psi^4_0}}{2\Psi^2_0}.
\end{equation}
\end{thm}

This section ends with the lemma stating any $\Psi$ satisfying (\ref{general_integrating_factor_eqn}) makes $d\bar{t}$ (\ref{exact_differential_eqn}) into an exact 1-form.
\begin{lem}\label{lem:exact_differential}
If the integrating factor $\Psi(t,\bar{r})$ satisfies
\begin{equation}\label{lem_exact_differential_eqn}
\frac{\partial}{\partial\bar{r}}(\Psi C)+\frac{\partial}{\partial t}(\Psi E) = 0,
\end{equation}
then the differential from $d\bar{t}=(\Psi C)dt-(\Psi E)d\bar{r}$ is exact.
\end{lem}
\begin{proof}
Suppose $\Psi$ satisfies (\ref{lem_exact_differential_eqn}).  By the following computation
\begin{equation}
curl(\mathbf{v}) =
\AutoAbs{\begin{array}{ccc}
\hat{\mathbf{i}}&\hat{\mathbf{j}}&\hat{\mathbf{k}}\\
\partial_t & \partial_{\bar{r}} & \partial_z \\
\Psi C & -\Psi E & 0
\end{array}}
=\left(\frac{\partial}{\partial\bar{r}}(\Psi C)-\frac{\partial}{\partial t}(-\Psi E)\right)\hat{\mathbf{k}}=\mathbf{0}
\end{equation}
this assumption is equivalent to $\Psi$ satisfying $curl(\mathbf{v})=\mathbf{0}$, where $\mathbf{v}=(\Psi C,-\Psi E, 0)$ is a vector in 3-space.  To simplify notation, define functions $g(t,\bar{r})\equiv\Psi E$ and $h(t,\bar{r})\equiv-\Psi E$.  Since $curl(\mathbf{v})=\mathbf{0}$, there exists a function $f(t,\bar{r})$ such that
\begin{equation}\label{f_partials}
\frac{\partial f}{\partial t}=\Psi C\equiv g\phantom{4} \text{ and }\phantom{4} \frac{\partial f}{\partial \bar{r}}=-\Psi E\equiv h.
\end{equation}

Let $\mathcal{C}$ be a curve in $(t,\bar{r})$ space parameterized by $\alpha$, denoted by $\mathcal{C}=(t(\alpha),\bar{r}(\alpha))$, going from the point $P_1=(t(\alpha_1),\bar{r}(\alpha_1))$ to the point $P_2=(t(\alpha_2),\bar{r}(\alpha_2))$, shown in Figure \ref{fig:t_r_bar_curve}.

\begin{figure}[!h]
\begin{pspicture}(5,4)(0,-0.5)
%\psgrid
%t and r_bar axises
\psline[linewidth=2pt]{->}(0,0)(5,0)
\psline[linewidth=2pt]{->}(0,0)(0,4)
%label the axises
\rput(-0.4,4){$t$}
\rput(5,-0.4){$\bar{r}$}
%draw the curve
\pscurve(1,1)(2.0,2.5)(3.2,2.2)(4,3)
\psdots(1,1)(4,3)
%label the curve and endpoints
\rput(3.2,1.8){$\mathcal{C}=(t(\xi),\bar{r}(\xi))$}
\rput(1.3,0.7){$P_1$}
\rput(4.2,3.3){$P_2$}
\end{pspicture}\caption{Arbitrary curve in $(t,\bar{r})$-space}
\label{fig:t_r_bar_curve}
\end{figure}
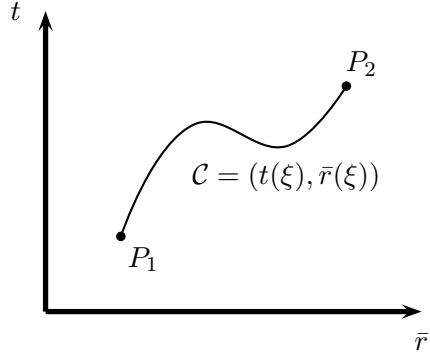

Now take the line integral of the 1-form $d\bar{t}$ along the curve $\mathcal{C}$ to obtain
\begin{equation}
\int_\mathcal{C}d\bar{t}=\int^{P_2}_{P_1}g dt + hd\bar{r} =\int^{\alpha_2}_{\alpha_1}\left[g(t(\alpha),\bar{r}(\alpha))\frac{dt}{d\alpha} + h(t(\alpha),\bar{r}(\alpha))\frac{d\bar{r}}{d\alpha}\right]d\alpha
\end{equation}
Using the chain rule along with (\ref{f_partials}), we have
\begin{equation}
\int_\mathcal{C}d\bar{t}=\int^{\alpha_2}_{\alpha_1}\frac{df}{d\alpha}d\alpha = f(t(\alpha_1),\bar{r}(\alpha_1)))-f(t(\alpha_2),\bar{r}(\alpha_2))).
\end{equation}
Hence, $\int_\mathcal{C}d\bar{t}$ is only dependent on the endpoints of the curve $\mathcal{C}$ which means $\int_\mathcal{C}d\bar{t}$ is independent of path; therefore, the 1-form $d\bar{t}$ is closed.  Since in subsequent chapters we only consider a convex domain $D=\{r_{min}\leq x\leq r_{max},t\geq0\}$, Poincare's Lemma \cite{rudi} states a closed 1-form on a convex domain is exact, proving the claim.
\end{proof}

\section{The TOV Metric}
\label{sec:tov_metric}
For the outer TOV solution in our shock wave simulation, we use the general relativistic static isothermal sphere with the equation of state $p=\frac{1}{3}\rho$ for the pure radiation phase as derived by Smoller and Temple in \cite{smolte}.  For completeness we summerize the results required here.  This TOV metric has the form
\begin{equation}\label{ch2_tov_in_ssc}
ds^2=-B(\bar{r})d\bar{t}^2+\left(\frac{1}{1-\frac{2\mathcal{G}M(\bar{r})}{\bar{r}}}\right)d\bar{r}^2 +\bar{r}^2d\Omega^2.
\end{equation}
The time metric component $B$ takes the form
\begin{equation}
B(\bar{r})=B_0(\bar{r})^{\frac{4\sigma}{1+\sigma}},
\end{equation}
where $\sqrt{\sigma}$ is the speed of sound, and the mass $M$ function is given as
\begin{equation}\label{ch2_tov_mass}
M(\bar{r})=4\pi\gamma\bar{r},
\end{equation}
where the parameter $\gamma$ is a constant dependent on the equation of state constant,
\begin{equation}\label{ch2_gamma}
\gamma=\frac{1}{2\pi\mathcal{G}}\left(\frac{\sigma}{1+6\sigma+\sigma^2}\right),
\end{equation}
which agrees with equation (3.4) in \cite{smolte}.
The fluid variables in the TOV metric are
\begin{equation}\label{ch2_tov_fluid_vars}
\rho(\bar{r})=\frac{\gamma}{\bar{r}^2},\phantom{4444}
v=0,
\end{equation}
where the velocity is zero since the TOV metric is static.

Unlike the FRW metric (\ref{ch2_frw_metric}), the TOV metric is already in standard Schwarzschild coordinates (\ref{ch2_ssc}) with
\begin{equation}
A(\bar{r}) = 1-\frac{2\mathcal{G}M(\bar{r})}{\bar{r}},\phantom{4444} B(\bar{r}) = B_0(\bar{r})^{\frac{4\sigma}{1+\sigma}}.
\end{equation}
both independent of the time coordinate $\bar{t}$.  Notice since
\begin{equation}
\frac{2\mathcal{G}M(\bar{r})}{\bar{r}} = 8\pi\mathcal{G}\gamma,
\end{equation}
this equality simplifies the $A$ metric component to be
\begin{equation}
A(\bar{r}) = 1-8\pi\mathcal{G}\gamma,
\end{equation}
a constant, independent of $\bar{r}$, for constant $\sigma$.

\section{One Parameter Family of Shock Wave Solutions}
\label{sec:matching_metrics}
With both the FRW and TOV metrics in standard Schwarzschild coordinates, we discuss the matching that stitches the two together.  For notation purposes, we use $A_{FRW}$ and $B_{FRW}$ to denote metric components of the FRW metric in standard Schwarzschild coordinates for the general form (\ref{frw2_in_ssc_with_eta}), regardless of which integrating factor is chosen, and $A_{TOV}$ and $B_{TOV}$ to denote the corresponding components for the TOV metric (\ref{ch2_tov_in_ssc}).  Assume the $(\bar{t},\bar{r})$ coordinates, describing the TOV and FRW metrics in standard Schwarzschild coordinates, represent a single coordinate system where both metrics are matched continuously across a shock surface (i.e. $A_{FRW}=A_{TOV}$ and $B_{FRW}=B_{TOV}$).  Matching the $A$ metric component gives us
\begin{equation}\label{matching_a}
A_{FRW}=1-v^2=1-\frac{2\mathcal{G}M}{\bar{r}}=A_{TOV}.
\end{equation}
Recall for both FRW metrics, FRW-1 and FRW-2, the relationship $v=\bar{r}/2t$ holds.  Using this relationship along with the density equation (\ref{frw_density}), the square of the velocity becomes
\begin{equation}\label{v_squared}
v^2=\frac{\kappa}{3}\rho\bar{r}^2.
\end{equation}
Substituting (\ref{v_squared}) into (\ref{matching_a}) and rearranging produces
\begin{equation}
M=\frac{\kappa}{6\mathcal{G}}\rho\bar{r}^3.
\end{equation}
By taking $c=1$ and using equation (\ref{ch2_kappa}) to replace $\kappa$, this equation simplifies to
\begin{equation}\label{shock_surface}
M(\bar{r})=\frac{4\pi}{3}\rho(t)\bar{r}^3.
\end{equation}
The equation (\ref{shock_surface}) defines the shock surface $\bar{r}=\bar{r}(t)$ implicitly in the $(t,\bar{r})$ coordinates.  To obtain the corresponding curve in the original FRW coordinates $(t,r)$, the substitution of $\bar{r}=R(t)r$ into (\ref{shock_surface}) is made to define the shock surface $r=r(t)$.  Interestingly, equation (\ref{shock_surface}) is independent of the integrating factor $\Psi$, meaning the shock surface is the same curve in $(t,\bar{r})$ space regardless of the integrating factor chosen.

To obtain $\Psi$, the other metric component must be matched at the shock surface $\bar{r}(t)$,
\begin{equation}\label{matching_b}
B_{FRW}=\frac{1}{\Psi^2(1-v^2)}=B_0(\bar{r})^{\frac{4\sigma}{1+\sigma}}=B_{TOV}.
\end{equation}
This matching allows us to finish the definition of $\Psi$ by defining the integrating factor constant $\Psi_0$ once the TOV time scale factor $B_0$ is chosen.  In the simulation setup, we reverse the roles of these constants; we choose $\Psi_0$ such that the coordinate speed of light (i.e. $\sqrt{AB}$) is one on the FRW side, for both FRW-1 and FRW-2, and equation (\ref{matching_b}) defines $B_0$.  Notice by choosing $\Psi_0$, which defines the rate at which time progresses uniformly across the matched metric, there remains one free parameter $\bar{r}$ or $t$ related by the shock surface equation (\ref{shock_surface}).  It is in this sense there exists a one parameter family of shock wave solutions.  In this paper, we consider the initial position of the discontinuity $\bar{r}$ as our parameter.  With the equations laid out for the FRW-1, FRW-2, TOV, and the matched metrics in standard Schwarzschild coordinates, we posses the required information for our locally inertial Godunov method, discussed in the next chapter.
   \chapter[%
      Short Title of 4th Ch.
   ]{%
      Locally Inertial Godunov Method
   }%
   \label{ch:frac_god_method}
This chapter is dedicated to the algorithm for the locally inertial Godunov method featuring dynamical time dilation.  Although in this thesis, the method is used to simulate the family of shock waves identified in Chapter \ref{ch:family_of_shock_waves}, the method can be applied to simulate general spherically symmetric flows.  This method is a modification of the locally inertial Glimm method by Groah and Temple \cite{groasmte,groate}, and many of the equations within this chapter are taken from their work.  In their paper, Groah and Temple devise the fractional Glimm method to prove the existence of shock wave solutions to the spherically symmetric Einstein equations for a perfect fluid.  This method is a technique to evolve solutions from initial profiles for the conserved quantities $u(t,x)$ and the metric $\mathbf{A}(t,x)$ in standard Schwarzschild coordinates satisfying the Einstein equations.  Recall, standard Schwarzschild coordinates take the form

\begin{equation}\label{ch3_ssc}
ds^{2} =-B(\bar{t},\bar{r})d\bar{t}^2+\frac{1}{A(\bar{t},\bar{r})}d\bar{r}^2+\bar{r}^{2}d\Omega^2.
\end{equation}
Here, we modify the locally inertial Glimm method by replacing the random choice Glimm step in Groah and Temple's work with an averaging Godunov step to obtain more consistent and less jagged solutions, with the ultimate goal of simulating these shock wave solutions in General Relativity.  Also, the purpose of the Groah and Temple construction is to prove an existence theorem; therefore, the development lacks some details needed to numerically construct the solution.  A goal of this chapter is thus to fill in these details, like an algorithm to find the middle state to the Riemann problem.  Our method can be interpreted as a locally inertial scheme, in the sense it exploits the locally flat character of spacetime.  Each grid cell is considered a locally flat frame, and we handle the time dilation between frames, by choosing a reference frame relative to which time can be synchronized.  The frame chosen is the one in which the factor $\sqrt{AB}$ is one.  We denote this reference frame as a {\it unitary frame}.  We also denote the factor $\sqrt{AB}$ as the {\it coordinate speed of light} to represent the rate at which light travels relative to the unitary frame.  Throughout this paper, we will refer to it as the speed of light, not to be confused with the speed of light constant $c=1$, and it will be clear from context which one is being used.  Note that for any single frame, the metric component $B$ (\ref{ch3_ssc}) can be rescaled by a change in the time coordinate to make the (coordinate) speed of light one in that frame, so the reference frame is quite arbitrary and not so important in tracking time.  However, it is important what this reference frame's speed of light is relative to the other frames around it, causing time dilation between the frames.  These ideas will be explored and expanded throughout the chapter.

In Section \ref{sec:ivp_special_relativity} we discuss the initial value problem in Special Relativity.  In this section, we study solving the Riemann problem in Minkowski spacetime, which is a unitary frame.  Solving the Riemann problem is at the heart of the Godunov step of our method, and its importance cannot be overstated.  Section \ref{sec:time_dilation} explains time dilation and how it affects the averaging of our Godunov step, extending our ability to perform this step on non-unitary frames.  With the background material set, we conclude in Section \ref{sec:frac_god_method} by stating the algorithm of the locally inertial Godunov method.  This method formulates the solution inductively and has four major steps: a Riemann problem step, a Godunov step (with time dilation), an ODE step, and an update step.

\section{The Initial Value Problem in Special Relativity}
\label{sec:ivp_special_relativity}
In this section we develop solutions to the relativistic compressible Euler equations in flat Minkowski spacetime for the case $p=\sigma\rho$ with a constant $\sigma$
\begin{equation}\label{minkowski_conservation_law}
\begin{split}
&\frac{\partial}{\partial t}\left\{\rho[\left(\frac{\sigma+c^2}{c^2}\right)\frac{v^2}{c^2-v^2}+1]\right\} + \frac{\partial}{\partial x}\left\{\rho(\sigma+c^2)\frac{v}{c^2-v^2}\right\} = 0,\\
&\frac{\partial}{\partial t}\left\{\rho(\sigma+c^2)\frac{v}{c^2-v^2}\right\} +\frac{\partial}{\partial x}\left\{\rho[(\sigma+c^2)\frac{v^2}{c^2-v^2}+\sigma]\right\}= 0,
\end{split}
\end{equation}
together with the initial conditions
\begin{equation}\label{minkowski_conservation_law_ics}
\rho(0,x) = \rho_0(x), \phantom{44}v(0,x)=v_0(x).
\end{equation}
Note, this equation (\ref{minkowski_conservation_law}) is our conservation law ((\ref{conservation_law00}) and (\ref{conservation_law01})) without the source term where $\sqrt{AB}=1$.  This problem is a specific case of the initial value problem for a general system of nonlinear hyperbolic conservation laws in the sense of Lax \cite{smol},
\begin{equation}\label{general_system}
\begin{split}
u_t+(F&(u))_x=0,\\
u(0,x&)=u_0(x).
\end{split}
\end{equation}
In our case, we have
\begin{equation}\label{conserved_quantities_defn}
u\equiv (u^0,u^1)=\left(\rho[\left(\frac{\sigma+c^2}{c^2}\right)\frac{v^2}{c^2-v^2}+1],
\phantom{4}\rho(\sigma+c^2)\frac{v}{c^2-v^2}\right),
\end{equation}
and
\begin{equation}\label{flux_defn}
F(u)\equiv (F^0,F^1)=\left(\rho(\sigma+c^2)\frac{v}{c^2-v^2},
\phantom{4}\rho[(\sigma+c^2)\frac{v^2}{c^2-v^2}+\sigma]\right).
\end{equation}
To distinguish between the two sets of variables, we refer to the pair $(\rho,v)$ as the {\it fluid variables} and to the pair $(u^0,u^1)$ in (\ref{conserved_quantities_defn}) as the {\it conserved quantities}.  Both of these variables play an important role in the locally inertial Godunov method.  The conserved quantities are needed in implementing the Godunov step, and the fluid variables are needed to solve the Riemann problem along with giving us more physical meaning behind our simulation.  Fortunately, Groah and Temple showed \cite{groasmte} there is a 1 - 1 correspondence between the fluid variables and the conserved quantities, giving us the following result
\begin{prop}\label{1_1_fluid_vars_conserved}
The mapping $(\rho,v)\rightarrow(u^0,u^1)$ is 1 - 1, and the Jacobian determinant of this mapping is both continuous and non-zero in the region $\rho>0$, $\AutoAbs{v}<c$.
\end{prop}
In the locally locally inertial Godunov method, we constantly transfer back and forth between the conserved quantities and the fluid variables, so the inversion of the mapping $(\rho,v)\rightarrow(u^0,u^1)$ defined in (\ref{conserved_quantities_defn}) is necessary to find the fluid variables as a function of the conserved quantities, stated in the following corollary to Proposition \ref{1_1_fluid_vars_conserved}
\begin{cor}
The mapping $(u^0,u^1)\rightarrow(\rho,v)$ takes the form
\begin{equation}\label{v_fn_of_conserved_quantities}
v(u^0,u^1)=\frac{c^2}{2\sigma u^1}\{(\sigma+c^2)u^0-\sqrt{(\sigma+c^2)^2(u^0)^2-4\sigma(u^1)^2}\}.
\end{equation}
\begin{equation}\label{rho_fn_of_conserved_quantities}
\rho(u^0,u^1)=\frac{(c^2-v^2)u^1}{(\sigma+c^2)v}.
\end{equation}
\end{cor}
\begin{proof}
Using the definition of the conserved quantities (\ref{conserved_quantities_defn}), both are solved as equations for $\rho$
\begin{equation}\label{rho1}
\rho=\frac{u^0}{\left(\frac{\sigma+c^2}{c^2}\right)\frac{v^2}{c^2-v^2}+1}=\frac{u^0(c^2-v^2)}{\frac{\sigma}{c^2}v^2+c^2},
\end{equation}
and
\begin{equation}\label{rho2}
\rho=\frac{(c^2-v^2)u^1}{(\sigma+c^2)v}.
\end{equation}
Setting these equations, (\ref{rho1}) and (\ref{rho2}), equal to each other and simplifying gives us the following quadratic equation in $v$
\begin{equation}
\frac{\sigma}{c^2}u^1v^2-(\sigma+c^2)u^0v+c^2u^1=0,
\end{equation}
so there exists two candidates for $v$
\begin{equation}\label{2_solns_for_v}
v=\frac{c^2}{2\sigma u^1}\{(\sigma+c^2)u^0\pm\sqrt{(\sigma+c^2)^2(u^0)^2-4\sigma(u^1)^2}\}.
\end{equation}
In order to determine which solution to use, we normalize the speed of light (i.e. $c=1$) and rewrite (\ref{2_solns_for_v}) as
\begin{equation}
v=\frac{u^0(\sigma +1)}{2u^1\sigma}\left(1\pm\sqrt{1-\frac{4(u^1)^2\sigma^2}{(u^0)^2(\sigma+1)^2}}\right).
\end{equation}
Since $u^0>u^1$ and $\sigma+1>\sigma$, the following inequalities hold $\frac{u^0(\sigma +1)}{2u^1\sigma}>1$ and $\frac{4(u^1)^2\sigma^2}{(u^0)^2(\sigma+1)^2}<1$.  In order to keep $v$ less than the speed of light (i.e. $v<1$), the minus sign is taken in (\ref{2_solns_for_v}) to obtain (\ref{v_fn_of_conserved_quantities}). Substituting this $v$ into (\ref{rho2}) gives us $\rho$ as (\ref{rho_fn_of_conserved_quantities}).
\end{proof}

The main focus of this section is to solve the Riemann problem for system (\ref{minkowski_conservation_law}).  The Riemann problem is the initial value problem with initial data $u_0(x)$ of a pair of constant states separated by a jump discontinuity at $x=0$,
\begin{equation}\label{rp_ics}
u_0(x) = \left\{
\begin{array}{ll}
u_L & x<0\\
u_R & x>0.
\end{array}
\right.
\end{equation}
Note, in light of Proposition \ref{1_1_fluid_vars_conserved}, the conserved quantities $u_L$ and $u_R$ are uniquely determined by the fluid variables $(\rho_L,v_L)$ and $(\rho_R,v_R)$, and the symbol $u$ can be interpreted as either set of variables throughout this paper.

A general theorem of Lax \cite{smol} states for any system of conservation laws (\ref{general_system}) which is strictly hyperbolic and genuinely nonlinear in each characteristic field, the Riemann problem for this system has a unique solution in the class of elementary waves if $u_R$ and $u_L$ are sufficiently close.  The proof for this result is a constructive one, relying on the structure of the state space.  More specifically, given a point $u_L$ in u-space, there exists a family of i-wave curves connecting the point $u_L$ to any other sufficiently close point $u_R$ by traversing a 1-wave curve followed by a 2-wave curve, where the middle state $u_M$ is the intersection of these two curves.  Depending on the direction along the 1-wave curve, a shock or rarefaction wave is between the states $u_L$ and $u_M$, and similarly for the 2-wave curve, the direction taken determines whether a shock or rarefaction wave is between the states $u_M$ and $u_R$.  As a side remark, each i-wave curve has second order contact at the point $u_L$ (i.e. their first two derivatives are equal at this point) \cite{smol}.  Figure \ref{fig:p_system_curves} gives us a sample of the web-like structure for these family of curves for a 2 system of conservation laws referred to as the (non-relativistic) $p$-system \cite{smol}, for $p=1$.  In this figure, the red dot represents the left state $u_L$.  The 1-shock curve is represented by the blue/red graph, and the 1-rarefaction is represented by the red/cyan graph.  The 2-shock and 2-rarefaction curves, originating from the 1-wave curve, are represented by green and yellow graphs, respectively, with the red stripped curves representing the ones emanating from the left state.  One can interpret the coordinate system of wave curves defined for each $u_L$ as giving us a road map providing directions from our starting point $u_L$ to our destination $u_R$: the 1-wave curve is the first street traveled and 2-wave curve next street traveled, and the middle state $u_M$ is where the streets traveled intersect.  The solution becomes the left state connected by a 1-wave to the middle state which is connected by a 2-wave to the right state.  This result is only a local one because for certain conservation laws (including $\gamma$-law gases), not all the points in u-space can be connected by this web structure, due to the formation of the so-called "vacuum" states \cite{smol}.

\begin{figure}
\begin{center}
\includegraphics[width=\textwidth]{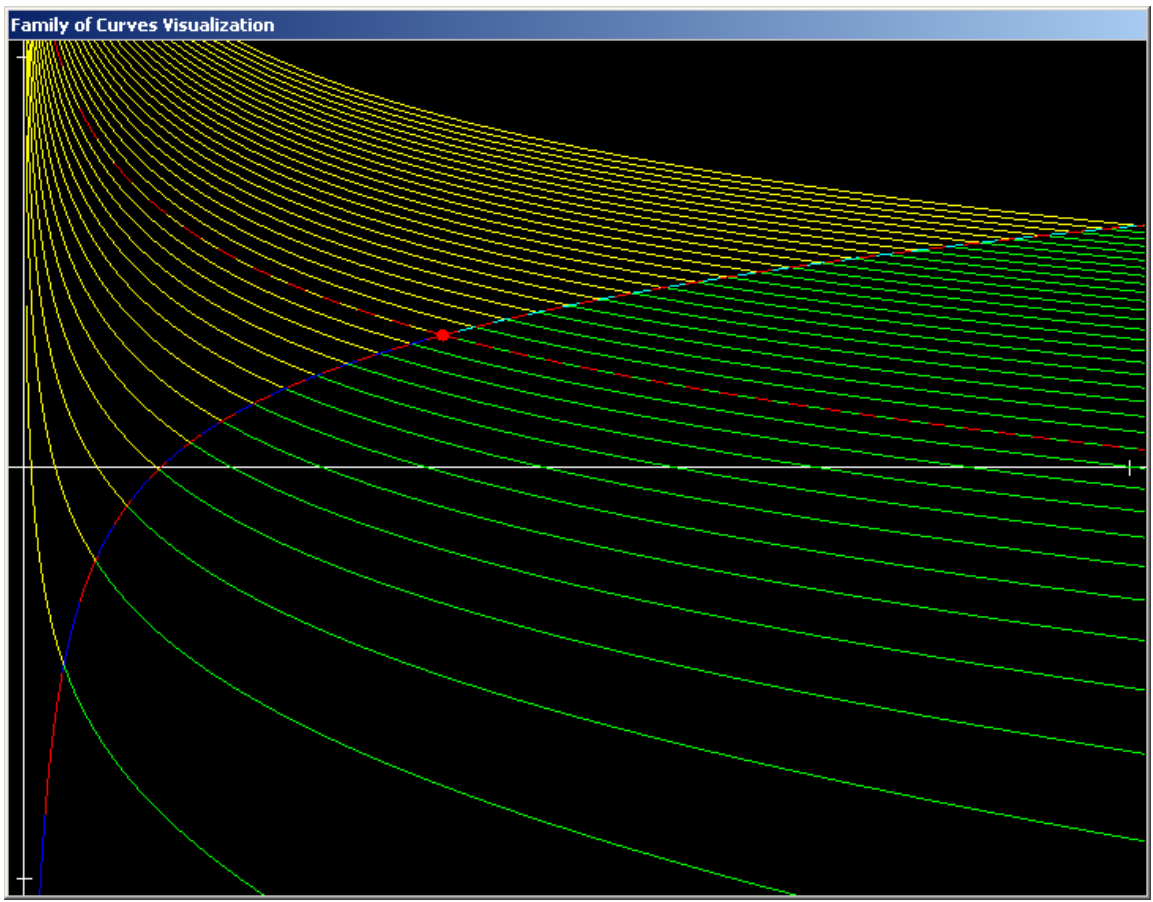}
\end{center}
\caption{Family of curves for the (non-relativistic) p-system}
\label{fig:p_system_curves}
\end{figure}

For other systems, this local connectivity of points in u-space can be extended to a global connectivity of points, where the "vacuum" states do not appear. Fortunately, for our system of interest (\ref{minkowski_conservation_law}), when $p=\sigma\rho$, the "vacuum" states do not appear and our i-wave curves cover the entire u-space.  Following Smoller and Temple \cite{smolte6}, Groah and Temple extend in \cite{groasmte} this general theorem of Lax our system
\begin{thm}
There exists a solution of the Riemann problem for system (\ref{minkowski_conservation_law}) with an equation of state $p=\sigma\rho$, $0<\sqrt{\sigma}<c$, as long as $u_L$ and $u_R$ satisfy
\begin{equation}
\rho_L>0,\phantom{44}\rho_R>0,
\end{equation}
and
\begin{equation}
-c<v_L<c,\phantom{44}-c<v_R<c.
\end{equation}
Moreover, the solution is given by a 1-wave followed by a 2-wave, satisfies $\rho>0$, and all speeds are bounded by $c$.  This solution is unique in the class of rarefaction and admissible shock waves.
\end{thm}

The main goal of this section is to use the details within the proof of this result to construct the explicit solution to the Riemann problem for any left and right states.  This section is organized into three subsections.  Subsection \ref{subsec:prelim} covers preliminary material needed to solve the Riemann problem.  Finding the middle state along with the type of elementary waves connecting it to the left and right states is discussed in Subsection \ref{subsec:middle_state}.  Bringing all these details together, Subsection \ref{subsec:solving_rp} explains how to solve the Riemann problem at any point of interest.
\subsection{Preliminaries}
\label{subsec:prelim}
In this subsection, equations needed to solve our Riemann problem (\ref{minkowski_conservation_law}) and (\ref{rp_ics}) are presented.  Along with the fluid variables, there is another set of variables, called the Riemann invariants, important in solving the Riemann problem efficiently.  The Riemann invariants are functions which are constant along the integral curves or rarefaction curves in the state space \cite{smol}.  The Riemann invariants $r$ and $s$ for our system (\ref{minkowski_conservation_law}) are
\begin{equation}\label{riemann_invariant_r}
r(\rho,v)=\frac{1}{2}\ln\left(\frac{c+v}{c-v}\right)-\sqrt{\frac{K}{2}}\ln(\rho),
\end{equation}
\begin{equation}\label{riemann_invariant_s}
s(\rho,v)=\frac{1}{2}\ln\left(\frac{c+v}{c-v}\right)+\sqrt{\frac{K}{2}}\ln(\rho),
\end{equation}
where
\begin{equation}
K=\frac{2\sigma c^2}{(\sigma+c^2)^2}.
\end{equation}
There is a typo in the statement of the Riemann invariants in the Groah and Temple paper \cite{groasmte} in equations (2.5.73), (2.5.74), (4.2.12), and (4.2.13).  As well as moving freely between the conserved quantities and the fluid variables, the ability to go back an forth between the fluid variables and the Riemann invariants is needed within the locally inertial Godunov method.

The fluid variables $(\rho,v)$ are recovered from the Riemann invariants $(r,s)$ by algebraic manipulation of (\ref{riemann_invariant_r}) and (\ref{riemann_invariant_s}).  More specifically, subtracting (\ref{riemann_invariant_r}) from (\ref{riemann_invariant_s}) to solve for the density
\begin{equation}\label{rho_fn_of_r_s}
\rho(r,s)=\exp\left\{\frac{s-r}{\sqrt{2K}}\right\},
\end{equation}
and adding (\ref{riemann_invariant_r}) to (\ref{riemann_invariant_s}) to solve for the velocity
\begin{equation}\label{v_fn_of_r_s}
v(r,s)=-\frac{c(1-e^{s+r})}{1+e^{s+r}}.
\end{equation}

To solve the Riemann problem, it is necessary to know the speed of the shock waves in the solution to determine which side of the discontinuity the point of interest is located.  The speed for the 1-shock is the following function of beta $\beta$
\begin{equation}\label{1_shock_speed}
s_1=c\sqrt{\frac{f_+(\beta)+\frac{\sigma}{c^2}}{f_+(\beta)+\frac{c^2}{\sigma}}},
\end{equation}
while the speed for the 2-shock is
\begin{equation}\label{2_shock_speed}
s_2=c\sqrt{\frac{f_-(\beta)+\frac{\sigma}{c^2}}{f_-(\beta)+\frac{c^2}{\sigma}}}.
\end{equation}
These shock speeds, (\ref{1_shock_speed}) and (\ref{2_shock_speed}), are calculated in a frame where the particle velocity $v$ is zero.  To obtain these quantities in an arbitrary frame, the Lorentz transformation law for velocities must be applied.  The transformation law is given as follows:  if in a Lorentz transformation, the barred frame moves with velocity $v$ measured in the unbarred frame with $s_i$ as the speed of the i-shock wave measured in the barred frame, and if $s$ denotes the speed of the shock in the unbarred frame, then
\begin{equation}\label{lorentz_velocity_xform_law}
s=\frac{v+s_i}{1+\frac{vs_i}{c^2}}.
\end{equation}

The eigenvalues $\lambda_i(\rho,v)$ of the system (\ref{minkowski_conservation_law}) are used to determine the speed of the rarefaction waves.  These eigenvalues are $\lambda_1 = -\sqrt{\sigma}$ and $\lambda_2 = \sqrt{\sigma}$ when the particle velocity is zero.  Using the Lorentz transformation law (\ref{lorentz_velocity_xform_law}) for velocity $v$, the eigenvalues become
\begin{equation}\label{eigenvalue1}
\lambda_1=\frac{v-\sqrt{\sigma}}{1-\frac{\sqrt{\sigma} v}{c^2}},
\end{equation}
and
\begin{equation}\label{eigenvalue2}
\lambda_2=\frac{v+\sqrt{\sigma}}{1+\frac{\sqrt{\sigma} v}{c^2}}.
\end{equation}

To find the solution of the Riemann problem within a rarefaction wave, we need to solve for the fluid variables based on the eigenvalues and the Riemann invariants.  More precisely, formulas for $v$ as a function of the eigenvalues, $\rho$ as a function of $(s,v)$, and $\rho$ as a function of $(r,v)$ are necessary to solve for a state in a rarefaction wave.  A quick calculation on (\ref{eigenvalue1}) and (\ref{eigenvalue2}) shows $v$ is a function of the first eigenvalue
\begin{equation}\label{inverse_eigenvalue1}
v(\lambda_1)=\frac{\lambda_1+\sqrt{\sigma}}{1+\frac{\sqrt{\sigma}\lambda_1}{c^2}},
\end{equation}
or is a function of the second eigenvalue
\begin{equation}\label{inverse_eigenvalue2}
v(\lambda_2)=\frac{\lambda_2-\sqrt{\sigma}}{1-\frac{\sqrt{\sigma}\lambda_2}{c^2}}.
\end{equation}
Another quick calculation on (\ref{riemann_invariant_r}) and (\ref{riemann_invariant_s}) provides us with $\rho$ as a function of $(r,v)$
\begin{equation}\label{rho_fn_of_r_v}
\rho(r,v)=\exp\left\{-\sqrt{\frac{2}{K}}\left(r-\frac{1}{2}\ln\left\{\frac{c+v}{c-v}\right\}\right)\right\},
\end{equation}
and $\rho$ as a function of $(s,v)$
\begin{equation}\label{rho_fn_of_s_v}
\rho(s,v)=\exp\left\{\sqrt{\frac{2}{K}}\left(s-\frac{1}{2}\ln\left\{\frac{c+v}{c-v}\right\}\right)\right\}.
\end{equation}
\subsection{Finding the Middle State}
\label{subsec:middle_state}
This subsection describes a numerical algorithm to find the middle state and the associated waves to the Riemann problem for our system (\ref{minkowski_conservation_law}).  To accomplish this task, not only does one need to determine in which region the right state is located, governing the set of curve equations to use, but one also needs to solve these curve equations which cannot be solved explicitly.  With this in mind, finding this middle state poses quite a challenge with the fluid variables $(\rho,v)$.  Fortunately, this problem is easier in the Riemann invariant coordinate system or the $rs$-plane.  This coordinate system simplifies the i-wave curves in the state space because the rarefaction curves become straight lines, resulting in a clearer segregation of the state space.  Another advantage is the i-wave curves are geometrically invariant across the $rs$-plane, so the shape of the wave curves are independent of the base point $(r_L,s_L)$, implying an i-wave curve at a point in this plane can be mapped by rigid translation onto any other point.  Within this subsection, we deal with the Riemann invariants as our variables, keeping in mind the goal is to obtain the middle state in the fluid variables.

To distinguish between the different coordinate systems, we denote the fluid variables $(\rho,v)$ by lower case $u$ and the Riemann invariants $(r,s)$ by upper case $U$.  More specifically, define
\begin{equation}
\begin{split}
U_R&\equiv(r_R(u_R),s_R(u_R))= (r_R(\rho_R,v_R),s_R(\rho_R,v_R)) \\ U_L&\equiv(r_L(u_L),s_L(u_L))= (r_L(\rho_L,v_L),s_L(\rho_L,v_L)),
\end{split}
\end{equation}
where the transformations (\ref{riemann_invariant_r}) and (\ref{riemann_invariant_s}) are used.  In the $rs$-plane, the 1-shock curve $S_1$ for the system is given by the following parametrization with respect to the $\beta$, $0\leq\beta < \infty$:
\begin{equation}\label{1_shock_curve}
\begin{split}
\Delta r&=r-r_L=-\frac{1}{2}\ln\{f_+(2K\beta)\}-\sqrt{\frac{K}{2}}\ln\{f_+(\beta)\} \equiv S^r_1(\beta), \\
\Delta s&=s-s_L=-\frac{1}{2}\ln\{f_+(2K\beta)\}+\sqrt{\frac{K}{2}}\ln\{f_+(\beta)\} \equiv S^s_1(\beta),
\end{split}
\end{equation}
and the 2-shock curve $S_2$ is given by:
\begin{equation}\label{2_shock_curve}
\begin{split}
\Delta r&=r-r_L=-\frac{1}{2}\ln\{f_+(2K\beta)\}-\sqrt{\frac{K}{2}}\ln\{f_-(\beta)\} \equiv S^r_2(\beta), \\
\Delta s&=s-s_L=-\frac{1}{2}\ln\{f_+(2K\beta)\}+\sqrt{\frac{K}{2}}\ln\{f_-(\beta)\} \equiv S^s_2(\beta),
\end{split}
\end{equation}
where
\begin{equation}
f_\mp(\beta)\equiv 1+\beta\left\{1\mp\sqrt{1+\frac{2}{\beta}}\right\},
\end{equation}
and
\begin{equation}
\beta\equiv\beta(v,v_L)=\frac{(\sigma+c^2)^2}{2\sigma^2}\frac{(v-v_L)^2}{(c^2-v^2)(c^2-v_L^2)}.
\end{equation}

In this coordinate system, both rarefaction curves are straight lines parallel to the coordinate axises along the positive directions in the $rs$-plane.  More precisely, the 1-rarefaction curve $R_1$ is:
\begin{equation}\label{1_rarefaction_curve}
\begin{split}
\Delta r&=r-r_L=\beta, \\
\Delta s&=s-s_L=0,
\end{split}
\end{equation}
and the 2-rarefaction curve $R_2$ is given by:
\begin{equation}\label{2_rarefaction_curve}
\begin{split}
\Delta r&=r-r_L=0 \\
\Delta s&=s-s_L=\beta,
\end{split}
\end{equation}
for the parameter $\beta$, $0\leq\beta < \infty$.  Notice within the equations for the shock and rarefaction curves (\ref{1_shock_curve})-(\ref{2_rarefaction_curve}) the differences $\Delta r$ and $\Delta s$ along the curves only depend on the parameter $\beta$, proving the geometric invariance mentioned earlier.  For notation purposes, we define the following sets based at a point $U_L$ for the shock curves
\begin{equation}\label{shock_set}
S_i(U_L)=\{(r,s):\exists \beta >0\text{ such that } s=s_L+S^s_i(\beta), r=r_L+S^r_i(\beta)\},\phantom{4444}i=1,2,
\end{equation}
and the rarefaction curves
\begin{equation}\label{rarefaction_set}
\begin{split}
R_1(U_L) =& \{(r,s): s = s_L\text{ and } r\geq r_L\},\\
R_2(U_L) =&\{(r,s): r = r_L\text{ and } s\geq s_L\}.
\end{split}
\end{equation}
Using this knowledge, a simpler family of curves is built to solve for the middle state $U_M$, as seen in Figure \ref{fig:rs_plane_curves}.

\begin{figure}
\begin{center}
\includegraphics[width=\textwidth]{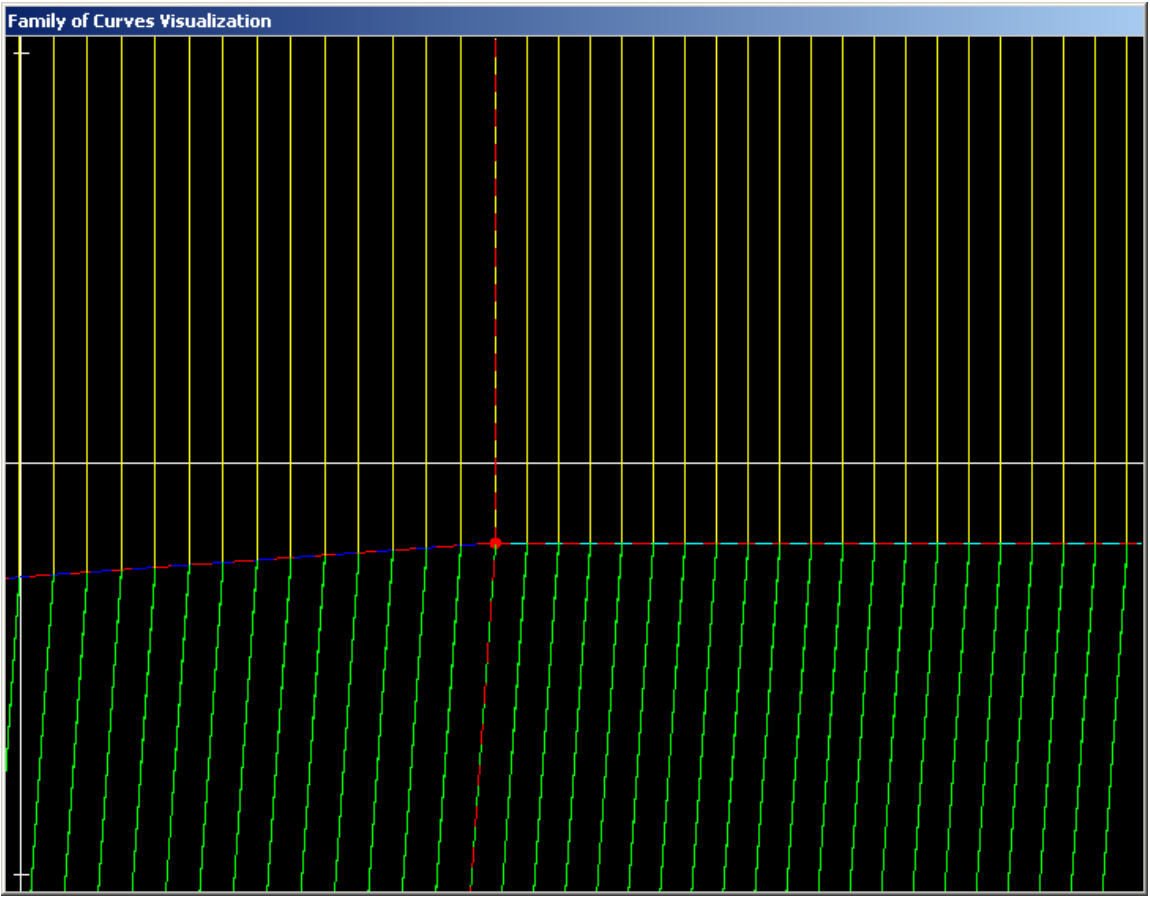}
\end{center}
\caption{Family of curves for our system in the $rs$-plane}
\label{fig:rs_plane_curves}
\end{figure}

Using these curves, the $rs$-plane is segregated into four disjoint open regions I, II, III, and IV as depicted in Figure \ref{fig:rs_plane}.  In each region, a possible right state is placed and denoted by $U^i$ for $i=1,2,3,4$.  This figure also shows the 1-wave and 2-wave curves emanating from $U_L$ as the thick lines with the 2-wave curves associated with the points $U^i$ for $i=1,2,3,4$ emanating from the 1-wave curve as dashed lines. Region I corresponds to connecting the point $U_L$ to the point $U^1$ by 1-rarefaction wave followed by a 2-shock wave.  Region II corresponds to the 1-shock/2-shock wave case, while region III and region IV corresponds to 1-shock/2-rarefaction and 1-rarefaction/2-rarefaction wave cases, respectively.  Since the wave curves are geometrically invariant, we only consider the change from the right to left states, or the quantities $\Delta r\equiv r_R-r_L$ and $\Delta s\equiv s_R-s_L$ where the right state is $U^i$ for $i=1,2,3,4$, depending on the region under consideration.  The middle state $U_M$ for all these cases is found by considering each quadrant in the $rs$-plane separately.  After the quadrant breakdown, we discuss the technique used to find the correct $\beta$(s) needed to solve the shock wave equations (\ref{1_shock_curve}) and (\ref{2_shock_curve}).  Note that the shock curves cannot be solved exactly, and a threshold of $\epsilon$ must be predefined, with the goal of solving the shock curve equations with a maximum error of $\epsilon$.

We start with the point $U^4$ in +/+ quadrant, which is determined by $\Delta r>0$ and $\Delta s>0$.  This quadrant is the same set as region IV, the 1-rarefaction/2-rarefaction case.  Since the rarefaction curves are straight lines, the middle state becomes $r_M=r_R$ and $s_M=s_L$.

Next we turn our attention to the point $U^3$ in the -/+ quadrant, where $\Delta r<0$ and $\Delta s>0$.  This quadrant is a subset of the region III, the 1-shock/2-rarefaction case.  Our goal is to find $\beta$ such that $(r_R,s^*)\in S_1(U_L)$. After finding $\beta$, the middle state will be $(r_M,s_M)=(r_R,s^*)$.

The point $U^1$ in the +/- quadrant, where $\Delta r>0$ and $\Delta s<0$, is similar to $U^3$ in the -/+ quadrant.  This quadrant is a subset of the region I, the 1-rarefaction/2-shock case.  Due to the geometric invariance of the 2-wave curve, we find a $\beta$ such that the point $(r^*,s_R)\in S_2(U_L)$, and the middle state will be $(r_M,s_M)=(r_L + (r_R-r^*),s_L)$.

Finally, consider the point $U^2$ in the -/- quadrant, where $\Delta r<0$ and $\Delta s<0$.  This quadrant is a superset of region II, the 1-shock/2-shock case, so a mechanism is needed to determine when $U^2$ is not in region II but in region I or III as discussed below in the two shock algorithm.  If $U^2$ is in region III, we find $\beta^1$ such that $U_M=(s_M,r_M)\in S_1(U_L)$ and $\beta^2$ such that $(s_R,r_R)\in S_2(U_M)$.  The middle state $U_M$ is, the intersection of the 1-shock curve $S_1(U_L)$ and the 2-shock curve $S_2(U_M)$ crossing the right state $U_R$.

\begin{figure}[!t]
\begin{pspicture}(8,8)(0,-0.5)
%\psgrid
%r and s axises
\psline{<->}(0,4)(8,4)
\psline{<->}(4,0)(4,8)
\rput(8,3.7){$r$}
\rput(3.7,8){$s$}
%draw the i-wave curves
\psline[linewidth=2pt]{->}(4,4)(4,6)
\rput(4.4,5.8){$R_2$}
\psline[linewidth=2pt]{->}(4,4)(6,4)
\rput(5.8,4.3){$R_1$}
\pscurve[linewidth=2pt]{->}(4,4)(3.9,3.2)(3,1.2)(2,0.2)
\rput(2.6,0.2){$S_2$}
\pscurve[linewidth=2pt]{->}(4,4)(3.2,3.9)(1.2,3)(0.2,2)
\rput(0.2,2.6){$S_1$}
%label the quadrants
\rput(6,6){IV}
\rput(2,6){III}
\rput(6,2){I}
\rput(2,2){II}
%draw and label states
\psdots(4,4)(4.5,1.2)(5,5.3)(1.2,5)(1.2,1.2)
\rput(4.4,4.3){$U_M$}
\rput(5,1.2){$U^1$}
\rput(5.5,5.3){$U^4$}
\rput(1.7,5){$U^3$}
\rput(1.7,1.2){$U^2$}
\pscurve[linestyle=dashed](5.5,4)(5.4,3.2)(4.5,1.2)
\pscurve[linestyle=dashed](1.8,3.3)(1.7,2.5)(1.2,1.2)
\psline[linestyle=dashed](5,4)(5,5.3)
\psline[linestyle=dashed](1.2,3)(1.2,5)
\end{pspicture}\caption{Labeling the $rs$-plane}
\label{fig:rs_plane}
\end{figure}
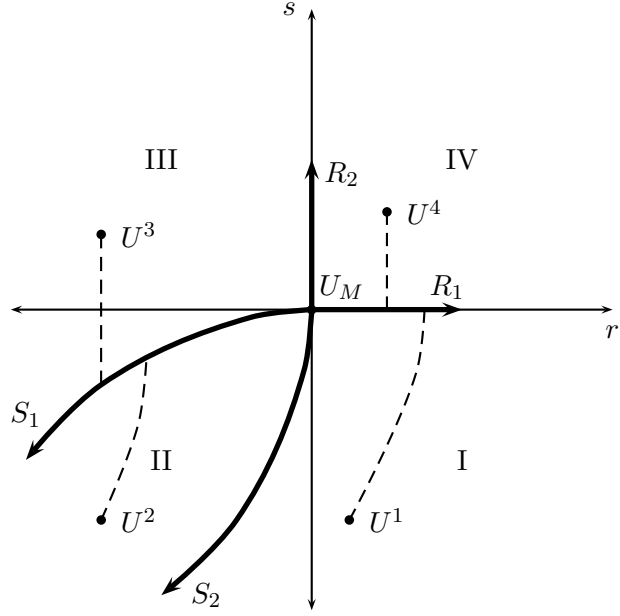

When solving one shock curve equation, as in the +/- or -/+ quadrant, a standard bisection method \cite{burdfa} is used to find $\beta$.  We show how we implement the bisection method for the point $U^3$, where a 1-shock is followed by a 2-rarefaction.   In this case, the bisection method is used to find the $\beta$ that satisfies the equation $\Delta r(\beta)\equiv r_R - (r_L + S^r_1(\beta)) = 0$, and our goal is to find a $\beta$ such that $\Delta r(\beta) < \epsilon$.  The bisection method requires a starting interval $[\beta_{min},\beta_{max}]$ where there exists a $\beta$ within this interval such that $\Delta r(\beta)=0$. The shock curves, $S^r_i(\beta)$ and $S^s_i(\beta)$ for $i=1,2$, are monotone decreasing functions of the variable $\beta$.  To find our interval, an initial guess, like $\beta = 10^5$, is chosen and $\Delta r(\beta)$ is computed.  If $\Delta r(\beta)>0$, our guess is too small (i.e. $r_R>r_L+S^r_1(\beta)$), and $\beta$ is decreased.  The power of our initial guess $\beta=10^k$ is decreased (i.e. $k=4,3,2,1,0,-1,\ldots$) until a $k$ is found such that $\Delta r(10^k)<0$, which means the guess is too big now.  With this $k$, we set $\beta_{max}=10^{k+1}$ and $\beta_{min}=0$.  On the other hand, if for the initial guess $\Delta r(\beta)<0$, it is too big (i.e. $r_R<r_L+S^r_1(\beta)$), and $\beta$ needs to be increased.  The power of our initial guess $\beta=10^k$ is increased (i.e. $k=6,7,8,9,\ldots$) until $k$ is found such that $\Delta r(10^k)>0$, which means our guess is too small now.  With this $k$, we set $\beta_{max}=10^k$ and $\beta_{min}=10^{k-1}$.  Either way, with the interval $[\beta_{max},\beta_{min}]$ established, the bisection method is implemented until a $\beta$ is found where $\Delta r(\beta) < \epsilon$.  For the -/+ quadrant, the same algorithm is used, but for the equation $\Delta s(\beta)\equiv s_R-(s_L+S^s_2(\beta))=0$.

For the point $U^2$ in the -/- quadrant, the pair $(\beta^1,\beta^2)$ solving the equations
\begin{equation}\label{delta_r}
\Delta r(\beta^1,\beta^2)\equiv r_R-(r_L+S^r_1(\beta^1)+S^r_2(\beta^2))=0
\end{equation}
and
\begin{equation}\label{delta_s}
\Delta s(\beta^1,\beta^2)\equiv s_R-(s_L+S^s_1(\beta^1)+S^s_2(\beta^2))=0
\end{equation}
is sought out, and our goal is to find $(\beta^1,\beta^2)$ such that $\Delta r(\beta^1,\beta^2) < \epsilon$ and $\Delta s(\beta^1,\beta^2) < \epsilon$.  To perform the bisection algorithm for both $\beta^1$ and $\beta^2$, intervals $[\beta^1_{min},\beta^1_{max}]$ and $[\beta^2_{min},\beta^2_{max}]$ are needed, just like solving for one shock equation.  Since the equation (\ref{delta_r}) is dominated by the parameter $\beta_1$, in order to get our interval $[\beta^1_{min},\beta^1_{max}]$, we assume $\beta^2=0$ simplifying (\ref{delta_r}) into $\Delta r(\beta^1)\equiv r_R - (r_L + S^r_1(\beta^1)) = 0$ and repeat the same procedure as above, increasing or decreasing the power $k$ in our guess $10^k$, for the parameter $\beta^1$.  Similarly, we repeat this process to find $[\beta^2_{min},\beta^2_{max}]$, by setting $\beta^1=0$ in (\ref{delta_s}) and working with the equation $\Delta s(\beta^2)\equiv s_R-(s_L+S^s_2(\beta^2))=0$.  If either initial guess $10^k$ gets too small, like $k<-20$, then we assume the point $U^2$ is in the wrong region.  Depending on which parameter it is, $\beta_1$ or $\beta_2$, the $U^2$ must be in region I or region III, respectively.  In either case, our problem boils down to only solving one shock curve equation, and it is handled according to the above procedure. Equipped with our initial guess $(\beta^1,\beta^2)$, we compare the errors $\Delta r(\beta^1,\beta^2)$ and $\Delta s(\beta^1,\beta^2)$ against each other.  If $\Delta r(\beta^1,\beta^2)<\Delta s(\beta^1,\beta^2)$, we perform a bisection method step on $\beta^1$; otherwise, we perform the step on $\beta^2$.  We repeat this process until $\Delta r(\beta^1,\beta^2) < \epsilon$ and $\Delta s(\beta^1,\beta^2) < \epsilon$.

\subsection{Solving the Riemann Problem}
\label{subsec:solving_rp}
Being able to find the middle state with the connecting waves enables us to build the solution to our Riemann problem in its entirety.  In particular, for an arbitrary point $(t,x)$, the solution $u(t,x)$ can be determined for any Riemann problem.  This subsection covers the details of finding this solution at an arbitrary point.  We also use this process to build a Riemann problem simulator.

Suppose we have a Riemann problem (\ref{rp_ics}) to our conservation law (\ref{minkowski_conservation_law}) and a point $(t,x)$, and we want to solve for the state $u_*\equiv u(t,x)=(\rho_*,v_*)$ as a solution to this problem.  Converting $u_L$ and $u_R$ over to Riemann invariants, using (\ref{riemann_invariant_r}) and (\ref{riemann_invariant_s}), and implementing the above procedure, the middle state $U_M$ is determined along with the 1-wave connecting to the left state and the 2-wave connecting to the right state.  Converting $U_M$ from Riemann invariant variables to fluid variables using (\ref{rho_fn_of_r_s}) and (\ref{v_fn_of_r_s}) gives us $u_M$.

Without loss of generality, we consider the 1-shock/2-rarefaction case, where the other cases are handled similarly.  Since the solution to the Riemann problem is self-similar, the ratio of our point of interest $s\equiv x/t$ is compared to the wave speeds, usually  done left to right.  We first compare this ratio against the 1-shock speed $s_1$, computed by the transformation law (\ref{lorentz_velocity_xform_law}) with $v=v_L$ and the speed ($s_i$) computed by (\ref{1_shock_speed}) using the $\beta$ found in the middle state algorithm.  If $s<s_1$, then the state proceeds the 1-shock, and $u_*=u_L$ is the solution.  Otherwise, we compute the speeds for the left and right sides of the 2-rarefaction fan.  The left speed $s^L_2$ is computed by (\ref{eigenvalue2}) with $u_M$ as the input, and the right speed $s^R_2$ is computed similarly using $u_R$ instead.  If $s_1<s<s^L_2$, then the state is between the shock and rarefaction waves, and $u_*=u_M$.  If $s^R_2<s$, then the state is beyond the rarefaction wave, and $u_*=u_R$.  If $s^L_2<s<s^R_2$, then the state is within the rarefaction wave, and more work needs to be done to find the solution.  We use the fact that within the rarefaction wave the solution varies smoothly, and every state $u$ between $u_M$ and $u_R$ moves at the speed $\lambda_2(u)$ \cite{smol}.  Since the speed for our state is $\lambda_2(u_*)=x/t$, the inversion of the second eigenvalue (\ref{inverse_eigenvalue2}) gives us the fluid velocity $v_*$.  With the fluid velocity, the density $\rho_*$ is found by (\ref{rho_fn_of_r_v}) where $r_R$ is used since this Riemann invariant is unchanged across the 2-rarefaction wave.  For the rarefaction wave case, our solution becomes $u_*=(\rho_*,v_*)$.  This completes the procedure to find the solution to the Riemann problem $u_*$ for any arbitrary point $(t,x)$ for the 1-shock/2-rarefaction case.

\begin{figure}
\begin{center}
\includegraphics[width=\textwidth]{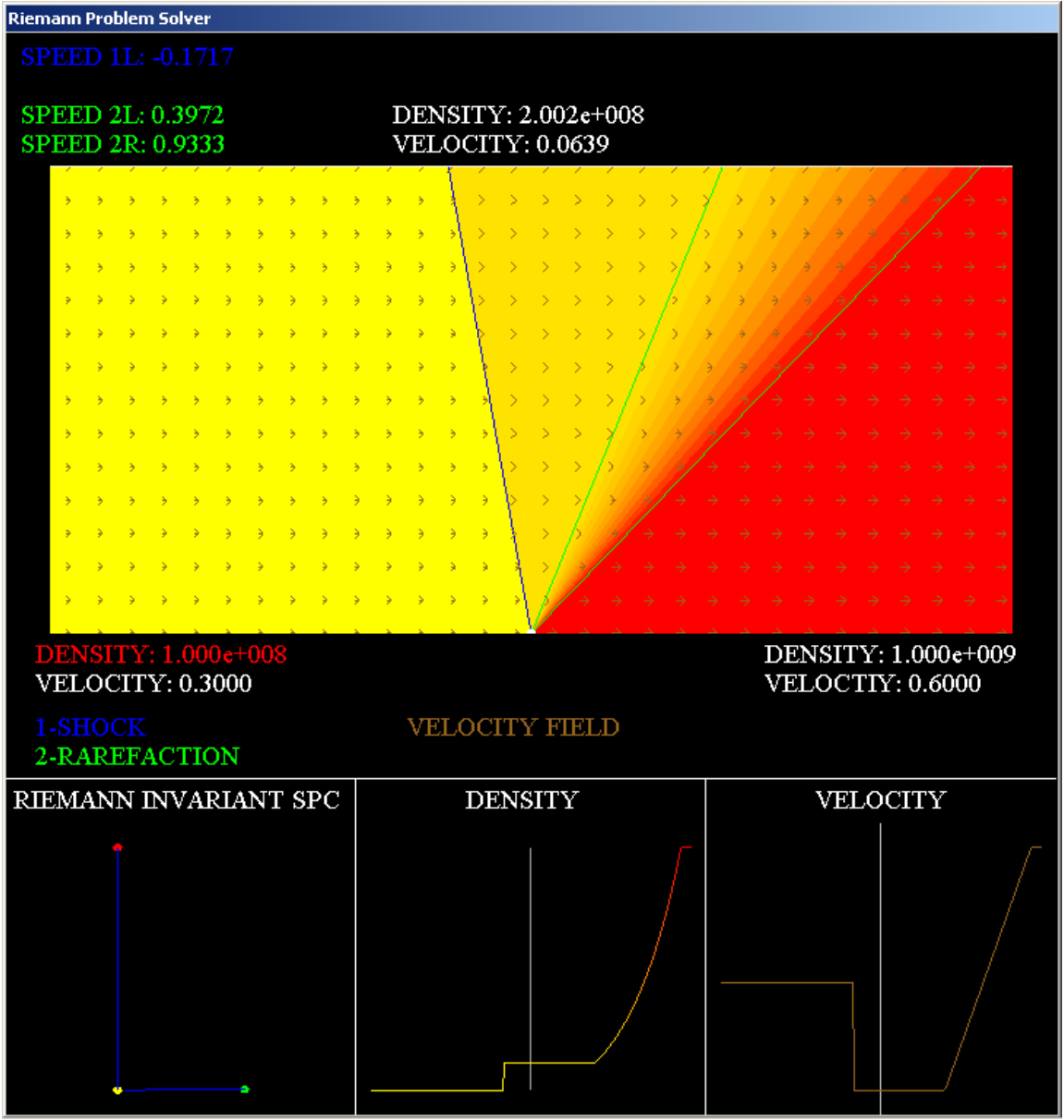}
\end{center}
\caption{The Riemann Problem Solver}
\label{fig:rp_solver}
\end{figure}

With all this knowledge of the relativistic compressible Euler equations, we construct a Riemann problem simulator to test our solver and show us solutions to the Riemann problem for this system of equations in Minkowski spacetime.  A glimpse of this simulator is shown in Figure \ref{fig:rp_solver}; for concreteness, the 1-shock and 2-rarefaction case is displayed, where the left and right states are chosen to be $u_L=(\rho_L,v_L)=(10^8,0.3)$ and $u_R=(\rho_R,v_R)=(10^9,0.6)$, respectively.  All the information needed to solve this Riemann problem is embedded into this figure.  The colored rectangle is the spacetime cell containing the Riemann problem.  The x-axis is horizontal while the t-axis is vertical, and the white half circle at the bottom center represents the origin $(t,x)=(0,0)$.  The speed of light is chosen to be one (i.e. $c=1$) so twice as much space relative to time is needed to enable light to travel in both directions.  The left/right state is displayed under the space time cell on the left/right side, and the middle state $u_M=(2.002\times10^8,0.0639)$ is recorded above this cell.  The density is displayed using a relative color map, where the highest value for the density ($10^9$) is red, and the lowest value ($10^8$) is yellow, with other values linearly interpolated between these two colors.  The velocity is represented by the brown arrows where the length of the stem represents magnitude, and the direction left/right represents negative/positive velocity.  The 1-shock wave is represented by the blue line with the corresponding speed ($s_1=-0.1717$) recorded on the top left.  The 2-rarefaction wave is represented by the multiple colored set of triangles bounded by green lines.  The green lines are the edges to the rarefaction wave, also referred to as a rarefaction fan, and the colored triangles represent the different states within this fan.  Shown on the top, the rarefaction wave has a the left speed of $s^L_2=0.3972$ and a right speed of $s^R_2=0.9333$.

On the bottom, there are three separate panels.  The left panel displays the Riemann invariant space with the three states associated with the Riemann problem under consideration.  The left state $u_L$, shown as a green dot, is connected to the middle state $u_M$, shown as a yellow dot, by the 1-shock curve.  Also, the middle state is connected to the right state $u_R$, shown as a red dot, by the 2-rarefaction curve.  The other two panels show the density and velocity profiles after one unit of time.  Here one can see, going left to right, the left state, the 1-shock wave, the middle state, the 2-rarefaction wave, and the right state in these graphs.

\section{Time Dilation Between Space Time Cells}
\label{sec:time_dilation}
\begin{figure}[t]
\begin{pspicture}(12,5)(0,-0.3)
%\psgrid
%draw twice the unitary speed of light grid cell
\psline(0,0)(2,0)
\psline(2,0)(2,1)
\psline(2,1)(0,1)
\psline(0,1)(0,0)
\psline(1,0)(0.5,1)
\psline(1,0)(1.5,1)
\psline{|-|}(-0.3,0)(-0.3,1)
\rput(-.8,0.5){$\Delta t$}
\psline{|-|}(2.3,0)(2.3,1)
\rput(2.8,0.5){$\Delta t_1$}
\psline{|-}(0,-0.3)(0.5,-0.3)
\psline{-|}(1.5,-0.3)(2,-0.3)
\rput(1,-0.3){$\Delta x$}
\rput(1,1.5){$\sqrt{A_1B_1}=2$}

%draw the unitary grid cell
\psline(5,0)(7,0)
\psline(7,0)(7,2)
\psline(7,2)(5,2)
\psline(5,2)(5,0)
\psline(6,0)(5.5,2)
\psline(6,0)(6.5,2)
\psline{|-|}(4.7,0)(4.7,1)
\rput(4.2,0.5){$\Delta t$}
\psline{|-}(7.4,0)(7.4,0.7)
\psline{-|}(7.4,1.3)(7.4,2)
\rput(7.4,1){$\Delta t_2$}
\psline{|-}(5,-0.3)(5.5,-0.3)
\psline{-|}(6.5,-0.3)(7,-0.3)
\rput(6,-0.3){$\Delta x$}
\rput(6,2.5){$\sqrt{A_2B_2}=1$}
\psline[linestyle=dashed](5,1)(7,1)

%draw half the unitary speed of light grid cell
\psline(10,0)(12,0)
\psline(12,0)(12,4)
\psline(12,4)(10,4)
\psline(10,4)(10,0)
\psline(11,0)(10.5,4)
\psline(11,0)(11.5,4)
\psline{|-|}(9.7,0)(9.7,1)
\rput(9.2,0.5){$\Delta t$}
\psline{|-}(12.4,0)(12.4,1.7)
\psline{-|}(12.4,2.3)(12.4,4)
\rput(12.4,2){$\Delta t_3$}
\psline{|-}(10,-0.3)(10.5,-0.3)
\psline{-|}(11.5,-0.3)(12,-0.3)
\rput(11,-0.3){$\Delta x$}
\rput(11,4.5){$\sqrt{A_3B_3}=\frac{1}{2}$}
\psline[linestyle=dashed](10,1)(12,1)

\end{pspicture}\caption{Effects of time dilation}
\end{figure}
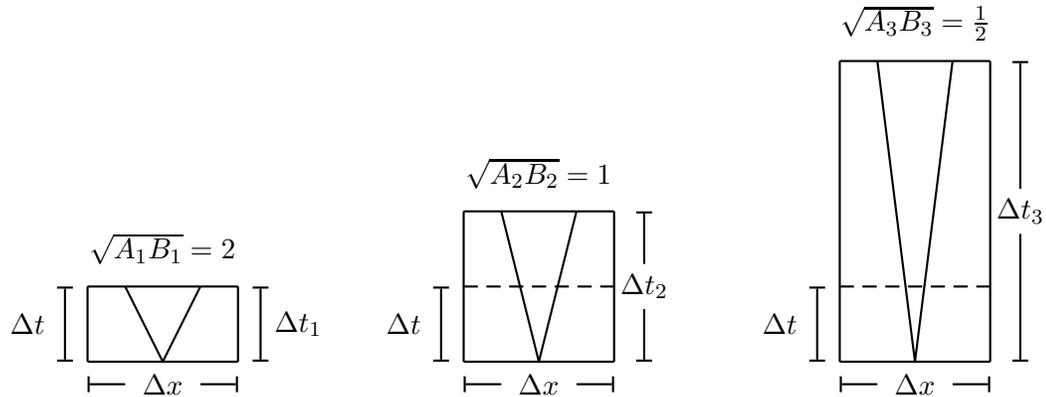
\label{fig:time_dilation}

With the capability of solving the Riemann problem for the relativistic compressible Euler equations in a unitary frame, like Minkowski spacetime, this section is dedicated to solving it in non-unitary frames.  A non-unitary frame has the (coordinate) speed of light different from one (i.e. $\sqrt{AB}\neq 1$).  This factor $\sqrt{AB}$, determined by the metric components, only changes the speeds of the waves, and when considering solutions to the Riemann problem, it has no effect on the states themselves \cite{groasmte}.  Since the locally inertial Godunov method has a fixed spatial distance for all the frames, this factor determines how time is sped up or slowed down relative to the unitary frame.  One way to view this phenomenon is this factor stretches or shrinks the height of our spacetime cell relative to other frames, as illustrated in Figure \ref{fig:time_dilation}.  This figure shows, from left to right, a frame of speed two, a unitary frame, and a frame of speed one-half.  One problem we face is one unit of time has different meanings in different frames.  For example, one unit of time in the left frame in Figure \ref{fig:time_dilation} corresponds to a half a unit of time in the middle frame and a fourth of a unit of time for the right frame.  In other words, we need a frame to give meaning to the quantity of time $t$, and we choose the unitary frame to be this time keeper for all of our simulations.  Another problem this factor causes is the need to unify all the times across the board, making sure to not break any one frame's Courant-Friedrichs-Levy (CFL) condition \cite{leve}.  To accomplish this unification, the frame where time moves the fastest sets the change in time $\Delta t$.  This reduction of time in the other frames is represented in Figure \ref{fig:time_dilation} as the dashed line.  The goal of this section is to determine the effect of shortening the time on all the other frames.

To proceed, we need to distinguish between two different spacetime cells.  In the last section, we dealt exclusively with a cell containing the whole Riemann problem posed, which we will refer to this cell as a {\it Riemann cell}.  Now, we consider a cell containing an entire Godunov step, and we refer to it as a {\it Godunov cell}.  In order to implement the Godunov step, there has to be at least three states $u_L$, $u_C$, and $u_R$, which poses two Riemann problems, and the step is performed on the center state $u_C$; therefore, the Godunov cell and the Riemann cell are staggered against one another.  Note that if we do not violate the CFL condition then the two Riemann problems cannot interact, and the solution $u(t,x)$ is completely determined.  Define $u^L_*$ to be the zero speed state for the left Riemann problem and $u^R_*$ the corresponding one for the right Riemann problem.  All of these details are displayed in Figure \ref{fig:god_cell}.

\begin{figure}[!t]
\begin{pspicture}(12,3.5)(-1,-0.5)
%\psgrid
%draw both space-time cells
\psline(0,0)(12,0)
\psline(12,0)(12,3)
\psline(12,3)(0,3)
\psline(0,3)(0,0)
\psline(6,0)(6,0.7)
\rput(6,1){$u_C$}
\psline(6,1.3)(6,3)
\psline(3,0)(2.4,3)
\psline(3,0)(5.4,3)
\psline(9,0)(7.6,3)
\psline(9,0)(10.5,3)
\rput(1.5,1){$u_L$}
\rput(10.5,1){$u_R$}
%dashed lines as the borders to the Godunov Cell
\psline[linestyle=dashed](3,0)(3,3)
\psline[linestyle=dashed](9,0)(9,3)
%label the Godunov Cell along with the zero speed states
\rput(6,-0.1){$\underbrace{\phantom{444444444444444444444444444444l}}$}
\rput(6,-0.6){Godunov Cell}
\rput(3,3.4){$u^L_*$}
\rput(9,3.4){$u^R_*$}
%label the time step
\psline{|-}(-0.5,0)(-0.5,1.2)
\rput(-0.5,1.5){$\Delta \tilde{t}$}
\psline{-|}(-0.5,1.8)(-0.5,3)
\end{pspicture}
\caption{The Godunov Cell}
\label{fig:god_cell}
\end{figure}

The following theorem expresses if the time in a Godunov cell is shortened, the resulting average is an affine combination of the original average and the center state, based on the ratio of the original and new time change.

\begin{thm}\label{shorten_time}
Let
\begin{equation}\label{max_time_in_cell}
\Delta\tilde{t}=\frac{\Delta x}{2\sqrt{AB}}
\end{equation}
represent the maximum time in a Godunov cell before the CFL condition is violated.  If  the change in time is shortened from $\Delta\tilde{t}$ to $\Delta t<\Delta\tilde{t}$ in a Godunov cell containing the solution to the Riemann problems $u(t,x)$, the average across that grid cell at time $\Delta\tilde{t}$, $\bar{u}(\Delta\tilde{t})$, and at time $\Delta t$, $\bar{u}(\Delta t)$, are related by
\begin{equation}\label{shorten_time_avg_relation}
\bar{u}(\Delta t) = \lambda\bar{u}(\Delta\tilde{t})+(1-\lambda)u_C
\end{equation}
where $\lambda=\frac{\Delta t}{\Delta\tilde{t}}<1$ is the ratio between the two times.
\end{thm}
\begin{proof}
Suppose we have a solution to both Riemann problems $u(t,x)$ in the Godunov cell, along with a maximum time $\Delta\tilde{t}$ (\ref{max_time_in_cell}) and a shorter time $\Delta t<\Delta\tilde{t}$.  Since each half of the Godunov cell contains a disjoint self-similar Riemann problem, we prove the claim on one half of the grid cell.  Without loss of generality, we study the left half (i.e. $0<x<\Delta x/2$) with the origin located at the bottom left of this cell.  We start by constructing the relationship between the function $u(t,x)$ at the two times.  The fastest speed in the Godunov cell (i.e. the speed of light) is referred to as $\alpha=\sqrt{AB}$.  By the self-similarity of the Riemann Problem, the following relationship holds
\begin{equation}
u(\Delta t,x)=\left\{
\begin{array}{ll}
u(\lambda^{-1}\Delta t,\lambda^{-1}x)=u(\Delta\tilde{t},\lambda^{-1}x) & \text{for } 0\leq x\leq\alpha \Delta t\\
u_C & \text{for } \alpha \Delta t\leq x\leq\frac{\Delta x}{2}.
\end{array}
\right.
\end{equation}
Figure \ref{fig:shorten_time} displays these details, where the center state $u_C$ is sectioned off by the speed of light of the left and right Riemann problems.
\begin{figure}[!h]
\begin{pspicture}(8,3.5)(-2,-0.5)
%\psgrid
%space-time cell
\psline(0,0)(6,0)
\psline(6,0)(6,3)
\psline(6,3)(0,3)
\psline(0,3)(0,0)
%dashed line
\psline[linestyle=dashed](0,2)(6,2)
\psline(0,0)(2.4,3)
\psdots(1.6,2)(2.4,3)
\rput(2.4,3.3){$u(\Delta \tilde{t},\lambda^{-1}x)$}
\rput(1,2.4){$u(\Delta t,x)$}
\psline{|-}(0,-0.4)(2.5,-0.4)
\rput(3,-0.4){$\Delta x$}
\psline{-|}(3.5,-0.4)(6,-0.4)
\psline{|-}(6.5,0)(6.5,0.7)
\rput(6.5,1){$\Delta t$}
\psline{-|}(6.5,1.3)(6.5,2)
\psline{|-}(-0.5,0)(-0.5,1.2)
\rput(-0.5,1.5){$\Delta \tilde{t}$}
\psline{-|}(-0.5,1.8)(-0.5,3)
%shade the region for the middle state u_m and label
\pspolygon[fillstyle=vlines](0,0)(3,3)(6,0)
\rput(3,0.5){\psframebox*{$u_C$}}
\end{pspicture}
\caption{The effect of shorten the time step within a Godunov cell}
\label{fig:shorten_time}
\end{figure}

The average at the time $\Delta t$ is directly calculated as
\begin{equation}\label{new_time_avg}
\begin{split}
\bar{u}(\Delta t)=&\frac{2}{\Delta x}\int^{\frac{\Delta x}{2}}_0u(\Delta t,x)dx =\frac{2}{\Delta x}\left(\int^{\alpha\Delta t}_0u(\Delta\tilde{t},\lambda^{-1}x)dx+\int^{\frac{\Delta x}{2}}_{\alpha\Delta t}u_Cdx\right)\\
=&\frac{2}{\Delta x}\left[\lambda\int^{\frac{\Delta x}{2}}_0u(\Delta\tilde{t},y)dy+u_C\left(\frac{\Delta x}{2}-\alpha\Delta t\right)\right].
\end{split}
\end{equation}
Since $\frac{\Delta x}{2}=\alpha\Delta\tilde{t}$, this implies
\begin{equation}
\left(\frac{\Delta x}{2}-\alpha\Delta t\right)=\frac{\Delta x}{2}\left(1-\frac{\alpha\Delta t}{\alpha\Delta\tilde{t}}\right) =\frac{\Delta x}{2}(1-\lambda),
\end{equation}
and (\ref{new_time_avg}) becomes
\begin{equation}
\bar{u}(\Delta t) = \lambda\bar{u}(\Delta\tilde{t})+(1-\lambda)u_C,
\end{equation}
proving the claim.
\end{proof}

Since the Godunov step only considers the zero speed states $u^L_*$ and $u^R_*$ and the factor $\sqrt{AB}$ has no affect on scaling a speed of zero, we expect the affect on the Godunov step of this time dilation to be minor.  A quick calculation for the left half of the Godunov cell shows
\begin{equation}
\begin{split}
\bar{u}(\Delta t) &= \lambda\bar{u}(\Delta\tilde{t})+(1-\lambda)u_C\\
&= (1-\lambda)u_C+\lambda u_C-\lambda\frac{2\Delta\tilde{t}}{\Delta x}[f(u_C)-f(u^L_*)]\\
&= u_C-\lambda\frac{2\Delta\tilde{t}}{\Delta x}[f(u_C)-f(u^L_*)]\\
&= u_C-\frac{2\Delta t}{\Delta x}[f(u_C)-f(u^L_*)].
\end{split}
\end{equation}
Hence, the only effect of a reduced time is to shorten the input time in the Godunov step.
\section{Locally Inertial Godunov Method}
\label{sec:frac_god_method}
Equipped with all the necessary tools, we state the algorithm of the locally inertial Godunov method.  This method is a fractional step scheme started by choosing some parameters.  In particular, we choose a minimum radius $r_{min}$, a maximum radius $r_{max}$, the number of spatial gridpoints $n$, and a start time $t_0$.  In our simulations, the number of spatial grid points $n$ is chosen to be a power of two (i.e. $n=2^k$ for some $k$).  From these parameters, the mesh width $\Delta x$ is determined to be
\begin{equation}\label{mesh_size}
\Delta x=\frac{r_{max}-r_{min}}{n-1},
\end{equation}
and is fixed throughout the scheme.   Let $(x_i,t_j)$ represent a mesh point in an unstaggered grid defined on the domain
\begin{equation}\label{solution_domain}
D=\{r_{min}\leq x_i\leq r_{max},t_j\geq t_0\}.
\end{equation}
The spatial points are defined as
\begin{equation}
x_i\equiv r_{min}+(i-1)\Delta x \text{ for } i=1,\ldots,n.
\end{equation}
Unlike the mesh width, the time step or the mesh height, $\Delta t$, changes from one time step to the next because there is no way to determine beforehand the smallest $\Delta t$ satisfying the CFL condition for every time step.  So for every time $t_j$, a new time step is computed by
\begin{equation}\label{time_step_defn}
\Delta t_j=\min\left\{\frac{\Delta x}{2\sqrt{A_{ij}B_{ij}}}\right\},
\end{equation}
where the minimum is taken over all the spatial gridpoints at time $t_j$ of the metric $\mathbf{A}_{ij}=(A_{ij},B_{ij})$, where these entries are defined shortly.  Starting at $t_0$, our temporal mesh points are defined as
\begin{equation}
t_j\equiv t_0+\sum^j_{k=1}\Delta t_k \text{ for } j=1,\ldots,\infty.
\end{equation}

We assume at our current time $t_j$ for $j\geq0$ there exists a solution $u(t_j,x)$ and $\mathbf{A}(t_j,x)$ for $(t_j,x)\in D$.  This solution is either provided as the starting solution at $t_0$ or from the last iteration of the locally inertial Godunov scheme constructed inductively.  To implement the method, this solution is discretized into piecewise constant states.  Discretizing the conserved quantities $u(t_j,x)$, let $u_{\Delta x}$ be given by piecewise constant states $u_{ij}$ at time $t=t^+_j$ as follows:
\begin{equation}\label{conserved_quantities_approx}
u_{\Delta x}(t,x) = u_{ij}\equiv u(t_j,x_i)\text{ for }x_{i-}\leq x < x_{i+}, t=t^+_j.
\end{equation}
For notational convenience, we denote $x_{i+}\equiv x_{i+\frac{1}{2}}$ and $x_{i-}\equiv x_{i-\frac{1}{2}}$ throughout this paper.

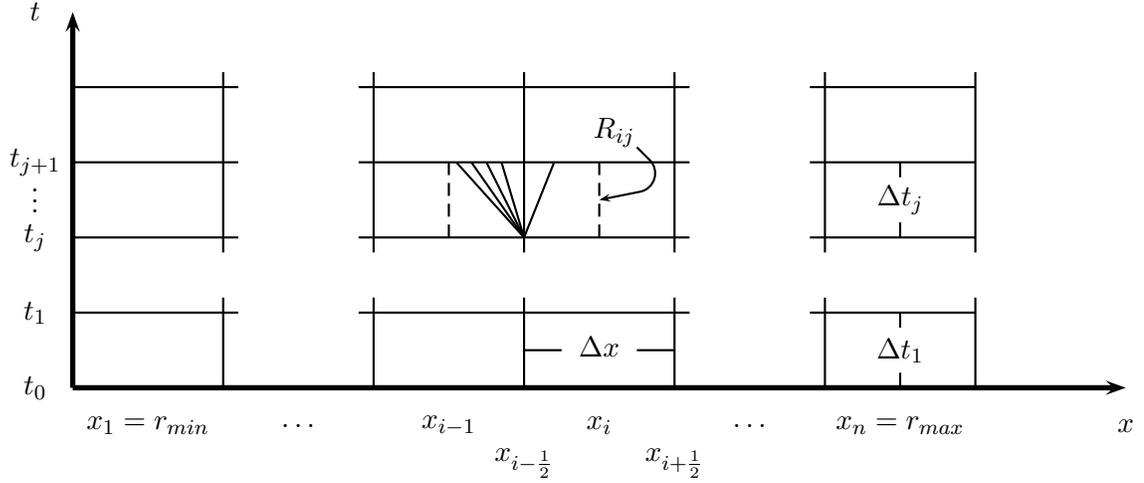
\begin{figure}[!t]
\begin{pspicture}(15,6)(-1,-1)
%horizontal lines
\psline[linewidth=2pt]{->}(0,0)(14,0)
\psline(0,1)(2.2,1)
\psline(0,2)(2.2,2)
\psline(0,3)(2.2,3)
\psline(0,4)(2.2,4)
\psline(3.8,1)(8.2,1)
\psline(3.8,2)(8.2,2)
\psline(3.8,3)(8.2,3)
\psline(3.8,4)(8.2,4)
\psline(9.8,1)(12,1)
\psline(9.8,2)(12,2)
\psline(9.8,3)(12,3)
\psline(9.8,4)(12,4)
%vertical lines
\psline[linewidth=2pt]{->}(0,0)(0,5)
\psline(2,0)(2,1.2)
\psline(2,1.8)(2,4.2)
\psline(4,0)(4,1.2)
\psline(4,1.8)(4,4.2)
\psline(6,0)(6,1.2)
\psline(6,1.8)(6,4.2)
\psline(8,0)(8,1.2)
\psline(8,1.8)(8,4.2)
\psline(10,0)(10,1.2)
\psline(10,1.8)(10,4.2)
\psline(12,0)(12,1.2)
\psline(12,1.8)(12,4.2)
%Riemann Problem Mesh Rectangle
\psline[linestyle=dashed](5,2)(5,3)
\psline[linestyle=dashed](7,2)(7,3)
\psline[linearc=0.25]{->}(7.5,3.2)(8,2.7)(7,2.5)
\psline(6,2)(5.1,3)
\psline(6,2)(5.3,3)
\psline(6,2)(5.5,3)
\psline(6,2)(5.7,3)
\psline(6,2)(6.4,3)
% Delta x and Delta t measurements
\psline(11,0)(11,0.2)
\psline(11,0.8)(11,1)
\rput(11,0.5){$\Delta t_1$}
\psline(11,2)(11,2.2)
\psline(11,2.8)(11,3)
\rput(11,2.5){$\Delta t_j$}
\psline(6,0.5)(6.5,0.5)
\psline(7.5,0.5)(8,0.5)
\rput(7,0.55){$\Delta x$}
%text
\rput(-.5,5){$t$}
\rput(-.5,0){$t_0$}
\rput(-.5,1){$t_1$}
\rput(-.5,2.6){$\vdots$}
\rput(-.5,2){$t_j$}
\rput(-.5,3){$t_{j+1}$}
\rput(14,-0.5){$x$}
\rput(1,-0.5){$x_1=r_{min}$}
\rput(3,-0.5){$\ldots$}
\rput(5,-0.5){$x_{i-1}$}
\rput(6,-1){$x_{i-\frac{1}{2}}$}
\rput(7,-0.5){$x_i$}
\rput(8,-1){$x_{i+\frac{1}{2}}$}
\rput(9,-0.5){$\ldots$}
\rput(11,-0.5){$x_n=r_{max}$}
\rput(7.2,3.4){$R_{ij}$}
%\end{pspicture}
\end{pspicture}\caption{The Riemann cell $R_{ij}$}
\label{fig:mesh_rectangle}
\end{figure}

We define the grid rectangle $R_{ij}$ so the mesh point $(x_{i-},t_j)$ is in the bottom center of it,
\begin{equation}\label{grid_rectangle}
R_{ij}\equiv\{x_{i-1}\leq x<x_{i},\phantom{4} t_j\leq t < t_{j+1}\}, \phantom{4444}1\leq i\leq n+1,\phantom{4}j\geq0,
\end{equation}
which is diagrammed in Figure \ref{fig:mesh_rectangle}.  Each grid rectangle is a Riemann cell, containing a solution to a distinct Riemann problem.  We are limited to solving Riemann problems within Riemann cells having a constant speed of light, as discussed in the last section.  To this end, the metric source $\mathbf{A}=(A,B)$ must be approximated by a constant value, denoted $\mathbf{A}_{ij}$, in each Riemann cell $R_{ij}$ throughout the simulation.  These constant values are established by setting
\begin{equation}\label{metric_approx}
\mathbf{A}_{\Delta x}(t,x) = \mathbf{A}_{ij}\equiv \mathbf{A}(t_j,x_{i-})\text{ for }(t,x)\in R_{ij}.
\end{equation}
This approximation makes $\mathbf{A}_{\Delta x}$ discontinuous along each line $x=x_i,i=1,\ldots,n$, at each time step $t=t_j$.

To implement the Godunov step, we need boundary profiles at the left and right boundaries along with the initial profiles at time $t_0$ because the Godunov step is a three point method, and the points $x_1$ and $x_n$ need left and right partners, respectively, to pose the boundary Riemann problems.  These boundary profiles are used to implement the boundary Riemann problems and are referred to as {\it ghost cells}.  The left and right ghost cells, located at the points $x_0$ and $x_{n+1}$, respectively, must be consistent with our numerical solution to the Einstein equations around these boundaries.  Any inconsistency in these boundary conditions would result in errors propagating into our solution, corrupting the data; therefore, data is needed for the left ghost cell $u_{0,j}$ and $\mathbf{A}_{0,j}$ along with the right ghost cell $u_{n+1,j}$ and $\mathbf{A}_{n+1,j}$ that are solutions to the Einstein equations synchronized with the data close to the boundary.  Figure \ref{fig:space_labels} displays the location of these ghost cells.

The discontinuities of the metric $\mathbf{A}_{\Delta x}$ are staggered relative to the approximate solution $u_{\Delta x}$ as illustrated in Figure \ref{fig:space_labels}.  This staggering puts constant metric values within each Riemann cell and constant conserved quantities states at the bottom of each Godunov cell.  Constant conserved quantities $u_{\Delta x}$ in each Godunov cell and a constant metric $\mathbf{A}_{\Delta x}$ in each Riemann cell enables us to pose Riemann problems in locally inertial coordinate frames where we are capable of solving the relativistic compressible Euler equations. More specifically, there is a Riemann problem at the bottom center of each Riemann cell $R_{ij}$
\begin{equation}\label{mesh_riemann_problem}
\begin{split}
&u_{t}+f(\mathbf{A}_{ij},u)_{x}=0\\
u_0(x) =& \left\{
\begin{array}{ll}
u_L=u_{i-1,j} & x<x_{i-}\\
u_R=u_{i,j} & x>x_{i-}.
\end{array}
\right.
\end{split}
\end{equation}
Let $u^{RP}_{ij}(t,x)$ denote the solution of (\ref{mesh_riemann_problem}) within the Riemann cell $R_{ij}$, and define
\begin{equation}\label{riemann_problem_soln}
u^{RP}_{\Delta x}(t,x) \equiv u^{RP}_{ij}(t,x)\text{ for }(t,x) \in R_{ij}
\end{equation}
as the Riemann problem step of the fractional step scheme.

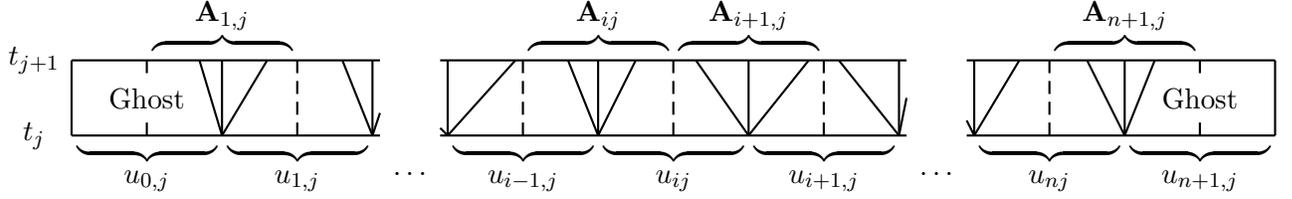
\begin{figure}[t]
\begin{pspicture}(16,2)
%\psgrid
%horizontal lines
\psline(0,0.5)(4.1,0.5)
\psline(4.9,0.5)(11.1,0.5)
\psline(11.9,0.5)(16,0.5)
\psline(0,1.5)(4.1,1.5)
\psline(4.9,1.5)(11.1,1.5)
\psline(11.9,1.5)(16,1.5)
%vertical lines
\psline(0,0.5)(0,1.5)
\psline(2,0.5)(2,1.5)
\psline(4,0.5)(4,1.5)
\psline(5,0.5)(5,1.5)
\psline(7,0.5)(7,1.5)
\psline(9,0.5)(9,1.5)
\psline(11,0.5)(11,1.5)
\psline(12,0.5)(12,1.5)
\psline(14,0.5)(14,1.5)
\psline(16,0.5)(16,1.5)
%vertical dashed lines
\psline[linestyle=dashed](1,0.5)(1,1.5)
\psline[linestyle=dashed](3,0.5)(3,1.5)
\psline[linestyle=dashed](6,0.5)(6,1.5)
\psline[linestyle=dashed](8,0.5)(8,1.5)
\psline[linestyle=dashed](10,0.5)(10,1.5)
\psline[linestyle=dashed](13,0.5)(13,1.5)
\psline[linestyle=dashed](15,0.5)(15,1.5)
%slanted lines to represent Riemann Solns
\psline(2,0.5)(1.7,1.5)
\psline(2,0.5)(2.6,1.5)
\psline(4,0.5)(3.6,1.5)
\psline(4,0.5)(4.1,0.8)
\psline(5,0.5)(4.9,0.6)
\psline(5,0.5)(5.9,1.5)
\psline(7,0.5)(6.6,1.5)
\psline(7,0.5)(7.5,1.5)
\psline(9,0.5)(8.3,1.5)
\psline(9,0.5)(9.8,1.5)
\psline(11,0.5)(10.2,1.5)
\psline(11,0.5)(11.1,1)
\psline(12,0.5)(11.9,0.7)
\psline(12,0.5)(12.6,1.5)
\psline(14,0.5)(13.5,1.5)
\psline(14,0.5)(14.4,1.5)
%labels
\rput(-.5,0.5){$t_j$}
\rput(-.5,1.5){$t_{j+1}$}
\rput(1,0.4){$\underbrace{\phantom{444444444l}}$}
\rput(1,-0.1){$u_{0,j}$}
\rput(3,0.4){$\underbrace{\phantom{444444444l}}$}
\rput(3,-0.1){$u_{1,j}$}
\rput(4.5,0){$\ldots$}
\rput(6,0.4){$\underbrace{\phantom{444444444l}}$}
\rput(6,-0.1){$u_{i-1,j}$}
\rput(8,0.4){$\underbrace{\phantom{444444444l}}$}
\rput(8,-0.1){$u_{ij}$}
\rput(10,0.4){$\underbrace{\phantom{444444444l}}$}
\rput(10,-0.1){$u_{i+1,j}$}
\rput(11.5,0){$\ldots$}
\rput(13,0.4){$\underbrace{\phantom{444444444l}}$}
\rput(13,-0.1){$u_{nj}$}
\rput(15,0.4){$\underbrace{\phantom{444444444l}}$}
\rput(15,-0.1){$u_{n+1,j}$}
\rput(2,1.6){$\overbrace{\phantom{444444444l}}$}
\rput(2,2.1){$\mathbf{A}_{1,j}$}
\rput(7,1.6){$\overbrace{\phantom{444444444l}}$}
\rput(7,2.1){$\mathbf{A}_{ij}$}
\rput(9,1.6){$\overbrace{\phantom{444444444l}}$}
\rput(9,2.1){$\mathbf{A}_{i+1,j}$}
\rput(14,1.6){$\overbrace{\phantom{444444444l}}$}
\rput(14,2.1){$\mathbf{A}_{n+1,j}$}
\rput(1,1){\psframebox*{Ghost}}
\rput(15,1){\psframebox*{Ghost}}
\end{pspicture}\caption{Staggering of the metric $\mathbf{A}$ and the solution $u$}
\label{fig:space_labels}
\end{figure}

Equipped with the solutions to the Riemann problems in each Riemann cell $R_{ij}$, we implement the Godnunov step to obtain the average of fluid variables $u^{RP}_{\Delta x}$ across the intervals $[x_{i-},x_{i+}]$ at the next time step $t_{j+1}$.  Since the metric $\mathbf{A}$ is different on both sides of $x_i$, separate averages must be taken over the left and right half cells and combined to obtain the true average.  In particular, let $\bar{u}^L_{ij}$ and $\bar{u}^R_{ij}$ be the average on the left and right half cells, respectively. Also, let $u^L_*=u^{RP}_{ij}(t^-_{j+1},x_{i-})$ and $u^R_*=u^{RP}_{ij}(t^-_{j+1},x_{i+})$ represent the zero speed states left and right Riemann problem, respectively, as shown in Figure \ref{fig:god_cell}. To perform the Godunov step on the left half cell, we compute
\begin{equation}\label{god_step_left}
\bar{u}^L_{ij}=u_{ij}-\frac{2\Delta t}{\Delta x}\{f(\mathbf{A}_{ij},u_{ij})-f(\mathbf{A}_{ij},u^L_*)\},
\end{equation}
and do the same for the right half cell
\begin{equation}
\bar{u}^R_{ij}=u_{ij}-\frac{2\Delta t}{\Delta x}\{f(\mathbf{A}_{ij},u^R_*)-f(\mathbf{A}_{ij},u_{ij})\},
\end{equation}
where the 2 accounts for the half cell calculations.  Taking the average of these results leads to
\begin{equation}\label{god_step_both}
\bar{u}_{ij}=\frac{1}{2}\{\bar{u}^L_{ij}+\bar{u}^R_{ij}\},
\end{equation}
defining our Godunov step of the method.

We proceed to define the ODE step.  Let $\hat{u}(t,u_0)$ denote the solution to the following initial value problem
\begin{equation}\label{ch3_ode_step_ivp}
\begin{split}
\hat{u}_t=G(\mathbf{A}_{ij},\hat{u},x)=g&(\mathbf{A}_{ij},\hat{u},x)- \mathbf{A}'\cdot\nabla_\mathbf{A}f(\mathbf{A}_{ij},\hat{u},x), \\
&\hat{u}(0)=u_0,
\end{split}
\end{equation}
where $G(\mathbf{A},\hat{u},x)=(G^0,G^1)$ takes the form
\begin{equation}\label{ch3_ode_step_G0_fluid}
G^0=-\frac{1}{2}\sqrt{AB}\left(\frac{c^2+\sigma^2}{c^2-v^2}\right)cv\frac{\rho}{x}\left\{2(\frac{1}{A}+1) -\frac{\kappa}{A}(c^2-\sigma^2)\rho x^2\right\},
\end{equation}
\begin{equation}\label{ch3_ode_step_G1_fluid}
G^1=-\frac{1}{2}\sqrt{AB}\left(\frac{c^2+\sigma^2}{c^2-v^2}\right)\frac{\rho}{x}\left\{4v^2 +(\frac{1}{A}-1)(c^2+v^2)+\frac{\kappa}{A}(\sigma^2-v^2)c^2\rho x^2\right\}.
\end{equation}

We define the approximate solution $u_{\Delta x}(t,x)$ and $\mathbf{A}_{\Delta x}(t,x)$ analytically to derive the piecewise formulas used to update the numerical scheme and to be used in the convergence proof of the next chapter.  The conserved quantities are defined by the formula
\begin{equation}\label{ch3_approx_soln}
u_{\Delta x}(t,x)=u^{RP}_{\Delta x}(t,x) +\int^t_{t_{j}}\{G(\mathbf{A}_{ij},\hat{u}(\xi-t_j,u^{RP}_{\Delta x}(t,x),x)\}d\xi
\end{equation}
Therefore, $u_{\Delta x}(t,x)$ is equal to $u^{RP}_{\Delta x}(t,x)$, the solution to the Riemann problems, plus a correction term from the ODE step of the method.  The metric is derived from the definition of the mass
\begin{equation}\label{compute_mass}
M_{\Delta x}(x,t)=M_{r_{min}}+\frac{\kappa}{2}\int^x_{r_{min}}u^0_{\Delta x}(r,t)r^2dr.
\end{equation}
In terms of these equations, define the metric as
\begin{equation}\label{compute_metric_a}
A_{\Delta x}(x,t)=1-\frac{2M_{\Delta x}(x,t)}{x},
\end{equation}
and
\begin{equation}\label{compute_metric_b}
B_{\Delta x}(x,t)=B_{r_0}\text{ exp}\int^x_{r_{min}}\left\{\frac{\{A_{\Delta x}(r,t)\}^{-1}-1}{r}+\frac{\kappa r}{A_{\Delta x}(r,t)}T^{11}_M(u_{\Delta x}(r,t))\right\}dr.
\end{equation}

Finally, in order to update the metric and conserved quantities, we use the Riemann problem averages $\bar{u}_{ij}$ to replace the Riemann problem solution $u^{RP}_{\Delta x}(t,x)$ and perform numerical integration on the analytical equations (\ref{ch3_approx_soln})-(\ref{compute_metric_b}).  This process leads us to define
\begin{equation}\label{ch3_update_conserved_quantities}
u_{i,j+1}=\bar{u}_{ij} +\left\{G(\frac{1}{2}(\mathbf{A}_{ij}+\mathbf{A}_{i+1,j}),\hat{u}(\xi-t_j,\bar{u}_{ij},x))\right\}\Delta t_j.
\end{equation}
The mass is
\begin{equation}\label{ch3_update_mass}
M_{i,j+1}=M_{r_{min}}+\sum_{k<i}\frac{\kappa}{2}\left(u^0_{\Delta x}(x_{k-},t_{j+1})x^2_{k-}\Delta x\right),
\end{equation}
with
\begin{equation}
u^0_{\Delta x}(x_{k-},t_{j+1})=\frac{1}{2}\{u^0_{k-1,j+1}+u^0_{k,j+1}\},
\end{equation}
and the metric becomes
\begin{equation}\label{ch3_update_metric_a}
A_{i,j+1}=1-\frac{2M_{i,j+1}}{x_{i-}},
\end{equation}
and
\begin{equation}\label{ch3_compute_metric_b}
B_{i,j+1}=B_{r_{min}}e^\tau,
\end{equation}
where
\begin{equation}
\tau= \left\{\sum_{k<i}\frac{\{A_{k,j+1}\}^{-1}-1}{x_{k-}}+\frac{\kappa x_{k-}}{A_{k,j+1}} T^{11}_M(u_{\Delta x}(x_{k-},t_{j+1}))\Delta x\right\},
%\tau=\left\{\sum_{x_{i-1/2}<x}\frac{\frac{1}{A_{\Delta x}(x_{i-1/2},t_{j+1})}-1}{x_{i-1/2}}+\frac{\kappa x_{i-1/2}}{A_{\Delta x}(x_{i-1/2},t_{j+1})}T^{11}_M(u_{\Delta x}(x_{i-1/2},t_{j+1}))\Delta x\right\}.
\end{equation}
with
\begin{equation}\label{cons_vars_in_between}
u_{\Delta x}(x_{k-},t_{j+1})=\frac{1}{2}\{u_{k-1,j+1}+u_{k,j+1}\}.
\end{equation}
Note that since the metric is staggered relative to the conserved quantities, we use the in between values, like $x_{k-}$ and $u_{\Delta x}(x_{k-},t_{j+1})$ in the update step. Let $\mathbf{A}_{i,j+1}=(A_{i,j+1},B_{i,j+1})$ denote the constant value for $\mathbf{A}_{\Delta x}$ on $R_{i,j+1}$.  This concludes the update step and completes the definition of the approximate solution $u_{\Delta x}$ and $\mathbf{A}_{\Delta x}$ by induction.

To summarize the method, after the setup of the gridpoints, Riemann cells, initial profiles, and ghost cells, the locally inertial Godunov method constructs the solution inductively with four major steps: a Riemann problem step, a Godunov step (with time dilation), an ODE step, and an update step.  The Riemann problem step is described in equations (\ref{mesh_riemann_problem}) and (\ref{riemann_problem_soln}).  Formulas (\ref{god_step_left})-(\ref{god_step_both}) denote the Godunov step.  The ODE step is detailed in (\ref{ch3_ode_step_ivp})-(\ref{ch3_approx_soln}).  Finally, equations (\ref{ch3_update_conserved_quantities})-(\ref{cons_vars_in_between}) express the update step.

   \chapter[%
      Short Title of 5th Ch.
   ]{%
      Convergence of the Method
   }%
   \label{ch:converge_of_method}

The focus of this chapter is proving the main theorem of this thesis.  This theorem states that if a solution $(u_{\Delta x},\mathbf{A}_{\Delta x})\rightarrow (u,\mathbf{A})$ using the locally inertial Godunov scheme converges and has a total variation bound at each time step, then $(u,\mathbf{A})$ are weak solutions to
\begin{equation}\label{ch5_conservation_law_source_cmpt}
\begin{split}
    u_{t}+f&(\mathbf{A},u)_{x}=g(\mathbf{A},u,x),\\
       &\mathbf{A}'=h(\mathbf{h},u,x),
\end{split}
\end{equation}
which is weakly equivalent to the Einstein equations (\ref{pde_system1})-(\ref{pde_system4}).  This proof is a modification of the Groah and Temple argument using the locally inertial Glimm scheme \cite{groasmte}, with a few differences.  The main difference is the solution update at each new time step.  In this paper, an average of the fluid variables is taken verses a random sampling, but the steps leading up to this update step are the same.  Another difference is the assumed total variation bound is used to bound the Riemann problem solutions, as opposed to Groah and Temple used wave strengths to bound the Riemann problem solutions.  Also, the time steps are now variable instead of constant.  The last difference is the inclusion of right boundary data along with the left boundary data because of the limited extent in space of a computer simulation.

There are two main points to the proof.  The first point is to show the discontinuities in the metric $\mathbf{A}$ along the boundary of Riemann cells are accounted for by the inclusion of the term
\begin{equation}
\mathbf{A}'\cdot\nabla_\mathbf{A}f(\mathbf{A}_{ij},\hat{u},x)
\end{equation}
in the ODE step (\ref{ch3_ode_step_ivp}).  The second point is to show the jump in the approximate solution $u_{\Delta x}$ along the time steps are of order $\Delta x$.  In their work \cite{groasmte}, Groah and Temple did not need the convergence and total variation assumptions because with the Glimm scheme, these assumptions are proven as long as there exists a total variation bound on the initial data, a truly remarkable feature of the scheme.  In this thesis, these assumptions are shown numerically by the simulation results in the subsequent chapters.  Thus, our theorem is perfectly suited to our numerical simulation: we numerically establish convergence and a total variation bound, and the theorem of this section proves that assuming these conditions, we can conclude convergence to a weak solution of the Einstein equations.

\section{Convergence to a Weak Solution}
The main theorem of this thesis is the following
\begin{thm}\label{thm:weak_soln}
Let $u_{\Delta x}(t,x)$ and $\mathbf{A}_{\Delta x}(t,x)$ be the approximate solution generated by the locally inertial Godunov method starting from the initial data $u_{\Delta x}(t_0,x)$ and $\mathbf{A}_{\Delta x}(t_0,x)$ for $t_0>0$.  Assume these approximate solutions exist up to some time $t_{end}>t_0$ and converge to a solution $(u_{\Delta x},\mathbf{A}_{\Delta x})\rightarrow (u,\mathbf{A})$ as $\Delta x\rightarrow 0$ along with a total variation bound at each time step $t_j$
\begin{equation}\label{total_variation_bound}
T.V._{[r_{min},r_{max}]}\{u_{\Delta x}(t_j,\cdot)\}<V,
\end{equation}
where $T.V._{[r_{min},r_{max}]}\{u_{\Delta x}(t_j,\cdot)\}$ represents the total variation of the function $u_{\Delta x}(t_j,x)$ on the interval $[r_{min},r_{max}]$.  Assume the total variation is independent of the time step $t_j$ and the mesh length $\Delta x$.  Then the solution $(u,\mathbf{A})$ is a weak solution to the Einstein equations (\ref{pde_system1})-(\ref{pde_system4}).
\end{thm}
\begin{proof}
Suppose we have approximate solutions $(u_{\Delta x},\mathbf{A}_{\Delta x})$ obtained by the locally inertial Godunov method that satisfy the hypothesis of the theorem.  Having a total variation bound at each time $t_j$ places a total variation bound on the inputs to all the Riemann problems posed at that time.  In \cite{groasmte}, Groah and Temple show a total variation bound on the inputs implies a total variation bound on the solution to the Riemann problem for any time $t$ such that $t_j\leq t<t_{j+1}$.  By the self similarity of the solution to the Riemann problem, this result also implies a total variation bound for any space coordinate within the Riemann cell.  More specifically, we have the following bounds:
\begin{equation}
T.V._{[x_{i-1},x_i]}\{u_{\Delta x}(t,\cdot)\}<V,
\end{equation}
and
\begin{equation}
T.V._{[t_j,t_{j+1})}\{u_{\Delta x}(\cdot,x)\}<V,
\end{equation}
for any $x$ and $t$ within the Riemann cell $R_{ij}$.

All the functions $f$, $G$, and $g$ derived in \cite{groasmte} are smooth, and it is the metric that is only Lipschitz continuous.  The smoothness of these functions is used throughout this proof.

Let $T=t_{end}-t_0$ be the overall time of the solution, and for each mesh length $\Delta x$ define the minimum time length
\begin{equation}
\Delta t\equiv \min_j\{\Delta t_j\}
\end{equation}
as the minimum over all the time lengths defined by (\ref{time_step_defn}).  By definition, this time length is proportional to the mesh length, $\Delta t \propto \Delta x$, implying $O(\Delta t)=O(\Delta x)$, and there exists a constant $C$ bounding all the time lengths, $\Delta t_j < C\Delta t$ for all $j$.  Throughout this chapter, let $C$ be a generic constant only depending on the bounds for the solution $[t_0,t_{end}]\times [r_{min},r_{max}]$.  This variable is created to unify all the time steps, and more importantly, used to calculate the maximum number of time steps needed to go from $t_0$ to $t_{end}$.

We now follow the development of Groah and Temple in \cite{groasmte}.  Recall, $u^{RP}_{\Delta x}(t,x)$ denotes the collection of the exact solutions in all the Riemann cells $R_{ij}$ for the Riemann problem of the homogenous system
\begin{equation}
u_{t}+f(\mathbf{A}_{ij},u)_{x}=0.
\end{equation}
So $u^{RP}_{\Delta x}(t,x)$ satisfies the weak form of this conservation law in each Riemann cell
\begin{equation}\label{rp_soln_weak_form}
\begin{split}
0=&\int\int_{R_{ij}}\left\{-u^{RP}_{\Delta x}\varphi_t - f(\mathbf{A}_{ij},u^{RP}_{\Delta x})\varphi_x\right\} dxdt \\
&+ \int_{R_i}\left\{u^{RP}_{\Delta x}(t_{j+1},x)\varphi(t_{j+1},x)-u^{RP}_{\Delta x}(t^+_j,x)\varphi(t_j,x)\right\}dx \\
&+\int_{R_j}\left\{f(\mathbf{A}_{ij},u^{RP}_{\Delta x}(t,x_{i})) \varphi(t,x_{i})\right.\\
&\left.-f(\mathbf{A}_{ij},u^{RP}_{\Delta x}(t,x_{i-1}))\varphi(t,x_{i-1})\right\}dt,
\end{split}
\end{equation}
where $\varphi$ is a smooth test function with $Supp(\varphi)\subset [t_0,t_{end})\times[a,b]$ for $a<r_{min}<r_{max}<b$.

Remember, $\hat{u}(t,u_0)$ denotes the solution to the ODE
\begin{equation}\label{ode_soln}
\begin{split}
\hat{u}_t=G(\mathbf{A}_{ij},\hat{u},x)=g&(\mathbf{A}_{ij},\hat{u},x)- \mathbf{A}'\cdot\nabla_\mathbf{A}f(\mathbf{A}_{ij},\hat{u},x), \\
&\hat{u}(0)=u_0.
\end{split}
\end{equation}
Therefore,
\begin{equation}
\hat{u}(t,u_0)=u_0+\int^t_0\left\{g(\mathbf{A}_{ij},\hat{u}(\xi,u_0),x) -\mathbf{A}'\cdot\nabla_\mathbf{A}f(\mathbf{A}_{ij},\hat{u}(\xi,u_0),x)\right\}d\xi.
\end{equation}
Also, recall $u_{\Delta x}$ denotes the approximate solution obtained using the fractional step method.  Since our fractional method takes the Riemann problem solution and feeds it into the ODE step, $u_{\Delta x}$ is defined on every Riemann cell $R_{ij}$ as
\begin{equation}\label{frac_step_approx_soln}
\begin{split}
u_{\Delta x}(t,x)=u^{RP}_{\Delta x}(t,x) +\int^t_{t_j}&\left\{g(\mathbf{A}_{ij},\hat{u}(\xi-t_j,u^{RP}_{\Delta x}(t,x)),x)\right.\\
-&\frac{\partial f}{\partial \mathbf{A}}\left.(\mathbf{A}_{ij},\hat{u}(\xi-t_j,u^{RP}_{\Delta x}(t,x)))\cdot\mathbf{A}'_{\Delta x}\right\}d\xi.
\end{split}
\end{equation}
This expression implies the error between the approximate solution and the Riemann problem solution is on the order of $\Delta x$; a fact that is repeatedly used throughout the proof.

Define the residual $\varepsilon=\varepsilon(u_{\Delta x},\mathbf{A}_{\Delta x},\varphi)$ of $u_{\Delta x}$ and $\mathbf{A}_{\Delta x}$ as the error of the solution in satisfying the weak form of the conservation law (\ref{ch5_conservation_law_source_cmpt}) by
\begin{equation}\label{residual}
\begin{split}
\varepsilon(u_{\Delta x},\mathbf{A}_{\Delta x},\varphi)\equiv&\int^{r_{max}}_{r_{min}}\int^{t_{end}}_{t_0}\left\{-u_{\Delta x}\varphi_t - f(\mathbf{A}_{\Delta x},u_{\Delta x})\varphi_x - g(\mathbf{A}_{\Delta x},u_{\Delta x},x)\varphi\right\}dxdt\\
&- I_1 - I_2\\
&\sum^{i=n+1}_{i=1,j}\int_{R_{ij}}\left\{-u_{\Delta x}\varphi_t - f(\mathbf{A}_{ij},u_{\Delta x})\varphi_x - g(\mathbf{A}_{ij},u_{\Delta x},x)\varphi\right\}dxdt\\
&- I_1 - I_2,
\end{split}
\end{equation}
where
\begin{equation}
I_1\equiv\int^{r_{max}}_{r_{min}}u_{\Delta x}(t^+_0,x)dx=\sum^{n+1}_{i=1}\int_{R_i}u_{\Delta x}(t^+_0,x)dx,
\end{equation}
and
\begin{equation}
\begin{split}
I_2\equiv&\int^{t_{end}}_{t_0}\left\{f(\mathbf{A}_{ij},u_{\Delta x}(t,r^+_{min}))\varphi(t,r^+_{min}) - f(\mathbf{A}_{ij},u_{\Delta x}(t,r^+_{max}))\varphi(t,r^+_{max})\right\}dt\\
=&\sum_j\int_{R_j}\left\{f(\mathbf{A}_{ij},u_{\Delta x}(t,r^+_{min}))\varphi(t,r^+_{min}) - f(\mathbf{A}_{ij},u_{\Delta x}(t,r^+_{max}))\varphi(t,r^+_{max})\right\}dt,
\end{split}
\end{equation}
The expression $\sum^{i=n+1}_{i=1,j}$ denotes a double sum where the index $i$ runs across all the spatial gridpoints, and the index $j$ runs across all the temporal gridpoints.  Remember, $n$ is the number of spatial gridpoints, and there are $n+1$ Riemann cells as depicted in Figure \ref{fig:space_labels}.  Our goal is to show $\varepsilon(u_{\Delta x},\mathbf{A}_{\Delta x},\varphi)=O(\Delta x)$ because if the approximation converges $(u_{\Delta x},\mathbf{A}_{\Delta x})\rightarrow (u,\mathbf{A})$ as $\Delta x\rightarrow 0$, then the limit function satisfies the condition of being a weak solution to the Einstein equations $\varepsilon(u,\mathbf{A},\varphi)=0$.

Substituting (\ref{frac_step_approx_soln}) into (\ref{residual}) gives us
\begin{equation}\label{residual_with_soln}
\begin{split}
\varepsilon=&\sum^{i=n+1}_{i=1,j}\int\int_{R_{ij}}\left\{-u^{RP}_{\Delta x}\varphi_t - f(\mathbf{A}_{ij},u_{\Delta x})\varphi_x - g(\mathbf{A}_{ij},u_{\Delta x},x)\varphi\right.\\
&- \varphi_t\int^t_{t_j}\left[g(\mathbf{A}_{ij},\hat{u}(\xi-t_j,u^{RP}_{\Delta x}(t,x)),x)\right.\\
&-\frac{\partial f}{\partial\mathbf{A}}\left.\left.(\mathbf{A}_{ij},\hat{u}(\xi-t_j,u^{RP}_{\Delta x}(t,x)))\cdot\mathbf{A}_{\Delta x}'\right]d\xi\right\}dxdt - I_1 - I_2.
\end{split}
\end{equation}
Define
\begin{equation}
\begin{split}
I^1_{ij}(t,x)\equiv&\int^t_{t_j}\left[g(\mathbf{A}_{ij},\hat{u}(\xi-t_j,u^{RP}_{\Delta x}(t,x)),x)\right.\\
&-\frac{\partial f}{\partial\mathbf{A}}\left.(\mathbf{A}_{ij},\hat{u}(\xi-t_j,u^{RP}_{\Delta x}(t,x)))\cdot\mathbf{A}_{\Delta x}'\right]d\xi
\end{split}
\end{equation}
Plugging the weak form of the conservation law (\ref{rp_soln_weak_form}) of each grid rectangle into (\ref{residual_with_soln}) gives us
\begin{equation}\label{residual_mid_step}
\begin{split}
\varepsilon=&\sum^{i=n+1}_{i=1,j}\int\int_{R_{ij}}\left\{\varphi_x[f(\mathbf{A}_{ij},u^{RP}_{\Delta x})-f(\mathbf{A}_{ij},u_{\Delta x})] - g(\mathbf{A}_{ij},u_{\Delta x},x)\varphi \right. \\
&\phantom{4444444444444444444444}\left.-\varphi_tI^1_{ij}(t,x)\right\}dxdt\\
&-I_1 - \sum^{i=n+1}_{i=1,j}\int_{R_i}\left\{u^{RP}_{\Delta x}(t^-_{j+1},x)\varphi(t_{j+1},x) - u^{RP}_{\Delta x}(t^+_j,x)\varphi(t_j,x)\right\}dx\\
&-I_2 -\sum^{i=n+1}_{i=1,j}\int_{R_j}\left\{f(\mathbf{A}_{ij},u^{RP}_{\Delta x}(t,x_{i}))\varphi(t,x_{i}) - f(\mathbf{A}_{ij},u^{RP}_{\Delta x}(t,x_{i-1}))\varphi(t,x_{i-1})\right\}dt.
\end{split}
\end{equation}
Note
\begin{equation}
\AutoAbs{f(\mathbf{A}_{ij},u^{RP}_{\Delta x})-f(\mathbf{A}_{ij},u_{\Delta x})}\leq C\Delta t
\end{equation}
which implies
\begin{equation}
\begin{split}
&\AutoAbs{\sum^{i=n+1}_{i=1,j}\int\int_{R_{ij}}\varphi[f(\mathbf{A}_{ij},u^{RP}_{\Delta x}) - f(\mathbf{A}_{ij},u_{\Delta x})]dxdt}\\
&\phantom{4444444444444444444}\leq C\AutoNorm{\varphi}_\infty \Delta t^2\Delta x\left(\frac{T}{\Delta t}\right)\left(n+1\right)= O(\Delta x)
\end{split}
\end{equation}
where the number of time steps is proportional to $T/\Delta t$ and the number of space steps is $O(1/\Delta x)$ by (\ref{mesh_size}).

Since $u^{RP}_{\Delta x}(t^+_j,x)=u_{\Delta x}(t^+_j,x)$, the following sum is rearranged to become
\begin{equation}
\begin{split}
-I_1 - &\sum^{i=n+1}_{i=1,j}\int_{R_i}\left\{u^{RP}_{\Delta x}(t^-_{j+1},x)\varphi(t_{j+1},x) - u^{RP}_{\Delta x}(t^+_j,x)\varphi(t_j,x)\right\}dx\\
&=\sum_{j\neq0}\int^{r_{max}}_{r_{min}}\left\{u_{\Delta x}(t^+_j,x) - u^{RP}_{\Delta x}(t^-_j,x)\right\}\varphi(t_j,x)dx\\
&=\sum_{j\neq 0}\int^{r_{max}}_{r_{min}}\varphi(t_j,x)\left\{u_{\Delta x}(t^+_j,x) - u_{\Delta x}(t^-_j,x)\right\}dx\\
&+ \sum_{j\neq 0}\int^{r_{max}}_{r_{min}}\varphi(t_j,x)\left\{u_{\Delta x}(t^-_j,x) - u^{RP}_{\Delta x}(t^-_j,x)\right\}dx,
\end{split}
\end{equation}
where the term $u_{\Delta x}(t_j,x)$ is added and subtracted to isolate the jump in the solution $u_{\Delta x}$ across the time step $t_j$.  We define this jump $\varepsilon_1=\varepsilon_1(u_{\Delta x},\mathbf{A}_{\Delta x},\varphi)$ as
\begin{equation}
\varepsilon_1(u_{\Delta x},\mathbf{A}_{\Delta x},\varphi)\equiv\sum_{j\neq0}\int^{r_{max}}_{r_{min}}\varphi(t_j,x)\left\{u_{\Delta x}(t^+_j,x) - u_{\Delta x}(t^-_j,x)\right\}dx,
\end{equation}
and this definition allows us to rewrite (\ref{residual_mid_step}) as
\begin{equation}\label{intro_epsilon1}
\begin{split}
\varepsilon= &\phantom{4}O(\Delta x) + \varepsilon_1 + \sum^{i=n+1}_{i=1,j}\int\int_{R_{ij}}\left\{-g(\mathbf{A}_{ij},u_{\Delta x},x)\varphi - \varphi_tI^1_{ij}(t,x)\right\}dxdt\\
&+ \sum_{j\neq0}\int^{r_{max}}_{r_{min}}\varphi(t,x)\left\{u_{\Delta x}(t^-_j,x) - u^{RP}_{\Delta x}(t^-_j,x)\right\}dx\\
&-I_2 - \sum^{i=n+1}_{i=1,j}\int_{R_j}\left\{f(\mathbf{A}_{ij},u^{RP}_{\Delta x}(t,x_{i}))\varphi(t,x_{i}) - f(\mathbf{A}_{ij}, u^{RP}_{\Delta x}(t,x_{i-1}))\varphi(t,x_{i-1})\right\}dt
\end{split}
\end{equation}
But the last sum is rearranged to cancel the boundary conditions as follows:
\begin{equation}\label{eliminate_i2}
\begin{split}
-I_2 - &\sum^{i=n+1}_{i=1,j}\int_{R_j}\left\{f(\mathbf{A}_{ij},u^{RP}_{\Delta x}(t,x_{i}))\varphi(t,x_{i}) - f(\mathbf{A}_{ij}, u^{RP}_{\Delta x}(t,x_{i-1}))\varphi(t,x_{i-1})\right\}dt\\
=&\sum^{i=n}_{i=1,j}\int_{R_j}\left\{f(\mathbf{A}_{i+1,j},u^{RP}_{\Delta x}(t,x_{i})) - f(\mathbf{A}_{ij},u^{RP}_{\Delta x}(t,x_{i}))\right\}\varphi(t,x_{i})dt\\
+&\sum_{j}\int_{R_j}\left\{f(\mathbf{A}_{1,j},u^{RP}_{\Delta x}(t,x_0)) - f(\mathbf{A}_{1,j},u_{\Delta x}(t,x_0))\right\}\varphi(t,x_0)dt\\
+&\sum_{j}\int_{R_j}\left\{f(\mathbf{A}_{n+1,j},u^{RP}_{\Delta x}(t,x_{n+1})) - f(\mathbf{A}_{n+1,j},u_{\Delta x}(t,x_{n+1}))\right\}\varphi(t,x_{n+1})dt,
\end{split}
\end{equation}
where
\begin{equation}
\begin{split}
&\AutoAbs{\sum_j\int_{R_j}\left\{f(\mathbf{A}_{1,j},u^{RP}_{\Delta x}(t,x_0)) - f(\mathbf{A}_{1,j},u_{\Delta x}(t,x_0))\right\}\varphi(t,x_0)dt}\\
&\leq\AutoNorm{\varphi}_\infty C\Delta t^2\left(\frac{T}{\Delta t}\right)=O(\Delta x),
\end{split}
\end{equation}
and similarly
\begin{equation}
\AutoAbs{\sum_j\int_{R_j}\left\{f(\mathbf{A}_{n+1,j},u^{RP}_{\Delta x}(t,x_{n+1})) - f(\mathbf{A}_{n+1,j},u_{\Delta x}(t,x_{n+1}))\right\}\varphi(t,x_{n+1})dt}=O(\Delta x).
\end{equation}
Note that the resulting double sum in (\ref{eliminate_i2}) lost a term, resulting in only $n$ terms.

To simplify the $I^1_{ij}$ term, we add and subtract a term deviating from it by an order of $\Delta x$, use integration by parts on the new term, and with the result add and subtract another term to reduce the expression further.  To this end, let
\begin{equation}
\begin{split}
I_{\Delta S}\equiv\sum^{i=n+1}_{i=1,j}\int\int_{R_{ij}}\varphi_t\int^t_{t_j} &\left\{g(\mathbf{A}_{ij},\hat{u}(\xi-t_j,u^{RP}_{\Delta x}(\xi,x)),x) - g(\mathbf{A}_{ij},\hat{u}(\xi-t_j,u^{RP}_{\Delta x}(t,x)),x)\right.\\
-\frac{\partial f}{\partial \mathbf{A}}&(\mathbf{A}_{ij},\hat{u}(\xi-t,u^{RP}_{\Delta x}(\xi,x)))\cdot\mathbf{A}_{\Delta x}'\\
+&\frac{\partial f}{\partial \mathbf{A}}\left.(\mathbf{A}_{ij},\hat{u}(\xi-t,u^{RP}_{\Delta x}(t,x)))\cdot\mathbf{A}_{\Delta x}'\right\}d\xi dx dt.
\end{split}
\end{equation}
From the total variation bound on the Riemann problems and the smoothness of $f$, this term is bounded by
\begin{equation}
\begin{split}
\AutoAbs{I_{\Delta S}}&\leq \sum^{i=n+1}_{i=1,j}\int\int_{R_{ij}}\AutoNorm{\varphi_t}_\infty\int^t_{t_j}C\phantom{4}T.V._{[x_{i-1},x_i]}\left\{u_{\Delta x}(\cdot,t_j)\right\}d\xi dxdt\\
&\leq \AutoNorm{\varphi_t}_\infty C\Delta t^2\Delta x\sum_jT.V._{[r_{min},r_{max}]}\{u_{\Delta x}(\cdot,t_j)\}\\
&\leq CV\AutoNorm{\varphi_t}_\infty\Delta x\Delta t^2\frac{T}{\Delta t}=O(\Delta x^2),
\end{split}
\end{equation}
and the above procedure reduces the term to
\begin{equation}\label{splitting_i1_ij}
\begin{split}
&-\int\int_{R_{ij}}\varphi_t I^1_{ij}(t,x)dxdt = I_{\Delta S} - \sum^{i=n+1}_{i=1,j}\int\int_{R_{ij}}\varphi_t\int^t_{t_j}\left\{ g(\mathbf{A}_{ij},\hat{u}(\xi-t_j,u^{RP}_{\Delta x}(\xi,x)),x)\right.\\
&\phantom{444444444}-\frac{\partial f}{\partial\mathbf{A}} \left.(\mathbf{A}_{ij},\hat{u}(\xi-t_j,u^{RP}_{\Delta x}(\xi,x)))\cdot\mathbf{A}_{\Delta x}'\right\}dxdt\\
&=O(\Delta x^2) - \sum^{i=n+1}_{i=1,j}\int_{R_i}\left\{\varphi(t_{j+1},x)\int^{t_{j+1}}_{t_j} \left[g(\mathbf{A}_{ij},\hat{u}(\xi-t_j,u^{RP}_{\Delta x}(\xi,x)),x) \right.\right.\\
&\phantom{444444444}- \frac{\partial f}{\partial\mathbf{A}}\left.(\mathbf{A}_{ij},\hat{u}(\xi-t_j,u^{RP}_{\Delta x}(\xi,x)))\cdot\mathbf{A}_{\Delta x}'\right]d\xi\\
&\phantom{444444444}-\left.\int^{t_{j+1}}_{t_j}\varphi[g(\mathbf{A}_{ij},u_{\Delta x},x) - \frac{\partial f}{\partial\mathbf{A}}(\mathbf{A}_{ij},u_{\Delta x})\cdot\mathbf{A}_{\Delta x}']d\xi\right\}dx\\
&=O(\Delta x^2)-\sum^{i=n+1}_{i=1,j}\int_{R_i}\left\{\varphi(t_{j+1},x)\int^{t_{j+1}}_{t_j}\left[ g(\mathbf{A}_{ij},\hat{u}(\xi-t_j,u^{RP}_{\Delta x}(t_{j+1},x)),x)\right.\right.\\
&\phantom{444444444}-\left.\frac{\partial f}{\partial\mathbf{A}}\left.(\mathbf{A}_{ij},\hat{u}(\xi-t_j,u^{RP}_{\Delta x}(t_{j+1},x)))\cdot\mathbf{A}_{\Delta x}'\right]d\xi\right\}dt+I_4+I_5,
\end{split}
\end{equation}
where
\begin{equation}
\begin{split}
I_4\equiv\sum^{i=n+1}_{i=1,j}&\int_{R_i}\left\{\varphi(t_{j+1},x)\int^{t_{j+1}}_{t_j}\left[ g(\mathbf{A}_{ij},\hat{u}(\xi-t_j,u^{RP}_{\Delta x}(t_{j+1},x)),x)\right.\right.\\
&- g(\mathbf{A}_{ij},\hat{u}(\xi-t_j,u^{RP}_{\Delta x}(\xi,x)),x) - \frac{\partial f}{\partial\mathbf{A}}(\mathbf{A}_{ij},\hat{u}(\xi-t_j,u^{RP}_{\Delta x}(t_{j+1},x)))\cdot\mathbf{A}_{\Delta x}'\\
&+ \left.\frac{\partial f}{\partial\mathbf{A}}\left.(\mathbf{A}_{ij},\hat{u}(\xi-t_j,u^{RP}_{\Delta x}(\xi,x)))\cdot\mathbf{A}_{\Delta x}'\right]d\xi\right\}dx,
\end{split}
\end{equation}
and
\begin{equation}
I_5\equiv\sum^{i=n+1}_{i=1,j}\int\int_{R_{ij}}\varphi\left[g(\mathbf{A}_{ij},u_{\Delta x},x)-\frac{\partial f}{\partial\mathbf{A}}(\mathbf{A}_{ij},u_{\Delta x})\cdot\mathbf{A}_{\Delta x}'\right]dxdt.
\end{equation}
Again by smoothness and the total variation bound, we have
\begin{equation}
\begin{split}
\AutoAbs{I_4}&\leq\AutoNorm{\varphi}_\infty\sum^{i=n+1}_{i=1,j}C\phantom{4}T.V._{[x_{i-1},x_i]}\left\{u_{\Delta x}(\cdot,t_j)\right\}\Delta x\Delta t\\
&\leq\AutoNorm{\varphi}_\infty C\Delta x\Delta t\sum_jT.V._{[r_{min},r_{max}]}\left\{u_{\Delta x}(\cdot,t_j)\right\}
=\AutoNorm{\varphi}_\infty CV\Delta x\Delta t\frac{T}{\Delta t}=O(\Delta x).
\end{split}
\end{equation}
Substituting (\ref{eliminate_i2}) and (\ref{splitting_i1_ij}) into (\ref{intro_epsilon1}) along with using (\ref{frac_step_approx_soln}) as an identity leaves us with
\begin{equation}\label{f_jumps_cancel}
\begin{split}
\varepsilon=&\phantom{4}O(\Delta x) + \varepsilon_1 - \sum^{i=n+1}_{i=1,j}\int\int_{R_{ij}}\varphi\frac{\partial f}{\partial\mathbf{A}}(\mathbf{A}_{ij},u_{\Delta x})\cdot\mathbf{A}_{\Delta x}'dxdt\\
&+\sum^{i=n}_{i=1,j}\int_{R_j}\varphi(t,x_{i})\left\{f(\mathbf{A}_{i+1,j},u^{RP}_{\Delta x}(t,x_{i})) - f(\mathbf{A}_{ij},u^{RP}_{\Delta x}(t,x_{i}))\right\}dt
\end{split}
\end{equation}
The second sum represents the jump in the flux function $f$, resulting from the discontinuities in the metric $\mathbf{A}$, and the first sum is the addition to the ODE step (\ref{ode_soln}) specifically designed to cancel these jumps in the flux.

To see how the cancelation works, we perform a Taylor expansion on the test function, and we add and subtract terms deviating by order $\Delta x$.  The first sum in (\ref{f_jumps_cancel}) is expanded as
\begin{equation}\label{df_da_expansion}
\begin{split}
\sum^{i=n+1}_{i=1,j}\int\int_{R_{ij}}&\varphi\frac{\partial f}{\partial\mathbf{A}}(\mathbf{A}_{ij},u_{\Delta x})\cdot\mathbf{A}_{\Delta x}'dxdt\\
&=\sum^{i=n+1}_{i=1,j}\int\int_{R_{ij}}\varphi(x_{i},t)\frac{\partial f}{\partial\mathbf{A}}(\mathbf{A}_{ij},u_{\Delta x})\cdot\mathbf{A}_{\Delta x}'dxdt+O(\Delta x)\\
&=\sum^{i=n+1}_{i=1,j}\int_{R_j}\varphi(x_{i},t)\int_{R_i}\left\{\frac{\partial f}{\partial\mathbf{A}}(\mathbf{A}_{ij},u_{\Delta x})\cdot\mathbf{A}_{\Delta x}'-\frac{\partial f}{\partial\mathbf{A}}(\mathbf{A}_{ij},u^{RP}_{\Delta x})\cdot\mathbf{A}_{\Delta x}'\right\}dxdt\\
&+\sum^{i=n+1}_{i=1,j}\int_{R_j}\varphi(x_{i},t)\int_{R_i}\left\{\frac{\partial f}{\partial\mathbf{A}}(\mathbf{A}_{ij},u^{RP}_{\Delta x})\cdot\mathbf{A}_{\Delta x}'-\frac{\partial f}{\partial\mathbf{A}}(\mathbf{A}_{ij},u^{RP}_{\Delta x}(x_i,t))\cdot\mathbf{A}_{\Delta x}'\right\}dxdt\\
&+\sum^{i=n+1}_{i=1,j}\int_{R_j}\varphi(x_{i},t)\int_{R_i}\left\{\frac{\partial f}{\partial\mathbf{A}}(\mathbf{A}_{ij},u^{RP}_{\Delta x}(x_i,t))\cdot\mathbf{A}_{\Delta x}'\right.\\
&\phantom{4444444444}\left.-\frac{\partial f}{\partial\mathbf{A}}(\mathbf{A}_{\Delta x}(x+\frac{\Delta x}{2},t_j),u^{RP}_{\Delta x}(x_i,t))\cdot\mathbf{A}_{\Delta x}'\right\}dxdt\\
&+\sum^{i=n+1}_{i=1,j}\int_{R_j}\varphi(x_{i},t)\int^{x_i}_{x_{i-1}}\frac{\partial f}{\partial\mathbf{A}}(\mathbf{A}_{\Delta x}(x+\frac{\Delta x}{2},t_j),u^{RP}_{\Delta x}(x_i,t))\cdot\mathbf{A}_{\Delta x}'dxdt+O(\Delta x)
\end{split}
\end{equation}
From the smoothness of $f$, each of the first three sums in equation (\ref{df_da_expansion}) are $O(\Delta x)$ for the following reasons: the first sum is order $\Delta x$ from the ODE step in the definition of the approximate solution $u_{\Delta x}$ (\ref{frac_step_approx_soln}), the second sum is order $\Delta x^2$ by the total variation bound on solutions to the Riemann problems, and the third sum is order $\Delta x$ by the Lipschitz continuity of the metric $\mathbf{A}$.  After these bounds are established, (\ref{df_da_expansion}) reduces to
\begin{equation}\label{df_da_result}
\begin{split}
\sum^{i=n+1}_{i=1,j}&\int\int_{R_{ij}}\varphi\frac{\partial f}{\partial\mathbf{A}}(\mathbf{A}_{ij},u_{\Delta x})\cdot\mathbf{A}_{\Delta x}'dxdt\\
&=\sum^{i=n+1}_{i=1,j}\int_{R_j}\varphi(x_{i},t)\int^{x_i}_{x_{i-1}}\frac{\partial f}{\partial\mathbf{A}}(\mathbf{A}_{\Delta x}(x+\frac{\Delta x}{2},t_j),u^{RP}_{\Delta x}(x_i,t))\cdot\mathbf{A}_{\Delta x}'dxdt+O(\Delta x)\\
&=\sum^{i=n+1}_{i=1,j}\int_{R_j}\varphi(x_{i},t)\int^{x_i}_{x_{i-1}}\frac{\partial f}{\partial x}(\mathbf{A}_{\Delta x}(x+\frac{\Delta x}{2},t_j),u^{RP}_{\Delta x}(x_i,t))dxdt+O(\Delta x)\\
&=\sum^{i=n+1}_{i=1,j}\int_{R_j}\varphi(t,x_{i})\left\{f(\mathbf{A}_{i+1,j},u^{RP}_{\Delta x}(t,x_{i+})) - f(\mathbf{A}_{ij},u^{RP}_{\Delta x}(t,x_{i+}))\right\}dt + O(\Delta x).
\end{split}
\end{equation}
Plugging this result (\ref{df_da_result}) into (\ref{f_jumps_cancel}) gives us
\begin{equation}
\begin{split}
\varepsilon=&\phantom{4}O(\Delta x) + \varepsilon_1\\
&-\sum_{j}\int_{R_j}\varphi(t,x_{n+1})\left\{f(\mathbf{A}_{n+2,j},u^{RP}_{\Delta x}(t,x_{n+1})) - f(\mathbf{A}_{n+1,j},u^{RP}_{\Delta x}(t,x_{n+1}))\right\}dt,
\end{split}
\end{equation}
where one term remains due to the mismatch in the number of terms in the spatial sum.  Clearly, this last term is $O(\Delta x)$.

So the residual boils down to
\begin{equation}
\varepsilon(u_{\Delta x},\mathbf{A}_{\Delta x},\varphi)=\varepsilon_1(u_{\Delta x},\mathbf{A}_{\Delta x},\varphi) + O(\Delta x),
\end{equation}
with all that remains to show is
\begin{equation}
\varepsilon_1=\sum_{j\neq0}\int^{r_{max}}_{r_{min}}\varphi(t_j,x)\left\{u_{\Delta x}(t^+_j,x) - u_{\Delta x}(t^-_j,x)\right\}dx=O(\Delta x).
\end{equation}

To estimate $\varepsilon_1$, we break up the sum by each time step $t_j$ and define
\begin{equation}\label{epsilon_j_1_defn}
\begin{split}
\varepsilon^j_1&\equiv\int^{r_{max}}_{r_{min}}\varphi(t_j,x)\left\{u_{\Delta x}(t^+_j,x) - u_{\Delta x}(t^-_j,x)\right\}dx\\
&=\sum_i\int^{x_{i+}}_{x_i-}\varphi(t_j,x)\left\{u_{\Delta x}(t^+_j,x) - u_{\Delta x}(t^-_j,x)\right\}dx,
\end{split}
\end{equation}
with $x_{i+}\equiv x_{i+\frac{1}{2}}$ and $x_{i-}\equiv x_{i-\frac{1}{2}}$.

Recall, the approximate solution for the new time step $t^+_j$ is computed by the Godunov step, using averages at the top of each Riemann cell $R_{ij}$.  In particular, the solution at each new time step is
\begin{equation}
u_{\Delta x}(t^+_j,x)\equiv\hat{u}(t_j-t_{j-1},\bar{u}(t_j),x))
\end{equation}
where
\begin{equation}
\bar{u}(t_j)\equiv\frac{1}{\Delta x}\int^{x_{i+}}_{x_i-}u^{RP}_{\Delta x}(t_j,x)dx
\end{equation}

To finish the proof, a lemma is needed, which is proven at the end of this chapter.  This lemma states the difference of the ODE step taken on an average verses the solution to the Riemann problem across the top of the Riemann cell is bounded by the total variation of the Riemann problem.
\begin{lem}\label{lem:avg_vs_rp_estimate}
Let $u^{RP}_{\Delta x}$ represent the solution of the Riemann problem in the Riemann cell $R_{i,j-1}$ and $\bar{u}_{\Delta x}(t)$ denote the average of the Riemann problem solution across Riemann cell.  Let $\hat{u}$ be the solution obtained by the ODE step (\ref{ode_soln}) and $\varphi$ be a smooth test function.  Then the following bound holds
\begin{equation}\label{avg_vs_rp_ineq}
\begin{split}
&\AutoAbs{\int^{x_{i+}}_{x_{i-}}\left\{\hat{u}(t_j-t_{j-1},\bar{u}_{\Delta x}(t_j),x)-\hat{u}(t_j-t_{j-1},u^{RP}_{\Delta x}(t_j,x),x)\right\}\varphi(t_j,x)dx}\\
&\phantom{44444444444444}\leq\phantom{4} C\AutoNorm{\varphi}_\infty\Delta x\Delta t \phantom{4}T.V._{[x_i,x_{i+1}]}\{u_{\Delta x}(t_j,\cdot)\}
\end{split}
\end{equation}
for some constant C.
\end{lem}
Using Lemma \ref{lem:avg_vs_rp_estimate}, (\ref{epsilon_j_1_defn}) is rewritten as solutions to the ODE step (\ref{ode_soln}) and bounded by
\begin{equation}
\begin{split}
\varepsilon^j_1&=\sum_i\int^{x_{i+}}_{x_{i-}}\varphi(t_j,x)\left\{ \hat{u}(t_j-t_{j-1},\bar{u}(t_j),x) - \hat{u}(t_j-t_{j-1},u^{RP}_{\Delta x}(t_j,x),x)\right\}dx\\
&\leq C\AutoNorm{\varphi}_\infty\Delta x\Delta t\sum_i T.V._{[x_{i-},x_{i+}]}\{u^{RP}_{\Delta x}(\cdot,t_j)\} = C\AutoNorm{\varphi}_\infty\Delta x\Delta t \phantom{4}T.V._{[r_{min},r_{max}]}\{u^{RP}_{\Delta x}(\cdot,t_j)\}
\end{split}
\end{equation}
By the total variation bound on $u_{\Delta x}(t_j,\cdot)$, the residual is bounded by
\begin{equation}
\varepsilon_1\leq\sum_{j\neq0}C\AutoNorm{\varphi}_\infty\Delta x\Delta t \phantom{4}T.V._{[r_{min},r_{max}]}\{u^{RP}_{\Delta x}(\cdot,t_j)\} \leq C\frac{T}{\Delta t}\Delta x\Delta tV = O(\Delta x).
\end{equation}
Therefore, $\varepsilon=O(\Delta x)$ and the proof is complete.
\end{proof}

To prove Lemma \ref{lem:avg_vs_rp_estimate}, a preliminary result is needed: given a function on a set of points the difference of the function between any point and the average is bounded by the total variation of that function on the set.  This result is provided by the following
\begin{lem}\label{lem:avg_bound}
Let $u(x)$ be a function on the set $[x_{i-},x_{i+}]$ and
\begin{equation}\label{avg_defn}
\bar{u}=\frac{1}{\Delta x}\int^{x_{i+}}_{x_{i-}}u(x)dx
\end{equation}
be the average of $u$ on this set.  Then we have
\begin{equation}\label{avg_bnded_by_tv}
|\bar{u}-u(x)|\leq\sup_{x_1,x_2\in[x_i,x_{i+1}]}|u(x_1)-u(x_2)|\leq T.V._{[x_{i-},x_{i+}]}\{u(\cdot)\}.
\end{equation}

\end{lem}
\begin{proof}
The second inequality is true by the definition of the total variation
\begin{equation}
\sup_{x_1,x_2\in[x_i,x_{i+1}]}|u(x_1)-u(x_2)|\leq T.V._{[x_{i-},x_{i+}]}\{u(\cdot)\}.
\end{equation}
To prove the first inequality, we assume it is false to obtain a contradiction, so suppose there exists $x_*\in[x_{i-},x_{i+}]$ such that
\begin{equation}\label{avg_inequality}
\sup_{x_1,x_2\in[x_i,x_{i+1}]}|u(x_1)-u(x_2)|<|\bar{u}-u(x_*)|.
\end{equation}
Relabel the u-coordinates by an isometry $\varphi:u\rightarrow v$ that maps the point $u(x_*)$ to the origin in the v-coordinates (i.e. $\varphi(u(x_*)) = \mathbf{0}$), and the vector $\bar{u}-u(x_*)$ in the direction of the 1st coordinate $v^1$, as show in Figure \ref{fig:avg_map}.

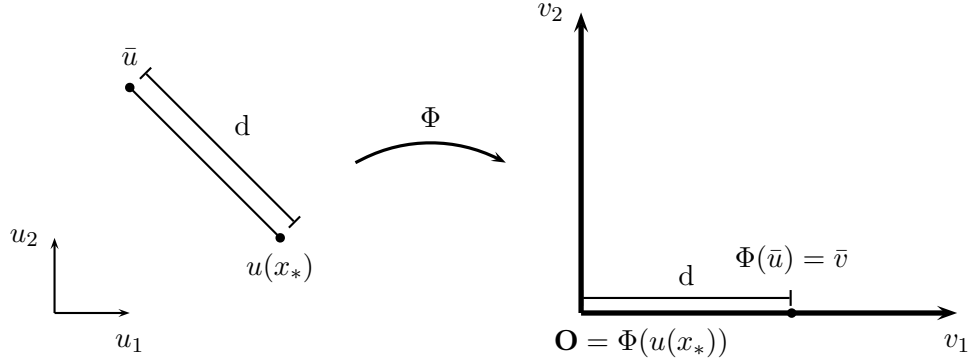
\begin{figure}[!h]
\begin{pspicture}(12,4)(0,-0.5)
%\psgrid
%u1 and u2 axises
\psline{->}(0,0)(1,0)
\psline{->}(0,0)(0,1)
%label the axises
\rput(-0.4,1){$u_2$}
\rput(1,-0.4){$u_1$}
%draw line from avg to arbitrary point and labels
\psline{*-*}(1,3)(3,1)
\psline{|-|}(1.2,3.2)(3.2,1.2)
\rput(1,3.4){$\bar{u}$}
\rput(3,0.6){$u(x_*)$}
\rput(2.5,2.5){d}
\psdots(9.8,0)
\psline{|-|}(7,0.2)(9.8,0.2)
\rput(8.4,0.5){d}
\rput(7.8,-0.4){$\mathbf{O}=\Phi(u(x_*))$}
\rput(9.8,0.7){$\Phi(\bar{u})=\bar{v}$}
%draw the map arrow and label
\psline[linewidth=1.2pt,linearc=2]{->}(4,2)(5,2.5)(6,2)
\rput(5,2.6){$\Phi$}
%v1 and v2 axises
\psline[linewidth=2pt]{->}(7,0)(12,0)
\psline[linewidth=2pt]{->}(7,0)(7,4)
%label the axises
\rput(6.6,4){$v_2$}
\rput(12,-0.4){$v_1$}
\end{pspicture}\caption{The isometry $\Phi:u\rightarrow v$}
\label{fig:avg_map}
\end{figure}

Since the average of a collection of points is independent of the coordinate system in which they are labeled in, we have
\begin{equation}
\bar{v}\equiv\frac{1}{\Delta x}\int^{x_{i+}}_{x_{i-}}v(x)dx=\frac{1}{\Delta x}\int^{x_{i+}}_{x_{i-}}\varphi(u(x))dx=\varphi\left(\frac{1}{\Delta x}\int^{x_{i+}}_{x_{i-}}u(x)dx\right)=\varphi(\bar{u})
\end{equation}
The following inequality holds by transforming equation (\ref{avg_inequality}) over to v-coordinates
\begin{equation}
|v(x)|=|u(x)-u(x_*)|<|\bar{u}-u(x_*)|=|\bar{v}|\phantom{4444}\forall x\in[x_{i-},x_{i+}],
\end{equation}
which implies
\begin{equation}
|\bar{v}|=\AutoAbs{\frac{1}{\Delta x}\int^{x_{i+1}}_{x_i}v(x)dx}\leq \frac{1}{\Delta x}\int^{x_{i+1}}_{x_i}|v(x)|dx< \frac{1}{\Delta x}\int^{x_{i+1}}_{x_i}|\bar{v}|dx =|\bar{v}|.
\end{equation}
This inequality $|\bar{v}|<|\bar{v}|$ is an obvious contradiction, proving the first inequality in (\ref{avg_bnded_by_tv}).
\end{proof}
Now we prove the lemma used in the proof of Theorem \ref{thm:weak_soln}

\vspace{.2cm}
\noindent{\bf Proof of Lemma \ref{lem:avg_vs_rp_estimate}}
Recall the solution to the ODE step has the form:
\begin{equation}
\hat{u}(t_j-t_{j-1},u^{RP}_{\Delta x}(t_j,x),x)=u^{RP}_{\Delta x}(t,x)+\int^{t_j}_{t_{j-1}}G(\mathbf{A}_{ij},u^{RP}_{\Delta x}(t,x),x)dt.
\end{equation}
This solution implies the LHS of (\ref{avg_vs_rp_ineq}) is written out as
\begin{equation}
\begin{split}
&\AutoAbs{\int^{x_{i+}}_{x_{i-}}\left\{\hat{u}(t_j-t_{j-1},\bar{u}_{\Delta x}(t_j),x)-\hat{u}(t_j-t_{j-1},u^{RP}_{\Delta x}(t_j,x),x)\right\}\varphi(t_j,x)dx}\\
=&\left\vert\int^{x_{i+}}_{x_{i-}}\left\{(\bar{u}_{\Delta x}(t_j)-u^{RP}_{\Delta x}(t_j,x))\right.\right. \\
&\left.+\int^{t_j}_{t_{j-1}}\left.(G(\mathbf{A}_{ij},\bar{u}_{\Delta x}(t),x)-G(\mathbf{A}_{ij},u^{RP}_{\Delta x}(t,x),x))dt\right\}\varphi(t_j,x)dx\right\vert \\
=&\left\vert\int^{x_{i+}}_{x_{i-}}\left\{\bar{u}_{\Delta x}(t_j)-u^{RP}_{\Delta x}(t_j,x)\right\}\varphi(t_j,x_i)dx\right.\\ &\left.+\int^{x_{i+}}_{x_{i-}}\int^{t_j}_{t_{j-1}}\left\{G(\mathbf{A}_{ij},\bar{u}_{\Delta x}(t_j),x)- G(\mathbf{A}_{ij},u^{RP}_{\Delta x}(t_j,x),x)\right\}dt\phantom{4}\varphi(t_j,x_i)dx\right\vert \\
& + O(\Delta x^2),
\end{split}
\end{equation}
where the test function in the first term is approximated by a Taylor expansion.  By the definition of the average function $\bar{u}$, the first term is zero.  By the smoothness of $G$, the bound (\ref{avg_vs_rp_ineq}) is proven by
\begin{equation}
\begin{split}
&\AutoAbs{\int^{x_{i+}}_{x_{i-}}\left\{\hat{u}(t_j-t_{j-1},\bar{u}_{\Delta x}(t_j),x)-\hat{u}(t_j-t_{j-1},u^{RP}_{\Delta x}(t_j,x),x)\right\}\varphi(t_j,x)dx}\\
&\phantom{4444}\leq\phantom{4}C\AutoNorm{\varphi}_\infty\Delta x\Delta t\sup_{x_{i-}<x<x_{i+}}\{\AutoAbs{\bar{u}_{\Delta x}(t_j)-u^{RP}_{\Delta x}(t_j,x)}\}\\
&\phantom{4444}\leq\phantom{4}C\AutoNorm{\varphi}_\infty\Delta x\Delta t \phantom{4}T.V._{[x_{i-},x_{i+}]}\{u_{\Delta x}(t_j,\cdot)\},
\end{split}
\end{equation}
where Lemma \ref{lem:avg_bound} is used to bound the difference between the average and the solution to the Riemann problem.
\begin{flushright}
$\Box$
\end{flushright}

   \chapter[%
      Short Title of 6th Ch.
   ]{%
      Continuous Models
   }%
   \label{ch:cont_models}
Before simulating the shock wave models, we briefly explore simulating more simple ones, which we denote as the {\it continuous models}.  The continuous models are the pure FRW-1, FRW-2, and TOV metrics developed in Chapter \ref{ch:family_of_shock_waves}, and these models are not as complex as the discontinuous shock wave models.  Our purpose in simulating these models is to use exact formulas to test the numerical convergence and the accuracy of our locally inertial Godunov scheme.  Each of the three models embodies a different time dilation scenario, each of which occur in the simulation of the shock wave models.  The FRW-1 metric is a model in which the (coordinate) speed of light is uniformly equal to one independent of time and space; therefore, time dilation and synchronization does not occur in this model.  In the TOV model the (coordinate) speed of light increases from one side of the simulated region to the other, so there exists time dilation between frames in this model.  Since the TOV metric is static, the synchronization of the clocks has no effect on the construction of the solution.  In the FRW-2 model, we choose an appropriate integrating factor to force the (coordinate) speed of light to be equal to one initially, but in this model the speed of light increases uniformly across the entire model as time progresses.  This model provides the scenario of a dynamical (coordinate) speed of light, so synchronizing time against the unitary frame must be handled correctly to obtain numerical convergence.  As a side note, it was serendipitous that we came across the FRW-2 model; this model allowed us to perfect our clock synchronization by providing an ideal test case.  Besides running these test cases, the simulation of the continuous models are important for correctly handling the non-interaction region and ghost cells for the shock wave model.  As discussed in more detail in the next chapter, emanating from the initial discontinuity is an interaction region surrounded by non-interaction regions on each side.  These non-interaction regions are the continuous models and their numerical accuracy and consistency is paramount to the success in simulating the shock wave model.

The models we consider in this paper have infinite extent, with the radii going from zero to infinity, and we refer to this infinite region of space as the universe of our simulation.  Even though theoretically space can extend out to infinity, the computer can only simulate a finite region of spacetime in the model, so we need to demarcate the minimum radius $\bar{r}_{min}$ and the maximum radius $\bar{r}_{max}$ of the simulated region, along with the number of gridpoints $n$ to simulate it.  Also, we need to decide when to begin the simulation of our models, so a start time $t_0$ is chosen.  We choose to set these parameters as
\begin{equation}\label{cont_std_params}
\bar{r}_{min}=3,\phantom{4444}\bar{r}_{max}=7,\phantom{4444}\bar{t}_0=15, \phantom{4444}n=2^{14}=16,384,
\end{equation}
in all the simulations within this chapter.  The units for these parameters are given meaning and are discussed in Chapter \ref{ch:units}.  There is an extra parameter for each of the models; it is the integrating factor constant $\Psi_0$ for the FRW-1 and FRW-2 models and the constant $B_0$ for the TOV model.  This parameter is a time scale factor that affects the rate at which the clocks run in the model, and it is assigned a value as we consider each model separately.

This chapter is organized into four sections.  The first three sections, Sections \ref{sec:frw1_results}-\ref{sec:frw2_results}, are dedicated to simulating the FRW-1, TOV, and FRW-2 metrics, respectively.  These sections start out by reiterating the equations developed in Chapter \ref{ch:family_of_shock_waves} that are needed to run the simulation.  More precisely, these equations are used to build the initial profiles and ghost cells. With this data set, we use our locally inertial Godunov scheme to simulate these models and provide a glimpse of them with various pictures.  In Section \ref{sec:converge_results}, we discuss how the errors and convergence rates are computed, and we record the numerical convergence in the simulation of all three continuous models.

\section{Simulating the FRW-1 Model}
\label{sec:frw1_results}

\begin{figure}
\begin{center}
\includegraphics[width=\textwidth]{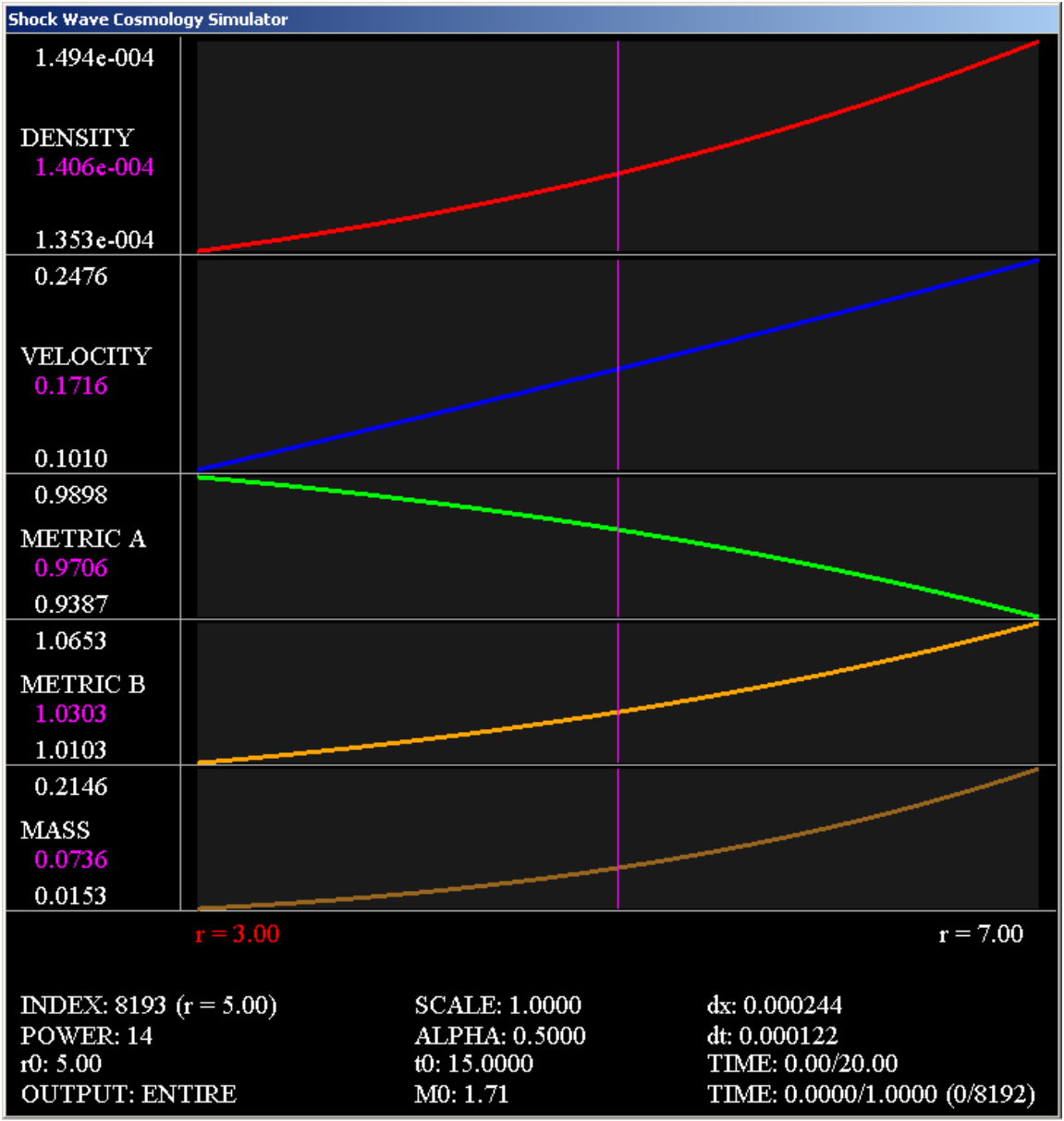}
\end{center}
\caption{Initial profiles for FRW-1 model}
\label{fig:frw1_init}
\end{figure}

Recall, using the coordinate transformation
\begin{equation}\label{ch5_frw1_xform_ssc}
\begin{split}
\bar{r}=&\sqrt{t}r,\\
\bar{t}=&\left\{1+\frac{\bar{r}^2}{4t^2}\right\}t=t+\frac{r^2}{4},
\end{split}
\end{equation}
the FRW metric in standard Schwarzschild coordinates (the FRW-1 form) is
\begin{equation}\label{ch5_frw1_in_ssc}
ds^{2} =-\frac{1}{1-v^2}d\bar{t}^2+\frac{1}{1-v^2}d\bar{r}^2+\bar{r}^{2}d\Omega^2,
\end{equation}
with the metric components
\begin{equation}\label{ch5_frw1_metric_in_ssc}
A(\xi) = 1-v(\xi)^2,\phantom{4444} B(\xi) = \frac{1}{1-v(\xi)^2},
\end{equation}
and the fluid variables
\begin{equation}\label{ch5_frw1_fluid_vars_in_ssc}
\rho(\xi,\bar{r}) =\frac{3v(\xi)^2}{\kappa\bar{r}^2},\phantom{4444} v(\xi)=\frac{1-\sqrt{1-\xi^2}}{\xi},
\end{equation}
where the variable $\xi$ is defined as $\xi=\bar{r}/\bar{t}$ and is a function of $v$ by the following relation
\begin{equation}\label{ch5_xi_v_relation}
\xi=\frac{2v}{1+v^2}.
\end{equation}

The integrating factor constant is suppressed, which is equivalent to setting $\Psi_0=1$.  With this parameter set, the initial profiles and ghost cells for the simulation are constructed by equations (\ref{ch5_frw1_in_ssc})-(\ref{ch5_xi_v_relation}) along with the standard parameters (\ref{cont_std_params}).  The resulting profiles are illustrated in Figure \ref{fig:frw1_init}.  The graphs for the density, velocity, metric component $A$, metric component $B$, and the mass function are sectioned into five panels displayed from top to bottom and color coded by the colors red, blue, green, yellow, and brown, respectively.  Each graph is scaled and shifted accordingly in the panel to ensure the top represents the maximum value and the bottom represents the minimum value of the graph in the region of space under consideration.  The numeric values for the minimum and maximum points on the graph are shown on the left side.  For example, the velocity profile is the blue graph second from the top with a minimum value of 0.1010 and maximum value of 0.2476 between the region of space $[r_{min},r_{max}]$.  For this particular example, the minimum value occurs at the left most point $r_{min}$, while the maximum occurs at $r_{max}$.  The vertical magenta line is a marker the user controls to examine a particular gridpoint in the simulation, with the corresponding numerical values displayed in magenta on the left side.  Another piece of information associated with this marker is the scale factor or (coordinate) speed of light ($\sqrt{AB}$) which is shown under the graphs and labeled "SCALE". Looking at the graphs in Figure \ref{fig:frw1_init}, all of the profiles are increasing functions except the metric component $A$ is decreasing. The graph of metric component $A$ in Figure \ref{fig:frw1_init} is the mirror image of the graph $B$, which confirms they are recipricals of one another (\ref{ch5_frw1_metric_in_ssc}).

\begin{figure}
\begin{center}
\includegraphics[width=\textwidth]{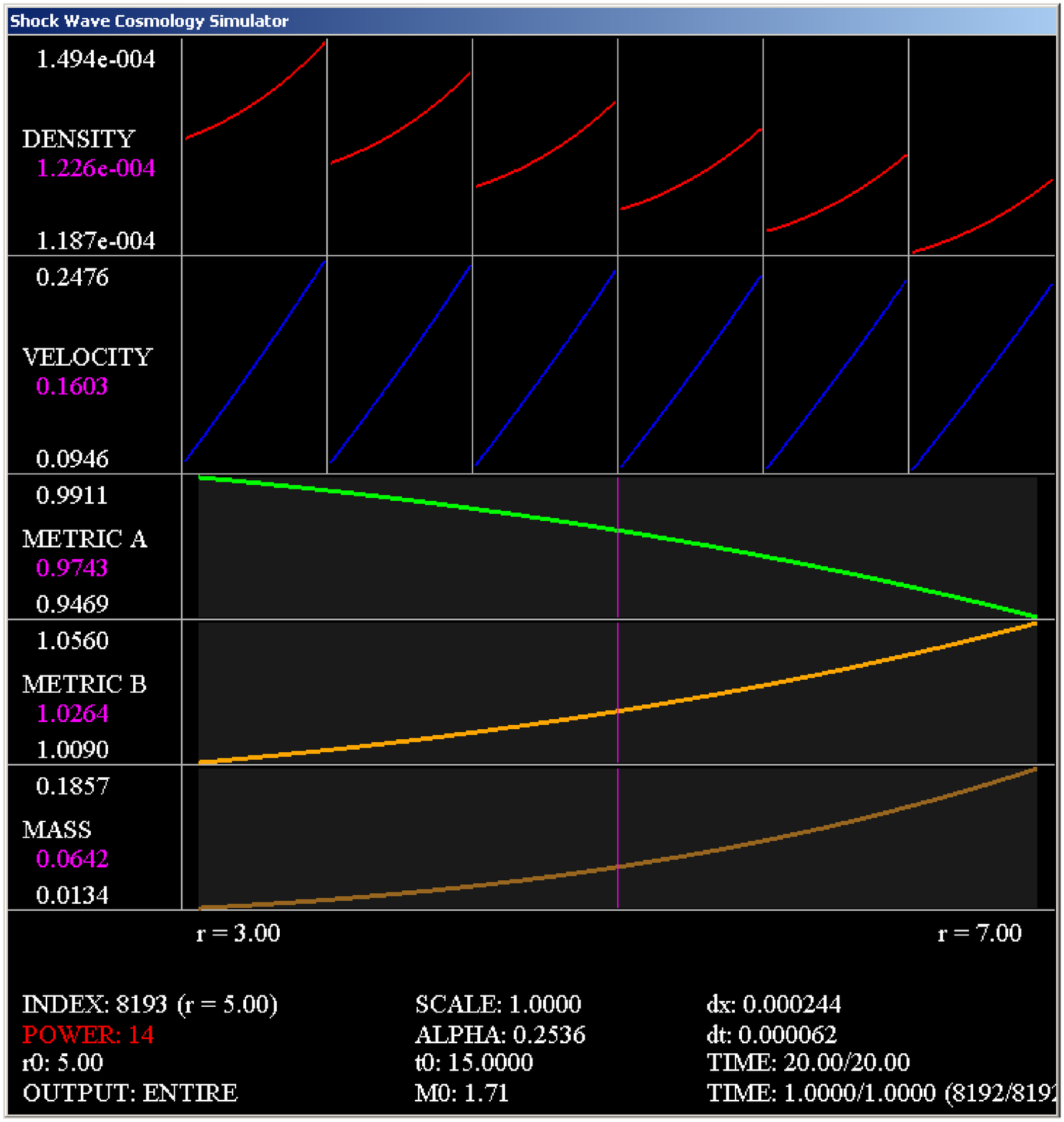}
\end{center}
\caption{Evolution of the FRW-1 model during a unit of time}
\label{fig:frw1_frames}
\end{figure}

The simulation is run for one unit of time, $t_{end} = t_0+1$, with the results depicted in Figure \ref{fig:frw1_frames}.  This figure features the evolution of the fluid variables $(\rho,v)$, showing snapshots of the density and velocity profiles at the times, going from the left to the right frame, $\bar{t}=\bar{t}_0$, $\bar{t}=\bar{t}_0+0.2$, $\bar{t}=\bar{t}_0+0.4$, $\bar{t}=\bar{t}_0+0.6$, $\bar{t}=\bar{t}_0+0.8$, and $\bar{t}=\bar{t}_{end}$.  The fluid variables decrease as time progresses, with the density decreasing faster than the velocity.  Since a positive velocity indicates matter is moving outward, we expect this decrease in density corresponding to a decrease in the mass function.  The metric has little change, and moreover, the solution as a whole has a very similar shape from the initial profiles seen in Figure \ref{fig:frw1_init}.  Notice the (coordinate) speed of light is identically equal to one, which is confirming the fact the metric components $A$ and $B$ are recipricals (\ref{ch5_frw1_metric_in_ssc}) at all times.

\section{Simulating the TOV Model}
\label{sec:tov_results}

\begin{figure}
\begin{center}
\includegraphics[width=\textwidth]{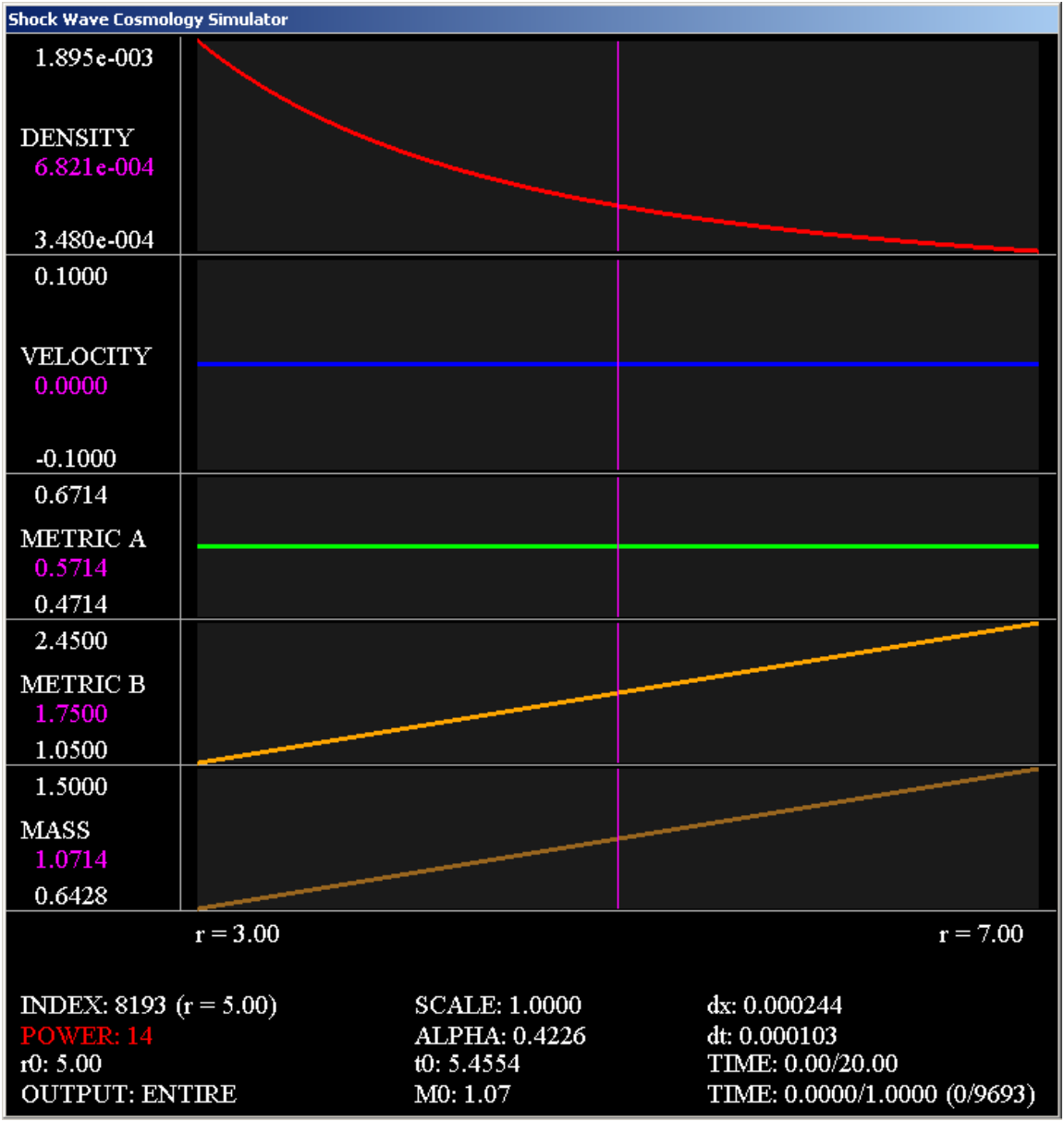}
\end{center}
\caption{Initial profiles for TOV model}
\label{fig:tov_init}
\end{figure}

Remember, the TOV metric takes the form
\begin{equation}\label{ch5_tov_in_ssc}
ds^2=-B(\bar{r})d\bar{t}^2+\left(\frac{1}{1-\frac{2\mathcal{G}M(\bar{r})}{\bar{r}}}\right)d\bar{r}^2 +\bar{r}^2d\Omega^2,
\end{equation}
with the metric components
\begin{equation}\label{ch5_tov_metric_in_ssc}
A(\bar{r}) = 1-8\pi\mathcal{G}\gamma,\phantom{4444} B(\bar{r}) = B_0(\bar{r})^{\frac{4\sigma}{1+\sigma}},
\end{equation}
and the fluid variables
\begin{equation}\label{ch5_tov_fluid_vars_in_ssc}
\rho(\bar{r})=\frac{\gamma}{\bar{r}^2},\phantom{4444} v(\bar{r})=0,
\end{equation}
where the parameter $\gamma$ is a constant dependent on $\sigma$
\begin{equation}\label{ch5_gamma}
\gamma=\frac{1}{2\pi\mathcal{G}}\left(\frac{\sigma}{1+6\sigma+\sigma^2}\right).
\end{equation}

\begin{figure}
\begin{center}
\includegraphics[width=\textwidth]{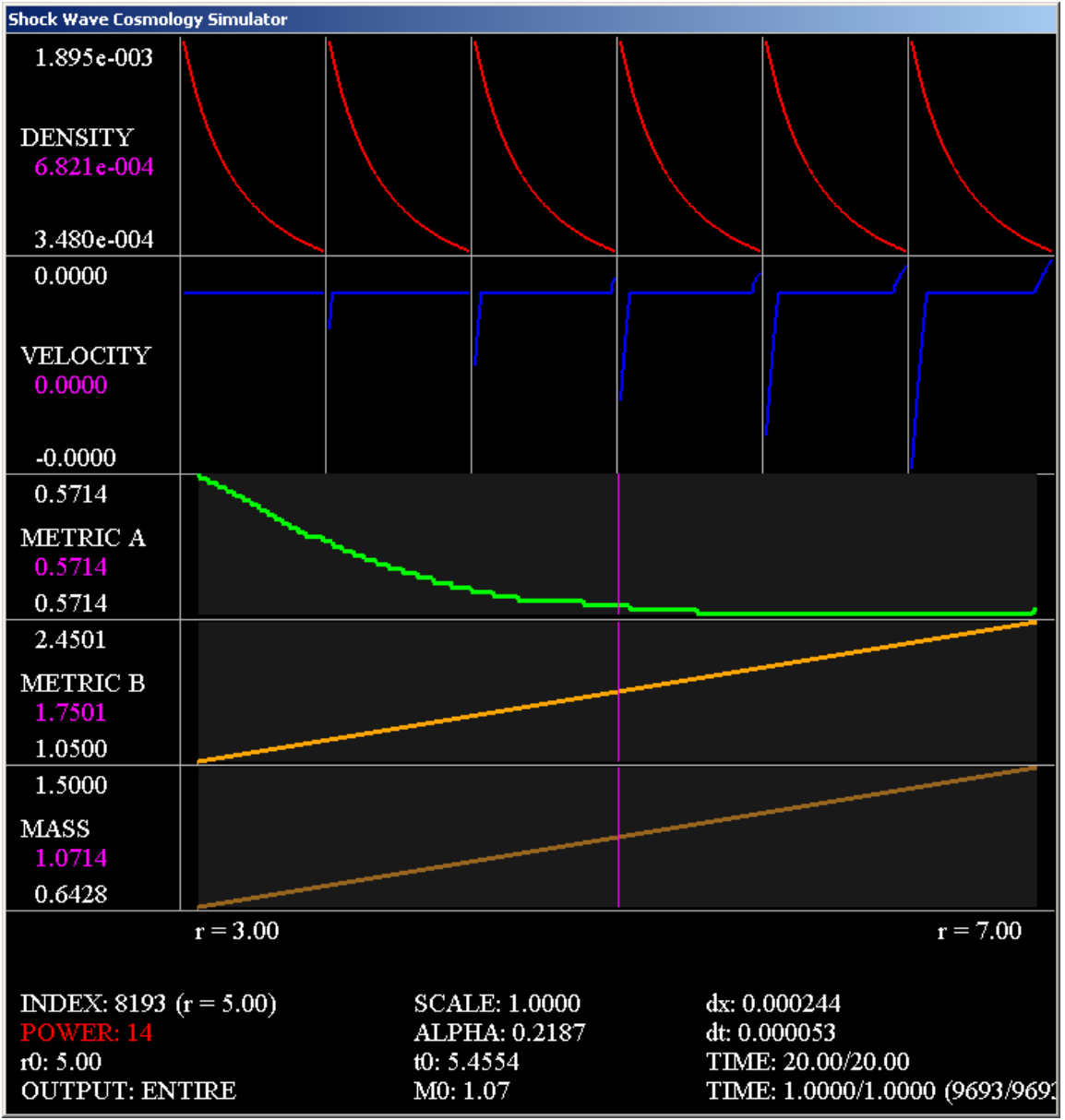}
\end{center}
\caption{Evolution of the TOV model during a unit of time}
\label{fig:tov_frames}
\end{figure}

Setting the extra parameter $B_0=1$ and using the standard parameters (\ref{cont_std_params}), the initial profiles and ghost cells are built from these equations (\ref{ch5_tov_in_ssc})-(\ref{ch5_gamma}) (with $\sigma=1/3$), and the resulting profiles are shown in Figure \ref{fig:tov_init}.  Notice the ghost cells can be set initially and do not need to be changed since the TOV metric is static, independent of time.  This model contrasts with the evolutionary nature of the FRW-1 model in Figure \ref{fig:frw1_init}.  In the TOV model, the density is a decreasing function, while the velocity and metric component $A$ are both constant functions.  The metric component $B$ and the mass function are always increasing functions since they are computed as the integral of positive values as seen by equations (\ref{compute_metric_b}) and (\ref{compute_mass}).  Since the metric $A$ is constant and the metric $B$ is increasing, the (coordinate) speed of light is increasing across the universe, so time dilation between the different frames occur in this model.  In particular,  the speed of light is 0.58 at $r_{min}$ and 0.89 at $r_{max}$.

Running the simulation for one unit of time, the evolution of the TOV metric along this time frame is shown in Figure \ref{fig:tov_frames}.  This metric is static, so we expect it to remain unchanged throughout this simulation.  The density, metric $B$, and the mass profiles are unchanged.  The constant profiles of the velocity and metric $A$ appear to have changed.  This appearance is due to the small numerical errors and the auto-zoom feature of the graphing tool which magnifies these errors.

\section{Simulating the FRW-2 Model}
\label{sec:frw2_results}

\begin{figure}
\begin{center}
\includegraphics[width=\textwidth]{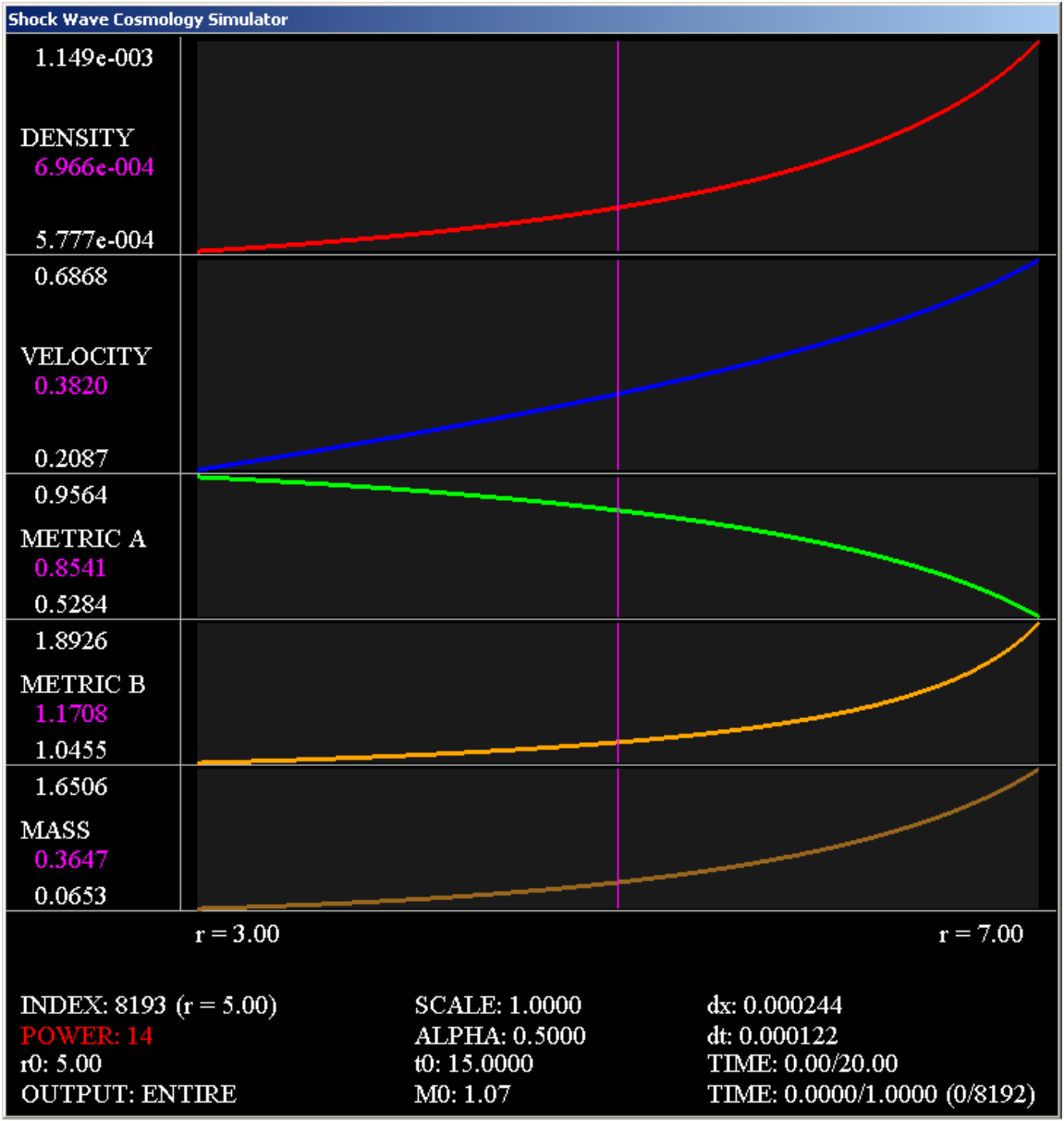}
\end{center}
\caption{Initial profiles for FRW-2 model}
\label{fig:frw2_init}
\end{figure}

Remember, using the coordinate transformation
\begin{equation}\label{ch5_frw2_xform_ssc}
\begin{split}
\bar{r}=&\sqrt{t}r,\\
\bar{t}=&\frac{\Psi_0}{2}\sqrt{\frac{4t^2+\bar{r}^2}{t}},
\end{split}
\end{equation}
the FRW metric in standard Schwarzschild coordinates (the FRW-2 form) is
\begin{equation}\label{ch5_frw2_in_ssc}
ds^{2} =-\frac{1}{\Psi^2(1-v^2)}d\bar{t}^2+\frac{1}{1-v^2}d\bar{r}^2+\bar{r}^{2}d\Omega^2,
\end{equation}
with the metric components
\begin{equation}\label{ch5_frw2_metric_in_ssc}
A(\bar{t},\bar{r}) = 1-v^2,\phantom{4444} B(\bar{t},\bar{r}) = \frac{1}{\Psi^2(1-v^2)},
\end{equation}
the integrating factor
\begin{equation}\label{ch5_frw2_integrating factor}
\Psi(\bar{t},\bar{r})=\Psi_0\sqrt{\frac{t}{4t^2+\bar{r}^2}},
\end{equation}
and the fluid variables
\begin{equation}\label{ch5_frw2_fluid_vars_in_ssc}
\rho(\bar{t},\bar{r}) =\frac{3}{4\kappa t^2},\phantom{4444} v(\bar{t},\bar{r})=\frac{\eta}{2}=\frac{\bar{r}}{2t}.
\end{equation}
Unlike FRW-1, the FRW-2 metric relies on the FRW time coordinate $t$, which is the following function of the new variables
\begin{equation}\label{ch5_frw2_t}
t(\bar{t},\bar{r})=\frac{\bar{t}^2+\sqrt{\bar{t}^4-\bar{r}^2\Psi^4_0}}{2\Psi^2_0}.
\end{equation}

\begin{figure}
\begin{center}
\includegraphics[width=\textwidth]{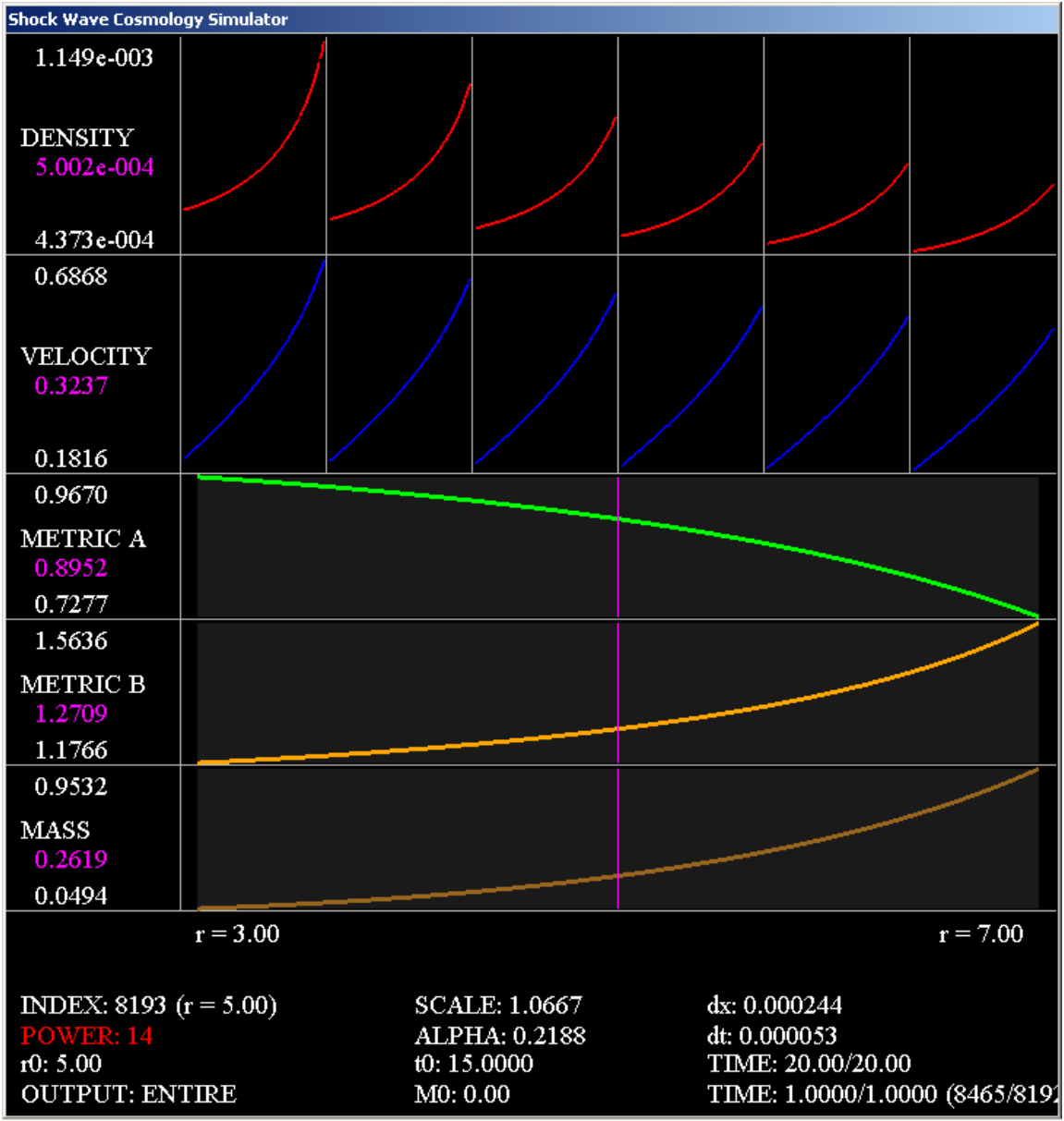}
\end{center}
\caption{Evolution of the FRW-2 metric during a unit of time}
\label{fig:frw2_frames}
\end{figure}

To match the FRW-1 metric more closely, we choose the integrating factor constant $\Psi_0$ so that the (coordinate) speed of light $\sqrt{AB}$ is equal to one.  To this end, we set
\begin{equation}\label{frw2_psi0_init}
\Psi_0=\sqrt{2\bar{t}_0}.
\end{equation}
This relationship is explained in more detail in the next chapter and is a rearrangement of equation (\ref{frw2_start_time}).  The initial profiles and ghost cells for the simulation are created by equations (\ref{ch5_frw2_in_ssc})-(\ref{ch5_frw2_t}), using the standard parameters (\ref{cont_std_params}) and $\Psi_0$ set in (\ref{frw2_psi0_init}).  These profiles are pictured in Figure \ref{fig:frw2_init}.  In this figure there is more curvature in the graphs of FRW-2 than the FRW-1 model (Figure \ref{fig:frw1_init}).  This curvature produces higher values in all the increasing graphs (density, velocity, metric $B$, and mass) and lower values in the decreasing graphs (metric $A$).

Allowing the simulation to run for one unit of time, the evolution of the FRW-2 metric is shown in Figure \ref{fig:frw2_frames}.  Again, the FRW-2 metric follows the same overall pattern as the FRW-1 metric (Figure \ref{fig:frw1_frames}), and the fluid variables are decreasing as time progresses.  We expect these similarities since both metrics, FRW-1 and FRW-2, are derived from the same FRW metric.  In the FRW-2 universe, the (coordinate) speed of light changes uniformly across the simulated region from 1 at $t_0$ to 1.0667 at $t_{end}$.  Thus, the FRW-2 model has a dynamic (coordinate) speed of light.  In contrast to the FRW-1 and TOV models, our simulation of the FRW-2 model tests the locally inertial Godunov method in a setting in which the (coordinate) speed of light is dynamic.

\section{Convergence Results}
\label{sec:converge_results}

This section discusses the numerical convergence results for all three of the continuous models.  These results are recorded for the FRW-1, TOV, and FRW-2 metrics in Table \ref{tab:frw1_converge}, \ref{tab:tov_converge}, and \ref{tab:frw2_converge}, respectively.  The tables are organized to show, from left to right, the number of gridpoints used, the density results, the velocity results, the metric $A$ results, and the metric $B$ results.  The variable results are partitioned into two values: the one-norm error and the convergence rate.  Since we possess exact formulas for the solution of all three models, the 1-norm error is numerically computed between the numerical solution and this exact solution.  The 1-norm is a natural norm to use for conservation laws because it requires integrating the function, and the weak form of the conservation law gives us information about these integrals.  The 1-norm is our chosen method for computing the error and showing numerical convergence in this paper.  The convergence rate is computed by taking the $\log_2$ of the ratio in successive errors, enabling us to measure the decrease in error relative to the increase in the number of gridpoints.  For example, a rate of 1 means that using twice the number of gridpoints reduces the error by half.  A rate less than 1 means the error is reduced by less than a half, while a rate greater than 1 means the error is reduced by more than a half.  Looking at the tables, all the errors are decreasing as the number of gridpoints increase with convergence rates around 1.  The convergence rate of 1 is expected because we are implementing a first order method on continuous solutions.  These results indicate the simulation is producing an accurate numerical representation of all three models, giving us the green light to simulate the discontinuous models in the following chapters.

\newpage
\begin{table}[h!]
\begin{center}
\begin{tabular}{|c|c|c|c|c|c|c|c|c|}
\hline
Number&\multicolumn{2}{|c|}{$\rho$}&\multicolumn{2}{|c|}{$v$}&\multicolumn{2}{|c|}{$A$}&\multicolumn{2}{|c|}{$B$}\\
\cline{2-9}
Gridpoints & Error & Rate & Error & Rate & Error & Rate & Error & Rate\\
\hline
64 & 2.660e-007 & N/A & 8.956e-004 & N/A & 1.526e-003 & N/A & 3.091e-003 & N/A\\
\hline
128 & 1.340e-007 & 0.99 & 4.543e-004 & 0.98 & 7.531e-004 & 1 & 1.529e-003 & 1\\
\hline
256 & 6.730e-008 & 0.99 & 2.289e-004 & 0.99 & 3.742e-004 & 1 & 7.607e-004 & 1\\
\hline
512 & 3.370e-008 & 1 & 1.149e-004 & 0.99 & 1.865e-004 & 1 & 3.794e-004 & 1\\
\hline
1024 & 1.690e-008 & 1 & 5.760e-005 & 1 & 9.310e-005 & 1 & 1.894e-004 & 1\\
\hline
2048 & 8.450e-009 & 1 & 2.880e-005 & 1 & 4.650e-005 & 1 & 9.470e-005 & 1\\
\hline
4096 & 4.230e-009 & 1 & 1.440e-005 & 1 & 2.320e-005 & 1 & 4.730e-005 & 1\\
\hline
8192 & 2.110e-009 & 1 & 7.220e-006 & 1 & 1.160e-005 & 1 & 2.370e-005 & 1\\
\hline
16384 & 1.060e-009 & 1 & 3.610e-006 & 1 & 5.810e-006 & 1 & 1.180e-005 & 1\\
\hline
\end{tabular}
\end{center}
\caption{Convergence results for the FRW-1 model}
\label{tab:frw1_converge}
\end{table}

\begin{table}[h!]
\begin{center}
\begin{tabular}{|c|c|c|c|c|c|c|c|c|}
\hline
Number&\multicolumn{2}{|c|}{$\rho$}&\multicolumn{2}{|c|}{$v$}&\multicolumn{2}{|c|}{$A$}&\multicolumn{2}{|c|}{$B$}\\
\cline{2-9}
Gridpoints & Error & Rate & Error & Rate & Error & Rate & Error & Rate\\
\hline
64 & 2.660e-007 & N/A & 8.956e-004 & N/A & 1.526e-003 & N/A & 3.091e-003 & N/A\\
\hline
128 & 1.340e-007 & 0.99 & 4.543e-004 & 0.98 & 7.531e-004 & 1 & 1.529e-003 & 1\\
\hline
256 & 6.730e-008 & 0.99 & 2.289e-004 & 0.99 & 3.742e-004 & 1 & 7.607e-004 & 1\\
\hline
512 & 3.370e-008 & 1 & 1.149e-004 & 0.99 & 1.865e-004 & 1 & 3.794e-004 & 1\\
\hline
1024 & 1.690e-008 & 1 & 5.760e-005 & 1 & 9.310e-005 & 1 & 1.894e-004 & 1\\
\hline
2048 & 8.450e-009 & 1 & 2.880e-005 & 1 & 4.650e-005 & 1 & 9.470e-005 & 1\\
\hline
4096 & 4.230e-009 & 1 & 1.440e-005 & 1 & 2.320e-005 & 1 & 4.730e-005 & 1\\
\hline
8192 & 2.110e-009 & 1 & 7.220e-006 & 1 & 1.160e-005 & 1 & 2.370e-005 & 1\\
\hline
16384 & 1.060e-009 & 1 & 3.610e-006 & 1 & 5.810e-006 & 1 & 1.180e-005 & 1\\
\hline
\end{tabular}
\end{center}
\caption{Convergence results for the TOV model}
\label{tab:tov_converge}
\end{table}

\begin{table}[h!]
\begin{center}
\begin{tabular}{|c|c|c|c|c|c|c|c|c|}
\hline
Number&\multicolumn{2}{|c|}{$\rho$}&\multicolumn{2}{|c|}{$v$}&\multicolumn{2}{|c|}{$A$}&\multicolumn{2}{|c|}{$B$}\\
\cline{2-9}
Gridpoints & Error & Rate & Error & Rate & Error & Rate & Error & Rate\\
\hline
64 & 8.060e-006 & N/A & 2.549e-003 & N/A & 1.025e-002 & N/A & 1.195e-002 & N/A\\
\hline
128 & 3.990e-006 & 1 & 1.263e-003 & 1 & 5.075e-003 & 1 & 6.048e-003 & 0.98\\
\hline
256 & 1.980e-006 & 1 & 6.288e-004 & 1 & 2.542e-003 & 1 & 3.227e-003 & 0.91\\
\hline
512 & 9.890e-007 & 1 & 3.136e-004 & 1 & 1.289e-003 & 0.98 & 1.899e-003 & 0.76\\
\hline
1024 & 4.940e-007 & 1 & 1.568e-004 & 1 & 6.240e-004 & 1 & 7.170e-004 & 1.4\\
\hline
2048 & 2.470e-007 & 1 & 7.830e-005 & 1 & 3.131e-004 & 0.99 & 3.704e-004 & 0.95\\
\hline
4096 & 1.230e-007 & 1 & 3.920e-005 & 1 & 1.579e-004 & 0.99 & 1.996e-004 & 0.89\\
\hline
8192 & 6.170e-008 & 1 & 1.960e-005 & 1 & 8.030e-005 & 0.98 & 1.181e-004 & 0.76\\
\hline
16384 & 3.080e-008 & 1 & 9.790e-006 & 1 & 3.890e-005 & 1 & 4.470e-005 & 1.4\\
\hline
\end{tabular}
\end{center}
\caption{Convergence results for the FRW-2 model}
\label{tab:frw2_converge}
\end{table}

   \chapter[%
      Short Title of 7th Ch.
   ]{%
      Shock Wave Model
   }%
   \label{ch:shock_wave_models}
In this chapter we perform the numerical simulation of a spherically symmetric general relativistic outgoing shock wave, together with the secondary incoming reflected wave.  This simulation demonstrates the secondary reflected wave is also a shock wave.  To reiterate, we interpret this result as the numerical resolution of the secondary reflected wave associated with the Smoller and Temple shock wave model \cite{smolte, smolte3, smolte4}.  That is, in \cite{smolte}, Smoller and Temple (Sm/Te) constructed an exact shock wave solution of the Einstein equations consisting of an inner FRW spacetime blasting outward into a TOV spacetime, such that the interface between them was a true fluid dynamical shock interface.  The metrics satisfied equations of state $p=\sigma\rho$ (FRW) and $\bar{p}=\bar{\sigma}\bar{\rho}$ (TOV), where $\sigma$ and $\bar{\sigma}$ were constant, consistent with an isothermal scenario.   Since the outer TOV solution was inverse square in the density, the solution could be interpreted as a blast wave propagating outward into a {\it static singular isothermal sphere}, c.f. \cite{smolte}.  To get exact formulas, the Sm/Te model assumed different sound speeds (temperatures) ahead and behind the shock wave.  This determined $\sigma$ as a function of $\bar{\sigma}$, which we interpret here as having the simplifying effect of eliminating the secondary reflected wave in the solution, thereby making the construction of exact formulas possible.

To simulate the secondary reflected wave and develop a picture of these solutions, we assume both TOV and FRW spacetimes satisfy the {\it same} equation of state, which is $p=\frac{c^2}{3}\rho$, for the pure radiation stage of the early universe.  To run the simulation, we use an exact formula for the FRW spacetime in standard Schwarzschild coordinates first constructed in \cite{smoltePNAS09}, detailed in Chapter \ref{ch:family_of_shock_waves}.   This transformation puts the FRW metric into the ``same coordinate system'' as the TOV metric, enabling us to match these metrics together, with the FRW metric on the inside, to form the FRW/TOV matched metric.  Since we have two transformations of the FRW metric into standard Schwarzschild coordinates, denoted as FRW-1 and FRW-2, we have two matched models, FRW-1/TOV and FRW-2/TOV, at our disposal to simulate.  This matching gives us exact expressions for the initial data (a point of continuity between FRW and TOV metric components, such that the density, pressure, and velocity suffer a jump discontinutiy) and boundary data.  This data is the information necessary to run the locally inertial Godunov scheme, developed in Chapter \ref{ch:frac_god_method}, to perform the numerical simulation of the interaction region between the two spacetime metrics.  In particular, since light speed is a speed limit for propagation of signals,  outside this lightcone the solution should remain FRW (on the inside) and TOV (on the outside).  We call this region contained within the light cone that emanates from the initial discontinuity the {\it region of interaction}.  Besides simulating shock wave formation, we also compute the {\it cone of light} and {\it cone of sound}, the light and sound information that propagates outward from the initial discontinuity.  Both the cone of light and the cone of sound emanate and expand away in both directions from the initial discontinuity.  Because light travels faster than sound, the cone of sound must be contained in the cone of light.   Since we expect from \cite{groasmte} that there is no lightlike propagation in spherically symmetric spacetimes, the cone of sound should be the true limit of propagation of information; therefore, outside the cone of sound we expect the solution should remain FRW (on the inside) and TOV (on the outside).    Figure \ref{fig:frw_tov_match} illustrates all these expectations.   This prediction is an important result we demonstrate in our numerical simulation.  In fact, since the numerical simulation involves integration of the metric components at each time step, we know of no mathematical proof that the weak solutions constructed by the Groah and Temple theorem \cite{groasmte} actually have the sound speed as a propagation speed for the information in the solution.

\begin{figure}[!h]
\begin{pspicture}(10,5)(0,-0.5)
%\psgrid
%t and r_bar axises
\psline[linewidth=2pt]{->}(0,0)(10,0)
\psline[linewidth=2pt]{->}(0,0)(0,4)
%draw the tick to represent the initial shock position
\psline(5,-0.1)(5,0.1)
%draw the cone of light and sound
\psline(5,0)(1,4)
\psline(5,0)(9,4)
\psline(5,0)(2.5,4)
\psline(5,0)(7.5,4)
%fill and label the interaction area
\pspolygon[fillstyle=crosshatch](5,0)(2.5625,3.9)(7.4375,3.9)
\rput(5,3){\psframebox*{$\begin{tabular}{c}Region of \\ Interaction\end{tabular}$}}
%label FRW and TOV sides along with the speed of light and sound
\rput(1.8,1.5){$\Large\mathbf{FRW}$}
\rput(8.2,1.5){$\Large\mathbf{TOV}$}
\psline{->}(8,4.4)(7.5,4)
\rput(9.3,4.4){Sound Speed}
\psline{->}(1.5,4.4)(1,4)
\rput(2.7,4.4){Light Speed}
%label the axises
\rput(-0.4,4){$\bar{t}$}
\rput(-0.4,0){$\bar{t}_0$}
\rput(5,-0.4){$\bar{r}_0$}
\rput(10,-0.4){$\bar{r}$}
\end{pspicture}\caption{Matching the FRW and TOV metrics}
\label{fig:frw_tov_match}
\end{figure}
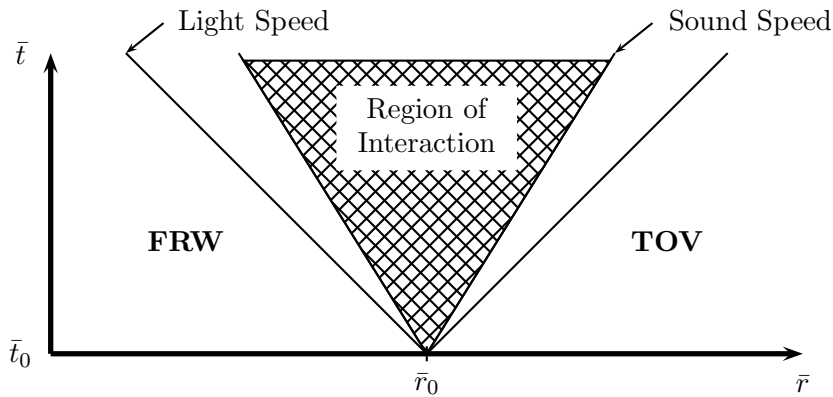

Just like the continuous models, we consider the entire infinite region of the FRW/TOV model as the universe of our simulation, and we need to define initial parameters to run the simulation.  Recall, our chosen parameter to the one parameter family of shock wave solutions is the radius of the initial discontinuity $\bar{r}_0$ separating the two metrics.  We choose to set these parameters as
\begin{equation}\label{std_params}
\bar{r}_{min}=3,\phantom{4444}\bar{r}_{max}=7,\phantom{4444}\bar{r}_0=5, \phantom{4444}n=2^{14}=16,384,
\end{equation}
in all the simulations in this chapter, unless otherwise stated. The units to these are explained in Chapter \ref{ch:units}.  Also, throughout this chapter, when we discuss mesh points, we use the space coordinate $x$ to match the notation of the locally inertial Godunov scheme as denoted in Chapter \ref{ch:frac_god_method}, and when we discuss points within our universe, we use the radial coordinate $r$ to explicitly express these points as a radial coordinate of the metric.  When we discuss the derivative of a function in our solution, we compute these derivatives numerically using a standard three point method \cite{burdfa}.

In Section \ref{sec:prelim} we introduce material needed for testing the accuracy and certainty of computing the correct solution in our numerical simulation. We start by computing the (coordinate) speed of light and sound.  We move to tracking the borders for both the FRW and TOV side of the simulation, and we conclude this section by finding the boundary condition for the TOV side.  We setup the FRW-1/TOV matched model simulation in Section \ref{sec:frw1_tov_setup}, determining the initial profiles and boundary conditions we can discretize and feed into the locally inertial Godunov scheme.  We use this setup in Section \ref{sec:frw1_tov_results} to run the simulation and obtain results.  We record the convergence of the entire simulation, using the successive mesh refinement technique introduced by Colella and Woodward in \cite{woodco}.  We determine the interaction region is not only contained in the cone of sound, but it is exactly synchronized with the cone of sound region.  We also show convergence of the non-interaction regions, the FRW-1 and TOV metrics, by computing the numerical error against their true solutions.  We resolve the question of the nature of the secondary wave in this solution as being another shock wave reflected back in towards the center.  We show the discontinuities in the fluid variables and metric derivatives, confirming these solutions to the Einstein equations are solutions in the weak sense of the theory of distributions.  Changing the parameter $\bar{r}$ results in a quantitatively different solution; thus, we have a one parameter family of quantitatively different shock wave solutions to the Einstein equations.  The FRW-2/TOV matched model simulation gets setup in Section \ref{sec:frw2_tov_setup} with the results shown in Section \ref{sec:frw2_tov_results}.  We numerically show the FRW-2/TOV matched model is the same solution as the FRW-1/TOV model, differing by a non-linear time coordinate transformation.  This second matching provides a pedagogically interesting numerical confirmation of the covariance of the Einstein equations in standard Schwarzschild coordinates.

\section{Preliminaries}
\label{sec:prelim}
In our simulation, we track the outer boundary of the lightlike and sound like information emanating from the initial discontinuity and compare them to the trajectory of the shock waves.  In order to make this comparison, we need to determine how fast light and sound propagates through our universe.  Since space and time are bent differently across the universe, we expect these speeds to be dependent on the point of interest.  We first compute the (coordinate) speed of light.  To simplify the calculation, we suppress the $\bar{r}^2d\Omega^2$ term and take the 2-dimensional metric in standard Schwarzschild coordinates
\begin{equation}
ds^2=-Bd\bar{t}^2+A^{-1}d\bar{r}^2=g_{ij}dx^idx^j.
\end{equation}
Consider a lightlike curve, $\gamma(\xi)=(\bar{t}(\xi),\bar{r}(\xi))$, which has the velocity vector
\begin{equation}
X(\xi)=\frac{d\bar{t}}{d\xi}\frac{\partial}{\partial\bar{t}}+\frac{d\bar{r}}{d\xi}\frac{\partial}{\partial\bar{r}}.
\end{equation}
Since $\gamma(\xi)$ is a lightlike curve, the length of this curve is zero or
\begin{equation}
0=g_{ij}X^iX^j=-B\left(\frac{d\bar{t}}{d\xi}\right)^2+A^{-1}\left(\frac{d\bar{r}}{d\xi}\right)^2.
\end{equation}
Solving this equation for the coordinate speed, we obtain the speed light travels as
\begin{equation}\label{speed_of_light}
l_\pm\equiv\frac{d\bar{r}}{d\bar{t}}=\pm\sqrt{AB},
\end{equation}
taking the plus or minus sign for light traveling in the positive or negative direction, respectively.

To find the speed at which sound travels in $(\bar{t},\bar{r})$-coordinates, we first note the sound speed is tied to the equation of state $p = \sigma\rho$ as being $\sqrt{\sigma}$.  This speed is relative to the background fluid at rest (i.e. $v=0$).  To find the speed of sound relative to fluid moving at a speed $v$, we need to apply the Lorentz transformation law for velocities (\ref{lorentz_velocity_xform_law}).  More specifically, the speed $\sqrt{\sigma}$ transforms into
\begin{equation}
\frac{v\pm\sqrt{\sigma}}{1+\frac{v\sqrt{\sigma}}{c^2}}.
\end{equation}
Again the plus/minus sign determines the direction in space being traveled.  In order to determine the $(\bar{t},\bar{r})$-coordinated speed, we magnify this speed by the factor $\sqrt{AB}$, and the speed of sound becomes
\begin{equation}\label{speed_of_sound}
s_{\pm}\equiv\sqrt{AB}\left(\frac{v\pm\sqrt{\sigma}}{1\pm\frac{v\sqrt{\sigma}}{c^2}}\right).
\end{equation}
Notice how the speed of sound is not only dependent on the metric $(A,B)$, but also on the movement of the fluid $v$, indicating the dependence on the medium the sound travels through.

For our simulation, we have a fixed speed of sound, namely $\sqrt{\sigma}=c^2/\sqrt{3}$, and based on the position in space, we know the values for $(A,B)$ and $v$ to determine both the speed of light (\ref{speed_of_light}) and the speed of sound (\ref{speed_of_sound}).  Hence, the new position of the light/sound information $\bar{r}$ is determined from the old position $\bar{r}_*$ after a time change of $\Delta\bar{t}$ by
\begin{equation}
\bar{r}=\bar{r}_* + s\Delta\bar{t},
\end{equation}
where we set the speed as $s=l_\pm$ or $s=s_\pm$ for the light or sound information, respectively, choosing the plus/minus sign depending on the direction traveled.

For the cone of sound verses the shock wave position comparison and error estimates, we want to know where the FRW metric stops and the interaction region begins; therefore, we need a mechanism to determine the border between the two.  Also, we need to find a similar border between the TOV metric and the interaction region.  Theoretically, the position of the shock wave can be used as the borders, but since the solution is numerically simulated with a first order method, these shock waves are smeared by numerical diffusion.  This diffusion bleeds into the metrics, causing numerical error in testing this region against the known model.  In order to handle these errors, we develop a criteria, based on studying the numerical solution, to determine where the numerical diffusion ends and the metric begins.  The border criteria we choose for the FRW side is where the spatial derivative of the fluid velocity first changes sign.  In terms of notation, the border is the first position $\bar{r}$ such that
\begin{equation}\label{frw_border}
\frac{\partial v}{\partial\bar{r}}(\bar{r})\frac{\partial v}{\partial\bar{r}}(\bar{r}+\Delta\bar{r})<0
\end{equation}
This criteria is sufficient because the velocity is an increasing function of $\bar{r}$ for the FRW metric in standard Schwarzschild coordinates and only decreases when it hits the shock.  To implement this criteria, we start the gridpoint under consideration at the minimum radius $x_i=x_1=\bar{r}_{min}$, which is in the FRW metric because of the boundary condition, and increase the index $i$ until the numerical derivative changes sign, giving us the point where the numerical diffusion first takes place.

Next, we explore the border criteria for the interaction region and the TOV metric. Again, the fluid velocity is used as our indicator.  Since the fluid velocity is theoretically zero, we would ideally find the point where it first becomes non-zero, but in the simulation it is only approximately zero, due to numerical error.  Instead, the border criteria we choose for the TOV side is where the absolute value of the spatial derivative of the fluid velocity first becomes greater than $0.01$.  More specifically, the border is the first position $\bar{r}$ such that
\begin{equation}\label{tov_border}
\AutoAbs{\frac{\partial v}{\partial\bar{r}}(\bar{r})}>0.01.
\end{equation}
This criteria detects the first significant change in the velocity and ignores small changes occurring from numerical error.  To implement this strategy, we start the gridpoint at the maximum radius $x_i=x_n=\bar{r}_{max}$, which is in the TOV metric, and decrease the index $i$ until the numerical derivative satisfies (\ref{tov_border}).

To run our simulation, we need boundary data from both the FRW and TOV metrics to maintain a consistent solution at the edges of the simulated region of the universe.  Much to our surprise, the TOV metric boundary condition changes although the TOV metric itself is time independent.  Even though the theory predicts there is a TOV metric matched continuously to the right side of the interaction region, the TOV metric allows for arbitrary changes of the time coordinate $\bar{t}$, and we do not know which of these is matched in the solution except at the initial matching by solving for the correct $B_0$.  More precisely, consider an arbitrary time coordinate map $\varphi:\tau\rightarrow\bar{t}$ applied to the TOV metric
\begin{equation}
ds^2=-B(\bar{r})d\bar{t}^2+\left(\frac{1}{1-\frac{2\mathcal{G}M(\bar{r})}{\bar{r}}}\right)d\bar{r}^2 +\bar{r}^2d\Omega^2.
\end{equation}
Since this metric is time independent, the effect of $\bar{t}=\varphi(\tau)$ is just a scaling of the $B$ metric component, resulting in the transformed metric of
\begin{equation}
ds^2=-\left(\frac{d\varphi}{d\bar{t}}\right)^2B(\bar{r})d\tau^2+\left(\frac{1}{1-\frac{2\mathcal{G}M(\bar{r})}{\bar{r}}}\right)d\bar{r}^2 +\bar{r}^2d\Omega^2.
\end{equation}
In $(\tau,\bar{r})$-coordinates, the metric is no longer time independent because the time metric component now has a time dependent scale factor, but the other metric component along with the fluid variables are still independent of time.  Since we refer to it throughout the chapter, we define
\begin{equation}\label{time_scale}
B^t(\bar{t})\equiv\left(\frac{d\varphi}{d\bar{t}}\right)^2
\end{equation}
to represent this scale factor caused by the arbitrary time coordinate transformation of the TOV metric.  This allowance for arbitrary time transformations gives us a slew of potential TOV metrics, differing only in the time metric component, to match in our solution.  One can view $B^t$ as a uniform change within the TOV metric of the rates the clocks move, and remarkably the FRW/TOV shock wave solution picks out the correct one of these TOV metrics at each time step based on the integration of $B$ up through the FRW metric and the interaction region.  To adjust to the change in the time scale, we determine where the TOV metric starts, based on the border criteria above, and match $B^t$ at that border to determine which TOV metric is used as the boundary condition.  This rematching also enables us to pick out the correct TOV metric model to test against the simulated solution in order to prove convergence of the TOV side of our universe, which will be discussed in more detail later.

To explain this in different words, the theory in \cite{groasmte} tells us the reason we must include the time scale factor in the $B$ component of the metric is because in the numerical method the metric values on the left hand boundary of the simulated region are imposed, but the metric values on the right hand boundary are determined by the integration of the equations.  Beyond the light cone on the right hand side, we know beforehand the method has to simulate the TOV spacetime, but we do not know a priori in which standard Schwarzschild coordinate system it is simulated.  Thus, since the standard Schwarzschild form allows for an arbitrary time scale factor on the metric $B$ component, the correct time translation must be included to get the correct $B$ component on the right hand boundary.

\section{FRW-1 and TOV Matched Model Setup}
\label{sec:frw1_tov_setup}
In this section, we cover the details of building the initial profiles and the boundary data required to setup the FRW-1/TOV model simulation.  We restate both metrics for ease of reference.  Using the coordinate transformation
\begin{equation}\label{ch6_frw1_xform_ssc}
\begin{split}
\bar{r}=&\sqrt{t}r,\\
\bar{t}=&\left\{1+\frac{\bar{r}^2}{4t^2}\right\}t=t+\frac{r^2}{4},
\end{split}
\end{equation}
the FRW metric in standard Schwarzschild coordinates (the FRW-1 form) is
\begin{equation}\label{ch6_frw1_in_ssc}
ds^{2} =-\frac{1}{1-v^2}d\bar{t}^2+\frac{1}{1-v^2}d\bar{r}^2+\bar{r}^{2}d\Omega^2,
\end{equation}
with the metric components
\begin{equation}\label{ch6_frw1_metric_in_ssc}
A(\xi) = 1-v(\xi)^2,\phantom{4444} B(\xi) = \frac{1}{1-v(\xi)^2},
\end{equation}
and the fluid variables
\begin{equation}\label{ch6_frw1_fluid_vars_in_ssc}
\rho(\xi,\bar{r}) =\frac{3v(\xi)^2}{\kappa\bar{r}^2},\phantom{4444} v(\xi)=\frac{1-\sqrt{1-\xi^2}}{\xi},
\end{equation}
where the variable $\xi$ is defined as $\xi=\bar{r}/\bar{t}$ and is a function of $v$ by the following relation
\begin{equation}\label{ch6_xi_v_relation}
\xi=\frac{2v}{1+v^2}.
\end{equation}

Whereas, the TOV metric is written
\begin{equation}\label{ch6_tov_in_ssc}
ds^2=-B(\bar{r})d\bar{t}^2+\left(\frac{1}{1-\frac{2\mathcal{G}M(\bar{r})}{\bar{r}}}\right)d\bar{r}^2 +\bar{r}^2d\Omega^2,
\end{equation}
with the metric components
\begin{equation}\label{ch6_tov_metric_in_ssc}
A(\bar{r}) = 1-8\pi\mathcal{G}\gamma,\phantom{4444} B(\bar{r}) = B_0(\bar{r})^{\frac{4\sigma}{1+\sigma}},
\end{equation}
and the fluid variables
\begin{equation}\label{ch6_tov_fluid_vars_in_ssc}
\rho(\bar{r})=\frac{\gamma}{\bar{r}^2},\phantom{4444} v(\bar{r})=0,
\end{equation}
where the parameter $\gamma$ is a constant dependent on $\sigma$
\begin{equation}\label{ch6_gamma}
\gamma=\frac{1}{2\pi\mathcal{G}}\left(\frac{\sigma}{1+6\sigma+\sigma^2}\right).
\end{equation}
Recall, for this simulation we let $\sigma=1/3$.

Assuming an initial start time $\bar{t}_0>0$ to be determined later,  we choose an initial radius for the discontinuity of $\bar{r}_0$, and we want to match the metric components $(A,B)$ continuously at the starting point $(\bar{t}_0,\bar{r}_0)$ for the initial discontinuity in the fluid variables.  We start by matching the metric $A$ component at this point
\begin{equation}\label{frw1_matching_a}
A_{FRW}(\bar{t}_0,\bar{r}_0)=1-v\left(\frac{\bar{r}_0}{\bar{t}_0}\right)^2=1-8\pi\mathcal{G}\gamma=A_{TOV}(\bar{t}_0,\bar{r}_0).
\end{equation}

Let $v_0=v(\bar{r}_0/\bar{t}_0)$ represent the fluid velocity at the discontinuity so (\ref{frw1_matching_a}) implies
\begin{equation}
v_0=\sqrt{8\pi\mathcal{G}\gamma}=\sqrt{\frac{4\sigma}{1+6\sigma+\sigma^2}},
\end{equation}
where we substituted (\ref{ch6_gamma}) for $\gamma$.  Take note that $v_0$ is independent of the our free parameter $r_0$, it is quite astounding the velocity of the fluid at the discontinuity remains the same regardless of the placement of the discontinuity.   Equipped with the value of $v_0$, we use (\ref{ch6_xi_v_relation}), rewritten as
\begin{equation}
\frac{\bar{r}_0}{\bar{t}_0}=\frac{2v_0}{1+v^2_0},
\end{equation}
to find the unknown starting time $\bar{t}_0$ as
\begin{equation}\label{frw1_start_time}
\bar{t}_0=\frac{\bar{r}_0(1+v^2_0)}{2v_0}.
\end{equation}
The independence of $v_0$ from $\bar{r}_0$ along with (\ref{frw1_start_time}) implies the initial start time is proportional to the initial radius of the discontinuity.  Finding $\bar{t}_0$ enables us to build the initial profile of the FRW-1 metric for any radial coordinate $\bar{r}<\bar{r}_0$ by computing $\xi=\bar{r}/\bar{t}_0$ and using equations (\ref{ch6_frw1_in_ssc})-(\ref{ch6_frw1_fluid_vars_in_ssc}).

To compute the TOV metric, the $A$ metric component is already determined beforehand by a constant (\ref{ch6_tov_metric_in_ssc}).  To find the other metric component, match it at the discontinuity
\begin{equation}\label{frw1_matching_b}
B_{TOV}(\bar{t}_0,\bar{r}_0)=B_0(\bar{r}_0)^{\frac{4\sigma}{1+\sigma}}=\frac{1}{1-v^2_0}=B_{FRW}(\bar{t}_0,\bar{r}_0),
\end{equation}
forcing the constant $B_0$ to take the form
\begin{equation}\label{frw1_matching_b0}
B_0=\frac{\bar{r}^{-\frac{4\sigma}{1+\sigma}}_0}{1-v^2_0}.
\end{equation}
With the TOV time scale $B_0$, we can build the TOV metric for any radial coordinate $\bar{r}>\bar{r}_0$ by using the equations (\ref{ch6_tov_in_ssc})-(\ref{ch6_tov_fluid_vars_in_ssc}).

Combining all this data together, we build the following functions $v_{init}(\bar{r})$, $\rho_{init}(\bar{r})$, $A_{init}(\bar{r})$, and $B_{init}(\bar{r})$ to use as the initial data at time $\bar{t}_0$ (\ref{frw1_start_time}), depending on the free parameter $\bar{r}_0$.  Because the other functions are based on $v$ in the FRW-1 space, we start by stating the fluid velocity
\begin{equation}\label{velocity_init}
v_{init}(\bar{r}) = \left\{
\begin{array}{ll}
\frac{1-\sqrt{1-\xi^2}}{\xi} & \bar{r}<\bar{r}_0\\
0 & \bar{r}>\bar{r}_0,
\end{array}
\right.
\end{equation}
where $\xi =\bar{r}/\bar{t}_0$.  Based on this function $v_{init}$, the initial density profile is
\begin{equation}\label{density_init}
\rho_{init}(\bar{r}) = \left\{
\begin{array}{ll}
\frac{3v_{init}^2}{\kappa\bar{r}^2} & \bar{r}<\bar{r}_0\\
\frac{\gamma}{\bar{r}^2} & \bar{r}>\bar{r}_0,
\end{array}
\right.
\end{equation}
and the metric $\mathbf{A}_{init}=(A_{init},B_{init})$ is giving by
\begin{equation}\label{metric_A_init}
A_{init}(\bar{r}) = \left\{
\begin{array}{ll}
1-v_{init}^2 & \bar{r}<\bar{r}_0\\
1-8\pi\mathcal{G}\gamma & \bar{r}>\bar{r}_0,
\end{array}
\right.
\end{equation}
and
\begin{equation}\label{metric_B_init}
B_{init}(\bar{r}) = \left\{
\begin{array}{ll}
\frac{1}{1-v_{init}^2} & \bar{r}<\bar{r}_0\\
B_0(\bar{r})^{\frac{4\sigma}{1+\sigma}} & \bar{r}>\bar{r}_0.
\end{array}
\right.
\end{equation}

One last piece of information needed to run the simulation is the boundary conditions or the ghost cells.  Recall, we need both boundary conditions to evolve according to the Einstein equations to maintain consistency at the edges of our simulated universe.  For the FRW-1 side, at the gridpoint $x_0$, the fluid velocity is a function of time $\bar{t}_j$
\begin{equation}\label{frw1_bndary_v}
v_{0,j}=\frac{1-\sqrt{1-\xi^2}}{\xi},
\end{equation}
where the variable $\xi$ is defined as $\xi=x_0/\bar{t}_j$.  With the fluid velocity, the fluid density becomes
\begin{equation}
\rho_{0,j} =\frac{3v^2_{0j}}{\kappa x^2_0}.
\end{equation}
Since the metric components are staggered relative to the fluid variables, we need to compute the half gridpoint
\begin{equation}\label{one_half_gridpoint}
x_{\frac{1}{2}}=x_0+\frac{\Delta x}{2},
\end{equation}
and use it to find the corresponding velocity
\begin{equation}
v_{\frac{1}{2},j}=\frac{1-\sqrt{1-\xi^2}}{\xi},
\end{equation}
for $\xi=x_{\frac{1}{2}}/\bar{t}_j$.  We use this velocity in the following computation for the metric components,
\begin{equation}\label{frw1_bndary_metric}
A_{1,j} = 1-v_{\frac{1}{2},j}^2,\phantom{4444} B_{1j} = \frac{1}{1-v_{\frac{1}{2},j}^2}.
\end{equation}

The boundary condition for the TOV side is easier to implement.  Since the TOV metric is independent of time, we set the data values during the initial setup of the function profiles, and they remain the same.  These static values work for the fluid variables and the metric component $A$, but the function $B$, specifically the time scale $B^t$ (\ref{time_scale}), changes and must be rematched during each time step, as discussed in the last section.  Using the above criteria for the TOV border (\ref{tov_border}), let $x_i$ be the gridpoint position of this border.  We rematch the time scale at time $\bar{t}_j$ by the following formula
\begin{equation}\label{rematch_b0}
B^t=B(\bar{t}_j,x_i)(x_i)^{-\frac{4\sigma}{1+\sigma}},
\end{equation}
where $B(\bar{t}_j,x_i)$ is the simulated solution at the coordinate $(\bar{t}_j,x_i)$.

\section{FRW-1 and TOV Matched Simulation Results}
\label{sec:frw1_tov_results}
\begin{figure}
\begin{center}
\includegraphics[width=\textwidth]{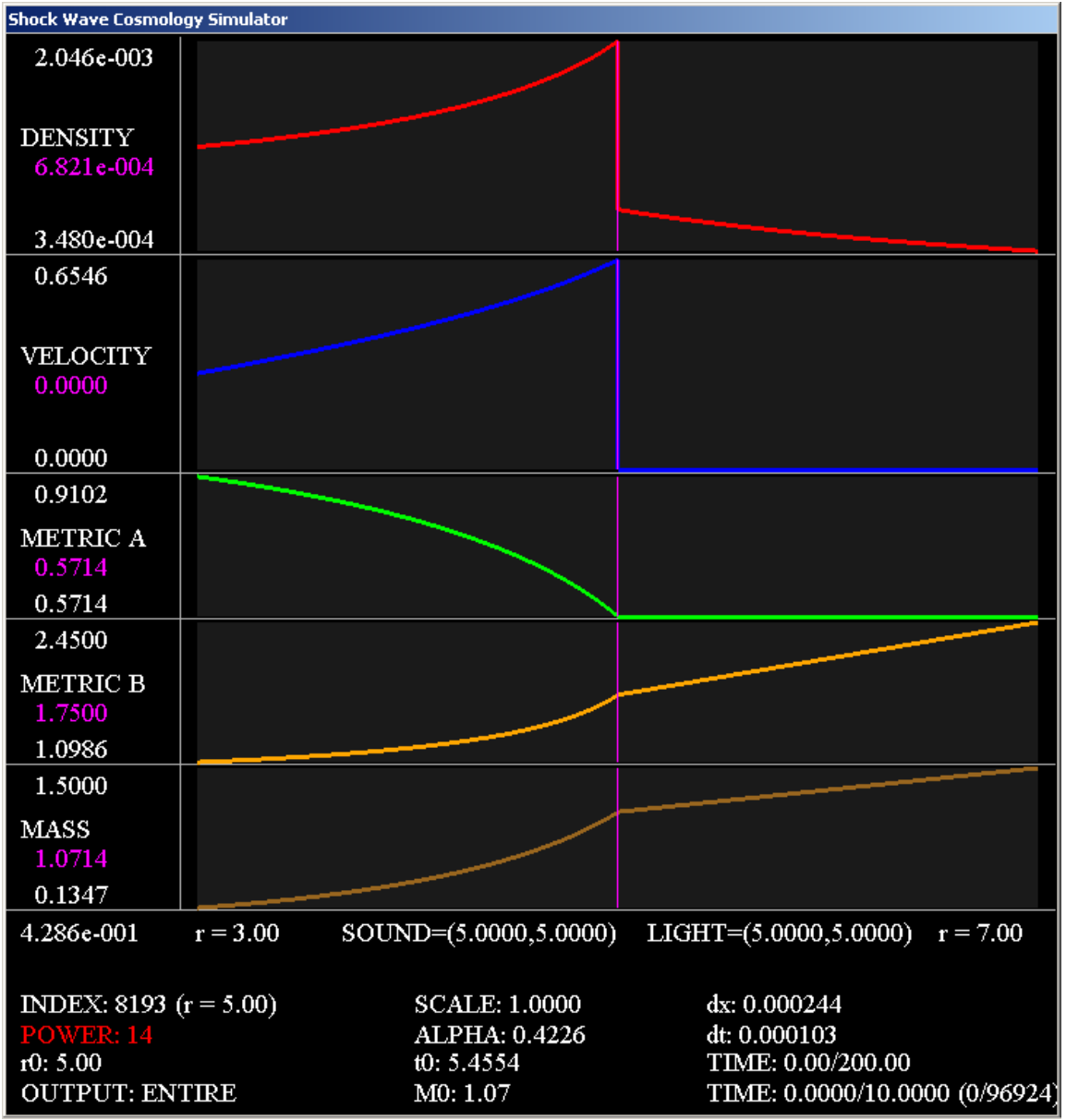}
\end{center}
\caption{Initial profiles}
\label{fig:frw1_tov_init}
\end{figure}

In this section, we look at results of the simulation for the FRW-1/TOV matched model.  We use the initial profiles (\ref{velocity_init})-(\ref{metric_B_init}) along with the boundary conditions (\ref{frw1_bndary_v})-(\ref{rematch_b0}) developed in the last section to run the simulation.  Figure \ref{fig:frw1_tov_init} shows the initial profiles with these parameters for the fluid variables $(\rho_{init},v_{init})$ and the metric $\mathbf{A}_{init}$, along with the mass.  By selecting the initial discontinuity at $\bar{r}_0=5$, equation (\ref{frw1_start_time}) gives the initial start time of $\bar{t}_0=5.4554$.  Notice how the discontinuities in the fluid variables jump down from the FRW-1 side to the TOV side.  Moreover, the FRW density $\rho$ and the TOV density $\bar{\rho}$ at this discontinuity are related by $\rho=3\bar{\rho}$.

\begin{figure}
\begin{center}
\includegraphics[width=\textwidth]{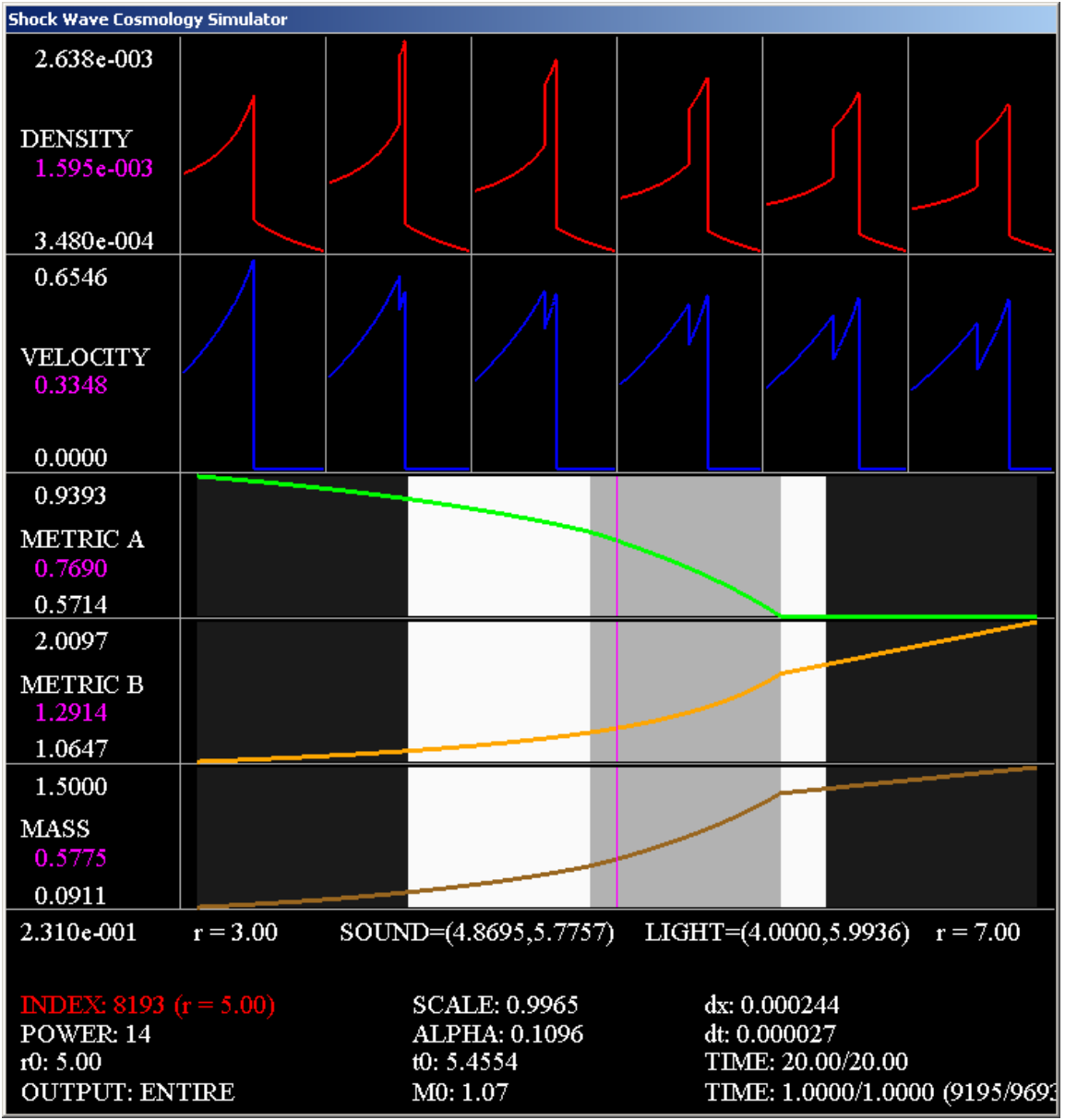}
\end{center}
\caption{Evolution of the fluid variables during a unit of time}
\label{fig:frw1_tov_frames}
\end{figure}

With these initial profiles, we run the simulation for one unit of time (i.e. $\bar{t}_{end}=\bar{t}_0+1$).  Figure \ref{fig:frw1_tov_frames} depicts the evolution of the fluid variables $(\rho,v)$, giving us a frame by frame view for the evolution of the fluid variables across this time frame, evenly distributed from the left frame at $\bar{t}_0$ to the right frame at $\bar{t}_{end}$.  After the initial time $\bar{t}_0$, two shock waves are formed, the stronger shock moving out toward the TOV side and the weaker shock moving in toward the FRW-1 side, creating a pocket of higher density expanding and interacting with both the FRW-1 and TOV metrics; therefore, the secondary wave to the strong shock for this solution of the Einstein equations is another shock wave, reflecting back in.

Next, we focus our attention to the resulting solution at the end time, $\bar{t}_{end}$.  Figure \ref{fig:frw1_tov_end} highlights where the two shock positions are relative to the cone of sound and the cone of light.  The cone of the light is represented by the white region while the cone of sound, embedded in the cone of light, is represented by the grey region.  Notice how the edges of the cone of sound line up with both shock wave positions, showing the interaction region between the two metrics lie completely in the cone of sound.  Since both characteristics and the edges to the cone of sound move at the speed of sound, we understand this result because these edges impinge on the shocks like a characteristic, so if one of the edges were to get slightly ahead or behind the shock position, then that edge would get pushed into the shock like all characteristics close to the shock.  This figure also displays the spatial derivatives in the metric components $A$ and $B$, the green and orange graphs, respectively.  These derivatives $(A',B')$, found using numerical differentiation,  have discontinuities aligned with the ones for the fluid variables at the edges to the cone of sound.  Looking back at Figure \ref{fig:frw1_tov_frames}, it shows the profiles for the metric $(A,B)$ as being continuous, so the metric is Lipschitz continuous, reinforcing the fact that we have a weak solution to the Einstein equations.

\begin{figure}
\begin{center}
\includegraphics[width=\textwidth]{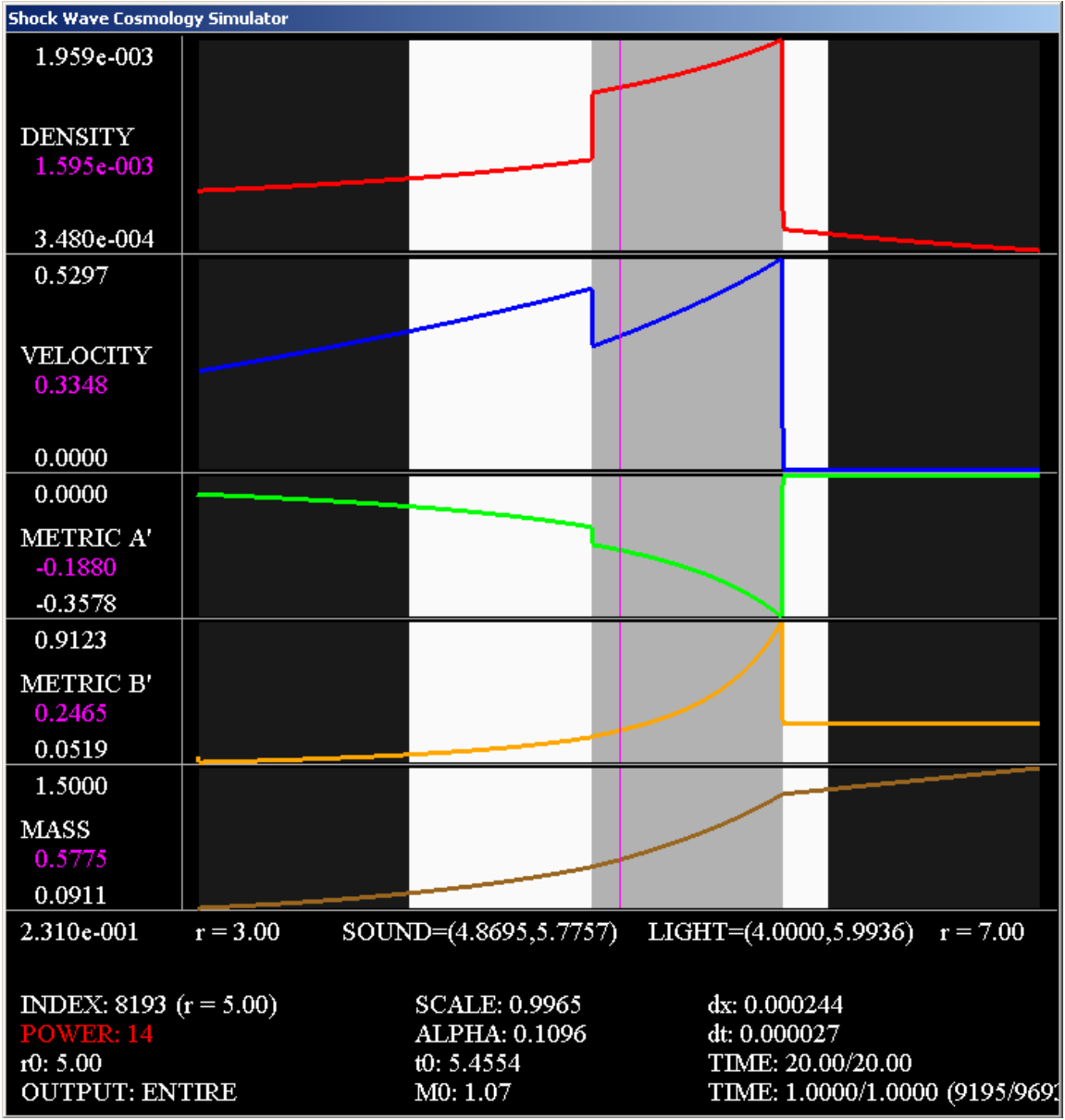}
\end{center}
\caption{Solution after a unit of time, showing the derivatives of the metric}
\label{fig:frw1_tov_end}
\end{figure}

The convergence of this solution, tested by successive mesh refinements, is shown in Table \ref{tab:frw1_tov_converge}.  Since we are implementing a first order method, we expect a convergence rate of 0.5 for the discontinuous fluid variables, while we expect a convergence rate of 1 for the continuous metric. Looking at the results in the table, the convergence rate of the fluid variables start around 0.5 for the big mesh sizes, as expected, and improve to 1 as we mesh refine, which is a little surprising.  In contrast, the convergence rate for the metric components start around 1 and on average continue to stay around 1 under mesh refinement, regardless of the high/low swing in the convergence rate for the metric $B$ component.  We believe this change in the rates is due to numerical error from integrating the metrics up across the universe.

\begin{table}
\begin{center}
\begin{tabular}{|c|c|c|c|c|c|c|c|c|}
\hline
Number&\multicolumn{2}{|c|}{$\rho$}&\multicolumn{2}{|c|}{$v$}&\multicolumn{2}{|c|}{$A$}&\multicolumn{2}{|c|}{$B$}\\
\cline{2-9}
Gridpoints & Error & Rate & Error & Rate & Error & Rate & Error & Rate\\
\hline
64 & 1.143e-004 & N/A & 4.657e-002 & N/A & 7.051e-003 & N/A & 3.146e-002 & N/A\\
\hline
128 & 8.490e-005 & 0.43 & 3.463e-002 & 0.43 & 3.710e-003 & 0.93 & 1.557e-002 & 1\\
\hline
256 & 5.970e-005 & 0.51 & 2.414e-002 & 0.52 & 1.817e-003 & 1 & 7.704e-003 & 1\\
\hline
512 & 4.000e-005 & 0.58 & 1.596e-002 & 0.6 & 9.243e-004 & 0.98 & 2.889e-003 & 1.4\\
\hline
1024 & 2.470e-005 & 0.7 & 9.741e-003 & 0.71 & 4.334e-004 & 1.1 & 1.974e-003 & 0.55\\
\hline
2048 & 1.410e-005 & 0.81 & 5.502e-003 & 0.82 & 2.568e-004 & 0.76 & 5.160e-004 & 1.9\\
\hline
4096 & 7.470e-006 & 0.92 & 2.866e-003 & 0.94 & 1.232e-004 & 1.1 & 4.172e-004 & 0.31\\
\hline
8192 & 3.740e-006 & 1 & 1.420e-003 & 1 & 7.100e-005 & 0.8 & 1.111e-004 & 1.9\\
\hline
16384 & 1.870e-006 & 1 & 7.063e-004 & 1 & 3.300e-005 & 1.1 & 1.024e-004 & 0.12\\
\hline
\end{tabular}
\end{center}
\caption{Successive mesh refinement convergence results}
\label{tab:frw1_tov_converge}
\end{table}

\begin{figure}
\begin{center}
\includegraphics[width=\textwidth]{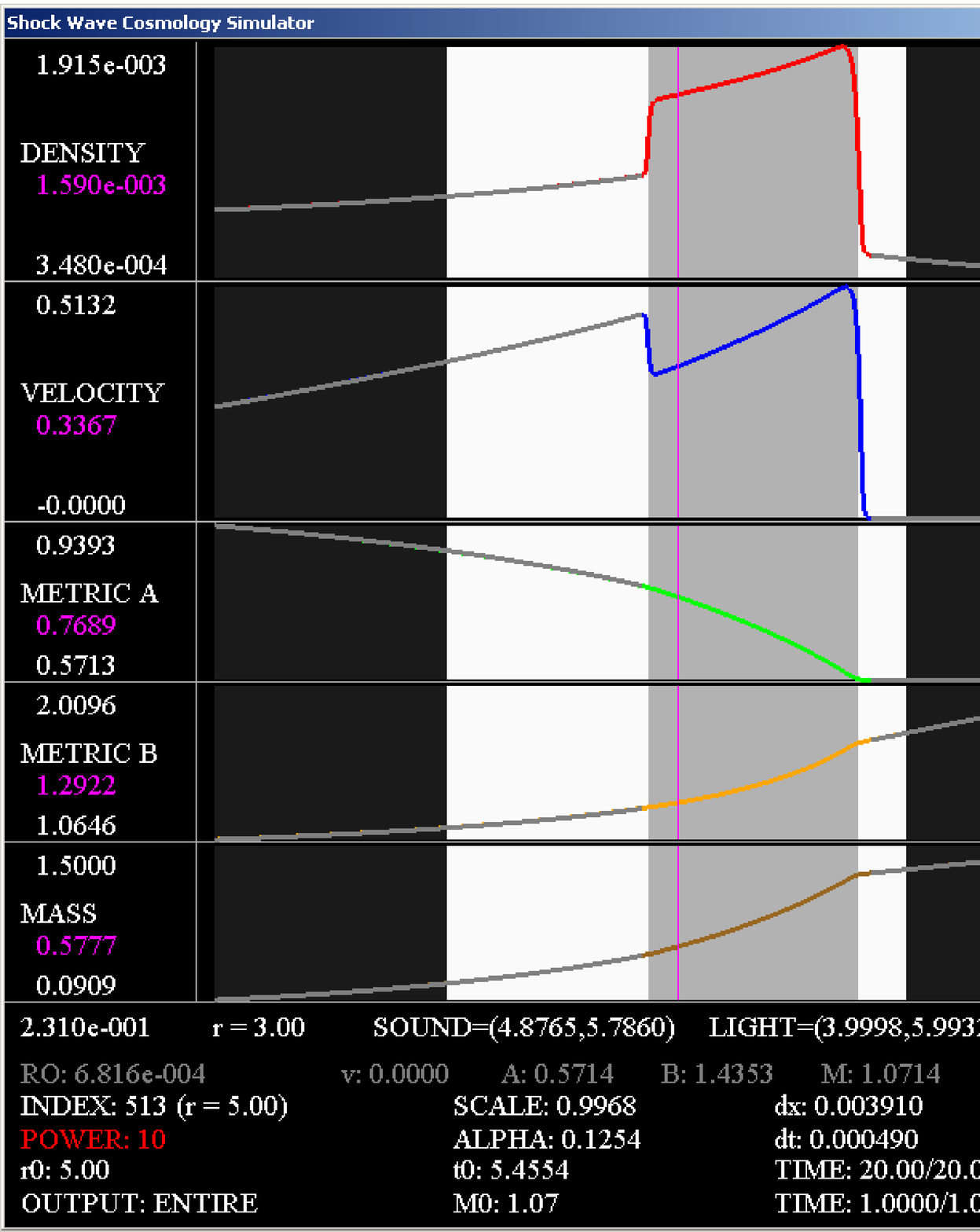}
\end{center}
\caption{Smearing of the shock waves for solution with less gridpoints ($n=2^{10}=1024$)}
\label{fig:frw1_tov_smeared}
\end{figure}

Next we study the preservation of the FRW-1 and TOV metrics outside of the interaction region.  Figure \ref{fig:frw1_tov_smeared} shows the numerical solution after one unit of time for a mesh with less gridpoints ($n=1024$).  In this figure, the numerical diffusion is more pronounced as opposed to a finer mesh, as seen in Figure \ref{fig:frw1_tov_end}.  This smearing of the shock wave is the reason for creating the border criteria discussed earlier.  This figure also highlights why each border criteria works.  The FRW border criteria is sufficient because one can see the fluid velocity on the FRW-1 side is increasing until it hits the smeared shock wave and starts decreasing.  On the other side, the TOV border criteria is sufficient because the fluid velocity is almost constant until the smeared shock wave causes a significant change.  Even though it looks constant on the TOV side, the fluid velocity is only close to being constant because there are numerical errors too small to be displayed on the graph shown.  These borders are displayed in the simulation by where the grey lines stop, as shown in both Figure \ref{fig:frw1_tov_smeared} and Figure \ref{fig:frw1_tov_test}.  Take notice on how both borders stop at the edge of the smearing of the shocks, and as we mesh refine, from Figure \ref{fig:frw1_tov_smeared} to Figure \ref{fig:frw1_tov_test}, they tend towards the edge of the cone of sound, as desired.  The colored lines verses the gray lines represent the simulated solution verses the model solution using equations (\ref{ch6_frw1_in_ssc})-(\ref{ch6_gamma}).  We record the convergence between the FRW-1 border, based on the criteria above, and the left edge of the cone of sound in Table \ref{tab:frw1_tov_shock_frw_side}.  Table \ref{tab:frw1_tov_shock_frw_side} shows, from left to right, the number of gridpoints, the position of the left edge to the cone of sound, the position of the FRW-1 border, the error between the two, and the rate of convergence associated with the error.  The corresponding results are displayed for the TOV side in Table \ref{tab:frw1_tov_shock_tov_side}, where we record the convergence between TOV border and the right edge of the cone of sound.  Notice how the error between the cone of sound and both borders are decreasing, approaching zero.  Moreover, the convergence rate is increasing, approaching 1, a linear rate which we would expect from a first order method.  These results gives us confidence in our method for determining the edges to the cone of sound and the FRW and TOV borders.

\begin{figure}
\begin{center}
\includegraphics[width=\textwidth]{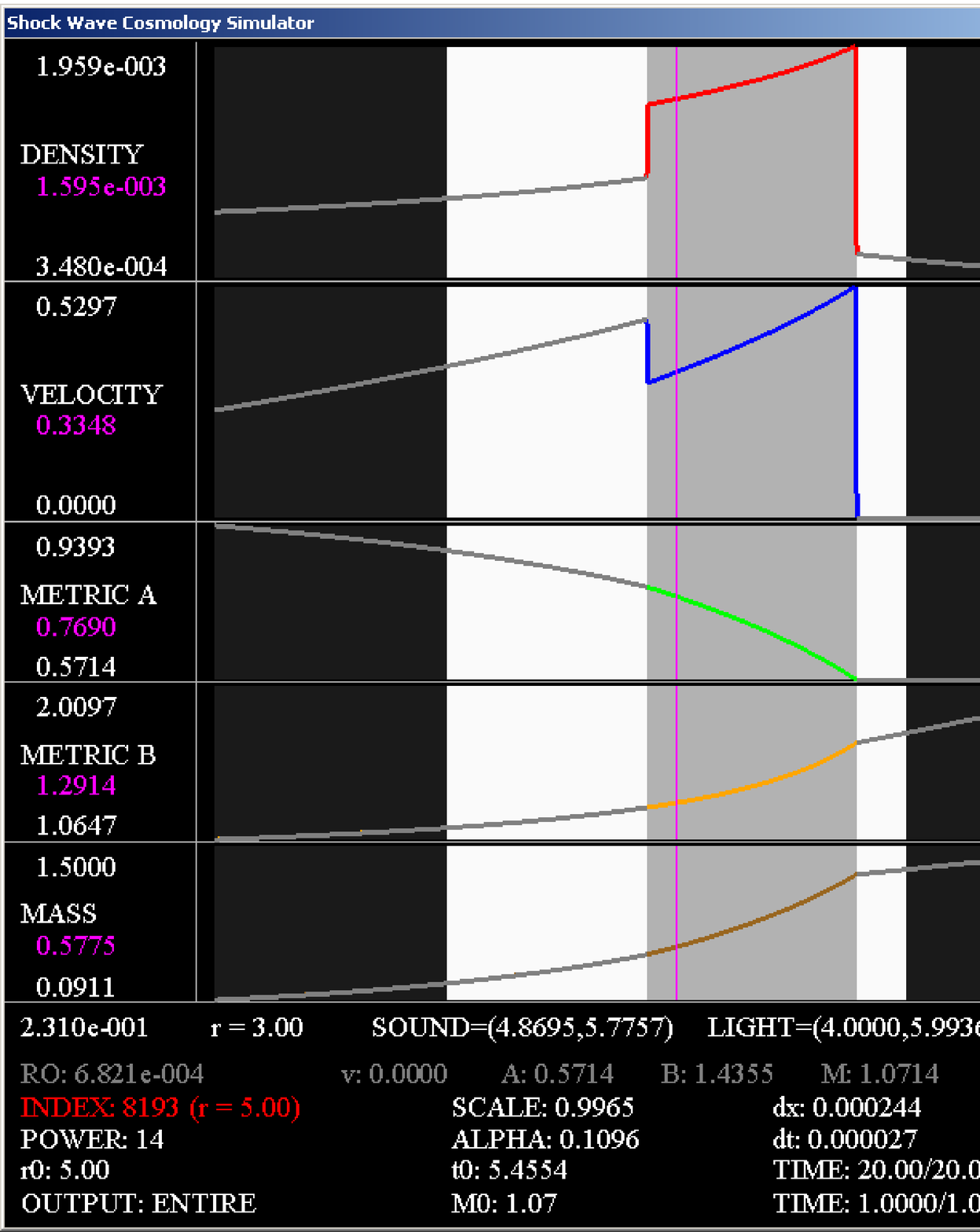}
\end{center}
\caption{Showing the model metrics against the simulated solution}
\label{fig:frw1_tov_test}
\end{figure}

\begin{table}
\begin{center}
\begin{tabular}{|c|c|c|c|c|}
\hline
 Gridpoints & Cone of Sound & Shock Wave & Error & Rate\\
\hline
64 & 4.9111 & 4.7778 & 0.13336 & N/A\\
\hline
128 & 4.9028 & 4.7953 & 0.10754 & 0.31\\
\hline
256 & 4.8896 & 4.8196 & 0.06999 & 0.62\\
\hline
512 & 4.8813 & 4.8395 & 0.0418 & 0.74\\
\hline
1024 & 4.8765 & 4.8495 & 0.02704 & 0.63\\
\hline
2048 & 4.8734 & 4.8583 & 0.01509 & 0.84\\
\hline
4096 & 4.8715 & 4.8628 & 0.00871 & 0.79\\
\hline
8192 & 4.8702 & 4.8655 & 0.00478 & 0.87\\
\hline
16384 & 4.8695 & 4.8671 & 0.00242 & 0.98\\
\hline
\end{tabular}
\end{center}
\caption{Shock wave verses cone of sound results for FRW side}
\label{tab:frw1_tov_shock_frw_side}
\end{table}

\begin{table}
\begin{center}
\begin{tabular}{|c|c|c|c|c|}
\hline
 Gridpoints & Cone of Sound & Shock Wave & Error & Rate\\
\hline
64 & 5.8134 & 6.3651 & 0.55165 & N/A\\
\hline
128 & 5.807 & 6.1181 & 0.31113 & 0.83\\
\hline
256 & 5.8003 & 5.9804 & 0.18006 & 0.79\\
\hline
512 & 5.7931 & 5.8963 & 0.10315 & 0.8\\
\hline
1024 & 5.786 & 5.8387 & 0.05267 & 0.97\\
\hline
2048 & 5.7808 & 5.81 & 0.0292 & 0.85\\
\hline
4096 & 5.7778 & 5.7927 & 0.01482 & 0.98\\
\hline
8192 & 5.7764 & 5.7845 & 0.00812 & 0.87\\
\hline
16384 & 5.7757 & 5.7797 & 0.00402 & 1\\
\hline
\end{tabular}
\end{center}
\caption{Shock wave verses cone of sound results for TOV side}
\label{tab:frw1_tov_shock_tov_side}
\end{table}

Equipped with a method for determining the position of the FRW-1 and TOV borders, these borders are used to show convergence of the metrics outside the interaction region.  These borders are indicators on where to stop computing the error between the simulated and model solutions.  For example, we consider the FRW-1 side error computation for $n=1024$.  Using Table \ref{tab:frw1_tov_shock_frw_side}, the numerical diffusion on the FRW-1 side ends at $r_*=4.8495$, so when computing the error, we only consider mesh points $x_i$ such that $x_1\leq x_i\leq r_*$.  As we mesh refine, the point $r_*$ increases, getting closer to the cone of sound as shown in Table \ref{tab:frw1_tov_shock_frw_side}, and the region where we are performing the error calculation expands, converging to the edge of the interaction region.  Using this procedure, the results for the convergence for the FRW-1 and TOV side are listed in Tables \ref{tab:frw1_tov_converge_frw_side} and \ref{tab:frw1_tov_converge_tov_side}, respectively.  These tables reveal both sides are converging to their respective model solution at the expected rate of 1.  We find it remarkable that the TOV metric is preserved, regardless of all the mix-up in the interaction region.  In particular, the TOV metric components remain intact even though we are integrating through the interaction region to obtain the metric.

\begin{table}[t]
\begin{center}
\begin{tabular}{|c|c|c|c|c|c|c|c|c|}
\hline
Number&\multicolumn{2}{|c|}{$\rho$}&\multicolumn{2}{|c|}{$v$}&\multicolumn{2}{|c|}{$A$}&\multicolumn{2}{|c|}{$B$}\\
\cline{2-9}
Gridpoints & Error & Rate & Error & Rate & Error & Rate & Error & Rate\\
\hline
64 & 7.670e-006 & N/A & 2.192e-003 & N/A & 4.945e-003 & N/A & 1.259e-002 & N/A\\
\hline
128 & 3.840e-006 & 1 & 1.072e-003 & 1 & 2.521e-003 & 0.97 & 6.422e-003 & 0.97\\
\hline
256 & 1.950e-006 & 0.98 & 5.324e-004 & 1 & 1.263e-003 & 1 & 3.230e-003 & 0.99\\
\hline
512 & 9.900e-007 & 0.98 & 2.664e-004 & 1 & 7.195e-004 & 0.81 & 1.797e-003 & 0.85\\
\hline
1024 & 4.940e-007 & 1 & 1.321e-004 & 1 & 3.256e-004 & 1.1 & 8.314e-004 & 1.1\\
\hline
2048 & 2.480e-007 & 0.99 & 6.610e-005 & 1 & 1.828e-004 & 0.83 & 4.566e-004 & 0.86\\
\hline
4096 & 1.240e-007 & 1 & 3.300e-005 & 1 & 8.240e-005 & 1.2 & 2.102e-004 & 1.1\\
\hline
8192 & 6.230e-008 & 1 & 1.650e-005 & 1 & 4.590e-005 & 0.84 & 1.148e-004 & 0.87\\
\hline
16384 & 3.120e-008 & 1 & 8.260e-006 & 1 & 2.060e-005 & 1.2 & 5.270e-005 & 1.1\\
\hline
\end{tabular}
\end{center}
\caption{Convergence results for the FRW side}
\label{tab:frw1_tov_converge_frw_side}
\end{table}

\begin{table}[h]
\begin{center}
\begin{tabular}{|c|c|c|c|c|c|c|c|c|}
\hline
Number&\multicolumn{2}{|c|}{$\rho$}&\multicolumn{2}{|c|}{$v$}&\multicolumn{2}{|c|}{$A$}&\multicolumn{2}{|c|}{$B$}\\
\cline{2-9}
Gridpoints & Error & Rate & Error & Rate & Error & Rate & Error & Rate\\
\hline
64 & 1.670e-007 & N/A & 1.626e-004 & N/A & 1.187e-003 & N/A & 6.794e-003 & N/A\\
\hline
128 & 1.360e-007 & 0.29 & 1.393e-004 & 0.22 & 8.266e-004 & 0.52 & 4.538e-003 & 0.58\\
\hline
256 & 8.260e-008 & 0.72 & 9.140e-005 & 0.61 & 4.773e-004 & 0.79 & 2.604e-003 & 0.8\\
\hline
512 & 4.540e-008 & 0.86 & 5.360e-005 & 0.77 & 2.900e-004 & 0.72 & 1.429e-003 & 0.87\\
\hline
1024 & 2.460e-008 & 0.88 & 2.930e-005 & 0.87 & 1.404e-004 & 1 & 7.464e-004 & 0.94\\
\hline
2048 & 1.270e-008 & 0.96 & 1.540e-005 & 0.93 & 7.970e-005 & 0.82 & 3.881e-004 & 0.94\\
\hline
4096 & 8.140e-009 & 0.64 & 2.530e-007 & 5.9 & 3.650e-005 & 1.1 & 1.946e-004 & 1\\
\hline
8192 & 4.160e-009 & 0.97 & 1.170e-008 & 4.4 & 2.030e-005 & 0.85 & 9.930e-005 & 0.97\\
\hline
16384 & 2.120e-009 & 0.97 & 3.670e-009 & 1.7 & 9.260e-006 & 1.1 & 4.920e-005 & 1\\
\hline
\end{tabular}
\end{center}
\caption{Convergence results for the TOV side}
\label{tab:frw1_tov_converge_tov_side}
\end{table}

\begin{figure}
\begin{center}
\includegraphics[width=\textwidth]{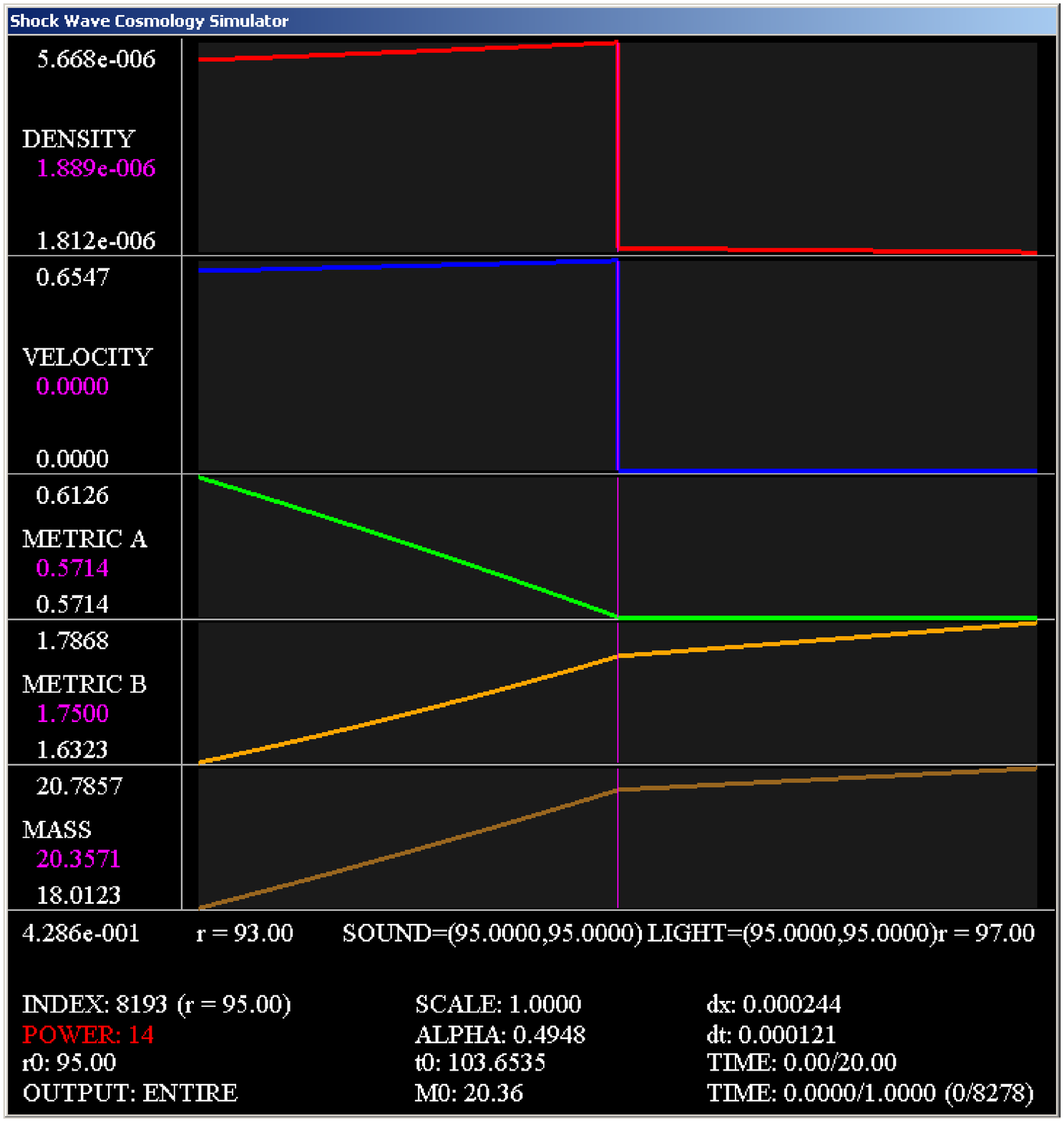}
\end{center}
\caption{The initial profiles for the shock radius at $\bar{r}_0=95$}
\label{fig:frw1_tov_init95}
\end{figure}

Since our free parameter is the position of the initial discontinuity, we have a one parameter family of solutions to the Einstein equations based on the initial radius, and we briefly explore changing this parameter.  Figure \ref{fig:frw1_tov_init95} shows the initial profile for shock position at $\bar{r}_0=95$, with $\bar{r}_{min} = 93$ and $\bar{r}_{max} = 97$.  We make a couple of observations of the initial profiles, comparing to Figure \ref{fig:frw1_tov_init}.  The fluid velocity along with the metric components at the initial discontinuity are the same, as determined in previous analysis.  All these profiles are just stretched out over a longer region of space, causing them to look more like straight lines.  There is less density and greater mass in this region of the universe since we are farther from the center of it.  Figure \ref{fig:frw1_tov_end95} shows the solution after one unit of time has passed, where again we have two shocks waves bounding a high density region.  Notice for this solution the weak shock has a positive speed as opposed to the former case, as seen in Figure \ref{fig:frw1_tov_end}.  The difference in speed from the earlier simulation $\bar{r}_0=5$ indicates we are dealing with a quantitatively different solution from before.  We explore changing this parameter many times and determine two conclusions in each case.  One is the resulting solution always has a region of higher density surrounded by two shock waves, a strong shock on the TOV side and a weak shock on the FRW-1 side, and the other is the shock waves have different speeds, resulting in quantitatively different solutions.  Hence, we truly have a one parameter family of quantitatively different shock wave solutions to the Einstein equations.

\begin{figure}
\begin{center}
\includegraphics[width=\textwidth]{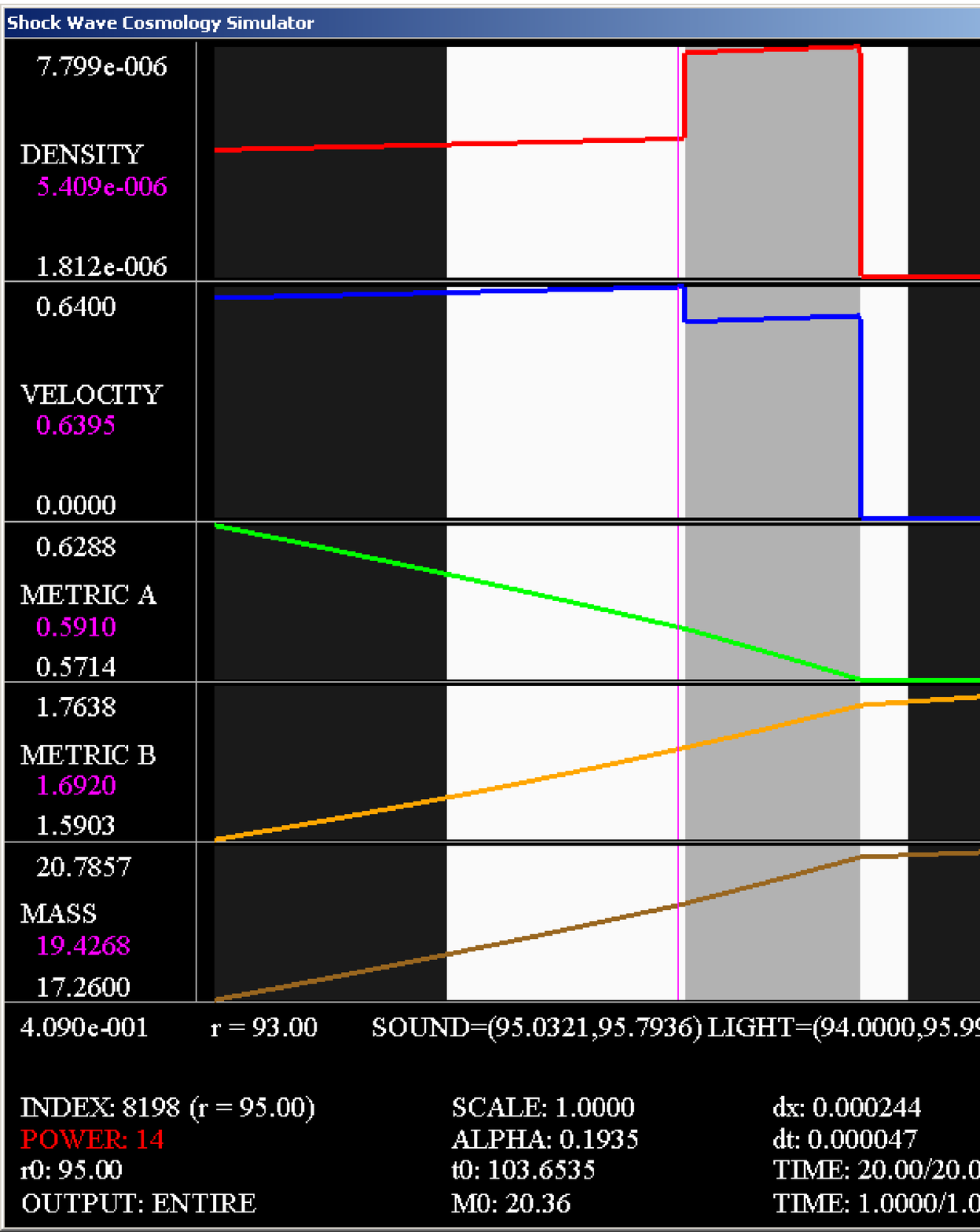}
\end{center}
\caption{Solution for the shock radius at $\bar{r}_0=95$ after a unit of time}
\label{fig:frw1_tov_end95}
\end{figure}

\section{FRW-2 and TOV Matched Model Setup}
\label{sec:frw2_tov_setup}
After exploring the FRW-1/TOV model simulation, we cover the details in the setup of the FRW-2/TOV model.  We recall the FRW-2 metric, so using the coordinate transformation
\begin{equation}\label{ch6_frw2_xform_ssc}
\begin{split}
\bar{r}=&\sqrt{t}r,\\
\bar{t}=&\frac{\Psi_0}{2}\sqrt{\frac{4t^2+\bar{r}^2}{t}},
\end{split}
\end{equation}
the FRW metric in standard Schwarzschild coordinates (the FRW-2 form) is
\begin{equation}\label{ch6_frw2_in_ssc}
ds^{2} =-\frac{1}{\Psi^2(1-v^2)}d\bar{t}^2+\frac{1}{1-v^2}d\bar{r}^2+\bar{r}^{2}d\Omega^2,
\end{equation}
with the metric components
\begin{equation}\label{ch6_frw2_metric_in_ssc}
A(\bar{t},\bar{r}) = 1-v^2,\phantom{4444} B(\bar{t},\bar{r}) = \frac{1}{\Psi^2(1-v^2)},
\end{equation}
the integrating factor
\begin{equation}\label{ch6_frw2_integrating factor}
\Psi(\bar{t},\bar{r})=\Psi_0\sqrt{\frac{t}{4t^2+\bar{r}^2}},
\end{equation}
and the fluid variables
\begin{equation}\label{ch6_frw2_fluid_vars_in_ssc}
\rho(\bar{t},\bar{r}) =\frac{3}{4\kappa t^2},\phantom{4444} v(\bar{t},\bar{r})=\frac{\eta}{2}=\frac{\bar{r}}{2t}.
\end{equation}
Remember, unlike FRW-1, the FRW-2 metric relies on the FRW time coordinate $t$, which is the following function of $(\bar{t},\bar{r})$
\begin{equation}\label{ch6_frw2_t}
t(\bar{t},\bar{r})=\frac{\bar{t}^2+\sqrt{\bar{t}^4-\bar{r}^2\Psi^4_0}}{2\Psi^2_0}.
\end{equation}
We match this metric to the TOV metric (equations (\ref{ch6_tov_in_ssc})-(\ref{ch6_gamma})), and we follow a procedure similar to the FRW-1/TOV matching.  Again, we choose the initial shock position $\bar{r}_0$, assuming an initial start time $\bar{t}_0>0$, and we want to match the FRW-2 and TOV metrics continuously at the point $(\bar{t}_0,\bar{r}_0)$.  Like before, we match the metric $A$ component as in equation (\ref{frw1_matching_a}) to obtain
\begin{equation}
v_0=\sqrt{8\pi\mathcal{G}\gamma}=\sqrt{\frac{4\sigma}{1+6\sigma+\sigma^2}},
\end{equation}
where $v_0$ represents the initial velocity at the interface.  Unlike before, we do not have a relationship between $v_0$ and $\bar{t}_0$ (\ref{ch6_xi_v_relation}) to solve for the initial time $\bar{t}_0$.  Instead, we possess a relationship between $v_0$ and the FRW time coordinate $t_0$ in equation (\ref{ch6_frw2_fluid_vars_in_ssc}) to find the FRW time coordinate at the discontinuous interface
\begin{equation}\label{frw_t_at_shock}
t_0(\bar{t}_0,\bar{r}_0)=\frac{\bar{r}_0}{2v_0}.
\end{equation}
In order to find the initial time $\bar{t}_0$, we solve for the integrating factor constant $\Psi_0$, requiring us to determine the integrating factor.  In the FRW-1 matched model, the integrating factor constant was suppressed, which is equivalent to setting $\Psi_0=1$, causing the (coordinate) speed of light $\sqrt{AB}$ on the FRW-1 side to be one.  We follow the same paradigm with the FRW-2 matched model by choosing $\Psi_0$ such that the speed of light is one on the FRW-2 side, which is equivalent to setting $\Psi=1$.  More specifically, we choose $\Psi_0$ such that at the discontinuity
\begin{equation}
\Psi=\Psi_0\sqrt{\frac{t_0}{4t^2_0+\bar{r}^2_0}}=1,
\end{equation}
implying the integrating factor constant must be
\begin{equation}\label{integrating_factor_constant}
\Psi_0=\sqrt{\frac{4t^2_0+\bar{r}^2_0}{t_0}}.
\end{equation}
Now we use (\ref{ch6_frw2_xform_ssc}) to solve for the initial time with $\bar{r}_0$, $t_0$ (\ref{frw_t_at_shock}), and $\Psi_0$ (\ref{integrating_factor_constant}) as
\begin{equation}\label{frw2_start_time}
\bar{t}_0=\frac{\Psi_0}{2}\sqrt{\frac{4t^2_0+\bar{r}^2_0}{t_0}}=\frac{\Psi^2_0}{2}.
\end{equation}
One might worry $\Psi$ does not equal one across the FRW-2 region for other values of $\bar{r}$ besides $\bar{r}_0$.  This concern is not an issue since we can substitute the time coordinate (\ref{frw2_start_time}) into the integrating factor equation (\ref{ch6_frw2_integrating factor}) to obtain
\begin{equation}
\Psi=\frac{\Psi^2_0}{2\bar{t}_0}=1,
\end{equation}
and $\bar{t}_0$ is independent of the spatial variable $\bar{r}$.  Equipped with $\bar{t}_0$ and $\Psi_0$, we are capable of solving for $t(\bar{t}_0,\bar{r})$ for any $\bar{r}< \bar{r}_0$ enabling us to compute the fluid variables and the metric components for the FRW-2 region.

For the TOV metric, the A metric component is the same constant as the FRW-1 case (\ref{ch6_tov_metric_in_ssc}), so we are left with matching the $B$ metric component.  Matching $B_{FRW}$ and $B_{TOV}$ at the coordinate $(\bar{t}_0,\bar{r}_0)$ provides us with
\begin{equation}\label{frw2_matching_b}
B_{TOV}(\bar{t}_0,\bar{r}_0)=B_0(\bar{r}_0)^{\frac{4\sigma}{1+\sigma}}=\frac{1}{\Psi^2(1-v^2_0)}=B_{FRW}(\bar{t}_0,\bar{r}_0).
\end{equation}
Since $\Psi=1$, the time scale constant $B_0$ becomes
\begin{equation}
B_0=\frac{\bar{r}^{-\frac{4\sigma}{1+\sigma}}_0}{1-v^2_0},
\end{equation}
exactly the same as the FRW-1 case (\ref{frw1_matching_b0}).  With $B_0$, we can compute the TOV metric for any radial coordinate $\bar{r}>\bar{r}_0$ by using the equations (\ref{ch6_tov_in_ssc})-(\ref{ch6_tov_fluid_vars_in_ssc}).

With the matching complete, we build the functions $v_{init}(\bar{r})$, $\rho_{init}(\bar{r})$, $A_{init}(\bar{r})$, and $B_{init}(\bar{r})$ to use as the initial data at time $\bar{t}_0$.  These initial profiles are similar to their FRW-1 matching counterparts (\ref{velocity_init})-(\ref{metric_B_init}), except the velocity is a function of $\eta$ instead of $\xi$, and we rely on the FRW time coordinate $t(\bar{t},\bar{r})$ defined in (\ref{ch6_frw2_t}).  Again, due to the dependence on the function $v$, we state the fluid velocity first
\begin{equation}\label{frw2_velocity_init}
v_{init}(\bar{r}) = \left\{
\begin{array}{ll}
\frac{\eta}{2} & \bar{r}<\bar{r}_0\\
0 & \bar{r}>\bar{r}_0,
\end{array}
\right.
\end{equation}
where $\eta =\bar{r}/t$.  The initial density profile is
\begin{equation}\label{frw2_density_init}
\rho_{init}(\bar{r}) = \left\{
\begin{array}{ll}
\frac{3}{4\kappa t^2} & \bar{r}<\bar{r}_0\\
\frac{\gamma}{\bar{r}^2} & \bar{r}>\bar{r}_0.
\end{array}
\right.
\end{equation}
Based the function $v_{init}$, the metric components are giving by
\begin{equation}\label{frw2_metric_A_init}
A_{init}(\bar{r}) = \left\{
\begin{array}{ll}
1-v_{init}^2 & \bar{r}<\bar{r}_0\\
1-8\pi\mathcal{G}\gamma & \bar{r}>\bar{r}_0,
\end{array}
\right.
\end{equation}
and
\begin{equation}\label{frw2_metric_B_init}
B_{init}(\bar{r}) = \left\{
\begin{array}{ll}
\frac{1}{\Psi^2(1-v_{init}^2)} & \bar{r}<\bar{r}_0\\
B_0(\bar{r})^{\frac{4\sigma}{1+\sigma}} & \bar{r}>\bar{r}_0.
\end{array}
\right.
\end{equation}

Since the boundary conditions are derived from the model equations, they contain the same discrepancies between the FRW-1 and FRW-2 matchings as the initial profiles.  For the FRW-2 side, at the gridpoint $x_0$ at time $\bar{t}_j$, we use (\ref{ch6_frw2_t}) to obtain the FRW time $t_j(\bar{t}_j,x_0)$ for computing the fluid velocity and density.   We find the fluid velocity
\begin{equation}\label{frw2_bndary_v}
v_{0,j}=\frac{x_0}{2t_j},
\end{equation}
and the fluid density
\begin{equation}
\rho_{0,j} =\frac{3}{4\kappa t^2_j}.
\end{equation}
Since the metric components are staggered relative to the fluid variables, we need to compute the half gridpoint as before (\ref{one_half_gridpoint}) and use it to find the corresponding velocity
\begin{equation}
v_{\frac{1}{2},j}=\frac{\eta}{2},
\end{equation}
for $\eta=x_{\frac{1}{2}}/t_j$.  We use this velocity in the following computation for the metric components,
\begin{equation}\label{frw2_bndary_metric}
A_{1j} = 1-v_{\frac{1}{2},j}^2,\phantom{4444} B_{1j} = \frac{1}{\Psi^2(1-v_{\frac{1}{2},j}^2)}.
\end{equation}

Since the TOV metric is independent of the FRW metric matched with it, the TOV boundary condition remains the same, and we use the same matching as in the FRW-1 case (\ref{rematch_b0}).

\section{FRW-2 and TOV Matched Simulation Results}
\label{sec:frw2_tov_results}
Using the initial profiles and boundary conditions developed in the last section, we run the simulation of the FRW-2/TOV matched model, comparing it with the FRW-1/TOV model results.  Figure \ref{fig:frw2_tov_init} shows the initial profiles for all the variables in the FRW-2/TOV matched model.  Comparing with the FRW-1 case (Figure \ref{fig:frw1_tov_init}), the graphs in both figures match exactly.  Actually, the only difference between the two models is the start time, $\bar{t}_0=10.9109$ instead of the previous time of $\bar{t}_0=5.4554$.  This similarity is one indication that we might have the same solution as the FRW-1 case.

\begin{figure}
\begin{center}
\includegraphics[width=\textwidth]{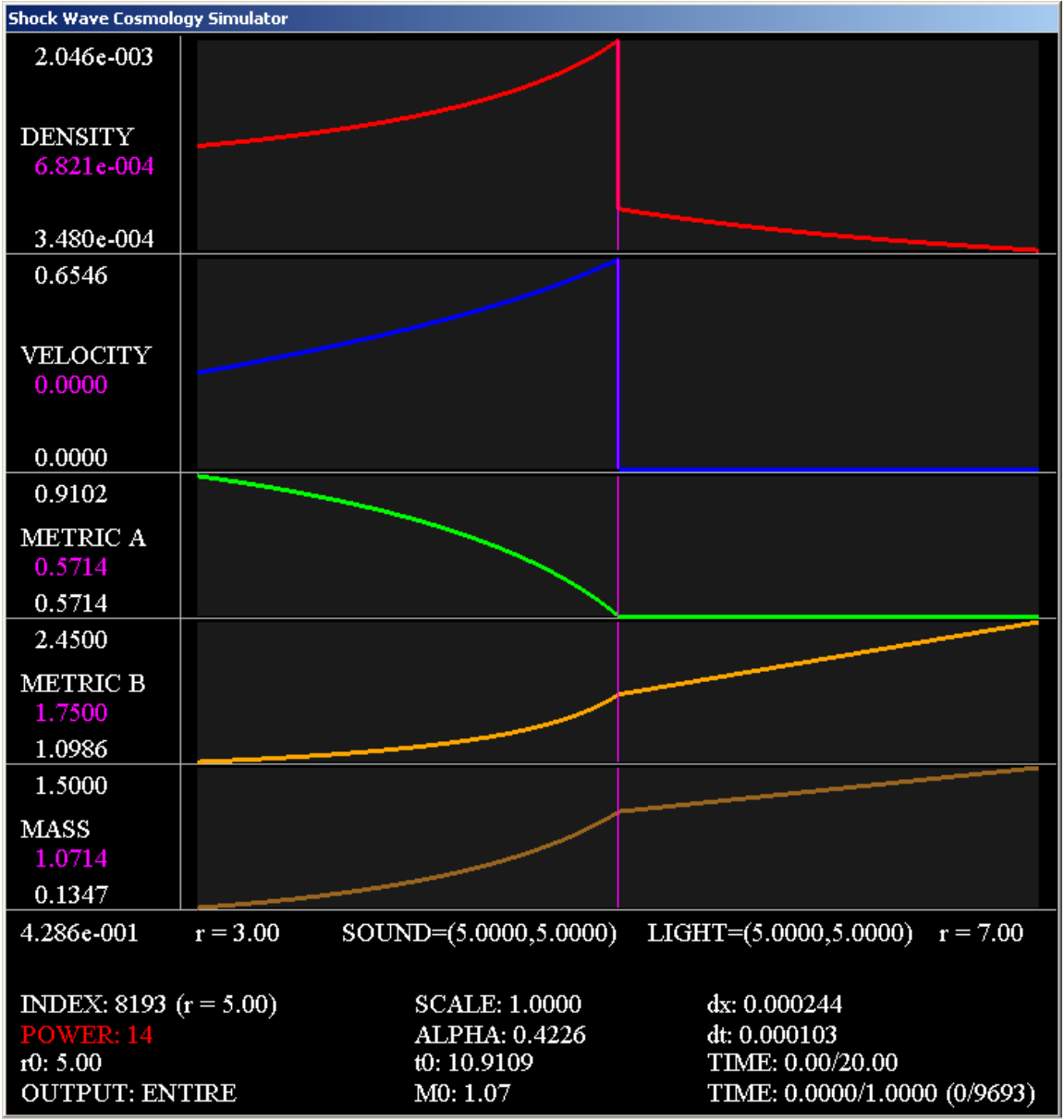}
\end{center}
\caption{Initial profiles}
\label{fig:frw2_tov_init}
\end{figure}

After running the simulation for one unit of time as shown in Figure \ref{fig:frw2_tov_end}, we notice more similarities to the FRW-1/TOV model.  More precisely,  if we compare with the FRW-1 case in Figure \ref{fig:frw1_tov_test}, both solutions have very similar features.  These features include the formation of two shocks with the stronger shock moving outward and all five graphs share the same shapes.  The main difference between the two solutions is the FRW-2 solution appears to have run for a longer period of time.  This difference is highlighted by the light and sound like information traveling slightly farther out from the initial discontinuity.  For example, the ingoing edge to the cone of sound is 4.8695 and 4.8576 in the FRW-1 and FRW-2 solutions, respectively, meaning the weak shock wave in the FRW-2 solution moved slightly farther out than the corresponding wave in the FRW-1 solution.  This outcome suggests we have the same solution but at different stages of their development.  Since both these solutions share the same initial data and the same TOV boundary condition, the only difference is the FRW boundary condition.  Looking deeper and comparing equations (\ref{frw1_bndary_v})-(\ref{frw1_bndary_metric}) to (\ref{frw2_bndary_v})-(\ref{frw2_bndary_metric}), if we assume the same FRW time $t$ at this boundary, the only change in the boundary condition is the metric component $B$ due to the presence of a non-constant integrating factor.  Intuitively, it makes sense to have the same solution at different times because a change in $B$ affects the measuring of time verses space.  This does not affect the interactions just how time and space are bent relative to each other, and since we fix spacial distances by setting constant grid cell sizes, it affects how time changes relative to each grid cell.  Hence, there is more evidence to support the claim both solutions are the same, just the clocks are not synchronized to produce the same result.

\begin{figure}
\begin{center}
\includegraphics[width=\textwidth]{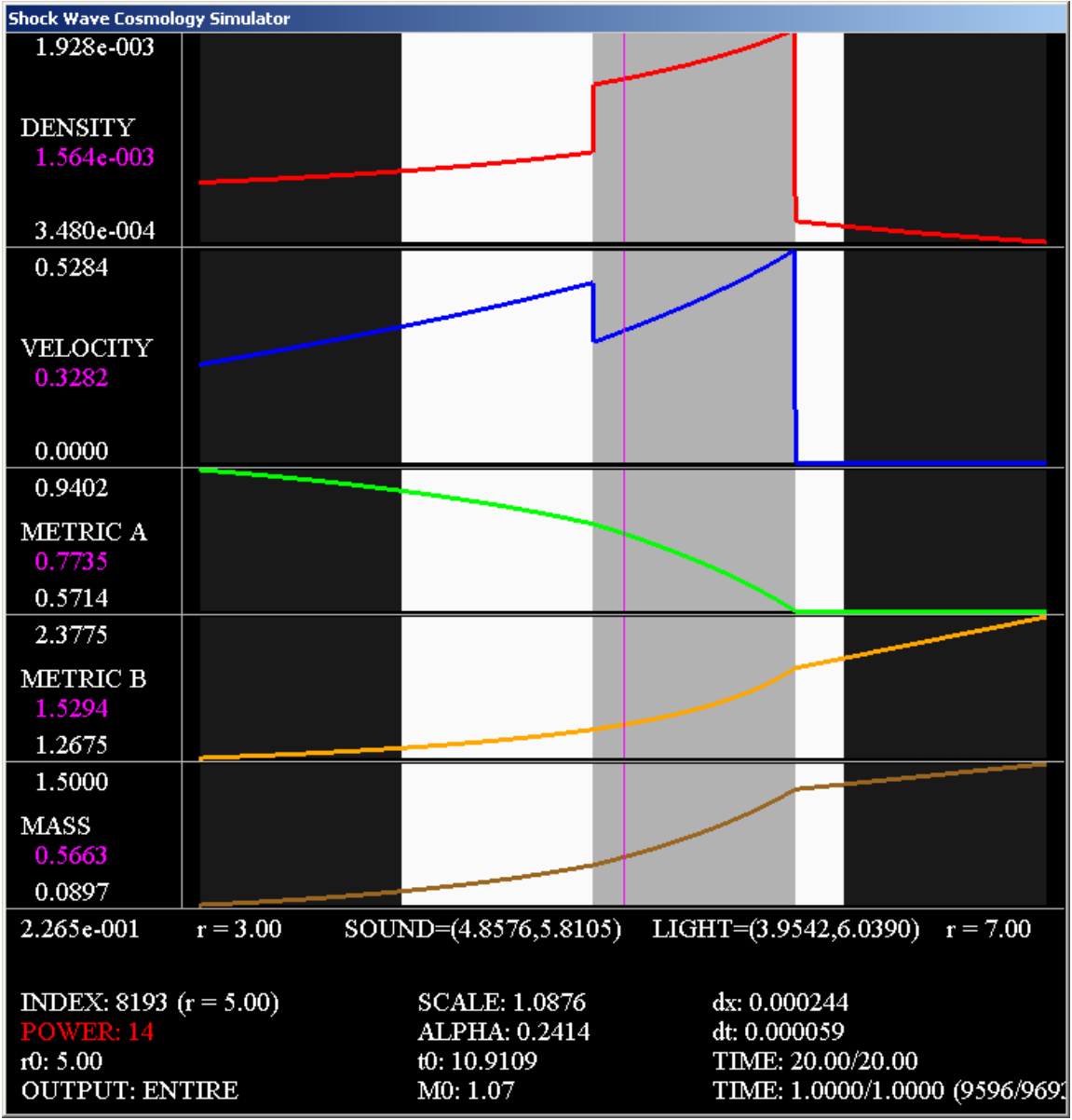}
\end{center}
\caption{Solution after a unit of time}
\label{fig:frw2_tov_end}
\end{figure}

We proceed under this assumption, a same solution at different times, and search for the FRW-2 time $\bar{t}_2$ where the solution at this time matches the solution of the FRW-1/TOV model at $\bar{t}_1=6.4554$, running for one unit of time.  If our assumption is correct, both models are mapping over the same region of spacetime in the original FRW metric.  For this mapping to work, both times $\bar{t}_1$ and $\bar{t}_2$ would have to correspond to the same FRW coordinate time $t$.  Of course, since $t$ is a function of $\bar{r}$ as well, this correspondence must hold for every $\bar{r}$ in the FRW region of the universe even though each $\bar{r}$ corresponds to a different time $t$.  To find $\bar{t}_2$, we pick a radius to match the time $t$ .  We use the left boundary radius $\bar{r}_{min}$ and find the corresponding FRW-2 time to be $\bar{t}_2=11.8688$.  This time is 0.9579 units of time after the initial start time of $\bar{t}_0=10.9109$, as observed in Figure \ref{fig:frw2_tov_init}, and this change in time is barely less than the change of one unit, explaining our earlier observation that the FRW-2/TOV solution is at a slightly later stage of development after one unit of time.

After finding this corresponding time $\bar{t}_2$, we run the FRW-2 model for 0.9579 units of time and show the results in Figure \ref{fig:frw2_tov_end2}.  Comparing against Figure \ref{fig:frw1_tov_test}, the solutions are almost exact, with the same shape and similar numbers, except for the $B$ metric component.  Interestingly enough, the metric $B$ component graph has the same shape but different numerical values, so if both graphs are placed on top of one another, they overlap, meaning the metric component in the FRW-1 case is stretched and shifted relative to the FRW-2 case.

\begin{figure}
\begin{center}
\includegraphics[width=\textwidth]{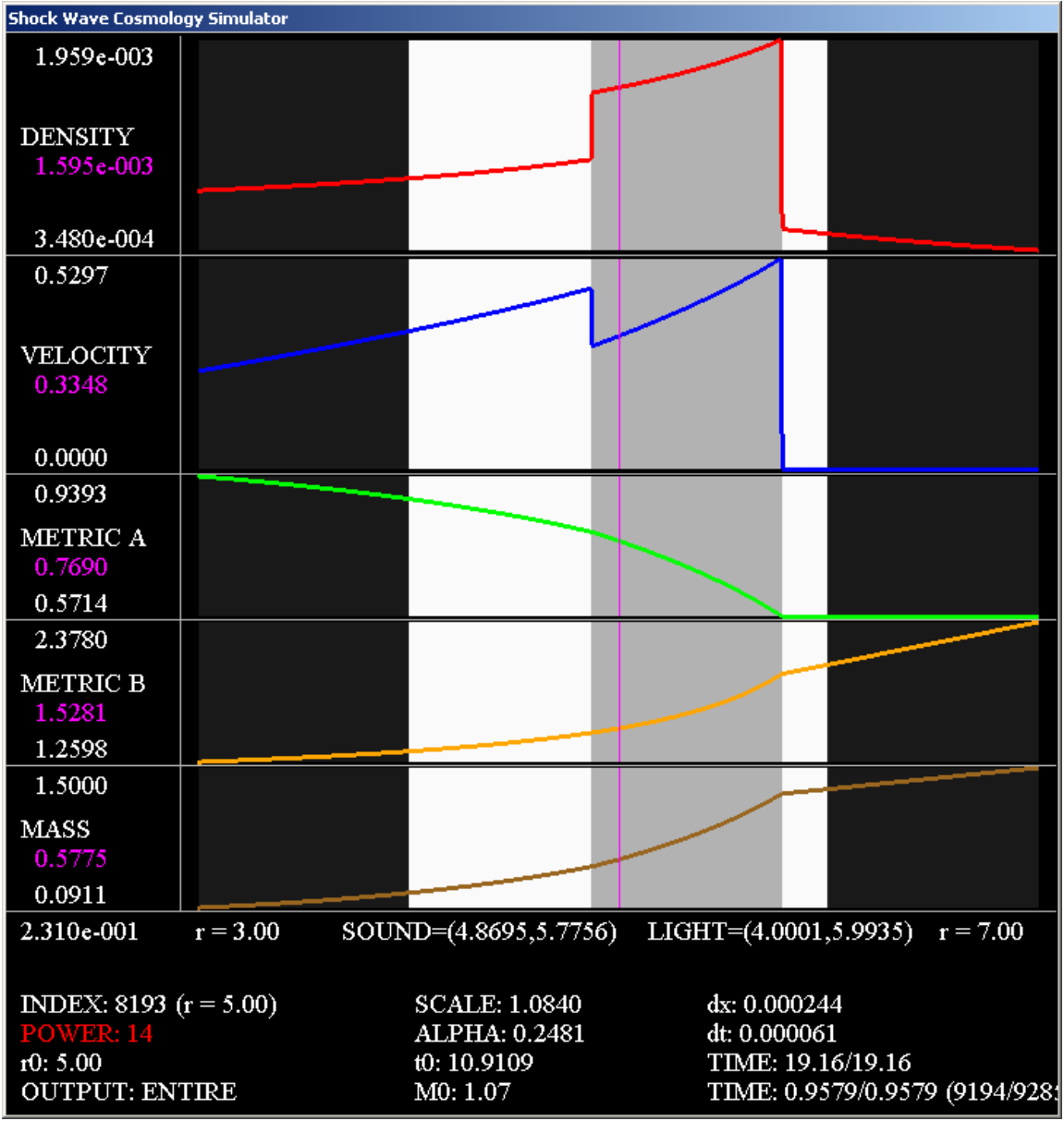}
\end{center}
\caption{Solution after 0.9579 units of time}
\label{fig:frw2_tov_end2}
\end{figure}

We continue on and test our same solution hypothesis numerically.  We accomplish this task by using the FRW-1/TOV solution at a fine mesh refinement, $n=16,384$, as the model to compare against.  We perform convergence calculations assuming the FRW-1 matched solution is the true solution, so we run the FRW-2/TOV matched model for 0.9579 units of time for different mesh sizes and compute the error between this solution and the FRW-1/TOV solution, testing the convergence rates.  One problem with this approach is the $B$ metric components between the two solutions are shifted and stretched from one another, so a direct error calculation for $B$ does not result in convergence.  To overcome this discrepancy, we build a map between the two metric components, which we label as $B_1$ and $B_2$ for the FRW-1/TOV and FRW-2/TOV models, respectively.  This map is built using the following data: let $B^1_{max}$ and $B^1_{min}$ represent the maximum and minimum values, respectively, for the $B$ metric component in the FRW-1/TOV model, and let $B^2_{max}$ and $B^2_{min}$ be the corresponding values for the FRW-2/TOV model.  We define the mapping from $B_1$ to $B_2$ as follows
\begin{equation}
B_2=\left(\frac{B^2_{max}-B^2_{min}}{B^1_{max}-B^1_{min}}\right)(B_1-B^1_{min})+B^2_{min}.
\end{equation}
The value $B^2_{min}$ represents the value the left most point gets mapped to while the scale factor $(B^2_{max}-B^2_{min})/(B^1_{max}-B^1_{min})$ represents the stretching.  For the mapping between the solutions in Figure \ref{fig:frw1_tov_test} and Figure \ref{fig:frw2_tov_end2}, $B^2_{min}$ is 1.2598 while the scale factor is 1.1833.  Using this map, the convergence results are shown in Table \ref{tab:frw2_tov_converge_to_frw1}, and the data shows the FRW-2/TOV model is converging to the FRW-1/TOV model at a linear rate, confirming our same solution at different times hypothesis numerically.

\begin{table}
\begin{center}
\begin{tabular}{|c|c|c|c|c|c|c|c|c|}
\hline
Number&\multicolumn{2}{|c|}{$\rho$}&\multicolumn{2}{|c|}{$v$}&\multicolumn{2}{|c|}{$A$}&\multicolumn{2}{|c|}{$B$}\\
\cline{2-9}
Gridpoints & Error & Rate & Error & Rate & Error & Rate & Error & Rate\\
\hline
64 & 2.463e-004 & N/A & 1.021e-001 & N/A & 6.514e-003 & N/A & 1.256e-002 & N/A\\
\hline
128 & 1.517e-004 & 0.7 & 5.883e-002 & 0.8 & 3.285e-003 & 0.99 & 6.849e-003 & 0.88\\
\hline
256 & 9.440e-005 & 0.68 & 3.788e-002 & 0.64 & 1.755e-003 & 0.9 & 3.527e-003 & 0.96\\
\hline
512 & 5.470e-005 & 0.79 & 2.149e-002 & 0.82 & 9.644e-004 & 0.86 & 1.763e-003 & 1\\
\hline
1024 & 2.800e-005 & 0.97 & 1.077e-002 & 1 & 4.872e-004 & 0.99 & 8.090e-004 & 1.1\\
\hline
2048 & 1.340e-005 & 1.1 & 5.132e-003 & 1.1 & 2.466e-004 & 0.98 & 3.589e-004 & 1.2\\
\hline
4096 & 5.920e-006 & 1.2 & 2.234e-003 & 1.2 & 1.111e-004 & 1.1 & 1.696e-004 & 1.1\\
\hline
8192 & 2.010e-006 & 1.6 & 7.468e-004 & 1.6 & 5.520e-005 & 1 & 1.150e-004 & 0.56\\
\hline
16384 & 1.980e-007 & 3.3 & 5.440e-005 & 3.8 & 1.790e-005 & 1.6 & 4.700e-005 & 1.3\\
\hline
\end{tabular}
\end{center}
\caption{FRW-2/TOV model verses FRW-1/TOV model convergence results}
\label{tab:frw2_tov_converge_to_frw1}
\end{table}

We find the discovery that the two FRW transformations produce the same result too uncanny to be just a random coincidence, and we search for a theoretical connection to justify this discovery.  Since we are dealing with a solution to the Einstein equations of a spherically symmetric metric in standard Schwarzschild coordinates, we take another look at the general form for this metric
\begin{equation}\label{ch6_ssc}
ds^{2} =-B(\bar{t},\bar{r})d\bar{t}^2+\frac{1}{A(\bar{t},\bar{r})}d\bar{r}^2+\bar{r}^{2}d\Omega^2.
\end{equation}
Let us consider the coordinate freedoms that preserves this metric; it turns out the only freedom of this metric is an arbitrary nonlinear transformation of the time coordinate $\bar{t}$.  We cannot change the radial coordinate $\bar{r}$ because the spheres of symmetry will no longer match, and we cannot change the time coordinate $\bar{t}$ in terms of $\bar{r}$ because this type of transformation would yield mixed terms violating the metric being in standard Schwarzschild coordinates.  Therefore, if the spacetime is in standard Schwarzschild coordinates, the only change of coordinates of this metric, using the same center, is an arbitrary nonlinear change in the time coordinate.  Since all the mappings from FRW into standard Schwarzschild coordinates differ by how the time coordinate $\bar{t}(t,\bar{r})$ gets mapped, based on different integrating factors, all these mapped FRW metrics must be the same solution which differ by a nonlinear change in the time coordinate.  Hence, both metrics FRW-1 and FRW-2 developed earlier are actually the same solution with a nonlinear change of time between the two, implying the FRW-1/TOV and FRW-2/TOV matched models are the same too.

With this fact in mind, we look for the transformation between the two time coordinates $\bar{t}_1$ and $\bar{t}_2$ for the FRW-1/TOV and FRW-2/TOV models, respectively.  Considering an arbitrary coordinate $(t,\bar{r})$, the FRW-1 time $\bar{t}_1$ is
\begin{equation}
\bar{t}_1=\left\{1+\frac{\bar{r}^2}{4t^2}\right\}t=\frac{4t^2+\bar{r}^2}{4t},
\end{equation}
and the FRW-2 time $\bar{t}_2$ is
\begin{equation}
\bar{t}_2=\frac{\Psi_0}{2}\sqrt{\frac{4t^2+\bar{r}^2}{t}},
\end{equation}
where we set the integrating factor constant
\begin{equation}
\Psi_0=\sqrt{\frac{4t^2+\bar{r}^2}{t}},
\end{equation}
as before (\ref{integrating_factor_constant}) to make $\Psi=1$ and initially match the $B$ metric component for the FRW-1/TOV and FRW-2/TOV models.  Studying these equations, we notice a couple of things.  The first one is the integrating factor constant chosen can be written as an explicit function of $\bar{t}_1$
\begin{equation}\label{psi_fn_t_bar_1}
\Psi_0=\sqrt{\frac{4t^2+\bar{r}^2}{t}}=2\sqrt{\frac{4t^2+\bar{r}^2}{4t}}=2\sqrt{\bar{t}_1}.
\end{equation}
Not only that, but the FRW-2 time can also be written as an explicit function of $\bar{t}_1$
\begin{equation}\label{t_bar_xform}
\bar{t}_2=\Psi_0\sqrt{\frac{4t^2+\bar{r}^2}{4t}}=\Psi_0\sqrt{\bar{t}_1}.
\end{equation}
Looking at the specific coordinate $(t_0,\bar{r}_0)$ for the initial discontinuity, we combine the equations (\ref{psi_fn_t_bar_1}) and (\ref{t_bar_xform}) to determine
\begin{equation}
\bar{t}_2=2\bar{t}_1,
\end{equation}
explaining the relationship between the FRW-1/TOV start time $\bar{t}_0=5.4554$ and the FRW-2/TOV start time $\bar{t}_0=10.9109$.  We also build the relationship between the two end times using (\ref{integrating_factor_constant}) as the integrating factor constant in the time mapping equation (\ref{t_bar_xform}).  Substituting the FRW-1/TOV end time $\bar{t}_1=6.4554$ into this equation (\ref{t_bar_xform}) gives us the corresponding FRW-2/TOV end time as $\bar{t}_2=11.8688$, matching the end time above to produce Figure \ref{fig:frw2_tov_end2} and to obtain the convergence results in Table \ref{tab:frw2_tov_converge_to_frw1}.

The other mystery on our hands is the shifting and stretching of the metric $B$ component.  The shifting is a result of the difference in the boundary data of $B$ between the FRW-1 and FRW-2 cases, caused by the different integrating factors.  This metric component is computed at the end of each time step by integration of the solution up from the boundary, and a shift in the boundary data of $B$ will result in a shift of the entire function $B$.  The stretching is caused by the effect on the metric of the coordinate transformation from the $\bar{t}_1$ to the $\bar{t}_2$ coordinate system.  More specifically, consider a time coordinate map $\varphi:\bar{t}_2\rightarrow\bar{t}_1$ from the FRW-2/TOV to the FRW-1/TOV time coordinate.  This map produces the following relationship between the differential forms
\begin{equation}
d\bar{t}^2_1=\left(\frac{d\varphi}{d\bar{t}_2}\right)^2d\bar{t}^2_2,
\end{equation}
causing a scaling or stretching between the $B$ metric components.  For our particular case, the time coordinate map is
\begin{equation}
\bar{t}_1=\varphi(\bar{t}_2)=\left(\frac{\bar{t}^2_2}{\Psi_0}\right),
\end{equation}
giving us a scale factor of
\begin{equation}
\left(\frac{d\varphi}{d\bar{t}_2}\right)^2=4\left(\frac{\bar{t}_2}{\Psi^2_0}\right)^2.
\end{equation}
Plugging in the values for $\Psi_0$ and $\bar{t}_2$, the scale factor equals 1.1833, matching the observed value derived numerically above.  This analysis gives us a procedure of mapping the FRW-1/TOV model over to the FRW-2/TOV model for any FRW-1 time $\bar{t}_1$, confirming our hypothesis that both models represent the same solution at different time coordinates.  Consequently, running the FRW-2/TOV simulation is just an exercise in the time coordinate invariance of this solution of the Einstein equations in standard Schwarzschild coordinates.

\section{Summary}
\label{sec:shock_wave_summary}
We provide the specific details in building the initial profiles (shown in Figure \ref{fig:frw1_tov_init}) and the boundary data required to run the simulation using the locally inertial Godunov method.  We have the first glimpse of a shock wave in general relativity by providing snapshots of the solution in Figure \ref{fig:frw1_tov_frames}.  This solution results in two shock waves, a strong outgoing wave and a weak incoming wave, enveloping a region of higher density.  Changing our one parameter $\bar{r}$ produces quantitatively different solutions, shown in Figure \ref{fig:frw1_tov_end95}.  This result enables us to resolve the secondary wave in the solution as a reflected shock wave.  This shock wave solution is a weak solution to the Einstein equations because of the discontinuities in the derivative of the metric shown in Figure \ref{fig:frw1_tov_end}, where these discontinuities are aligned with the corresponding discontinuities in the fluid variables.  We show convergence of the region of interaction to the cone of sound by recording the convergence between the position of the shock waves and the edges of the cone of sound in Tables \ref{tab:frw1_tov_shock_frw_side} and \ref{tab:frw1_tov_shock_tov_side}.  This result affirms that the region of interaction is precisely the region of the cone of sound, and again at this stage we have no formal mathematical proof of this.  The numerical determination of these borders help us demonstrate the numerical convergence of the FRW and TOV metrics in the non-interaction regions, using the exact solutions, in Tables \ref{tab:frw1_tov_converge_frw_side} and \ref{tab:frw1_tov_converge_tov_side} because these borders gives us the boundary of integration in the calculation of the one-norm.

In simulating the FRW-2/TOV model, we force the initial profiles for this model to closely resemble the FRW-1/TOV model, and after running the FRW-2/TOV model, we notice many similarities to its counterpart, providing evidence that these are the same solution at different times.  We test and confirm this hypothesis numerically by proving convergence between the two solutions in Table \ref{tab:frw2_tov_converge_to_frw1}.  We proceed to find the theoretical justification of these numerical results and determine the two solutions differ by a non-linear time coordinate transformation.  This result can be interpreted as a numerical confirmation of the covariance of solutions to the Einstein equations in standard Schwarzschild coordinates.  All these results give us confidence in the correctness of our locally inertial Godunov method implementation and in the accuracy of our solution.

   \chapter[%
      Short Title of 8th Ch.
   ]{%
      Time Reversal Model
   }%
   \label{ch:time_rev_model}
Now that we have explored the shock wave model by simulating the FRW/TOV matched model forward in time, we consider reversing time and running the matched spacetime backward in time.  Not only are we interested in the structure of the resulting solution, but we conjecture reversing time will send the strong wave into the center of the universe in hopes of forming a black hole.  By sending matter into the center, we expect the density and consequently the mass function to increase as time unfolds.  This increase at a fixed radius $\bar{r}$ will cause the black hole criteria
\begin{equation}\label{black_hole_criteria}
\frac{2\mathcal{G}M(\bar{t},\bar{r})}{\bar{r}}=1,
\end{equation}
to eventually be satisfied for some time-space coordinate $(\bar{t},\bar{r})$.  Since it is referred to throughout the chapter, we define the {\it black hole number} $\mu(\bar{t},\bar{r})$ as
\begin{equation}\label{black_hole_number}
\mu(\bar{t},\bar{r})\equiv\frac{2\mathcal{G}M(\bar{t},\bar{r})}{\bar{r}}.
\end{equation}

Section \ref{sec:rev_time} discusses the validity of a reversed time solution along with the setup to implement it.  We show the difference of this setup from the forward time solution is just the sign of the times, $\bar{t}$ and $t$, from positive to negative, so the initial profiles and boundary conditions developed in Section \ref{sec:frw1_tov_setup} are used again here with this slight change.  In Section \ref{sec:rev_frw1_tov_results}, we run this reversed time simulation for one unit of time to obtain results of the solution.  In particular, we demonstrate numerical convergence of the entire solution, the interaction and non-interaction regions, and within the non-interacting regions, numerical convergence to the pure FRW and TOV metrics is also verified.  As in the forward time case, we demonstrate the region of interaction is synchronized with the cone of sound region.  This solution contains two expansion waves surrounding a region of under density, as opposed to the two shock waves in the forward time model.  Both of these waves are an expansion wave in curved spacetime which we interpret as a generalized rarefaction wave, and we refer to it throughout as a rarefaction wave.  Since we are dealing with rarefaction waves, the metric is continuous, producing a strong solution to the Einstein equations.  After one unit of time, the strong rarefaction wave is heading toward the center of the universe, bringing mass along for the ride, and getting closer to the black hole criteria (\ref{black_hole_criteria}).  In Section \ref{sec:black_hole_form}, we continue the flow of time to see the beginning of what we believe is the formation of a black hole, where we obtain a black hole number of 0.9218.  Although we wish to get even closer to satisfying the black hole criteria (\ref{black_hole_criteria}), we are content with passing the Buchdahl stability limit of 0.9, which loosely states, after reaching this limit, a star is unstable and is subject to it own gravitational collapse.  That is, the theorem of Buchdahl is whenever a mass gets within 9/8ths of its Schwzarzschild radius, or in our language its black hole number gets within 0.9, there doesn't exist a static solution of the TOV equations with a finite pressure at the center capable of holding the star up \cite{smolte2}.  This fact gives us confidence that not only does black hole formation occur in the reverse FRW/TOV model, but our simulation gives us a first glimpse at it.

\section{Reversing Time}
\label{sec:rev_time}

We argue the time reversibility of the FRW metric.  Recall, the FRW metric is given by
\begin{equation}\label{ch7_frw_metric_with_k}
ds^2=-dt^2+R^2(t)\left\{\frac{1}{1-kr^2}dr^2+r^2d\Omega^2\right\}.
\end{equation}
Plugging (\ref{ch7_frw_metric_with_k}) into the Einstien field equations (\ref{einstein_eqns}), along with the perfect and comoving fluid assumption, gives the following pair of constraint equations on the functions $R(t)$, $\rho(t)$, and $p(t)$:
\begin{equation}\label{frw_general_rho_ode}
p=-\rho-\frac{R\dot{\rho}}{3\dot{R}},
\end{equation}
\begin{equation}\label{frw_general_r_ode}
\dot{R}^2+k=\frac{8\pi\mathcal{G}}{3}\rho R^2,
\end{equation}
developed by Smoller and Temple in \cite{smolte}.  For the case of interest, when the metric is spatially flat ($k=0$) with an isothermal equation of state ($p=\sigma\rho$), these constraint equations (\ref{frw_general_rho_ode}) and (\ref{frw_general_r_ode}) simplify to
\begin{equation}\label{frw_rho_ode}
\dot{\rho}=\frac{(\sigma +1)\rho}{3R}\dot{R},
%p=-\rho-\frac{R\dot{\rho}}{3\dot{R}},
\end{equation}
\begin{equation}\label{frw_r_ode}
\dot{R}^2=\frac{8\pi\mathcal{G}}{3}\rho R^2.
\end{equation}
There exists explicit solutions, $\rho(t)$ and $R(t)$, to these pair of ODEs for $t>0$ \cite{smolte}.  Knowing these solutions exist, we show $\rho(-t)$ and $R(-t)$ for $t<0$, the reversed time solutions, are also solutions to these ODEs.  Define $\tau \equiv -t$, and since $\tau > 0$ there exists solutions $\rho(\tau)$ and $R(\tau)$ to (\ref{frw_rho_ode}) and (\ref{frw_r_ode}).  Using these solutions and the fact $\frac{dt}{d\tau}=-1$, we show
\begin{equation}
-\dot{\rho}=\frac{dt}{d\tau}\frac{d\rho}{dt}=\rho_\tau=\frac{(\sigma +1)\rho}{3R}R_\tau =\frac{(\sigma +1)\rho}{3R}\frac{dR}{dt}\frac{dt}{d\tau} =-\frac{(\sigma +1)\rho}{3R}\dot{R},
\end{equation}
implying
\begin{equation}
\dot{\rho}=\frac{(\sigma +1)\rho(-t)}{3R(-t)}\dot{R}.
\end{equation}
The other ODE is easily satisfied by
\begin{equation}
\dot{R}^2=\left(\frac{dt}{d\tau}\dot{R}\right)^2=R^2_\tau=\frac{8\pi\mathcal{G}}{3}\rho R^2 =\frac{8\pi\mathcal{G}}{3}\rho(-t) R^2(-t).
\end{equation}
Hence, the time reversed solutions $\rho(-t)$ and $R(-t)$ for $t<0$ are also solutions for the FRW metric.  Notice the FRW metric is invariant under this time reversal $\tau=-t$:
\begin{equation}
ds^2=-d\tau^2+R^2(\tau)\left\{\frac{1}{1-kr^2}dr^2+r^2d\Omega^2\right\}= -dt^2+R^2(-t)\left\{\frac{1}{1-kr^2}dr^2+r^2d\Omega^2\right\}.
\end{equation}

We turn our attention to how this time reversal affects the FRW-1 metric.  If we repeat the process in Chapter \ref{ch:family_of_shock_waves} to obtain the coordinate transformation, the time reversed transformation becomes
\begin{equation}
\begin{split}
\bar{r}=&\sqrt{-t}r,\\
\bar{t}=&\left\{1+\frac{\bar{r}^2}{4t^2}\right\}t,
\end{split}
\end{equation}
and the resulting equations for the fluid variables $(\rho,v)$ and the metric components $(A,B)$ remain the same.  The only difference is that $t<0$ implies $\bar{t}<0$ and consequently, $\xi=\bar{r}/\bar{t}$ changes sign from positive to negative.  This change in sign only affects the fluid velocity since all the other equations in the metric posses $v^2$ terms in them.  With this in mind, we use the results in the FRW-1/TOV matched model process in Chapter \ref{ch:shock_wave_models} with a negative time $\bar{t}<0$.  In particular, we use the same equations (\ref{velocity_init})-(\ref{rematch_b0}) for the initial profiles and the boundary conditions with a negative time, increasing as the simulation runs.  The only difference is the fluid velocity at the discontinuity $v_0$ is chosen to be negative instead of positive.

\section{Reverse FRW and TOV Matched Simulation Results}
\label{sec:rev_frw1_tov_results}

\begin{figure}
\begin{center}
\includegraphics[width=\textwidth]{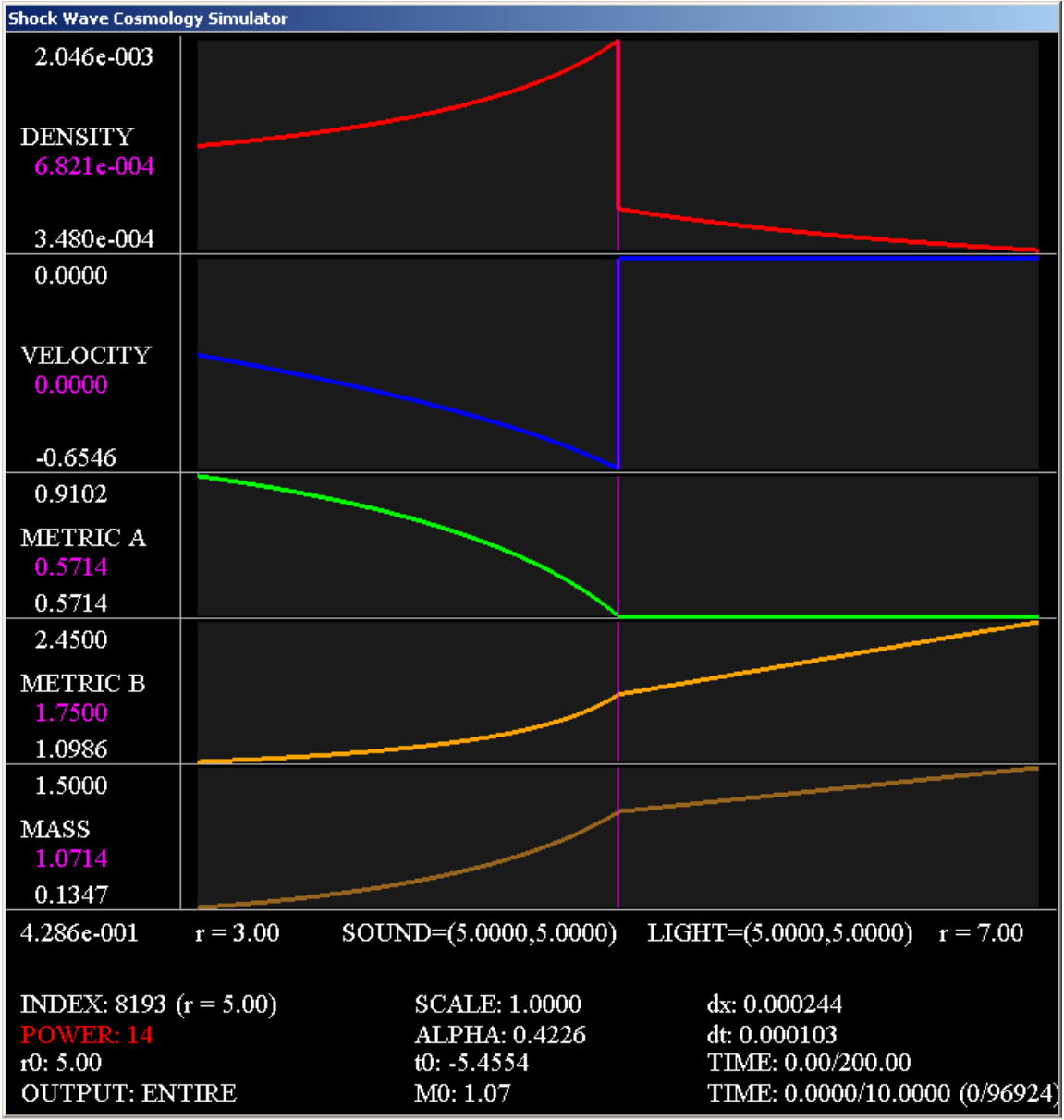}
\end{center}
\caption{Initial profiles}
\label{fig:frw1_rev_tov_init}
\end{figure}

To better compare and contrast the forward and reverse solutions, we use the same initial parameters as before (\ref{std_params}).  Figure \ref{fig:frw1_rev_tov_init} shows the initial profiles for the fluid variables, the metric components, and the mass function.  Comparing these profiles to the corresponding ones for the forward time in Figure \ref{fig:frw1_tov_init}, we see they match except the fluid velocity along with the start time $\bar{t}_0=-5.4554$ has changed signs, as expected from the above analysis.

\begin{figure}
\begin{center}
\includegraphics[width=\textwidth]{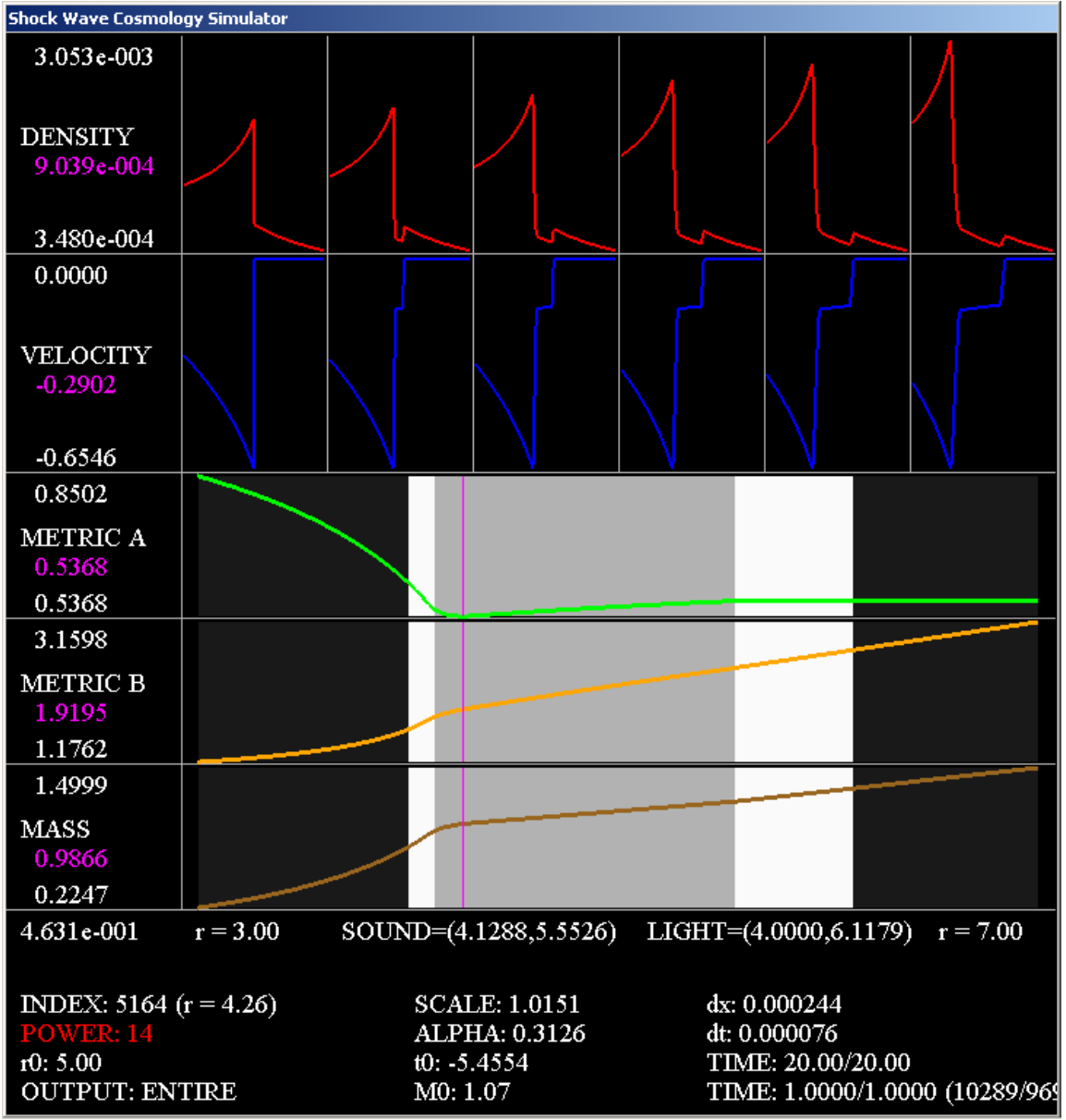}
\end{center}
\caption{Evolution of the fluid variables during a unit of time}
\label{fig:frw1_rev_tov_frames}
\end{figure}

After setting up the initial profiles, we run the simulation for one unit of time, but since time is running in reverse, $\bar{t}_0$ is getting closer to zero, evolving toward the big bang instead of away from it (i.e. $\bar{t}_{end}=\bar{t}_0+1=-4.4554$).  Figure \ref{fig:frw1_rev_tov_frames} gives a frame by frame view for the evolution of the fluid variables across this time frame, evenly distributed from the left frame at $\bar{t}_0$ to the right frame at $\bar{t}_{end}$.  From the initial discontinuity, two rarefaction waves are formed, a stronger wave moving in towards the FRW region and a weaker one moving out towards the TOV region, and between the two waves a pocket of lower density is formed.  These results differ from the forward time case where there are two shock waves surrounding a region of higher density.  As time progresses, the rarefaction waves and the pocket are both expanding.  To the left of this pocket, there is a density spike growing and moving in towards the FRW region, meaning more matter is falling into the center of the universe as time progresses.  One can notice that indeed the FRW region is accumulating more mass from the start time (Figure \ref{fig:frw1_rev_tov_init}) to the end time (Figure \ref{fig:frw1_rev_tov_frames}).  If more mass keeps coming into a finite radius, then $\mu$ will continue to grow and eventually the black hole criteria will be satisfied (\ref{black_hole_criteria}), leading us to believe, given enough time, a black hole will evolve out of this solution. We explore this idea in more detail in the next section.

\begin{figure}
\begin{center}
\includegraphics[width=\textwidth]{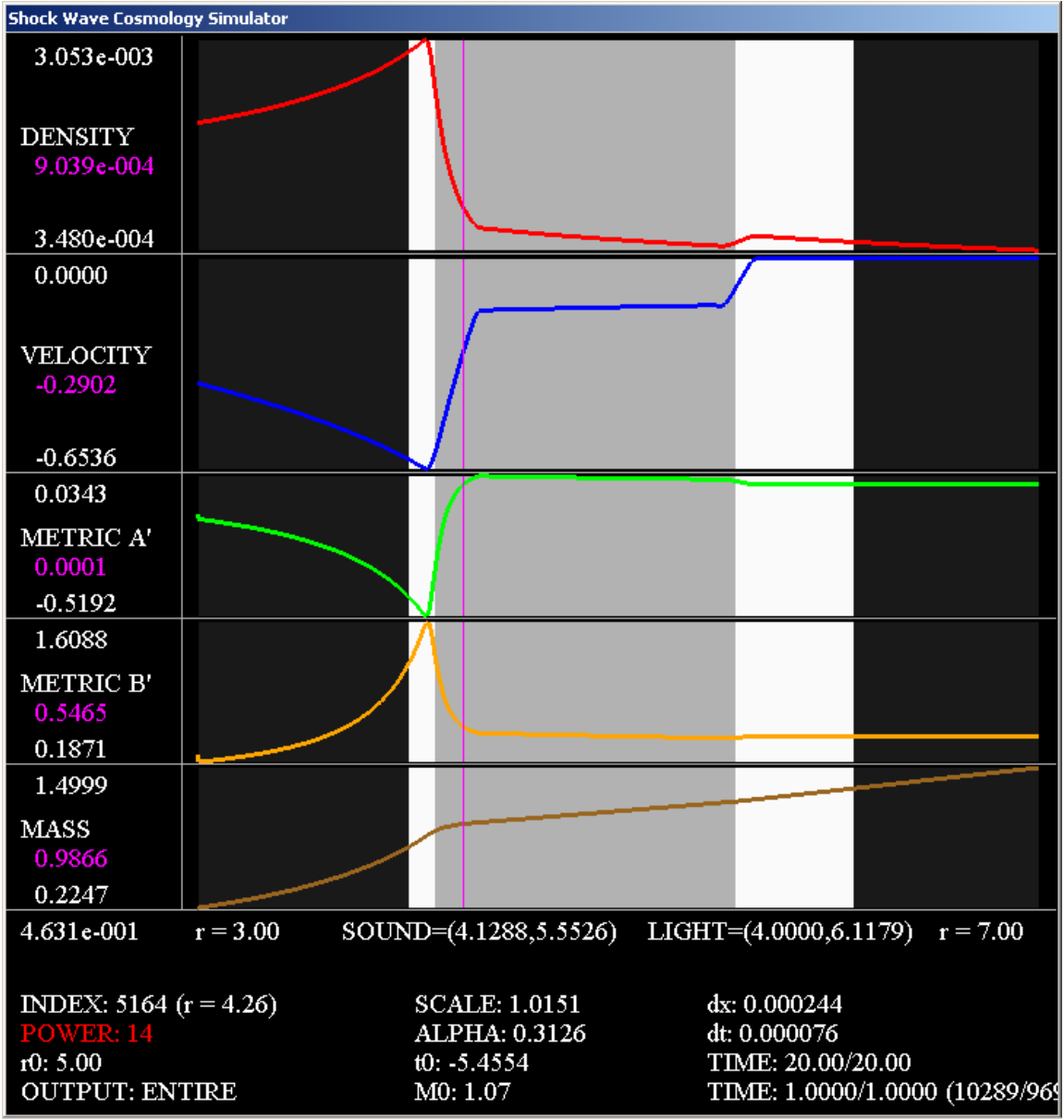}
\end{center}
\caption{Solution after a unit of time, showing the derivatives of the metric}
\label{fig:frw1_rev_tov_end}
\end{figure}

We turn our attention to the end result at time $\bar{t}_{end}$, where Figure \ref{fig:frw1_rev_tov_end} gives us a more detailed look at the fluid variables.  This figure allows us to see the rarefaction waves in their entirety, where one might mistake them as shock waves in Figure \ref{fig:frw1_rev_tov_frames}.  Comparing to the forward time model (Figure \ref{fig:frw1_tov_end}), the cone of sound is wider and shifted more to the left, while the cone of light remains the same.  This shift toward the FRW side is caused by the negative velocity forcing the fluid and consequently all sound like information inward instead of outward. Observe the derivatives of the metric components are continuous, producing no jump discontinuities as seen in the forward time case; without these discontinuities, we no longer have to settle for a weak solution, so the reversed time model is a strong solution to the Einstein equations.

\begin{table}[!t]
\begin{center}
\begin{tabular}{|c|c|c|c|c|c|c|c|c|}
\hline
Number&\multicolumn{2}{|c|}{$\rho$}&\multicolumn{2}{|c|}{$v$}&\multicolumn{2}{|c|}{$A$}&\multicolumn{2}{|c|}{$B$}\\
\cline{2-9}
Gridpoints & Error & Rate & Error & Rate & Error & Rate & Error & Rate\\
\hline
64 & 1.472e-004 & N/A & 5.334e-002 & N/A & 2.707e-002 & N/A & 9.879e-002 & N/A\\
\hline
128 & 1.126e-004 & 0.39 & 3.815e-002 & 0.48 & 1.346e-002 & 1 & 5.532e-002 & 0.84\\
\hline
256 & 8.210e-005 & 0.46 & 2.694e-002 & 0.5 & 6.751e-003 & 1 & 3.299e-002 & 0.75\\
\hline
512 & 5.850e-005 & 0.49 & 1.889e-002 & 0.51 & 4.063e-003 & 0.73 & 2.632e-002 & 0.33\\
\hline
1024 & 4.090e-005 & 0.52 & 1.301e-002 & 0.54 & 2.042e-003 & 0.99 & 1.592e-002 & 0.73\\
\hline
2048 & 2.790e-005 & 0.55 & 8.770e-003 & 0.57 & 8.794e-004 & 1.2 & 8.348e-003 & 0.93\\
\hline
4096 & 1.860e-005 & 0.58 & 5.764e-003 & 0.61 & 4.801e-004 & 0.87 & 5.295e-003 & 0.66\\
\hline
8192 & 1.220e-005 & 0.62 & 3.685e-003 & 0.65 & 2.705e-004 & 0.83 & 3.312e-003 & 0.68\\
\hline
16384 & 7.750e-006 & 0.65 & 2.294e-003 & 0.68 & 1.622e-004 & 0.74 & 2.210e-003 & 0.58\\
\hline
\end{tabular}
\end{center}
\caption{Successive mesh refinement convergence results}
\label{tab:frw1_rev_tov_converge}
\end{table}

\begin{figure}
\begin{center}
\includegraphics[width=\textwidth]{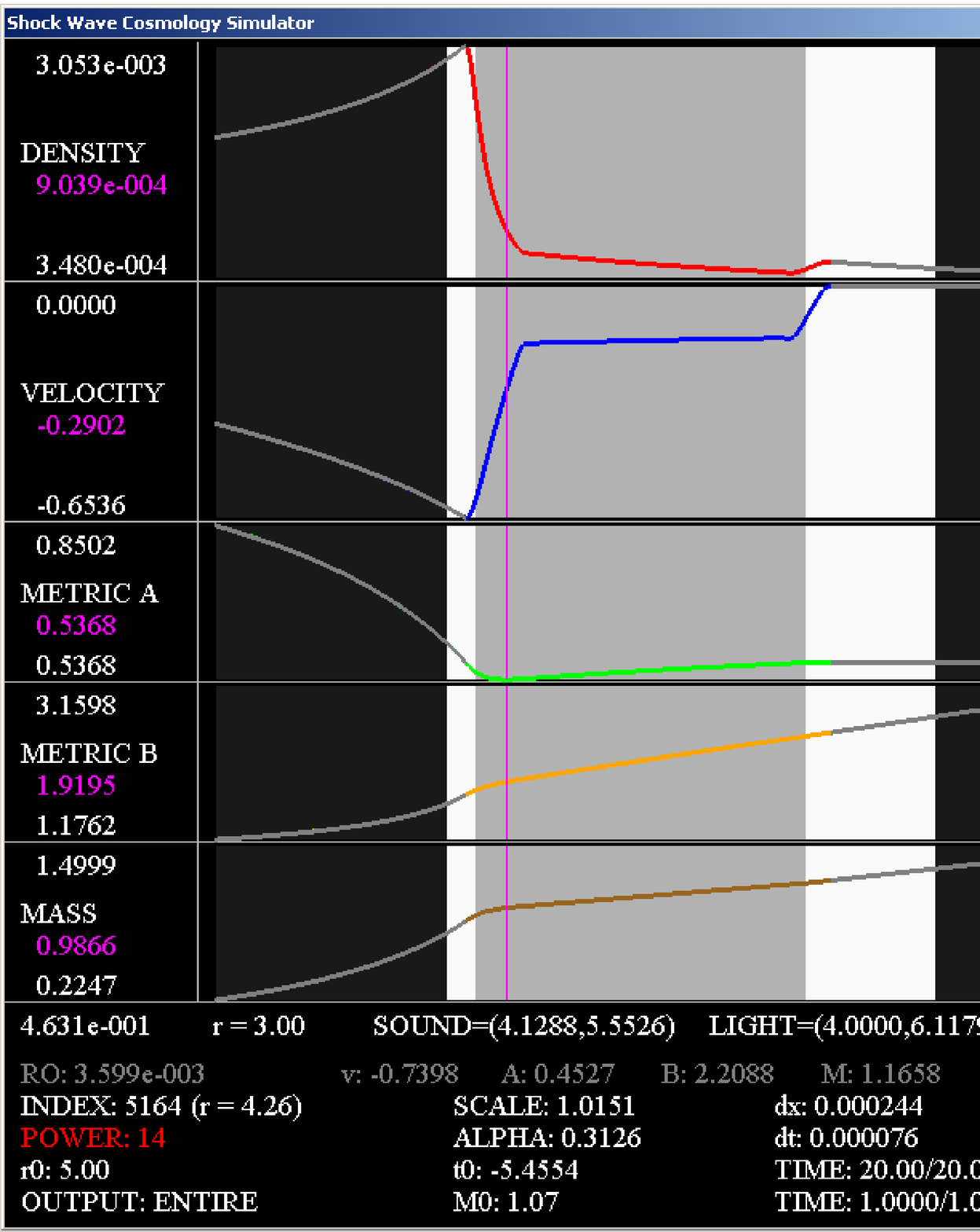}
\end{center}
\caption{Showing the model metrics against the simulated solution}
\label{fig:frw1_rev_tov_test}
\end{figure}

Again, we use the successive mesh refinement technique to test for the convergence of this solution, recorded in Table \ref{tab:frw1_rev_tov_converge}.  Since we are dealing with a continuous solution using a first order method, we expect convergence rates around 1 across all the variables.  Looking at Table \ref{tab:frw1_rev_tov_converge}, the only variable close to this rate is the metric component $A$, which has an average convergence rate of about 0.9.  The fluid variables have a rate close to 0.5, while the other metric component has a rate close to 0.7.  It is unclear why these rates are lower than the forward model (Table \ref{tab:frw1_tov_converge}), where the discontinuities of the shock waves should make these rates lower, but numerical convergence is obtained at an adequate rate.

\begin{table}
\begin{center}
\begin{tabular}{|c|c|c|c|c|}
\hline
 Gridpoints & Cone of Sound & Shock Wave & Error & Rate\\
\hline
64 & 4.1349 & 3.9524 & 0.18247 & N/A\\
\hline
128 & 4.1346 & 3.9764 & 0.15823 & 0.21\\
\hline
256 & 4.1341 & 4.0039 & 0.13016 & 0.28\\
\hline
512 & 4.1334 & 4.0254 & 0.10796 & 0.27\\
\hline
1024 & 4.1324 & 4.044 & 0.0884 & 0.29\\
\hline
2048 & 4.1314 & 4.0591 & 0.07229 & 0.29\\
\hline
4096 & 4.1304 & 4.0716 & 0.05888 & 0.3\\
\hline
8192 & 4.1295 & 4.0797 & 0.04982 & 0.24\\
\hline
16384 & 4.1288 & 4.0865 & 0.04226 & 0.24\\
\hline
\end{tabular}
\end{center}
\caption{Shock wave verses cone of sound results for FRW side}
\label{tab:frw1_rev_tov_shock_frw_side}
\end{table}

\begin{table}
\begin{center}
\begin{tabular}{|c|c|c|c|c|}
\hline
 Gridpoints & Cone of Sound & Shock Wave & Error & Rate\\
\hline
64 & 5.4398 & 6.1111 & 0.67131 & N/A\\
\hline
128 & 5.4816 & 5.9921 & 0.51056 & 0.39\\
\hline
256 & 5.5086 & 5.8863 & 0.37767 & 0.43\\
\hline
512 & 5.5258 & 5.818 & 0.29218 & 0.37\\
\hline
1024 & 5.5368 & 5.7644 & 0.22764 & 0.36\\
\hline
2048 & 5.5438 & 5.7259 & 0.18218 & 0.32\\
\hline
4096 & 5.5482 & 5.6989 & 0.15072 & 0.27\\
\hline
8192 & 5.5509 & 5.6795 & 0.12864 & 0.23\\
\hline
16384 & 5.5526 & 5.6659 & 0.1133 & 0.18\\
\hline
\end{tabular}
\end{center}
\caption{Shock wave verses cone of sound results for TOV side}
\label{tab:frw1_rev_tov_shock_tov_side}
\end{table}

We examine the preservation of the FRW and TOV metrics outside the region of interaction.  Looking back at Figure \ref{fig:frw1_rev_tov_end}, the ends of the rarefaction waves are outside the cone of sound.  Before, there was numerical diffusion of shock waves bleeding outside the cone of sound, but we do not expect to see a similar phenomenon with a rarefaction wave.  Since the speed at which matter travels is bound by the speed of sound, this fact is a cause for concern.  Like the forward model, we need to determine where the wave ends, and the border criteria developed in the Chapter \ref{ch:shock_wave_models} is used here.  On the FRW side, the border criteria is where the derivative of the fluid velocity changes sign, and instead of increasing then decreasing in the forward time model, the fluid velocity is decreasing until it hits the rarefaction wave and increases, satisfying this criteria.  On the TOV side, the border criteria is a significant change in the velocity from zero which is still satisfied since the fluid velocity has a sudden decrease at the rarefaction wave, as opposed to the sharp increase at the shock wave in the forward time model.  Looking at Figure \ref{fig:frw1_rev_tov_test}, the borders are indicated by where the gray curves stop and shown to be at the end of both rarefaction waves, as desired.  Next, we numerically compare the edges of both rarefaction waves against the cone of sound.  The results for the FRW side are recorded in Table \ref{tab:frw1_rev_tov_shock_frw_side}, while the TOV side is recorded in Table \ref{tab:frw1_rev_tov_shock_tov_side}.  The convergence rates are significantly lower from before, but numerical convergence is still obtained, albeit a little slower.  It is uncertain why both convergence rates are much slower, but these results give us confidence, as we mesh refine, that the outer edges to the rarefaction waves will converge to the edges of the cone of sound.

\begin{table}
\begin{center}
\begin{tabular}{|c|c|c|c|c|c|c|c|c|}
\hline
Number&\multicolumn{2}{|c|}{$\rho$}&\multicolumn{2}{|c|}{$v$}&\multicolumn{2}{|c|}{$A$}&\multicolumn{2}{|c|}{$B$}\\
\cline{2-9}
Gridpoints & Error & Rate & Error & Rate & Error & Rate & Error & Rate\\
\hline
64 & 3.460e-005 & N/A & 5.098e-003 & N/A & 3.780e-003 & N/A & 1.571e-002 & N/A\\
\hline
128 & 2.140e-005 & 0.69 & 2.993e-003 & 0.77 & 2.255e-003 & 0.75 & 8.818e-003 & 0.83\\
\hline
256 & 1.320e-005 & 0.7 & 1.749e-003 & 0.78 & 1.277e-003 & 0.82 & 4.659e-003 & 0.92\\
\hline
512 & 7.580e-006 & 0.8 & 9.778e-004 & 0.84 & 7.336e-004 & 0.8 & 2.733e-003 & 0.77\\
\hline
1024 & 4.270e-006 & 0.83 & 5.381e-004 & 0.86 & 4.261e-004 & 0.78 & 1.627e-003 & 0.75\\
\hline
2048 & 2.360e-006 & 0.85 & 2.921e-004 & 0.88 & 2.268e-004 & 0.91 & 8.586e-004 & 0.92\\
\hline
4096 & 1.320e-006 & 0.84 & 1.591e-004 & 0.88 & 1.171e-004 & 0.95 & 4.378e-004 & 0.97\\
\hline
8192 & 6.830e-007 & 0.95 & 8.220e-005 & 0.95 & 5.800e-005 & 1 & 2.086e-004 & 1.1\\
\hline
16384 & 3.620e-007 & 0.92 & 4.300e-005 & 0.93 & 2.960e-005 & 0.97 & 1.051e-004 & 0.99\\
\hline
\end{tabular}
\end{center}
\caption{Convergence results for the FRW side}
\label{tab:frw1_rev_tov_converge_frw_side}
\end{table}

\begin{table}
\begin{center}
\begin{tabular}{|c|c|c|c|c|c|c|c|c|}
\hline
Number&\multicolumn{2}{|c|}{$\rho$}&\multicolumn{2}{|c|}{$v$}&\multicolumn{2}{|c|}{$A$}&\multicolumn{2}{|c|}{$B$}\\
\cline{2-9}
Gridpoints & Error & Rate & Error & Rate & Error & Rate & Error & Rate\\
\hline
64 & 4.030e-007 & N/A & 5.032e-004 & N/A & 2.933e-003 & N/A & 1.128e-002 & N/A\\
\hline
128 & 2.400e-007 & 0.74 & 3.083e-004 & 0.71 & 1.732e-003 & 0.76 & 5.863e-003 & 0.94\\
\hline
256 & 1.390e-007 & 0.8 & 1.776e-004 & 0.8 & 1.060e-003 & 0.71 & 2.978e-003 & 0.98\\
\hline
512 & 7.640e-008 & 0.86 & 9.740e-005 & 0.87 & 5.264e-004 & 1 & 1.567e-003 & 0.93\\
\hline
1024 & 4.120e-008 & 0.89 & 5.230e-005 & 0.9 & 2.379e-004 & 1.1 & 8.415e-004 & 0.9\\
\hline
2048 & 2.160e-008 & 0.93 & 2.650e-005 & 0.98 & 1.250e-004 & 0.93 & 4.260e-004 & 0.98\\
\hline
4096 & 9.410e-009 & 1.2 & 9.020e-006 & 1.6 & 6.750e-005 & 0.89 & 2.113e-004 & 1\\
\hline
8192 & 4.680e-009 & 1 & 2.760e-006 & 1.7 & 3.930e-005 & 0.78 & 1.007e-004 & 1.1\\
\hline
16384 & 2.480e-009 & 0.92 & 9.790e-007 & 1.5 & 2.050e-005 & 0.94 & 4.980e-005 & 1\\
\hline
\end{tabular}
\end{center}
\caption{Convergence results for the TOV side}
\label{tab:frw1_rev_tov_converge_tov_side}
\end{table}

Using these border criteria, we test the numerical convergence of the FRW and TOV metrics outside the region of interaction.  Recall, these borders act as a marker where we stop computing the error between the simulated and model metrics, which is discussed more thoroughly in Chapter \ref{ch:shock_wave_models}.  Using this procedure, the convergence results for the FRW and TOV region are recorded in Tables \ref{tab:frw1_rev_tov_converge_frw_side} and \ref{tab:frw1_rev_tov_converge_tov_side}, respectively.  Looking at both tables, both regions are converging to their respective models at the appropriate rate of 1 for a first order method.  Again, we note how remarkable it is the metric components in the TOV region remain preserved after integration through the FRW metric and the region of interaction.

\section{Black Hole Formation}
\label{sec:black_hole_form}

\begin{figure}
\begin{center}
\includegraphics[width=\textwidth]{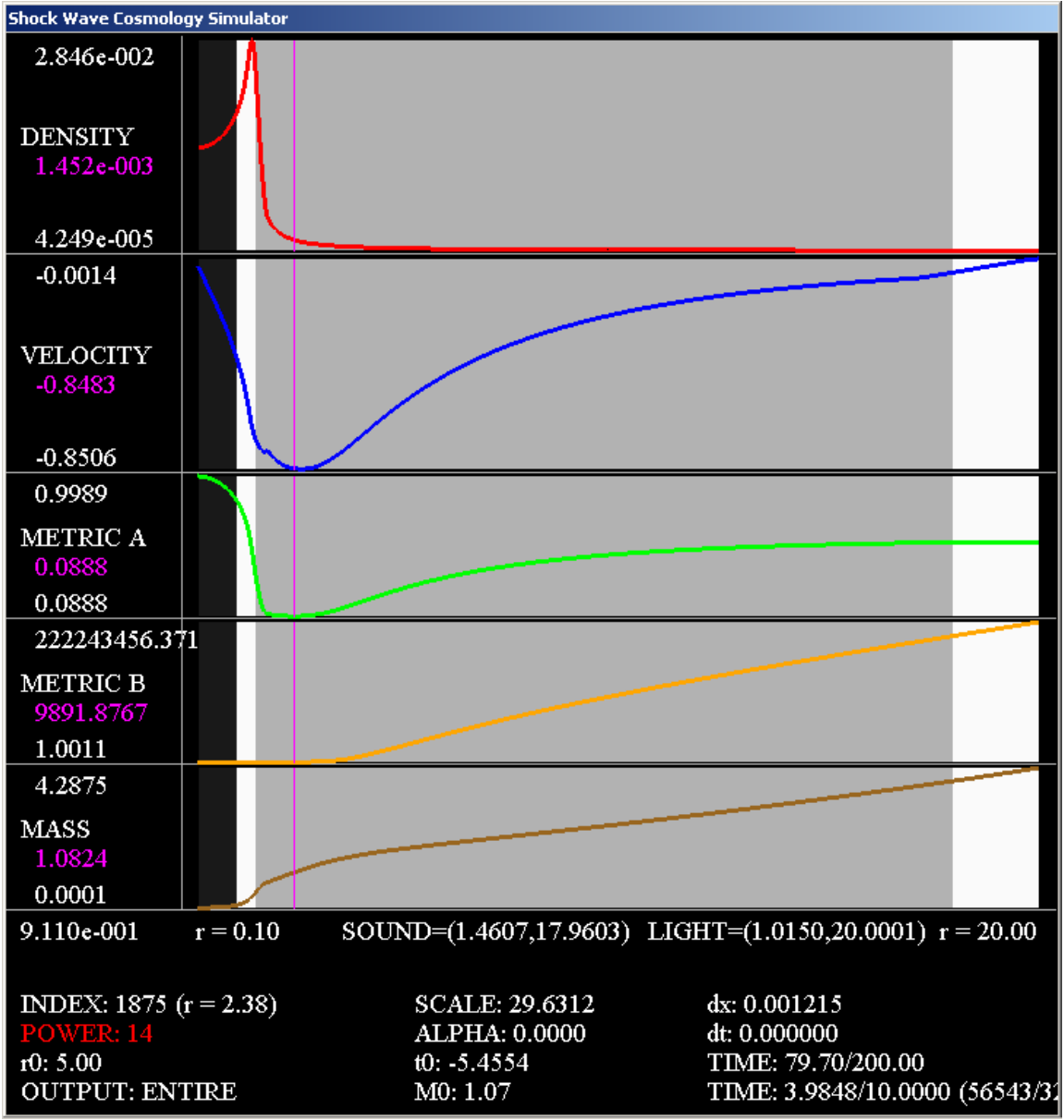}
\end{center}
\caption{Approaching the black hole with the interaction region hitting the TOV boundary}
\label{fig:black_hole1}
\end{figure}

\begin{figure}
\begin{center}
\includegraphics[width=\textwidth]{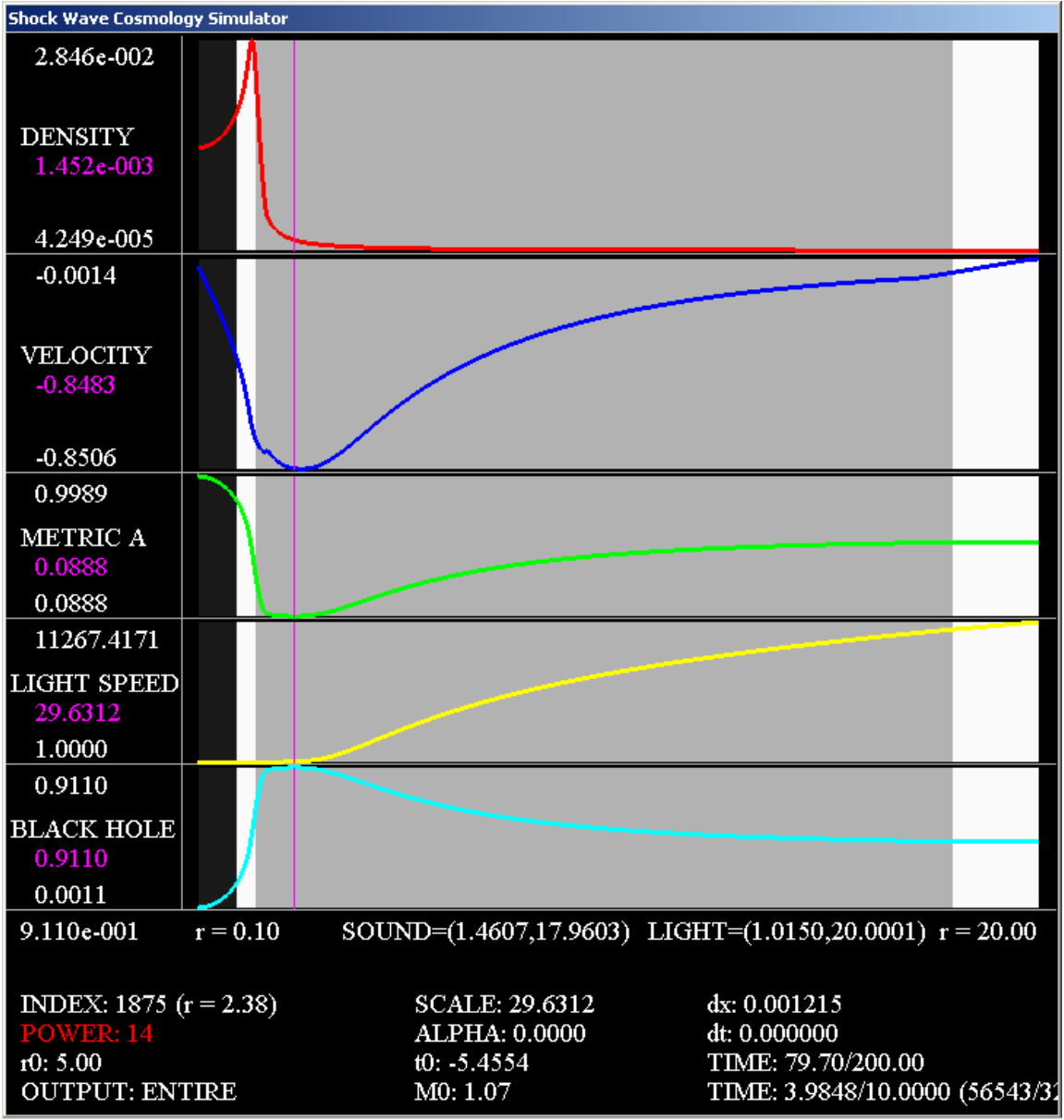}
\end{center}
\caption{Looking at the black hole criteria and the speed of light}
\label{fig:black_hole2}
\end{figure}

In the previous section, as time progresses in the reversed time model more mass falls in towards the center of the universe.  If this trend continues, then it is inevitable, given enough time, a black hole will form.  In order to run the simulation long enough for a black hole to form, our initial parameters must be changed.  Using the previous parameters (\ref{std_params}), the region of interaction hits the boundaries before the black hole forms, and once this happens, the solution is no longer valid because our boundary data is inaccurate.  In order to solve this problem, we expand the region of space under consideration by setting the minimum radius $\bar{r}_{min}=0.1$ and the maximum radius $\bar{r}_{max}=20$.  We leave the other parameters the same, namely, the discontinuity position $\bar{r}_0=5$ and the number of gridpoints $n=16,384$.

\begin{figure}
\begin{center}
\includegraphics[width=\textwidth]{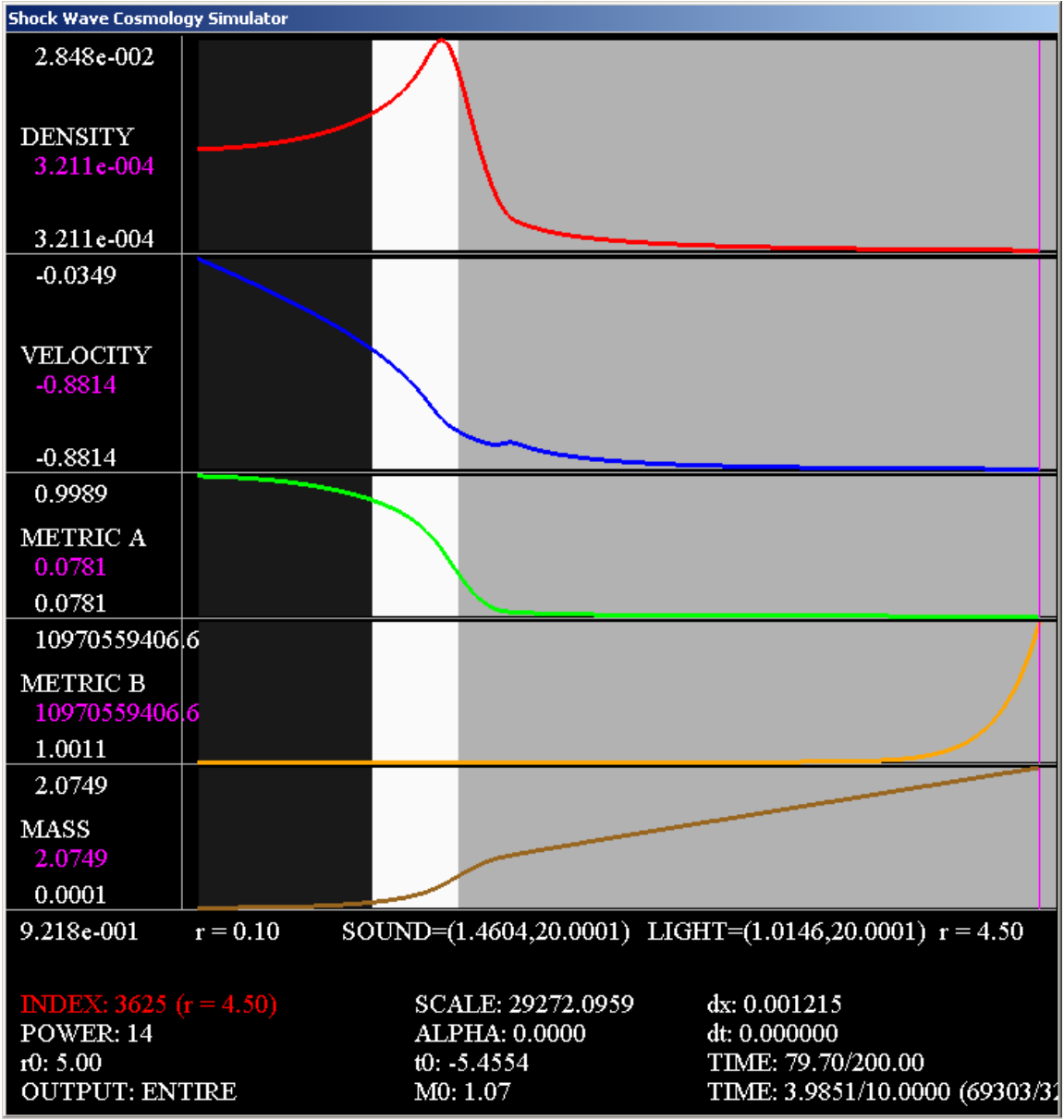}
\end{center}
\caption{Continuing black hole formation until it hits the boundary data}
\label{fig:black_hole3}
\end{figure}

Running this simulation, the interaction region hits the TOV boundary before the black hole forms, with the end result shown in Figure \ref{fig:black_hole1}.  We cannot proceed further because the TOV boundary data is no longer valid.  At this stage of the solution, the highest value of the black hole number is $\mu=0.9110$, located where the vertical magenta line is placed at a radius of $\bar{r}=2.38$.  This radius corresponds to the dip in the velocity and metric component $A$.  It also occurs right after the density spike which causes the mass to increase dramatically in this region of the universe.  Comparing to an earlier stage of the solution in Figure \ref{fig:frw1_rev_tov_test}, at the edge of the incoming rarefaction wave, the density has continued to increase while velocity has become more negative, so matter has been coming in towards the center at a faster rate.  Another thing to notice is the metric $B$ component has grown extremely large at the TOV edge of our simulation.  This large metric $B$ component causes a huge time dilation between both edges of our simulation, which can be seen more clearly in Figure \ref{fig:black_hole2} where we graph the (coordinate) speed of light $\sqrt{AB}$ in yellow and the black hole number $\mu$ in cyan.  The black hole criteria is the most satisfied in a region of space around the point $\bar{r}=2.38$ where $\mu=0.9110$.  Now more extreme relativistic effects come into play because the significant time dilation, measured by the speed of light, between the two edges is around 11,267 to 1, which means light travels 11,267 times faster at one edge of the simulated region to the other.  This is a problem because it causes our future time steps to be smaller, so the evolution of the system with greatest $\mu$ is slower than the development close to the TOV border.  As a final note, there is a blip in the velocity profile as seen in Figure \ref{fig:black_hole1}.  Since there is much data to support the accuracy of our numerical simulation, we have no reason to believe this blip is due to numerical error.  We interpret this blip in the velocity as an unknown phenomenon, and it is a future research project to determine the nature of it.

\begin{figure}
\begin{center}
\includegraphics[width=\textwidth]{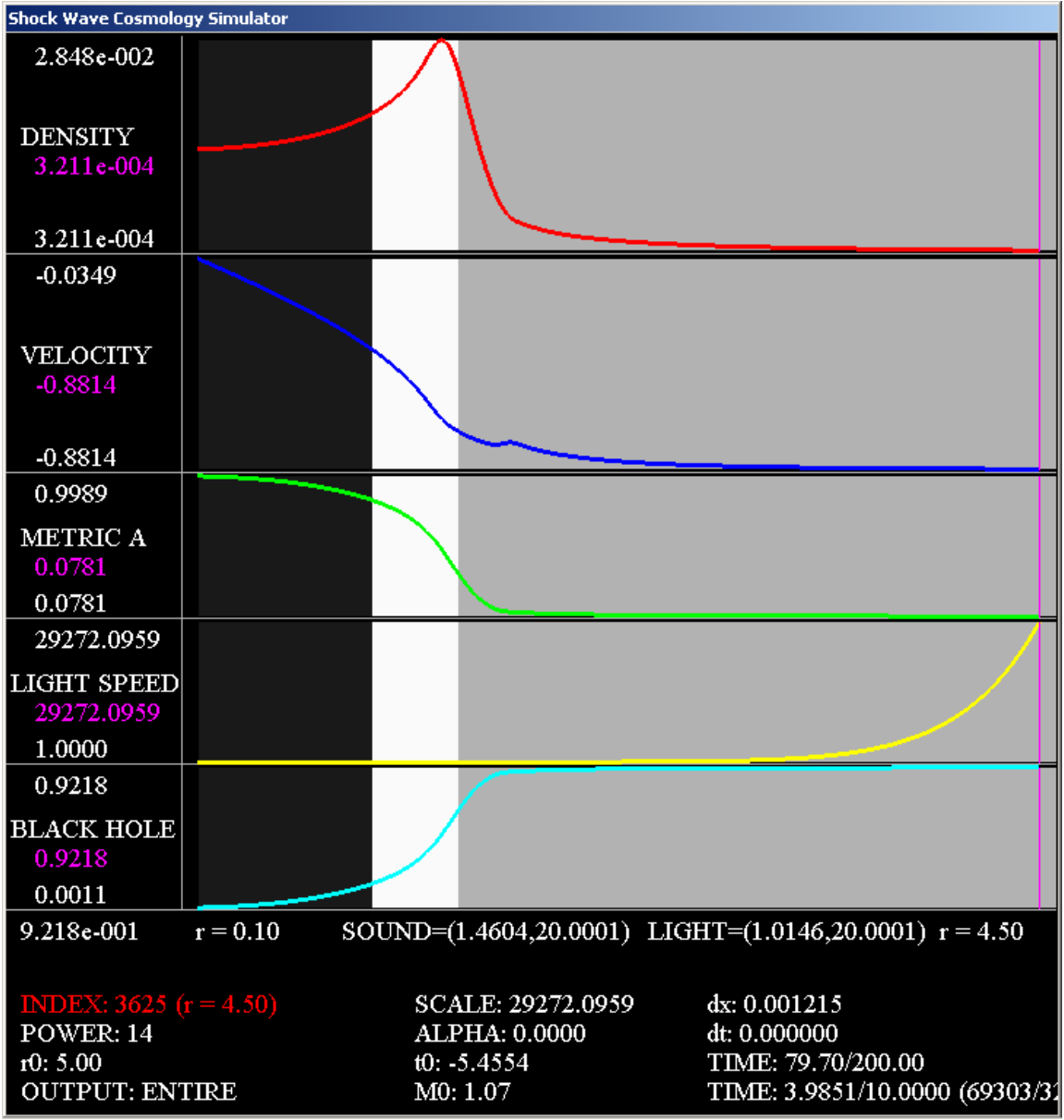}
\end{center}
\caption{Black hole criteria and the speed of light after the boundary data hits the black hole formation}
\label{fig:black_hole4}
\end{figure}

In order to get closer to the formation of the black hole, we could try to expand the region of simulated space again, but the huge time dilation causes the region of interaction to grow extremely fast, quickly reaching the new border while only a small amount of time passes near the region where the black hole appears to be forming (i.e. $\mu=0.9110$ at $\bar{r}=2.38$).  Regardless of how far we put the maximum radius, we never get closer than $\mu=0.9110$ in the black hole criteria (\ref{black_hole_criteria}).  Our alternative strategy is to realize we are only interested in the region of space where the black hole is forming, where $\mu$ is greatest.  With this in mind, to continue the evolution of the black hole, we start chopping off the right side of our region of simulated space; even though the right most grid cell is no longer valid, the one next to it (on the left side) is still a valid solution.  This process is like zooming into the region of the black hole formation.  More precisely, the grid point $x_{n+1}$ is used as the boundary data throughout the entire simulation.  After the TOV edge has been reached, we use the gridpoint $x_n$ as the new boundary data for the next time step, discarding the data associated with the gridpoint $x_{n+1}$.  During the subsequent time steps, we discard the current boundary data at the gridpoint $x_i$ and replace it with its predecessor $x_{i-1}$.  Using this procedure and stopping when the radius of the right boundary matches the radius with the greatest value of the black hole number, we produce Figure \ref{fig:black_hole3} and Figure \ref{fig:black_hole4}.  The black hole number is $\mu=0.9218$ at a radius of $\bar{r}=4.5$, so we are a little closer to black hole formation.  More interestingly, this radius for the greatest $\mu$ is moving farther out, from $\bar{r}=2.38$ to $\bar{r}=4.5$, meaning the region of black hole formation is expanding which is shown in Figure \ref{fig:black_hole4}.  In Figure \ref{fig:black_hole3}, one notices the metric component $B$, along with the speed of light in Figure \ref{fig:black_hole4}, is growing exponentially at this radius, causing the time dilation to be more extreme.  Here, the time dilation from the center to this radius is 29,517 to 1, two and a half times greater than before.

\section{Summary}
\label{sec:time_rev_summary}
We began the chapter by arguing for the validity of the existence of a reversed time solution, and then setup the initial profiles and boundary conditions needed to start the simulation.  It turns out the only difference from the forward model is that positive time is replaced by negative time, and this difference only affects a change in sign of the velocity, as shown in Figure \ref{fig:frw1_rev_tov_init}.  Snapshots of the solution are recorded in Figure \ref{fig:frw1_rev_tov_frames}, showing the reverse time solution containing rarefaction waves surrounding a region of under density.  Unlike the forward time model, the derivative of the metric is continuous in Figure \ref{fig:frw1_rev_tov_end}, so this model is a strong solution to the Einstein equations.    We mimic the convergence results obtained in the forward time model.  In particular, in Table \ref{tab:frw1_rev_tov_converge}, we record numerical convergence of the entire solution taken as a whole, including the interaction and non-interaction regions, using successive mesh refinement.  In Tables \ref{tab:frw1_rev_tov_shock_frw_side} and \ref{tab:frw1_rev_tov_shock_tov_side}, we also record numerical convergence of the interaction region to the cone of sound, using border tracking of the edges to the simulated rarefaction waves.  In Tables \ref{tab:frw1_rev_tov_converge_frw_side} and \ref{tab:frw1_rev_tov_converge_tov_side}, we record the numerical convergence of the non-interaction regions, using the true solutions.  Like the forward time model, these results affirm the region of interaction is precisely the region of the cone of sound (no formal mathematical proof of this is known).

As time progresses in this solution, more and more mass is falling in towards the center of the universe.  This fact leads us to believe that if the solution were to progress in time beyond the simulation time limit, the black hole number (\ref{black_hole_number}) would continue to increase and eventually the black hole criteria (\ref{black_hole_criteria}) would be achieved exactly, resulting in black hole formation.  Our method is to expand the region of simulation by running the simulation until the TOV region is completely gone, giving us a value of $\mu=0.9110$.  We then continue to run the simulation by discarding the incorrect boundary data and zooming into the region of black hole formation, and by this refined method we obtain convergence all the way up to the black hole criteria value $\mu=0.9218$.  These results taken together give us confidence in our solution as a whole and provide strong support for our claim that this solution has black hole formation contained within it.  As a final note, with the successful results obtained here and Chapter \ref{ch:shock_wave_models}, we have simulated general relativistic analogs of both the 1 and 2 family of shock and rarefaction waves, the elementary waves of conservation laws.  This fact gives us confidence in the ability and accuracy of our locally inertial Godunov method to simulate any solution to the Einstein equations for a perfect fluid in standard Schwarzschild coordinates.

   \chapter[%
      Short Title of 9th Ch.
   ]{%
      Putting Units into the Simulation
   }%
   \label{ch:units}
We end our discussion of shock wave formation by putting back the units to give more physical meaning to our simulations.  To place units on our numerical values requires us to understand the consequence of setting Newton's gravitational constant $\mathcal{G}$ and Einstein's speed of light constant $c$ both to one.  The effect of this is the unification of the units of time, length, and mass into one set of units, which we choose as the units of mass in our discussion here.  Understanding this unification allows us to recover the other units, time and length.  To end this chapter, we consider our simulation on the solar and galactic scale by transforming our numerical values to familiar units, giving us a firm grasp of the scale and magnitude involved within our simulations.

To determine the effect of setting the speed of light constant $c$ to one, we start by considering the TOV metric
\begin{equation}\label{ch9_tov_in_ssc}
ds^2=-B(\bar{r})d\bar{t}^2+\left(\frac{1}{1-\frac{2\mathcal{G}M(\bar{r})}{\bar{r}}}\right)d\bar{r}^2,
\end{equation}
where we dropped the $\bar{r}^2d\Omega^2$ to simplify the notation.  Since it represents a measurement,  the metric must have units of length or time attached to it; these units can be either associated with the coordinates $(\bar{t},\bar{r})$ or the metric components themselves.  Here, we choose the units to be associated to the coordinates $(\bar{t},\bar{r})$, and we use units of length (i.e. kilometers) to be the measurement of the metric.  This implies the time coordinate $\bar{t}$ is measured in units of length, and the speed of light constant converts units of time (i.e. seconds) into units of length.  More precisely, let $\bar{t}_*$ be the time measured in units of time, so the time coordinate becomes
\begin{equation}
\bar{t}=c\bar{t}_*,
\end{equation}
where the units are represented by units of length scale.  One can view this as the time coordinate absorbs the speed of light constant and transforms it into units of length. This absorbtion is equivalent to setting $c=1$, as done in our simulations.  In particular, the following holds
\begin{equation}
\frac{d\bar{r}}{d\bar{t}_*}=c \Longleftrightarrow \frac{d\bar{r}}{d(c\bar{t}_*)}=\frac{d\bar{r}}{d\bar{t}}=1,
\end{equation}
so the coordinate speed of light is scaled down to one in $(\bar{t},\bar{r})$ coordinates.  To recover the units of time from the time coordinate $\bar{t}$ requires a simple calculation
\begin{equation}
\bar{t}_*=\frac{\bar{t}}{c}.
\end{equation}
Thus, the speed of light constant enables us to make equivalent the units of time and length, a truly remarkable insight of Einstein's theory.

To determine the effect of setting Newton's gravitational constant $\mathcal{G}$ to one, we perform dimensional analysis on (\ref{ch9_tov_in_ssc}) to determine the dimensions of $\mathcal{G}$.  To this end, let $L$ represent units of length, $T$ represent units of time, and $M$ represent units of mass.  As stated before, the metric components of (\ref{ch9_tov_in_ssc}) are dimensionless; therefore, the dimensional analysis of the space metric component gives us
\begin{equation}
1=\left[\frac{2\mathcal{G}M(\bar{r})}{\bar{r}}\right]=\frac{M\left[\mathcal{G}\right]}{L},
\end{equation}
where $[\cdot]$ is the units of the variable.  Rearranging this equation, we have
\begin{equation}\label{g_units}
\left[\mathcal{G}\right]=\frac{L}{M}.
\end{equation}
In the same way that $c=1$ makes the units of time equivalent to the units of length, setting $\mathcal{G}=1$ equates units of length to units of mass.  More specifically, let $\bar{r}_*$ be the radius measured in units of length, so the space coordinate in units of mass becomes
\begin{equation}
\bar{r}=\frac{\bar{r}_*}{\mathcal{G}},
\end{equation}
where the space coordinate absorbs Newton's gravitational constant.  This absorbtion is equivalent to setting $\mathcal{G}=1$ in (\ref{ch9_tov_in_ssc}), which is done in our simulations.  For $\bar{r}$ in units of mass, we recover $\bar{r}_*$ in units of length by simply computing
\begin{equation}\label{mass_to_length}
\bar{r}_*=\mathcal{G}\bar{r}.
\end{equation}
To take this one step further, we also make $\bar{t}$ in units of mass related to $\bar{t}_*$ in the units of length  by
\begin{equation}\label{mass_to_time}
\bar{t}_*=\frac{\mathcal{G}}{c}\bar{t}.
\end{equation}
Therefore, by setting $c=1$ and $\mathcal{G}=1$, we unify all the units (time, length, and mass) into one set of units.  For our discussion, we choose the mass to be this unified set of units, where the units of length and time can be recovered with equations (\ref{mass_to_length}) and (\ref{mass_to_time}), respectively.

We also consider the units for the fluid variables $(\rho, v)$.  The velocity $v$ is measured as a percentage to the speed of light and is thus unitless.  With units of length equivalent to units of mass, the units for the density $\rho$ in our simulation are
\begin{equation}
\left[\rho\right]=\frac{M}{L^3}=\frac{1}{M^2}.
\end{equation}
To recover the density in units of $M/L^3$, we compute
\begin{equation}
\rho_*=\frac{1}{\mathcal{G}^3}\rho.
\end{equation}

To determine the value of Newton's gravitational constant, we perform dimensional analysis to recover the factors of $c$ that have been suppressed by setting $c=1$.  Looking up the units, the Newton's true gravitational constant, denoted as $\hat{\mathcal{G}}$, has units of
\begin{equation}\label{g_hat_units}
\hat{\mathcal{G}}=\frac{L^3}{MT^2}.
\end{equation}
Comparing (\ref{g_units}) to (\ref{g_hat_units}), the relationship between them must be
\begin{equation}
\mathcal{G}=\frac{\hat{\mathcal{G}}}{c^2}.
\end{equation}

To obtain numerical values for $\mathcal{G}$ and $c$, the units for the variables $\bar{t}_*$, $\bar{r}_*$, and $M$ must be chosen.  We choose to measure time in seconds (sec), the length scale in kilometers (km), and mass in solar masses $M_\odot$.  This choice makes Newton's gravitational constant
\begin{equation}\label{g_constant}
\mathcal{G}=1.47664 \text{ km}/M_\odot,
\end{equation}
and the speed of light constant
\begin{equation}\label{c_constant}
c\approx 3\times10^5 \text{ km/sec}.
\end{equation}

With the ability to recover the proper units, we interpret some of the numerical values of our simulated FRW/TOV model in Figure \ref{fig:frw1_tov_test} on a solar scale.  To interpret the results on a solar scale, we choose the unit of mass to be represented by a solar mass ($1M = 1M_\odot$), where $M_\odot=1.989\times10^{30}\text{ g}$ is the mass of our sun.  Using this scale, the range of values for the variables are
\begin{eqnarray}
0.09 \text{ }M_\odot &<&M < 1.5\text{ }M_\odot,\label{mass_value}\\
3 \text{ }M_\odot &<&\bar{r} < 7\text{ }M_\odot,\\
5.46 \text{ }M_\odot &<&\bar{t} < 6.46\text{ }M_\odot,\\
3.48\times10^{-4} \text{ }1/M^2_\odot &<&\rho < 1.92\times10^{-3}\text{ }1/M^2_\odot,\\
0 &<&v < 0.53.\label{velocity_value}
\end{eqnarray}
Using the constants (\ref{g_constant}) and (\ref{c_constant}) and equations (\ref{mass_to_length}) and (\ref{mass_to_time}), we convert the mass to the familiar units of
\begin{eqnarray}
0.09 \text{ }M_\odot &<&M < 1.5\text{ }M_\odot,\label{solar_mass_value}\\
4.43 \text{ km} &<&\bar{r}_* < 10.37\text{ km},\\
2.69\times10^{-5} \text{ sec} &<&\bar{t}_* < 3.18\times10^{-5}\text{ sec},\\
1.08\times10^{-4} M_\odot/\text{km}^3 &<&\rho_* < 5.95\times10^{-4}M_\odot/\text{km}^3,\\
0 \text{ km/sec} &<&v < 1.59\times10^5 \text{ km/sec}.\label{solar_velocity_value}
\end{eqnarray}
Hence, our simulation is dealing with mass equivalent to 1.4 times our sun across a distance of 5.94 kilometers for a time frame of 4.9 microseconds.  To give more meaning to the density, the average density of the sun is $7.04\times10^{-19} M_\odot/\text{km}^3$, so our simulation is dealing with densities of at least $1.53\times10^{14}$ times greater than the average density of the sun.  In this coordinate system, matter is moving at speeds up to $1.59\times10^5 \text{ km/sec}$.

To consider galactic scales instead, one unit of mass is set equal to the mass of the Milky Way, $1.8\times10^{11}\text{ } M_\odot$.  The effect of this equivalency is a direct scale by the mass of the Milky Way of the solar scale results (\ref{solar_mass_value})-(\ref{solar_velocity_value}), except for the velocity.  In this setting, the simulated values (\ref{mass_value})-(\ref{velocity_value}) become
\begin{eqnarray}
1.62\times10^{10} \text{ }M_\odot &<&M < 2.7\times10^{11}\text{ }M_\odot,\\
0.084 \text{ light-years} &<&\bar{r}_* < 0.2\text{ light-years},\\
56 \text{ days} &<&\bar{t}_* < 66\text{ days},\\
1.94\times10^{7} M_\odot/\text{km}^3 &<&\rho_* < 1.07\times10^{8}M_\odot/\text{km}^3,\\
0 \text{ km/sec} &<&v < 1.59\times10^5 \text{ km/sec}.
\end{eqnarray}
Therefore, on the galactic scale, our simulation is dealing with mass equivalent to 1.4 times our Milky Way across a distance of 0.12 light-years in a time frame of 10 days.
   \appendix

   \chapter[%
      Short Title of Appendix A
   ]{%
      Programming Code for the Simulation
   }%
   \label{ch:sim_code}
This appendix provides an overview of all the source code used to perform the simulations throughout this thesis.  A copy of this source code is provided on the attached CD, which is approximately 8,000 lines contained in 22 files.  This code is written in C++, with a Microsoft C++ 5.0 compiler, using the OpenGL graphics library on the Windows operating system.  Our goal here is to give the reader a sense of the organization of the code to enable the reader to find specific areas of interest within the code.

To start, the main file (\verb"main.cpp") initializes the displayed window and the simulation itself, and contains the code specific to the Windows OS.  This file also handles the interface between the operating system and the simulation.  The rest of the code is segregated into different classes in accordance to an object-oriented programming paradigm.  Each of these classes has two files associated with them, a header \verb"*.h" file and an implementation file \verb"*.cpp".  The header file contains the definition of the class and its associated variables and functions, and the implementation file contains the code for these functions.  To reduce the amount of redundant code, we created a class hierarchy for the simulators, displayed in Figure \ref{fig:code_graph}.  Classes on top of the hierarchy are base classes and the connected nodes represent children classes that inherit there functionality from the base class.  We provide a brief explanation of the functionality of the classes in hierarchy:
\begin{itemize}
\item The CSimulator class is a base class to provide a template in which all simulation classes are derived.  It unifies the basic structure and interface of all of our simulator classes.

\item The CRsolver class is a base class to provide a template for a Riemann solver to the compressible Euler equations, both the relativistic and classical.  The algorithms for finding the middle state and solving the Riemann problem, as discussed in Subsections \ref{subsec:middle_state} and \ref{subsec:solving_rp}, respectively, are contained in this class.  Also, the code to display the Riemann problem simulator in Figure \ref{fig:rp_solver} is in this class.

\item The CCosmoSim class contains the code to simulate and display the cosmology models of Chapters \ref{ch:cont_models}-\ref{ch:time_rev_model}.  This class contains the algorithm for the locally inertial Godunov method with dynamical time dilation, defined in Chapter \ref{ch:frac_god_method}.  Since the code became quite large, we split the implementation file into two, the numerics in \verb"ccosmosim.cpp" and the graphics in \verb"ccosmosim_graphics.cpp".

\item The CRelEulerSolver class is the Riemann solver for the relative compressible Euler equations.  This class is derived from the CRsolver class, and it contains the specific equations tied to the relative compressible Euler equations, scribed in Section \ref{sec:ivp_special_relativity}.

\item The CClasEulerSolver class is the Riemann solver for the classical compressible Euler equations \cite{smol}.  This class was used to test our Riemann solver on a more simple set of equations to find errors in the algorithms contained in the CRsovler class.  This class has no affect on the simulations in this paper and is left over code from the testing stage of our simulations.
\end{itemize}

\begin{figure}[t]
\begin{pspicture}(14,6)(2,0)
%\psgrid
\rput(9,5){\ovalnode{A}{CSimulator}}
\rput(4.5,3){\ovalnode{B}{CRsolver}}
\rput(13.5,3){\ovalnode{C}{CCosmoSim}}
\rput(2,1){\ovalnode{D}{CClasEulerSolver}}
\rput(7,1){\ovalnode{E}{CRelEulerSolver}}
\ncline{A}{B}
\ncline{A}{C}
\ncline{B}{D}
\ncline{B}{E}

\end{pspicture}\caption{The object oriented structure of the classes in the simulation code}
\end{figure}
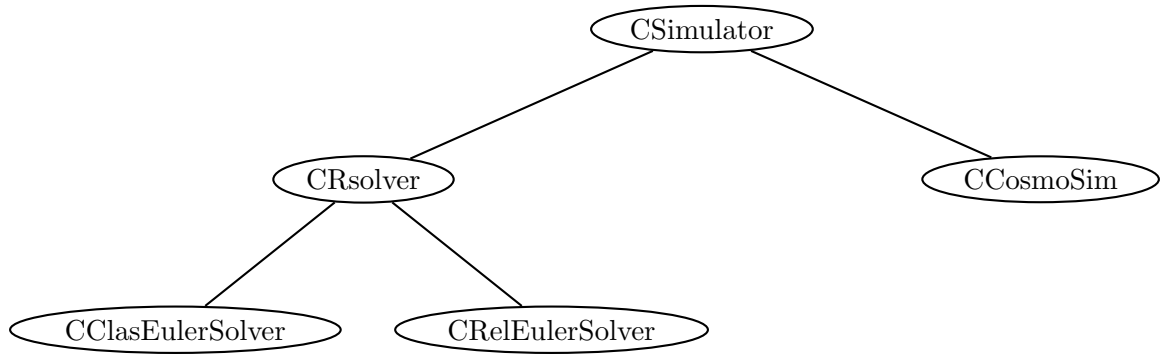\label{fig:code_graph}

To handle the rendering of our simulations, we use a number of stand alone classes to serve specific graphical needs.  We present a list of these classes along with a brief explanation of their utility:
\begin{itemize}
\item The Color class is used to consolidate the 3 color components into one vector.  This class allows for easy manipulation and use of colors in OpenGL.

\item The Font class gives us a method to display all the text used in our simulations.

\item The Interval class is a concise method of keeping track of intervals on the real line.

\item The ColorArray class enables us to keep track of an array of colors.

\item The ColorMap class uses the Interval and ColorArray classes to build a mapping of an interval on the real line to a set of discrete colors.  In particular, this class is used to display the density in an array of different colors in Figure \ref{fig:rp_solver}.
\end{itemize}

   \backmatter

   \bibliographystyle{amsalpha-fi-arxlast}
   \bibliography{DissertationBibliography}
\end{document}